\documentclass[twoside,11pt]{article}
\usepackage{jair, theapa, rawfonts}

\jairheading{70}{2021}{1-15}{9/2021}{9/2022}
\ShortHeadings{FO-Rewritability of Two-Dimensional Temporal Ontology-Mediated Queries}{Artale, Kontchakov, Kovtunova, Ryzhikov, Wolter \& Zakharyaschev}
\firstpageno{1}

\usepackage{enumitem}
\usepackage{latexsym}
\usepackage{amsmath,amsthm,amssymb}
\usepackage[only bigsqcap]{stmaryrd} %% for \bigsqcap
\usepackage{pifont} % for \ding
\usepackage{booktabs}
\usepackage{multirow}
\usepackage{thm-restate}

\usepackage{tikz}
\usetikzlibrary{shapes,arrows,positioning,calc,patterns,arrows.meta}

\tikzset{>=latex, 
    point/.style={circle,semithick,fill=white,draw=black,minimum size=1.3mm,inner sep=0pt},
    spoint/.style={circle,semithick,fill=gray,draw=black,minimum size=1.3mm,inner sep=0pt},
    qpoint/.style={circle,semithick,fill=white,draw=black,minimum size=2mm,inner sep=0pt},
    ppoint/.style={circle,semithick,fill=gray,draw=black,minimum size=2mm,inner sep=0pt},
    wpoint/.style={circle,thick,fill=white,draw=black,minimum size=4.5mm,inner sep=1pt},
    graypoint/.style={circle,very thin,fill=white,draw=gray,minimum size=1.1mm,inner sep=0pt},
    tpoint/.style={circle,semithick,fill=white,draw=black,minimum size=1.2mm,inner sep=0pt},
    object-timeline/.style={ultra thick, gray!30},
    time-guideline/.style={ultra thin},
    time-guideline2/.style={dash dot,gray}
}

\newtheorem{theorem}{Theorem}
\newtheorem{lemma}[theorem]{Lemma}

\newtheorem{proposition}[theorem]{Proposition}
\newtheorem{example}[theorem]{Example}

\newtheorem{definition}[theorem]{Definition}

\newcommand{\LTL}{\textsl{LTL}}
\newcommand{\DL}{\textsl{DL-Lite}}
\newcommand{\EL}{\mathcal{EL}}
\newcommand{\OWLQL}{\textsl{OWL\,2\,QL}}
\newcommand{\OWLEL}{\textsl{OWL\,2\,EL}}

\newcommand{\Xallop}{^{\smash{\Box\raisebox{1pt}{$\scriptscriptstyle\bigcirc$}}}}
\newcommand{\Xnext}{^{\smash{\raisebox{1pt}{$\scriptscriptstyle\bigcirc$}}}}
\newcommand{\Xbox}{^{\smash{\Box}}}

\newcommand{\LogSpace}{\textsc{L}}
\newcommand{\NP}{\textsc{NP}}
\newcommand{\NL}{\textsc{NL}}
\newcommand{\coNP}{\textsc{coNP}}
\newcommand{\PSpace}{\textsc{PSpace}}
\newcommand{\ExpSpace}{\textsc{ExpSpace}}
\newcommand{\NCo}{{{\ensuremath{\textsc{NC}^1}}}}
\newcommand{\ACz}{{\ensuremath{\textsc{AC}^0}}}
\newcommand{\LogTime}{\textsc{LogTime}}
\newcommand{\ExpTime}{\textsc{ExpTime}}
\newcommand{\NExpTime}{\textsc{NExpTime}}
\newcommand{\PTime}{\textsc{P}}

\newcommand{\FOLTLI}{\textsl{FOLTL$_1$}}
\newcommand{\FO}{\textup{FO}}
\newcommand{\RPR}{\textup{RPR}}
\newcommand{\FOE}{\FO(<,\equiv)}

\newcommand{\bool}{\textit{bool}}
\newcommand{\gbool}{\textit{g-bool}}
\newcommand{\horn}{\textit{horn}}
\newcommand{\rhorn}{\textit{horn}^+}
\newcommand{\krom}{\textit{krom}}
\newcommand{\core}{\textit{core}}

\newcommand{\FF}{{\scriptscriptstyle F}}
\newcommand{\PP}{{\scriptscriptstyle P}}

\newcommand{\nm}[1]{\textsf{#1}}

\newcommand{\nxt}{{\ensuremath\raisebox{0.25ex}{\text{\scriptsize$\bigcirc$}}}}
\newcommand{\Rnext}{\nxt_{\!\scriptscriptstyle F}}
\newcommand{\Lnext}{\nxt_{\!\scriptscriptstyle P}}

\newcommand{\Rdiamond}{\Diamond_{\!\scriptscriptstyle F}}
\newcommand{\Ldiamond}{\Diamond_{\!\scriptscriptstyle P}}

\newcommand{\Rbox}{\rule{0pt}{1.4ex}\Box_{\!\scriptscriptstyle F}}
\newcommand{\Lbox}{\rule{0pt}{1.4ex}\Box_{\!\scriptscriptstyle P}}

\newcommand{\Si}{\mathbin{\mathcal{S}}}
\newcommand{\U}{\mathbin{\mathcal{U}}}

\newcommand{\Z}{\mathbb{Z}}
\newcommand{\N}{\mathbb{N}}

\newcommand{\I}{\mathcal{I}}
\newcommand{\J}{\mathcal{J}}
\newcommand{\A}{\mathcal{A}}
\newcommand{\TO}{\mathcal{O}}
\newcommand{\T}{\mathcal{T}}
\newcommand{\R}{\mathcal{R}}
\newcommand{\C}{\mathcal{C}}
\newcommand{\K}{\mathcal{K}}

\newcommand{\can}{{\C_{\TO,\A}}}

\newcommand{\q}{\boldsymbol q}
\newcommand{\rew}{\boldsymbol{Q}}

\newcommand{\ind}{\mathsf{ind}}
\newcommand{\tem}{\mathsf{tem}}
\newcommand{\ans}{\mathsf{ans}}
\newcommand{\cl}{\mathsf{cl}}

\newcommand{\subo}{\mathsf{sub}_{\TO}}
\newcommand{\subt}{\mathsf{sub}_{\T}}
\newcommand{\subr}{\mathsf{sub}_{\R}}

\newcommand{\lang}{\mathcal{L}}
\newcommand{\frag}{{\boldsymbol{c}}}
\newcommand{\fragr}{{\boldsymbol{r}}}
\newcommand{\op}{{\boldsymbol{o}}}

\newcommand{\ka}[1]{\sigma_\A(#1)}
\newcommand{\kpar}{k}

\newcommand{\OMPIQ}{OMPIQ}
\newcommand{\OMQ}{OMQ}
\newcommand{\OMAQ}{OMAQ}

\newcommand{\avec}[1]{\boldsymbol{#1}}
\newcommand{\SA}{\mathfrak S_\A}
\newcommand{\GOA}{\mathfrak G_{\TO,\A}}

\newcommand{\type}{\tau}
\newcommand{\beam}{\boldsymbol{b}}
\newcommand{\rod}{\boldsymbol{r}}
\newcommand{\tp}{\tau}
\newcommand{\rtp}{\rho}

\newcommand{\parity}{\textsc{Parity}}

\allowdisplaybreaks

\begin{document}

\title{First-Order Rewritability and Complexity of Two-Dimensional Temporal Ontology-Mediated Queries}

\author{\name Alessandro Artale \email artale@inf.unibz.it\\
\addr KRDB Research Centre, Free University of Bozen-Bolzano, Italy
\AND
\name Roman Kontchakov \email roman@dcs.bbk.ac.uk\\
\addr Department of Computer Science, Birkbeck, University of London, U.K.
       \AND
\name Alisa Kovtunova \email alisa.kovtunova@tu-dresden.de\\
\addr Chair for Automata Theory, Technische Universit\"at Dresden, Germany
\AND
\name Vladislav Ryzhikov \email vlad@dcs.bbk.ac.uk \\
        \addr Department of Computer Science, Birkbeck, University of London, U.K.
       \AND
       \name Frank Wolter \email wolter@liverpool.ac.uk \\
        \addr Department of Computer Science, University of Liverpool, U.K.
       \AND
       \name Michael Zakharyaschev \email michael@dcs.bbk.ac.uk \\
        \addr Department of Computer Science, Birkbeck, University of London, U.K.\\
        School of Data Analysis and Artificial Intelligence, HSE University, Moscow, Russia
       }

\maketitle

%%%%%%%%%%%%%%%%%%%%%%%%%%%%%%%%%%%%%%%%%%%%%%%%%%%%%%%%%%%%%%%%%%%%%%

\begin{abstract}
Aiming at ontology-based data access to temporal data, we design two-dimensional temporal ontology and query languages by combining logics from the (extended) \DL{} family with linear temporal logic \LTL{} over discrete time $(\Z,<)$. Our main concern is first-order rewritability of ontology-mediated queries (OMQs) that consist of a 2D ontology and a positive temporal instance query. Our target languages for FO-rewritings are two-sorted FO$(<)$---first-order logic with sorts for time instants ordered by the built-in precedence relation $<$ and for the domain of individuals---its extension $\FOE$ with the standard congruence predicates \mbox{$t \equiv 0 \pmod n$}, for any fixed $n > 1$, and $\FO(\RPR)$ that admits relational primitive recursion. In terms of circuit complexity, $\FOE$- and $\FO(\RPR)$-rewritability guarantee answering OMQs in uniform \ACz{} and \NCo{}, respectively.

We proceed in three steps. First, we define a hierarchy of 2D \DL/\LTL{} ontology languages and investigate the FO-rewritability of OMQs with atomic queries by constructing projections onto 1D \LTL{} OMQs and employing recent results on the FO-rewritability of propositional \LTL{} OMQs. As the projections involve deciding  consistency of ontologies and data, we also consider the consistency problem for our languages. While the  undecidability of consistency for 2D ontology languages with expressive Boolean role inclusions might be expected, we also show that, rather surprisingly, the restriction to Krom
and Horn role inclusions leads to decidability (and \ExpSpace-completeness), even if one admits full Booleans on concepts. As a final step, we lift some of the rewritability results for atomic OMQs to OMQs with expressive positive temporal instance queries. The lifting results are based on an in-depth study of the canonical models and only concern Horn ontologies.
\end{abstract}

%%%%%%%%%%%%%%%%%%%%%%%%%%%%%%%%%%%%%%%%%%%%%%%%%%%%%%%%%%%%%%%%%%%%%%

% !TEX root = TDL-Lite.tex

\section{Introduction}
\label{intro}

Ontology-based data access~\cite{CDLLR07}, also known as virtual knowledge graphs~\cite{DBLP:journals/dint/XiaoDCC19}, 
has recently become one of the most successful applications of ontologies. The main aim of ontology-based data access (OBDA, for short) is to facilitate access to possibly heterogeneous, distributed and incomplete data for non-IT-users. To this end, an ontology is employed to provide both a user-friendly and uniform vocabulary for formulating queries and a conceptual model of the domain for capturing background knowledge and obtaining more complete answers. Thus, instead of querying data directly by means of often convoluted and \emph{ad hoc} database queries, one can use ontology-mediated queries (OMQs) of the form $\q = (\TO, \varphi)$ with an ontology $\TO$ and a query~$\varphi$ over the familiar and natural vocabulary provided by $\TO$. The answers to $\q$ over a data instance~$\A$ (which can often be obtained via mappings from the original data) are then those tuples of individual names from $\A$ that satisfy $\varphi$ in every model of $\TO$ and $\A$. After nearly 20 years of research, OMQ answering is now well understood both in theory and real-world applications; consult~\citeA{DBLP:conf/rweb/BienvenuO15,DBLP:conf/ijcai/XiaoCKLPRZ18,DBLP:journals/dint/XiaoDCC19} for surveys. 

One of the main challenges in OBDA has been to identify ontology languages that strike a good balance between the expressive power required for conceptual modelling and querying on the one hand and the computational complexity of answering OMQs on the other. The ontology languages employed for OMQ answering nowadays are either based on description logics---DLs for short---\cite{BaaderCalvaneseEtAl2007,DBLP:books/daglib/0041477}, and in particular the \DL{} family~\cite{PLCD*08,ACKZ09}, or extensions of datalog and various forms of tuple-generating dependences~\cite{Abitebouletal95} such as linear and sticky sets of tgds~\cite{DBLP:journals/ws/CaliGL12,DBLP:journals/ai/CaliGP12} and existential rules~\cite{DBLP:journals/ai/BagetLMS11}. 
The data complexity of answering OMQs, where only the data is regarded as an input, while the OMQ is deemed to be fixed (or negligibly small compared to the data), and the rewritability of OMQs into conventional database queries that can be directly evaluated over the data without ontology reasoning have emerged as the two most important measures of the efficiency of OMQ answering.
Thus,  the \DL{}-based ontology language \OWLQL{}\footnote{https://www.w3.org/TR/owl2-profiles}\!, which was standardised by W3C specifically for OBDA and supported by systems such as \textsc{Mastro}\footnote{https://www.obdasystems.com} and \textsc{Ontop}\footnote{https://ontopic.ai}\!, ensures rewritability of all OMQs with conjunctive queries into first-order (FO) queries, i.e., essentially SQL queries~\cite{Abitebouletal95}. Complexity-wise, it means that such OMQs can be answered in \textsc{LogTime}-uniform \ACz{}, one of the smallest complexity classes~\cite{Immerman99}.

In many applications, data comes with timestamps indicating at which moment of time a fact holds. For instance, suppose we have a database on the submission and publication of papers in the area of computer science collected from various sources on the Web and elsewhere. The database may contain, among others, the facts
\begin{align*}
	& \nm{underSubmissionTo}(a,\, \text{JACM},\, \text{Feb2017}),
	&& \nm{UnderSubmission}(b,\, \text{Jan2021}),\\
	& \nm{underSubmissionTo}(a,\, \text{JACM},\, \text{Sep2020}),
	&&  \nm{Published}(b,\, \text{Oct2021}), \\
	& \nm{authorOf}(\text{Bob},\, b,\, \text{May2014}),
	&& \nm{Journal}(\text{JACM},\, \text{Jan1954})
\end{align*}
stating that paper $a$ was under submission to JACM in February 2017
and September~2020; paper $b$, authored by Bob in May 2014, was under
submission in January 2021 (to an unknown venue) and was published
(again in an unknown venue) in October 2021; JACM was a journal in
January 1954. Observe that the predicates in the snippet above have a
timestamp as their last argument (e.g.,~$\text{Feb2017}$) and either
one or two object domain arguments (e.g., $a$, $b$, $\text{Bob}$,
$\text{JACM}$).

While existing standard ontology languages  could use concrete datatypes
to support answering queries over timestamped data (e.g., the 
datatype \textsf{xsd:dateTime} represents timestamps in \textsl{OWL}), they do
not have the expressive power to model even basic temporal
aspects of the domain, and thus support the formulation of queries
taking temporal domain knowledge into account.
In particular, \textsl{OWL} does not have comparison operators over 
concrete domains, the so-called concrete predicates~\cite{LutzAHS05}, and so is
unable to express any temporal constraints for domain concepts. 
For instance, in the context of publishing papers, both ontology engineers and users need means for referring to a
discrete linearly ordered temporal precedence relation in order to
define axioms such as `all published papers have
previously been accepted' and formulate queries such as `find all
journals $x$ and months $t$ such that all papers under submission at~$x$ in $t$ were eventually published.'

In this article, we combine DLs from the \DL{} family with the well-established linear temporal logic \LTL{}~\cite{DBLP:books/cu/Demri2016} to obtain a hierarchy of ontology languages that support temporal conceptual modelling and OMQ answering over temporal data. Our main aim is to explore the trade-off between the expressive power and computational complexity/FO-rewritability of OMQs formulated in these combined languages that are interpreted over the two-dimensional Cartesian products of an object domain and a discrete linear order representing a flow of time.

Combinations of DLs with temporal formalisms have been widely investigated since the pioneering work of~\citeA{DBLP:conf/aaai/Schmiedel90} and~\citeA{DBLP:conf/epia/Schild93} in the early 1990s; we refer the reader to~\citeA{gkwz,Baader:2003:EDL:885746.885753,DBLP:reference/fai/ArtaleF05,DBLP:conf/time/LutzWZ08} for surveys and \citeA{DBLP:conf/dlog/PagliarecciST13,DBLP:journals/tocl/ArtaleKRZ14,DBLP:conf/kr/Gutierrez-BasultoJ014,DBLP:conf/ijcai/Gutierrez-Basulto15,DBLP:conf/ecai/Gutierrez-Basulto16,DBLP:journals/tocl/BaaderBKOT20} for more recent developments. However, the main reasoning task targeted in this line of research was concept satisfiability rather than OMQ answering, and the general aim was to identify and tailor combinations of temporal and DL constructs that ensure decidability of concept satisfiability with acceptable combined complexity.\!\footnote{In a nutshell, the rule of thumb is that to be decidable a temporalised DL should be embeddable into the \emph{monodic} fragment of first-order temporal logic, where no temporal operator is applied to a formula with two free variables~\cite{DBLP:journals/apal/HodkinsonWZ00}. Even tiny additions to one-variable temporal FO such as the `elsewhere' quantifier lead to undecidability~\cite{DBLP:journals/tocl/HampsonK15}.} In contrast, our main concern in this article is OMQ answering, and thus we focus instead on the data complexity and FO-rewritability of ontology-mediated queries.

We use the standard discrete time model with the integers $\Z$ and the order $<$ as precedence. The temporal operators supplied by \LTL{} are $\Rnext$ (at the next moment of time), $\Rdiamond$ (eventually), $\Rbox$ (always in the future),
$\mathcal{U}$ (until), and their past-time counterparts $\Lnext$ (at the previous moment),
$\Ldiamond$ (some time in the past), $\Lbox$ (always in the past) and $\mathcal{S}$ (since); e.g.,~\cite{DBLP:books/daglib/0077033,Gabbayetal94,DBLP:books/cu/Demri2016}.

Following the DL terminology, we refer to data instances as ABoxes; in the ABox snippet~$\A$  above, we call the binary relations \nm{UnderSubmission}, \nm{Published} and \nm{Journal} \emph{concept names} and the ternary relations \nm{underSubmissionTo} and \nm{authorOf} \emph{role names}. In general, an ABox is a finite set of atoms of the form $A(a,\ell)$ and $P(a,b,\ell)$, where $a$, $b$ are individual names, $\ell\in \Z$ is a timestamp, $A$ a concept name, and $P$ a role name.

The combined \DL{}/\LTL{} ontologies we consider are finite sets of inclusions between concepts and between roles in the style of \DL{}, in which temporal operators and DL constructs (e.g., intersection $\sqcap$ or union $\sqcup$) can be applied to concept and role names to construct compound concepts and roles. Concept and role inclusions are assumed to be true at all moments of time. In contrast to standard \DL{} (and generally DLs), here we treat roles in the same way as concepts and thus also allow Boolean operators to be applied to roles. We illustrate the expressive power of our languages using the domain of computer science papers. By applying temporal operators to roles we can state, for instance, that \nm{underSubmissionTo} is convex using
\begin{equation}\label{pub1}
\Ldiamond \nm{underSubmissionTo} \sqcap \Rdiamond \nm{underSubmissionTo}
\sqsubseteq \nm{underSubmissionTo},
\end{equation}
and that \nm{authorOf} is rigid (does not change in time):
\begin{equation}\label{pub2}
\Ldiamond \nm{authorOf} \sqsubseteq \nm{authorOf}, \qquad \Rdiamond \nm{authorOf} \sqsubseteq \nm{authorOf}.
\end{equation}
Similarly, we can postulate that concept \nm{Journal} is rigid.\!\footnote{Subtler modelling would be to state that each instance of \nm{Journal} remains a journal within its lifespan, which could be done using an additional concept `exists' that, for any moment of time, comprises those objects that exist at that moment~\cite<e.g.,>{gkwz}.}   
We could state inclusion similar to~\eqref{pub1} also for the concept name \nm{UnderSubmission}, but observe that, while convexity of role $\nm{underSubmissionTo}$ corresponds to the widespread policy that a paper can be submitted to the same venue only once, convexity of concept $\nm{UnderSubmission}$ implies that rejected papers cannot be submitted to another venue. So far, the role names have not been linked to the corresponding concept names. From \DL{} we inherit the capability of doing so via domain and range restrictions. Thus, we can extend our ontology with the equivalences 
	\begin{equation}\label{pub3}
		\exists \nm{underSubmissionTo} \equiv \nm{UnderSubmission}, \qquad
		\exists \nm{publishedIn} \equiv \nm{Published},
	\end{equation}
the first of which, for example, says that every article that is under submission to some venue is submitted (and the other way round). Observe that a concept inclusion stating that $\nm{UnderSubmission}$ and $\nm{Published}$ are disjoint will imply  that the roles $\nm{underSubmissionTo}$ and $\nm{publishedIn}$ are also disjoint. The converse does not hold as there can be submitted papers that have already been published elsewhere. Additional vocabulary items can be introduced by stating, for example, that one can only submit to conferences or journals:
	\begin{equation}\label{pub4}
		\exists \nm{underSubmissionTo}^{-} \sqsubseteq \nm{Conference} \sqcup \nm{Journal},
	\end{equation}
where the role $\nm{underSubmissionTo}^{-}$ is the inverse of \nm{underSubmissionTo}.
We denote by $\TO$ the ontology snippet with inclusions~\eqref{pub1}--\eqref{pub4}.
To illustrate OMQ answering, consider the atomic query $\varphi_{1}=\nm{UnderSubmission}$ that asks to find all pairs $(x,t)$ such that paper $x$ is under submission at time point $t$. By inclusions~\eqref{pub1} and~\eqref{pub3} of $\TO$, $(b,\text{Jan2021})$ and all pairs in the interval $(a,\text{Feb2017}),\dots,(a,\text{Sep2020})$ are answers to the OMQ $\q_1 = (\TO,\varphi_{1})$ over $\A$. Next, consider the OMQ $\q_2 = (\TO,\varphi_{2})$ with instance query $\varphi_2 = \exists \nm{authorOf}.(\nm{UnderSubmission}\sqcap \Rdiamond \nm{Published})$ asking for pairs $(x,t)$ such that $x$ is an author of
a paper that is under submission at time point~$t$ but eventually published. Then, by inclusion~\eqref{pub2}, $(\text{Bob},\text{Jan2021})$ is the single answer to this query over $\A$. Observe that both OMQs are rewritable into $\FO(<)$, \emph{two-sorted FO} with quantification over the convex closure of the time points in the ABox with the temporal precedence relation~$<$ supplemented by quantification over the ABox individuals. In fact, the following formulas are $\FO(<)$-rewritings of~$\q_1$ and~$\q_2$, respectively:
\begin{align*}
& \rew_{1}(x,t) \  = \  \nm{UnderSubmission}(x,t) \ \ \vee \ \ \ \nm{underSubmissionTo}(x,y,t)  \ \ \vee \ \ {} \\ 
& \hspace*{4em}\exists y, t', t''\, \bigl((t'< t < t'') \wedge
               \nm{underSubmissionTo}(x,y,t') \wedge \nm{underSubmissionTo}(x,y,t'')\bigr),\\[6pt]
& \rew_{2}(x,t) \  = \  \exists y, t'\, \bigl(\nm{authorOf}(x,y,t') \wedge \rew_{1}(y,t) \wedge {} \\ 
& \hspace*{13em} \exists t''\,\bigl[(t''>t) \wedge \bigl(\nm{Published}(y,t'') \lor \exists z\,\nm{publishedIn}(y,z,t'') \bigr)\bigr]\bigr).
\end{align*}
It follows that answering such OMQs is in \ACz{} for data complexity and can be implemented using conventional relational database management systems (RDBMSs).

In this article, we determine classes of OMQs all of which are $\FO(<)$-rewritable. To illustrate, call a concept \emph{basic} if it is of the form $A$ or $\exists S$, where $A$ is a concept name, and $S$ a role name or the inverse thereof. The set of $\mathcal{C}^{\textit{Prior}}$-concepts is obtained from basic concepts using arbitrary Boolean connectives and Prior's temporal operators $\Lbox$, $\Ldiamond$, $\Rbox$ and $\Rdiamond$~\cite{prior:1956b,DBLP:conf/spin/Vardi08}. %ono1980on,Burgess84
$\mathcal{C}^{\textit{Prior}}$-concept inclusions (CIs) are inclusions between $\mathcal{C}^{\textit{Prior}}$-concepts. Also, let $\mathcal{R}^{\Diamond}_{\smash{\rhorn}}$ denote the class of role inclusions (RIs) $R_{1}\sqcap \cdots \sqcap R_{n}\sqsubseteq R$, where the $R_i$ are roles possibly prefixed by $\Ldiamond$ or $\Rdiamond$ and $R$ is either $\bot$ or a role possibly prefixed by $\Lbox$ or $\Rbox$. The ontology
$\TO$ above is a union of $\mathcal{C}^{\textit{Prior}}$-CIs and $\mathcal{R}^{\Diamond}_{\smash{\rhorn}}$-RIs.
An atomic OMQ (\OMAQ{}, for short) takes the form $(\TO,B)$, where $B$ is a basic concept.
We obtain the following rewritability result:

\medskip
\noindent
\textbf{Theorem A.}
\emph{All \OMAQ{}s $(\TO,B)$, where $\TO$ is a union of $\mathcal{C}^{\textit{Prior}}$-CIs and $\mathcal{R}^{\Diamond}_{\smash{\rhorn}}$-RIs, are $\FO(<)$-rewritable.} 

\medskip

Our next aim is to go beyond \OMAQ{}s, as in Theorem~A, and admit 2D instance queries that support both DL and \LTL{} constructs. In fact, our main expressive instance query language is given by positive temporal concepts, $\varkappa$, that are constructed from concept names using $\sqcap$, $\sqcup$, any temporal operators of \LTL{}, and the DL construct $\exists S.\varkappa$, where $S$ is a role name or its inverse.
An OMQ $(\TO,\varkappa)$ with such a $\varkappa$ is called an ontology-mediated  positive instance query (or \OMPIQ{}).
We cannot expect Theorem~A to generalise to \OMPIQ{}s as already for the atemporal ontology consisting of the CI $\top \sqsubseteq A \sqcup B$ answering \OMPIQ{}s (even without temporal operators in the query) is \coNP-hard~\cite{DBLP:journals/jiis/Schaerf93}, which implies non-$\FO(<)$-rewritability. This \coNP-hardness result also holds for \OMPIQ{}s with a single temporal CI $A \sqsubseteq \Ldiamond B$ in the ontology. Nevertheless, we can define an expressive ontology language
with \OMPIQ{}s rewritable into $\FO(<)$.
Let $\mathcal{C}^{\Diamond}_{\smash{\horn}}$ denote the class of
CIs $C_{1} \sqcap \cdots \sqcap C_{n} \sqsubseteq C$, where the $C_{i}$ are basic concepts possibly prefixed by operators of the form $\Ldiamond,\Lbox,\Rdiamond,\Rbox$, and $C$ is either $\bot$ or a basic concept possibly prefixed by $\Lbox$ or~$\Rbox$.

\medskip
\noindent
\textbf{Theorem B.}
\emph{All \OMPIQ{}s $(\TO,\varkappa)$, where $\TO$ is the union of $\mathcal{C}^{\Diamond}_{\smash{\horn}}$-CIs and $\mathcal{R}^{\Diamond}_{\smash{\rhorn}}$-RIs, are $\FO(<)$-rewritable.}

\medskip

Notice that the ontology $\TO$ from our example without CI~\eqref{pub4} is covered by Theorem~B. The combined ontology languages considered up to now do not use the temporal operators~$\Lnext$ and $\Rnext$. In fact, as observed by~\citeA{AIJ21},  already the \OMAQ{} $(\{A \sqsubseteq \Rnext B,B \sqsubseteq \Rnext A\},A)$ is not $\FO(<)$-rewritable.
Following this work, we extend \mbox{$\FO(<)$} to obtain a suitable target language for rewriting OMQs using $\Lnext$ and $\Rnext$ with the data complexity still in
$\ACz$. Let \mbox{$\FOE$} be the extension of $\FO(<)$ with the unary congruence predicates $t\equiv 0 \pmod{n}$, for any fixed~$n > 1$. To illustrate our rewritability results with the target language $\FOE$, let a  $\DL_{\core}\Xnext$ ontology contain any inclusion $\vartheta_{1} \sqsubseteq \vartheta_{2}$ or $\vartheta_{1} \sqcap \vartheta_{2} \sqsubseteq \bot$, in which $\vartheta_{1},\vartheta_{2}$ are
either basic concepts or roles possibly prefixed by the operators $\Lnext$ and/or $\Rnext$. In other words, $\DL_{\core}\Xnext$ is the extension of the \DL{} dialect underpinning \OWLQL{} with the operators $\Lnext$ and~$\Rnext$.

\medskip
\noindent
\textbf{Theorem C.}
\emph{All \OMPIQ{}s $(\TO,\varkappa)$ with a $\DL_{\core}\Xnext$ ontology $\TO$ are $\FOE$-rewritable.}

\medskip

Observe that rigidity of concepts and roles as in~\eqref{pub2} can also be expressed using $\Lnext$ and $\Rnext$. To cover ontologies that are able to capture~\eqref{pub1} and more general Horn-shaped inclusions as well as $\Lnext$ and $\Rnext$, we have to go beyond $\FOE$ as the target language for rewritings and admit some form of recursion. Again following~\citeA{AIJ21}, we consider the extension $\FO(\RPR)$ of $\FO(<)$ with relational primitive recursion.
Rewritability into $\FO(\RPR)$ implies that OMQ answering is in $\NCo ~\subseteq~ \LogSpace$ for data complexity. Note that $\FO(\RPR)$-queries can be expressed in SQL with recursion or procedural extensions, which are, in general, less efficient and not always supported by RDBMSs. A $\DL_{\horn}\Xallop$ ontology contains
CIs and RIs of the form $\vartheta_{1}\sqcap \cdots \sqcap \vartheta_{n}\sqsubseteq \vartheta$, where each $\vartheta_{i}$ is a basic concept/role possibly prefixed by operators
of the from $\Rnext, \Rdiamond, \Rbox,\Lnext, \Ldiamond, \Lbox$, and $\vartheta$ is either $\bot$ or a basic concept/role possibly prefixed by operators of the form $\Rnext, \Rbox,\Lnext, \Lbox$.

\medskip
\noindent
\textbf{Theorem D.}
\emph{All \OMPIQ{}s $(\TO,\varkappa)$ with a $\DL_{\horn}\Xallop$ ontology $\TO$ are $\FO(\RPR)$-rewritable.}

\medskip

Theorems~A--D only show a few `impressions' of the results obtained in this article. To give a natural hierarchy of the combined \DL{}/\LTL{} OMQs and facilitate a systematic study of their rewritability and data complexity, we present the ontology languages in a rather different way than in Theorems~A--D, using a clausal normal form to be introduced in Section~\ref{sec:tdl}. An overview of our rewritability and complexity results for ontologies in the normal form will be provided in Table~\ref{TDL-table-omaq} (Section~\ref{sec:DL-Lite}) and Table~\ref{TDL-table-ompiq} (Section~\ref{OBDAfor Temporal}). The theorems above are then obtained by  straightforward polynomial-time normalisations.

\smallskip

We establish our results via a series of reductions to simpler cases. Our first main step
\begin{description}
\item[(prj)] projects two-dimensional \DL{}/\LTL{} OMAQs onto one-dimensional \LTL{} OMQs.
\end{description}
Here, by 1D \LTL{} we mean propositional \LTL{} speaking, intuitively, about how a single individual develops in time; input data is given by 1D ABoxes containing assertions of the form $A(\ell)$ with $\ell\in \Z$, and answers to 1D OMQs are sets of timestamps from the 1D ABox. \citeA{AIJ21} obtain 1D rewritability results for the target languages $\FO(<)$, $\FOE$ and $\FO(\RPR)$ using automata-theoretic machinery with a few model-theoretic insights. We apply those results here in a black-box manner. Thus, despite the fact that ultimately the results obtained in this article heavily rely on automata theory, it is only present implicitly as the vehicle used by~\citeA{AIJ21} to obtain 1D rewritings.

It turns out, however, that reduction~\textbf{(prj)} is not always possible and, even if possible, it requires careful treatment.  Recall that the interaction between concepts and roles in classical atemporal \DL{} ontologies is by design rather weak: it can be captured by the one-variable fragment of FO if all role inclusions are Horn and by the two-variable fragment of FO  otherwise~\cite{DBLP:conf/kr/KontchakovRWZ20}. Temporalised concepts and roles interpreted over 2D Cartesian products of \DL{} and \LTL{} structures make the combined logic more expressive than both of its components.  For example, answering OMAQs with Booleans on roles turns out to be similar to reasoning with the two-variable first-order \LTL{}, which is known to be undecidable~\cite{DBLP:journals/apal/HodkinsonWZ00}; see Theorem~\ref{thm:undec}. Even more unexpectedly, as we show in Theorem~\ref{thm:unexpected}, answering OMAQs with Horn-shaped CIs and RIs having $\Box$-operators on roles only is $\NCo$-complete (and so needs recursion), while the OMAQs in the respective \DL{} and \LTL{} component fragments are $\FO(<)$-rewritable: in fact, the combined logic is capable of expressing the operators $\Rnext$ and $\Lnext$ using $\Rbox$, $\Lbox$ and $\sqcap$. 

We construct a projection of 2D OMAQs with Boolean CIs and Horn RIs onto 1D \LTL{} \OMAQ{}s by showing that the interaction between the component logics can be captured using (exponentially many, in general) `connecting axioms' in the form of an implication  between propositional variables possibly prefixed by a temporal operator that, in many cases, has to be $\Rnext$ or $\Lnext$, which explains the unexpected increase of expressivity mentioned above. This projection is, however, sound and complete only if the input OMAQs are evaluated over ABoxes that are consistent with the ontology. So, before applying~\textbf{(prj)}, a preliminary reduction step is required:
\begin{description} 
\item[(con)] reduce FO-rewritability of \OMAQ{}s to FO-rewritability of $\bot$-free quasi-\OMAQ{}s\\ (restricted \OMPIQ{}s with the same rewritability properties as \OMAQ{}s),
\end{description}
whose ontology is consistent with any ABox, ensuring soundness and completeness of~\textbf{(prj)}. 

Reduction~\textbf{(con)} brings up the problem of deciding whether an ABox is consistent with an ontology. As this problem is of independent interest for temporal conceptual modelling, we consider not only the data complexity of deciding consistency but also its combined complexity. Our results are given in Table~\ref{table:consistency:tdl-lite} (Section~\ref{sec:consistency}) and provide a systematic investigation into the way how expressive role inclusions affect the complexity of consistency. While we confirm that full Boolean expressive power on roles leads to undecidability, we also show that, rather surprisingly, the restriction to Krom and Horn RIs leads to decidability and \ExpSpace-completeness, even if one admits full Booleans on concepts. The latter result is one of the very few instances breaking the monodicity barrier in temporal first-order logic, according to which in almost all instances reasoning about the temporal evolution of binary relations enables the encoding of undecidable tiling problems~\cite{gkwz}. The upper bounds are proved by relating the temporal DLs considered here to the one-variable fragment of first-order \LTL.

The final and technically most difficult step of our construction shows how to
\begin{description}
\item[(lift)]  lift rewritability and data complexity results from \OMAQ{}s to \OMPIQ{}s. 
\end{description}
This is again not always possible. In particular, as already mentioned above, if the ontology language admits disjunction, answering \OMPIQ{}s is \coNP-hard. We thus focus on \OMPIQ{}s whose ontology is given in the combined Horn \DL{}/\LTL{} language. We construct our rewritings inductively from the known rewritings of the constituent \OMAQ{}s by describing possible embeddings of the query into the canonical model. The classical atemporal FO-rewritings of \DL{} OMQs of, say~\citeA{PLCD*08,DBLP:journals/jacm/BienvenuKKPZ18}, rely upon the key property of \DL{} that can be characterised as  `concept locality' in the sense that the concepts a given individual belongs to in the canonical model only depend on the ABox concepts containing that element, if any, and the roles adjacent to it, but not on other individuals. The rewritability results in the temporal case are only possible because concept locality comes together with the (exponential) periodicity of the temporal evolution of each pair of domain elements in the canonical models, which can also be captured in FO for any given Horn \OMPIQ{}.

We have chosen to focus in this investigation on OMQs with positive instance queries because, under the open-world semantics, answers to queries involving negation, implication, or universal quantification are typically rather uninformative. To illustrate, consider the OMQ `find all journals $x$ and months $t$ such that all papers under submission to $x$ at $t$ are eventually published in $x$', 
\begin{align*}
\varphi(x,t) \ = \ \forall y\, \big(\nm{underSubmissionTo}(y,x,t) \rightarrow \Rdiamond \nm{publishedIn}(y,x,t)\big),
\end{align*}
where $\Rdiamond \nm{publishedIn}(y,x,t)$ is an atom built from a positive temporal role $\Rdiamond \nm{publishedIn}$.
Under the open-world semantics, this query has no answers
over any ABox simply because, for any journal and month, we can add a paper under submission that was never published. The epistemic semantics for OMQs with FO-queries proposed by~\citeA{CalvaneseGLLR07} and partially realised in SPARQL~\cite{GlimmOgbuji13} yields more useful answers. In the query above, the answers $(x,t)$ are then such that every paper $y$ \emph{known} to be under submission to~$x$ at $t$ is eventually published in $x$.
We observe that, as expected from the atemporal and 1D temporal cases, for OMQs with FO-queries that use positive temporal concepts and roles as atoms (e.g., $\Rdiamond \nm{publishedIn}(y,x,t)$) interpreted under the epistemic semantics, the data complexity is the same as for \OMPIQ{}s.

The article is structured as follows. In the remainder of this section, we briefly discuss related work.  Section~\ref{sec:tdl} introduces our classification of temporal \DL{} logics, and  Section~\ref{sec:tomq} defines and illustrates \OMAQ{}s and \OMPIQ{}s with their semantics. Section~\ref{sec:consistency} investigates the combined complexity of checking consistency of knowledge bases in our formalisms. Section~\ref{sec:DL-Lite} identifies classes of FO-rewritable \OMAQ{}s by projection to \LTL, while Section~\ref{OBDAfor Temporal} lifts the obtained results from Horn \OMAQ{}s to \OMPIQ{}s. First-order temporal OMQs under the epistemic semantics are briefly considered in Section~\ref{sec:omq}. We conclude in Section~\ref{conclusions} by discussing open problems and directions of further research.

% !TEX root = TDL-Lite.tex

\subsection{Related Work}
\label{related}

We have already discussed the main related work on atemporal OBDA and on combining DLs and temporal logic. In this section, we focus on two related topics: (1) work on OMQ answering over temporal data with discrete linear time and (2) work on temporal deductive databases. For detailed recent surveys of temporal OMQ answering in general we refer the reader to~\citeA{AKKRWZ:TIME17,DBLP:conf/rweb/RyzhikovWZ20}.

In temporal OMQ answering, there is a basic distinction between formalisms where temporal constructs are added to \emph{both} ontology and query languages (as in this article) and those in which only the query language is temporalised while ontologies are given in some standard atemporal language. The main advantage of keeping the ontology language atemporal is that the increase in the complexity of query answering compared to the atemporal case (observed also in this article) is less severe, if any. Temporal OMQ answering with atemporal ontologies has been introduced and investigated in the context of semantic technologies for situation awareness; a detailed introduction is given by~\citeA{DBLP:journals/ki/BaaderBKTT20}. A basic query language, \LTL-CQ, proposed  in this framework is obtained from \LTL{} formulas by replacing occurrences of propositional variables by arbitrary conjunctive queries (CQs). In this framework, a fundamental distinction is between answering  \LTL-CQs without any additional temporal constraints and answering  \LTL-CQs  when some concept and/or role names are declared rigid so that their interpretation does not change over time. \citeA{DBLP:conf/cade/BaaderBL13,DBLP:journals/ws/BaaderBL15} analyse the problem of answering \LTL-CQs with respect to $\mathcal{ALC}$ and $\mathcal{SHQ}$ ontologies and show, for example, that \LTL-CQ answering under $\mathcal{ALC}$ ontologies is \ExpTime-complete in combined complexity without rigid names but 2\ExpTime-complete with rigid names.
\citeA{DBLP:conf/gcai/BorgwardtT15,DBLP:conf/ijcai/BorgwardtT15} and \citeA{DBLP:conf/ausai/BaaderBL15} analyse the complexity of answering \LTL-CQs for weaker ontology languages such as $\EL$ (see below) and \DL{}, while~\citeA{DBLP:conf/frocos/BorgwardtLT13,DBLP:journals/ws/BorgwardtLT15} study  rewritability of \LTL-CQs. 
\citeA{DBLP:journals/semweb/BourgauxKT19} investigate the problem of querying inconsistent data, and \citeA{DBLP:conf/aaai/Koopmann19} proposes an extension to probabilistic data. As the monitoring of systems based on data streams is a major application area for these logics, it is of interest to be able to answer queries without storing the whole ABox; for such results on the query language \LTL-CQs, we refer the reader to~\citeA{DBLP:journals/ws/BorgwardtLT15}.

In the monitoring context, the temporal query language STARQL inspired by SPARQL has 
been proposed and investigated~\cite{OzcepM14}: its expressive power is quite different from
\LTL-CQs as it extends SPARQL with  time window operators and comes with an epistemic semantics similar to ours in Section~\ref{sec:omq}. %
Efficient query answering under the assumption that only restricted portions of data, or time windows, are available for querying is the central problem of a more recent research direction of stream reasoning~\cite<e.g.,>{DBLP:journals/semweb/DellAglioEHP19,DBLP:journals/ai/BeckDE18,DBLP:conf/esws/AjileyeMH21}.

We next consider the work on OMQ answering with temporal ontologies that combine \LTL{} with $\EL$ instead of \DL{}. $\EL$ underpins the \OWLEL{} profile of \textsl{OWL\,2}~\cite{DBLP:books/daglib/0041477}, and, unlike $\DL$, it admits qualified existential restrictions, $\exists R.C$, but not inverse roles. Since OMQ answering with atemporal $\EL$ is  \PTime-complete for data complexity, a more expressive target language than $\FO(<)$ is required. \citeA{DBLP:conf/ijcai/Gutierrez-Basulto16} consider combinations of fragments of \LTL{} and $\EL$ and investigate the complexity and rewritability of atomic queries. Since those combinations contain no RIs,  the only temporal expressive power in the ontology language for roles is to declare them rigid (CIs can contain temporal operators). The target language for rewritings is the extension datalog$_{1s}$ of datalog with a single successor function $\textit{succ}$ introduced by~\citeA{Chomicki:1988:TDD:308386.308416}, and answering atomic OMQs is shown to be \PTime-complete for data and \PSpace-complete for combined complexity in an $\EL$/\LTL{} combination without rigid roles, and \PSpace-complete for data and in \ExpTime{} for combined complexity if rigid roles can occur only on the left-hand side of CIs. For acyclic ontologies, rewritability into the extension of $\FO(<)$ with the standard numerical predicate~$+$ is obtained. \citeA{BorgwardtFK21} investigate an extension of $\mathcal{ELH}_\bot$ with diamond-type temporal operators for bridging gaps in sparse temporal data and show that
answering metric temporal rooted CQs with guarded negation under the minimal-world semantics is \PTime-complete for data and \ExpSpace-complete for combined complexity. 

The extension of relational databases to temporal deductive databases has been an active area of research for almost 40 years. It is beyond the scope of this discussion to provide a survey of this field, and we refer the reader
instead to the Encyclopaedia of Database Systems for brief introductions  and pointers to the literature~\cite{DBLP:reference/db/2018}. Of particular relevance here is the work on the extension of logic programming languages by  temporal operators or by terms for the natural numbers or integers representing time points and possibly constraints over them in the sense of constraint databases~\cite{DBLP:books/sp/Kuper00}. Two basic such languages are Templog,  which extends Prolog with the operators $\Rbox$, $\Rdiamond$, and $\Rnext$ (restricted to admit Herbrand models) interpreted over the natural numbers~\cite{DBLP:books/bc/tanselCGSS93/BaudinetCW93} and the aforementioned datalog$_{1s}$, which offers terms $0, \textit{succ}(0), \dots$; for further relevant languages consult~\citeA{DBLP:conf/cdb/Revesz00}. We share with this research direction the issue of having to deal with infinite intended models over the integers and thus answering queries over infinite domains. Another main concern of temporal deductive databases is the problem of finding finite representations of infinite answers~\cite{DBLP:books/bc/tanselCGSS93/BaudinetCW93}. This problem is beyond the scope of this article, where we follow the standard approach, which is  taken in both database theory and OBDA, and only look for answers to OMQs from the finite active domain of the input ABox. 

The expressive power of the ontologies and OMQs considered in this article and the queries considered in temporal datalog extensions are incomparable. A major insight of research into OBDA was the need for existential restrictions on the right-hand side of concept inclusions. The main datalog extensions considered in temporal deductive databases do not admit such existential quantifiers (except possibly over the time points). On the other hand, plain datalog without additional temporal features is already much more expressive than OMQs with ontologies in Horn \DL{} dialects without time. As temporal extensions of datalog inherit this expressive power, query answering is much harder in data complexity than in the languages we are investigating.
Finally, we mention the recent temporalisations of datalog using the operators of metric temporal logic \textsl{MTL} rather than \LTL{}~\cite{DBLP:journals/jair/BrandtKRXZ18,KR2020-79,DBLP:conf/ijcai/WalegaGKK20,DBLP:conf/aaai/CucalaWGK21}.

% !TEX root = TDL-Lite.tex

\section{Temporal \DL{} Logics}
\label{sec:tdl}

We begin by defining the ontology languages we investigate in the context of temporal OBDA via FO-rewriting. The W3C standardised ontology language \OWLQL{} for atemporal OBDA is based on the \DL{} family~\cite{CDLLR07,ACKZ09} of description logics that was designed as a compromise between two aims:
\begin{itemize}
\item[--] the logics should be expressive enough to represent basic conceptual modelling constructs (thereby providing a link between relational databases and ontologies), and

\item[--] simple enough to guarantee uniform reducibility of OMQ answering to standard database query evaluation (that is, FO-rewritability, which also ensures that answering OMQs can be done in \ACz{} for data complexity).
\end{itemize}
Conceptual data models for temporal databases~\cite{DBLP:journals/tods/ChomickiTB01} were analysed and encoded in 2D temporalised \DL{} logics  by~\citeauthor{DBLP:journals/tocl/ArtaleKRZ14}~\citeyear{DBLP:journals/tocl/ArtaleKRZ14}, but the rewritability properties of those logics have remained open.
The temporal \DL{} logics we are about to define are constructed similarly to those of~\citeA{DBLP:journals/tocl/ArtaleKRZ14}. However, following the standard atemporal \OWLQL{}, we opted not to include cardinality constraints in our languages as their interaction with role inclusions ruins FO-rewritability already in the atemporal case~\cite{ACKZ09}. On the other hand, we
allow role inclusion axioms of the same shape as concept inclusion axioms.

The alphabet of our temporal \DL{} logics contains countably infinite sets of \emph{individual names} $a_0,a_1,\dots$, \emph{concept names} $A_0,A_1,\dots$, and \emph{role names} $P_0,P_1,\dots$. We then construct \emph{roles} $S$, \emph{temporalised roles} $R$, \emph{basic concepts} $B$, and \emph{temporalised concepts} $C$ by means of the following grammar:
\begin{align*}
& S \ \ ::=\ \ P_i \ \ \mid \ \ P_i^-,
&& R \ \ ::=\ \ S \ \ \mid\ \ \Rbox R \ \ \mid \ \ \Lbox R \ \ \mid\ \ \Rnext R \ \ \mid \ \ \Lnext R,\\
& B \ \ ::= \ \ A_i \ \ \mid \ \ \exists S,
&& C \ \ ::=\ \ B  \ \ \mid\ \ \Rbox C \ \ \mid \ \ \Lbox C \ \ \mid\ \ \Rnext C \ \ \mid \ \ \Lnext C
\end{align*}
with the \emph{temporal operators} $\Rbox$ (always in the future), $\Lbox$ (always in the past), $\Rnext$ (at the next moment), and $\Lnext$ (at the previous moment).
A \emph{concept} or \emph{role inclusion} takes the~form
\begin{equation}\label{axiom1}
\vartheta_1 \sqcap \dots \sqcap \vartheta_k ~\sqsubseteq~ \vartheta_{k+1} \sqcup \dots \sqcup \vartheta_{k+m},
\end{equation}
where the $\vartheta_i$ are all temporalised concepts of the form $C$ or, respectively, temporalised roles of the form $R$. As usual, we denote the empty $\sqcap$ by $\top$ and the empty $\sqcup$ by $\bot$.
A  \emph{TBox} $\T$ and an \emph{RBox} $\R$ are finite sets of concept inclusions (CIs, for short) and, respectively, role inclusions (RIs); their union $\TO= \T\cup \R$ is called an \emph{ontology}. In the sequel, $\TO$, $\T$ and $\R$ (possibly with decorations) always denote an ontology, TBox and RBox, respectively.

Following~\citeA{AIJ21}, we classify ontologies depending on the form of their concept and role inclusions and the temporal operators that are allowed to occur in them. Let $\frag,\fragr \in \{\gbool,\bool, \horn, \krom,\core\}$ and $\op \in \{\Box, \nxt,\Box\nxt\}$. We denote by $\DL_{\frag/\fragr}^{\op}$ the \emph{temporal description logic} whose ontologies contain inclusions of the form~\eqref{axiom1} with the (future and past) operators indicated by $\op$ (for example, $\op = \Box$ means that only $\Rbox$ and $\Lbox$ can be used); in addition, CIs in $\DL_{\frag/\fragr}^{\op}$ satisfy the following restrictions on $k$ and~$m$ in~\eqref{axiom1} indicated by $\frag$:
\begin{description}\itemsep=0pt
\item[(horn)] $m\leq 1$ if $\frag = \horn$,

\item[(krom)] $k + m\leq 2$ if $\frag = \krom$,

\item[(core)] $k + m\leq 2$ and $m \leq 1$ if $\frag = \core$,

\item[(g-bool)] $k \geq 1$ and any $m$ if $\frag = \gbool$,\footnote{Here, {\bf g} stands for `guarded'~\cite{DBLP:journals/jphil/AndrekaNB98}.}

\item[(bool)] any $k$ and $m$ if $\frag = \bool$,
\end{description}
and RIs in $\DL_{\frag/\fragr}^{\op}$  satisfy analogous restrictions on $k$ and $m$ indicated by $\fragr$. Whenever  $\frag = \fragr$, we use a single subscript: $\DL^\op_{\frag} = \DL^\op_{\frag/\frag}$. We shall also require the fragments $\smash{\DL_{\frag/\rhorn}^{\op}}$ of $\smash{\DL_{\frag/\horn}^{\op}}$ that disallow temporal operators on the left-hand side of RIs; without temporal operators, $\rhorn$ coincides with $\horn$. Containment between the fragments are shown in Fig.~\ref{fig:fragment:inclusions}. Note that all of our logics feature disjointness inclusions of the form $\vartheta_1 \sqcap \vartheta_2 \sqsubseteq \bot$. However, unlike the standard atemporal \DL{} logics~\cite{CDLLR07,ACKZ09}, which can have various types of CIs but allow only \emph{core} RIs (of the form $S_1 \sqsubseteq S_2$ and $S_1 \sqcap S_2 \sqsubseteq \bot$), we treat CIs and RIs in a uniform way and impose restrictions on the clausal structure of CIs and RIs separately (the complexity  of reasoning with such atemporal DLs is discussed in Section~\ref{underlying-DLs} below).

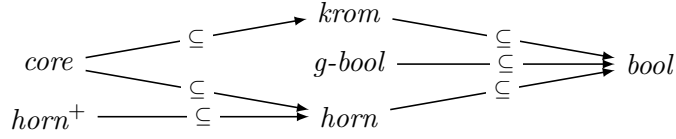
\begin{figure}[t]
\centering
\begin{tikzpicture}[xscale=4,yscale=0.7,semithick]
\node (c) at (0,0) {\core};
\node (k) at (1,1) {\krom};
\node (h) at (1,-1) {\horn};
\node (m) at (0,-1) {$\rhorn$};
\node (b) at (2,0) {\bool};
\node (g) at (1,0) {\gbool};
\begin{scope}\scriptsize
\draw[->] (c) -- node[fill=white] {$\subseteq$} (k);
\draw[->] (c) -- node[fill=white] {$\subseteq$} (h);
\draw[->] (k) -- node[fill=white] {$\subseteq$} (b);
\draw[->] (h) -- node[fill=white] {$\subseteq$} (b);
\draw[->] (m) -- node[fill=white] {$\subseteq$} (h);
\draw[->] (g) -- node[fill=white] {$\subseteq$} (b);
\end{scope}
\end{tikzpicture}
\caption{Containment between fragments (only the shown containments hold).}\label{fig:fragment:inclusions}
\end{figure}

An \emph{ABox} (or \emph{data instance}), $\A$, is a finite set of atoms of the form $A_i(a,\ell)$ and $P_i(a,b,\ell)$, where $a$ and $b$ are individual names and $\ell \in \Z$ is a \emph{timestamp}.  We write $P^-(a,b,\ell)\in\A$ whenever $P(b,a,\ell)\in\A$. We denote by $\ind(\A)$ the set of individual names that occur in~$\A$, by $\min \A$ and $\max \A$ the minimal and maximal integer numbers occurring in~$\A$, and set $\tem(\A) = \bigl\{\,n \in \Z\mid \min \A \leq n \leq \max \A\,\bigr\}$. To simplify constructions and without much loss of generality, we assume that $0 = \min \A$ and $1 \le \max \A$, implicitly adding `dummies' such as $D(a,0)$ and $D(a,1)$ if necessary, where $D$ is a fresh concept name (which will never be used in queries).
A $\DL^{\op}_{\frag/\fragr}$ \emph{knowledge base} (KB) is a pair $(\TO,\A)$, where $\TO$ is a $\DL^{\op}_{\frag/\fragr}$ ontology and $\A$ an ABox. The \emph{size}~$|\TO|$ of an ontology $\TO$ is the number of occurrences of symbols in~$\TO$; the size of a TBox, RBox, ABox, and knowledge base is defined analogously assuming that the numbers (timestamps) in ABoxes are given in \emph{unary}.

A (\emph{temporal}) \emph{interpretation} is a pair $\I = (\Delta^\I,\cdot^{\I(n)})$, where $\Delta^\I \ne\emptyset$ and, for  each $n \in \Z$,
\begin{equation}\label{interpretation}
\I(n)=(\Delta^\I, a_0^\I,\dots ,A_0^{\smash{\I(n)}}, \dots
,P_0^{\smash{\I(n)}},\dots )
\end{equation}
is a standard (atemporal) description logic interpretation with  $a_i^{\I}\in \Delta^\I$, $A_i^{\smash{\I(n)}}\subseteq \Delta^\I$
and $P_i^{\smash{\I(n)}}\subseteq\Delta^\I\times\Delta^\I$.
Thus, we assume that the \emph{domain} $\Delta^\I$ and the interpretations $a_i^\I\in \Delta^\I$ of the individual names are the same for all \mbox{$n\in \Z$}. (However, we do not adopt the unique name assumption, which does not affect our results.)
The description logic and temporal constructs are interpreted in~$\I(n)$ as follows:
\begin{align}
\nonumber
(P^-_i)^{\I(n)} &~=~ \bigl\{\, (u,v) \mid (v,u) \in P_i^{\I(n)} \,\bigr\}, &
(\exists S)^{\I(n)} &~=~ \bigl\{\, u \mid  (u,v) \in S^{\I(n)}, \text{ for some } v  \,\bigr\},\\[4pt]
\label{eq:semantics:box}
(\Rbox \vartheta)^{\I(n)} &~=~ \bigcap_{k>n}
  \vartheta^{\I(k)}, &
(\Lbox \vartheta)^{\I(n)} &~=~ \bigcap_{k<n}
  \vartheta^{\I(k)},\\
\label{eq:semantics:next}
(\Rnext \vartheta)^{\I(n)} &~=~ \vartheta^{\I(n+1)}, &
(\Lnext \vartheta)^{\I(n)} &~=~ \vartheta^{\I(n-1)}.
\end{align}
CIs and RIs are interpreted in $\I$ \emph{globally} in the sense that inclusion~\eqref{axiom1} is \emph{true} in $\I$ if
\begin{equation*}
\vartheta_1^{\I(n)} \cap \dots \cap \vartheta_k^{\I(n)} ~\subseteq~ \vartheta_{k+1}^{\I(n)} \cup \dots \cup \vartheta_{k+m}^{\I(n)},  \qquad\text{for \emph{all} $n \in \Z$}.
\end{equation*}
As usual, $\bot$ (the empty $\sqcup$) is interpreted by $\emptyset$, and $\top$ (the empty $\sqcap$) by $\Delta^{\smash{\I}}$ for concepts and  by $\Delta^{\smash{\I}}\times\Delta^{\smash{\I}}$ for roles. Given an inclusion $\alpha$, we write $\I \models \alpha$ if $\alpha$ is true in $\I$.
We call $\I$ a \emph{model} of $(\TO,\A)$ and write
$\I\models (\TO,\A)$ if
$\I \models \alpha$ for all $\alpha \in \TO$, $a^{\smash{\I}}\in A^{\smash{\I(\ell)}}$   for all $A(a,\ell)\in \A$, and $(a^{\smash{\I}},b^{\smash{\I}})\in P^{\smash{\I(\ell)}}$ for all $P(a,b,\ell)\in \A$.
We say that $\TO$ is \emph{consistent} if there is an interpretation $\I$, a \emph{model of $\TO$}, such that $\I\models\alpha$,  for all $\alpha \in \TO$; we also say that  $\A$ is \emph{consistent with} $\TO$ if there is a model of $(\TO,\A)$.
For an inclusion $\alpha$, we write $\TO\models \alpha$ if $\I \models \alpha$ for every model $\I$ of $\TO$. A concept $C$ is
\emph{consistent with} $\TO$ if there is a model $\I$ of $\TO$ and $n\in \Z$ such that $C^{\I(n)}\ne\emptyset$; consistency of roles with $\TO$ is defined analogously.

In the sequel, we assume that our RBoxes are \emph{closed under taking the inverses of roles in RIs} in the sense that, together with every RI (say, $\Lnext P_1 \sqcap \Rbox P^-_2 \sqsubseteq P^-_3$), the RBox contains the corresponding RI for the inverse roles ($\Lnext P^-_1 \sqcap \Rbox P_2 \sqsubseteq P_3$ in this instance). However, to avoid notational clutter, we do not mention RIs with inverse roles in our examples.

It is not hard to see that, for any $\smash{\DL_{\frag/\fragr}^{\op}}$ ontology $\TO$, one can construct a $\smash{\DL_{\frag/\fragr}^{\op}}$ ontology $\TO'$, possibly using some fresh concept and role names, such that $\TO'$ contains no nested temporal operators, the size of $\TO'$ is linear in the size of $\TO$, and $\TO'$ is a model-conservative  extension\footnote{An ontology $\TO'$ is a \emph{model conservative extension} of an ontology $\TO$ if $\TO'$ entails $\TO$, the signature of $\TO$ is contained in the signature of $\TO'$, and every model of $\TO$ can be extended to a model of $\TO'$ by providing interpretations of the fresh symbols of $\TO'$ and leaving the domain and the interpretation of the symbols in $\TO$ unchanged.}
of $\TO$, which implies that $\TO$ and $\TO'$ give the same certain answers to queries  (to be formally defined in Section~\ref{sec:tomq}). For example, the inclusion $\Rbox\Lnext A \sqsubseteq B$ in $\TO$ can be replaced with two inclusions $\Lnext A \sqsubseteq A'$ and $\Rbox A' \sqsubseteq B$, where $A'$ is a fresh name. In what follows and where convenient, we assume without loss of generality that our ontologies do not contain nested temporal operators.

Although
we do not include the standard \LTL{} operators $\Rdiamond$ (eventually), $\Ldiamond$ (some time in the past), $\U$ (until) and $\Si$ (since) in our ontology languages, all of them can be expressed in the $\bool$ fragment (but not necessarily in smaller fragments). For example, $A \sqsubseteq \Rdiamond B$ can be simulated by two $\krom$ inclusions with $\Rbox$ and a fresh name $A'$: namely, $A \sqcap \Rbox A' \sqsubseteq \bot$ and $\top \sqsubseteq A' \sqcup B$. However,  it cannot be expressed in $\core$ or $\horn$ fragments; more details and examples can be found in~\cite{AIJ21}.
Each of our ontology languages can say that a concept $A$ is \emph{expanding} (by means of $A \sqsubseteq \Rnext A$ in the languages with $\nxt$ and $A \sqsubseteq \Rbox A$ in the languages with $\Box$) or \emph{rigid} (using, in addition, $A \sqsubseteq \Lnext A$  and $A \sqsubseteq \Lbox A$, respectively), and similarly for roles.
Finally, note that $\smash{\DL\Xbox_{\horn/\rhorn}}$ extends the temporal ontology language \textsl{TQL}~\cite{ArtaleKWZ13}, which only allows $\Ldiamond$ and $\Rdiamond$ on the left-hand side of $\horn$ CIs and RIs: $\Ldiamond A \sqsubseteq B$ is equivalent to $A \sqsubseteq \Rbox B$. Thus, \emph{convexity} axioms such as~\eqref{pub1} for both concepts and roles  are expressible in this language:
\begin{gather*}
 \nm{underSubmissionTo}  \sqsubseteq \Rbox \nm{wasUST}, \qquad\qquad
 \nm{underSubmissionTo}  \sqsubseteq \Lbox \nm{willBeUST},\\
 \nm{wasUST} \sqcap \nm{willBeUST}  \sqsubseteq \nm{underSubmissionTo},
\end{gather*}
where $\nm{wasUST}$ and $\nm{willBeUST}$ are two fresh roles names.

\subsection{Remarks on the Underlying Description Logics}\label{underlying-DLs}

We conclude Section~\ref{sec:tdl} with a brief summary of what is known about the computational complexity of reasoning with the DL fragments of our temporal \DL{} logics; for more details the reader is referred to~\citeA{DBLP:conf/kr/KontchakovRWZ20}.

\begin{table}[t]
\centering%
\renewcommand{\tabcolsep}{6pt}%
\begin{tabular}{ccccc}\toprule
role & \multicolumn{4}{c}{concept inclusions}\\
inclusions & \bool & \krom  &   \horn & \core\\\midrule
 \bool & \multicolumn{2}{c}{\NExpTime{}} & \multicolumn{2}{c}{\NExpTime{}}   \\
\gbool &  \multicolumn{2}{c}{\ExpTime{}} & \multicolumn{2}{c}{\ExpTime{}}  \\
 \krom & \NP{}   & \NL{}  &  \NP{}   & \NL{}  \\
 \horn &  \NP{}  & \PTime{}  &  \PTime{}  & \PTime{}  \\
 \core &  \NP &  \NL &  \PTime{}  & \NL{} \\\bottomrule\\[-6pt]
 \multicolumn{5}{c}{\small a) Consistency: combined complexity.}
\end{tabular}
\hfill
\renewcommand{\tabcolsep}{6pt}%
\begin{tabular}{ccc}\toprule
role & \multicolumn{2}{c}{concept inclusions}\\
inclusions\ & \bool \hfil \krom &  \horn   \hfil \core\\\midrule
  \bool & \coNP{} &  \coNP{} \\
  \gbool  &   \coNP{} &   \coNP{}    \\
  \krom &   \coNP{}/in \ACz{} & \coNP{}/in \ACz{} \\
  \horn &   \coNP{}/in \ACz{} &  in \ACz{} \\
  \core &   \coNP{}/in \ACz{} & in \ACz{}  \\\bottomrule\\[-6pt]
 \multicolumn{3}{c}{\small b) Instance/atomic queries: data complexity.}
\end{tabular}
\caption{Complexity of the underlying (non-temporal) description logics. For data complexity, a single result is shown when it holds for both instance queries ($\mathcal{ELUI}$-concepts) and atomic queries (concept names).}
\label{DL-sat}
\label{DL-IC}
\end{table}

Table~\ref{DL-sat}a shows the known results on the \emph{combined complexity} of checking whether a given knowledge base is consistent (with both ontology and ABox regarded as input). Observe 
first that the most expressive \DL{} logic, which admits Boolean CIs and RIs, is contained in the two-variable fragment FO$_2$ of first-order logic (FO)
in the sense that, for every inclusion, one can construct in linear time an equivalent sentence in FO$_{2}$. From the \NExpTime{} upper bound for consistency in FO$_{2}$
we then obtain \NExpTime{} membership for consistency for all our \DL{} logics. A matching lower bound is proved by~\citeA{DBLP:conf/kr/KontchakovRWZ20} for the logic admitting Boolean CIs and RIs by observing that the extension $\mathcal{ALC}^{\top,\text{id},\cap,\neg,-}$ of $\mathcal{ALCI}$ with Boolean operators for roles captures
FO$_{2}$~\cite{DBLP:conf/csl/LutzSW01} and that, for every CI without the identity role in $\mathcal{ALC}^{\top,\text{id},\cap,\neg,-}$, one can construct in polynomial time a model conservative extension in the \DL{} logic.
By disallowing RIs of the form $\top~\sqsubseteq~ \vartheta_1 \sqcup \dots \sqcup \vartheta_m$, we obtain a \DL{} logic whose CIs and RIs are translated into \emph{guarded} FO$_{2}$-sentences~\cite{DBLP:journals/jphil/AndrekaNB98}. Observe that no such restriction is needed for CIs: for instance, the standard FO-translation of the CI $\top \sqsubseteq A \sqcup B$ is the guarded sentence $\forall x\, \bigl((x=x) \to A(x) \lor B(x)\bigr)$. This logic inherits from guarded FO$_{2}$ that consistency is
\ExpTime-complete. The \DL{} logics with Horn/Krom RIs are polynomially reducible (preserving Horness/Kromness) to propositional logic. Finally, the \DL{} logics with core RIs are well-documented in the literature: the one with core CIs goes under the monikers $\DL_{\mathcal{R}}$~\cite{CDLLR07} and $\DL^{\mathcal{H}}_\core$~\cite{ACKZ09}; the one with Horn CIs is known as $\DL_{\mathcal{R},\sqcap}$~\cite{DBLP:conf/kr/CalvaneseGLLR06} and $\DL^{\mathcal{H}}_\horn$~\cite{ACKZ09}; the remaining two logics are called $\DL_{\krom}^\mathcal{H}$ and $\DL_{\bool}^\mathcal{H}$~\cite{ACKZ09}.

In the context of answering ontology-mediated queries to be discussed in the sequel, we are interested in the \emph{data complexity} (when only the ABox is regarded as input) of the \emph{instance checking problem} for concepts constructed from concept names using $\sqcap$, $\sqcup$ and the qualified existential restriction $\exists S.C$ (in other words, $\mathcal{ELUI}$-concepts). The \coNP{} upper bound in Table~\ref{DL-IC}b follows, e.g., from the results on the two-variable FO with counting quantifiers~\cite{DBLP:journals/iandc/Pratt-Hartmann09}, the \coNP{} lower bound for the ontology $\{\top \sqsubseteq A \sqcup B\}$ with a single Krom CI was established by~\citeauthor{DBLP:journals/jiis/Schaerf93}~\citeyear{DBLP:journals/jiis/Schaerf93}; see also Section~\ref{OBDAfor Temporal}. This Krom CI can obviously be expressed by a Krom RI; using the guarded RI $P \sqsubseteq R \sqcup Q$, one can capture the CI~$A \sqsubseteq B \sqcup C$, for which instance checking is also \coNP-hard~\cite<e.g.,>{DBLP:conf/kr/GerasimovaKKPZ20}. The \ACz{} upper bound (FO-rewritablity, to be more precise) for Horn CIs and core RIs is a classical result of~\citeauthor{DBLP:conf/kr/CalvaneseGLLR06}~\citeyear{DBLP:conf/kr/CalvaneseGLLR06,CDLLR07}, which can be readily extended to Horn RIs~\cite{DBLP:conf/kr/KontchakovRWZ20}. The \ACz{} upper bound can also be extended to Krom RIs for atomic queries (which is essentially the consistency problem). On the other hand, for \DL{} ontologies with  (guarded) Boolean RIs, \coNP-hardness already holds for atomic queries (concept names) as we can encode instance queries in the ontology by using the correspondence with $\mathcal{ALC}^{\top,\text{id},\cap,\neg,-}$ mentioned above; see Theorem~\ref{app:th:coNP-gbool} in Appendix~\ref{sec:app-atemporal}.

%*************

\section{Ontology-Mediated Atomic and Positive Instance Queries}
\label{sec:tomq}

We use as our main query language positive temporal concepts and roles.
The most important alternative options are (fragments of) two-sorted first-order logic, and we first briefly motivate our choice. The restriction to positive fragments is standard
in OBDA since the use of negation, implication, or universal quantification hardly yields satisfactory answers to queries. For those connectives, an epistemic semantics appears to be more appropriate; we shall discuss it in Section~\ref{sec:omq}. Using (variations of) \LTL{} to query temporal data has a long tradition, going back more than 30 years~\cite{DBLP:conf/ictl/Chomicki94}. 
In fact, by Kamp's celebrated theorem~\cite{phd-kamp}, propositional temporal logic over discrete (and more general Dedekind complete) linear orders has the same expressive power as monadic first-order logic, and so \LTL{} supplies a user-friendly query language without sacrificing expressivity. An analogous result can be proved for the positive fragment of \LTL~\cite{AIJ21}.
It follows that as far as querying the temporal evolution of a single individual is concerned, we do not lose any expressive power compared to positive monadic first-order logic. 
In contrast, in the 2D case our language does have less expressive power than positive two-sorted first-order logic since already the atemporal part of our queries only captures tree-shaped positive existential queries rather than arbitrary positive existential queries. We believe this restriction is justified as it leads to a user-friendly variable-free query language complementing the language used in the ontology. The extension of our results to a language supporting arbitrary positive existential queries is non-trivial and left for future work.

We now introduce our basic language for querying temporal knowledge bases. It consists of  \emph{positive temporal concepts}, $\varkappa$, and \emph{positive temporal roles}, $\varrho$, that are defined by the following grammar:
\begin{align*}
\varkappa  \quad&::=\quad  \top \quad \mid\quad A_k \quad \mid\quad \exists S.\varkappa \quad \mid \quad \varkappa_1 \sqcap \varkappa_2 \quad \mid \quad \varkappa_1 \sqcup \varkappa_2 \quad \mid \quad \mathop{\boldsymbol{op}_1} \varkappa \quad \mid \quad \varkappa_1 \mathbin{\boldsymbol{op}_2}\varkappa_2,\\
\varrho  \quad&::=\quad S \quad \mid \quad \varrho_1 \sqcap \varrho_2 \quad \mid \quad \varrho_1  \sqcup \varrho_2 \quad \mid \quad \mathop{\boldsymbol{op}_1} \varrho \quad \mid \quad \varrho_1 \mathbin{\boldsymbol{op}_2}  \varrho_2,
\end{align*}
where $\boldsymbol{op}_1\in \{ \Rnext, \Rdiamond, \Rbox,
\Lnext, \Ldiamond, \Lbox\}$ and $\boldsymbol{op}_2\in\{\U,
\Si\}$. Let $\I = (\Delta^\I,\cdot^{\I(n)})$ be an  interpretation.  The \emph{extensions} $\varkappa^{\I(n)}$ of~$\varkappa$ in $\I$, for $n \in \Z$, are determined using~\eqref{eq:semantics:box}--\eqref{eq:semantics:next} and the following:
\begin{align*}
 \top^{\I(n)} & = \Delta^{\I},\qquad\qquad (\exists S.\varkappa)^{\I(n)} = \bigl\{\, u \in \Delta^\I \mid  (u,v) \in S^{\I(n)}, \text{ for some } v\in \varkappa^{\I(n)} \,\bigr\},\hspace*{-20em}\\
  (\varkappa_1 \sqcap \varkappa_2)^{\I(n)} & = \varkappa_1^{\I(n)} \cap \varkappa_2^{\I(n)}, & (\varkappa_1 \sqcup \varkappa_2)^{\I(n)} & = \varkappa_1^{\I(n)} \cup \varkappa_2^{\I(n)},\\
 (\Rdiamond \varkappa)^{\I(n)}  & = \bigcup_{k>n} \varkappa^{\I(k)}, &
(\Ldiamond \varkappa)^{\I(n)}  & = \bigcup_{k<n} \varkappa^{\I(k)},\\
 (\varkappa_1 \U \varkappa_2)^{\I(n)}  & =
  \bigcup_{k>n}\,\bigl(\varkappa_2^{\I(k)}
  \cap \hspace*{-0.3em}\bigcap_{n < m < k} \hspace*{-0.3em}\varkappa_1^{\I(m)}\bigr), &
(\varkappa_1 \Si \varkappa_2)^{\I(n)} & =
  \bigcup_{k<n}\,\bigl(\varkappa_2^{\I(k)}
  \cap \hspace*{-0.3em}\bigcap_{k < m < n} \hspace*{-0.3em}\varkappa_1^{\I(m)}\bigr).
\end{align*}
The definition of $\varrho^{\I(n)}$ is analogous. Note that positive temporal concepts $\varkappa$ and roles $\varrho$ include all temporalised concepts $C$ and roles $R$, respectively ($\exists S$ is a shortcut for $\exists S.\top$).

A $\DL^{\op}_{\frag/\fragr}$ \emph{ontology-mediated positive instance query} (\OMPIQ)
is a pair of the form $\q = (\TO, \varkappa)$ or $\q = (\TO,
\varrho)$, where $\TO$ is a $\DL^{\op}_{\frag/\fragr}$ ontology, $\varkappa$
is a positive temporal concept, and $\varrho$ a positive temporal role
(which can use all temporal operators, not necessarily only those in
$\op$). If $\varkappa$ is a basic concept (i.e., $A$ or $\exists S$) and $\varrho$ a role, then we refer to $\q$ as an \emph{ontology-mediated atomic query} (\OMAQ). The \emph{size}, $|\varkappa|$, $|\varrho|$, and $|\q|$, is the number of occurrences of symbols in $\varkappa$, $\varrho$ and $\q$, respectively. 

A \emph{certain answer} to an \OMPIQ{} $(\TO,\varkappa)$ over an ABox $\A$ is a pair $(a,\ell)\in \ind(\A)\times \tem(\A)$ such that  $a^{\I} \in \varkappa^{\smash{\I(\ell)}}$ for every model $\I$ of
$(\TO, \A)$. A \emph{certain answer} to $(\TO,\varrho)$
over $\A$ is a triple $(a,b,\ell)\in \ind(\A)\times\ind(\A)\times \tem(\A)$ such that $(a^{\I},b^{\I}) \in \varrho^{\smash{\I(\ell)}}$ for every
model~$\I$ of~$(\TO, \A)$.
The set of all certain answers to $\q$ over $\A$ is denoted by $\ans(\q,\A)$.
As a technical tool in our constructions, we also require `certain answers' in which $\ell$ ranges over the whole~$\Z$ rather than only the finite \emph{active temporal domain} $\tem(\A)$; we denote the (possibly infinite) set of such certain answers \emph{over $\A$ and $\Z$} by $\ans^\Z(\q,\A)$.

\begin{example}\label{ex1}\em
Suppose
$\T = \{\,A \sqsubseteq \exists P\,\}$,  $\R = \{\, P \sqsubseteq \Rnext Q \,\}$, and $\varkappa = \exists P. \Rnext \exists Q^-.B$.
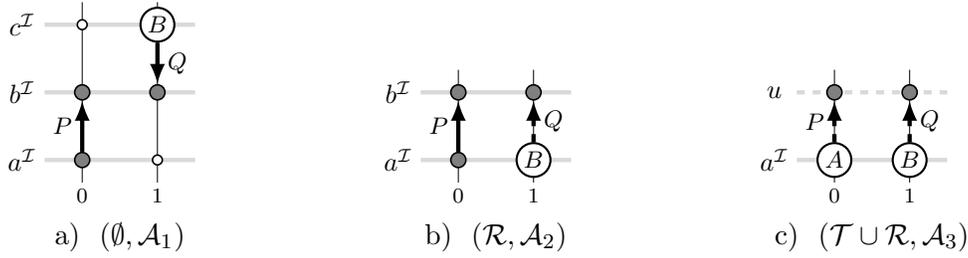
\begin{figure}[t]
\centerline{%
\begin{tikzpicture}[xscale=1, yscale=0.6]\footnotesize
\begin{scope}
\begin{scope}\small
\draw[object-timeline] (-0.5, 0) -- ++(2,0); \node at (-0.8, 0) {$a^\I$};
\draw[object-timeline] (-0.5, 1.5) -- ++(2,0); \node at (-0.8, 1.5) {$b^\I$};
\draw[object-timeline] (-0.5, 3) -- ++(2,0); \node at (-0.8, 3) {$c^\I$};
\end{scope}
\begin{scope}\scriptsize
\draw[time-guideline] (0,-0.5) -- ++(0,4); \node at (0,-0.8) {$0$};
\draw[time-guideline] (1,-0.5) -- ++(0,4); \node at (1,-0.8) {$1$};
\end{scope}
\node[ppoint] (a0) at (0,0) {};
\node[ppoint] (b0) at (0,1.5) {};
\node[point] (c0) at (0,3) {};
\node[point] (a1) at (1,0) {};
\node[ppoint] (b1) at (1,1.5) {};
\node[wpoint] (c1) at (1,3) {$B$};
\begin{scope}[ultra thick]
\draw[->] (a0) to  node[midway,left] {$P$} (b0);
\draw[->] (c1) to  node[midway,right] {$Q$} (b1);
\end{scope}
\node at (0.5, -1.7) {\normalsize a) \ $(\emptyset, \A_1)$};
\end{scope}
\begin{scope}[xshift=50mm]
\begin{scope}\small
\draw[object-timeline] (-0.5, 0) -- ++(2,0); \node at (-0.8, 0) {$a^\I$};
\draw[object-timeline] (-0.5, 1.5) -- ++(2,0); \node at (-0.8, 1.5) {$b^\I$};
\end{scope}
\begin{scope}\scriptsize
\draw[time-guideline] (0,-0.5) -- ++(0,2.5); \node at (0,-0.8) {$0$};
\draw[time-guideline] (1,-0.5) -- ++(0,2.5); \node at (1,-0.8) {$1$};
\end{scope}
\node[ppoint] (a0) at (0,0) {};
\node[ppoint] (b0) at (0,1.5) {};
\node[wpoint] (a1) at (1,0) {$B$};
\node[ppoint] (b1) at (1,1.5) {};
\begin{scope}[ultra thick]
\draw[->] (a0) to  node[midway,left] {$P$} (b0);
\draw[->,dashed] (a1) to  node[midway,right] {$Q$} (b1);
\end{scope}
\node at (0.5, -1.7) {\normalsize b) \ $(\R, \A_2)$};
\end{scope}
\begin{scope}[xshift=100mm]
\begin{scope}\small
\draw[object-timeline] (-0.5, 0) -- ++(2,0); \node at (-0.8, 0) {$a^\I$};
\draw[object-timeline,dashed] (-0.5, 1.5) -- ++(2,0); \node at (-0.8, 1.5) {$u$};
\end{scope}
\begin{scope}\scriptsize
\draw[time-guideline] (0,-0.5) -- ++(0,2.5); \node at (0,-0.8) {$0$};
\draw[time-guideline] (1,-0.5) -- ++(0,2.5); \node at (1,-0.8) {$1$};
\end{scope}
\node[wpoint] (a0) at (0,0) {$A$};
\node[ppoint] (b0) at (0,1.5) {};
\node[wpoint] (a1) at (1,0) {$B$};
\node[ppoint] (b1) at (1,1.5) {};
\begin{scope}[ultra thick]
\draw[->,dashed] (a0) to  node[midway,left] {$P$} (b0);
\draw[->,dashed] (a1) to  node[midway,right] {$Q$} (b1);
\end{scope}
\node at (0.5, -1.7) {\normalsize c) \ $(\T\cup \R, \A_3)$};
\end{scope}
\end{tikzpicture}
}%
\caption{Fragments of models in Example~\ref{ex1}.}\label{fig:ex1}
\end{figure}%
For $\A_1 = \{ P(a,b,0), Q(c, b, 1), B(c,1)\}$ and $\q_1 = (\emptyset, \varkappa)$, we have $\ans(\q_1,\A_1) = \{(a,0)\}$ because $a^\I \in \varkappa^{\I(0)}$ for any model $\I$ of $\A_1$; see Fig.~\ref{fig:ex1}a.
For $\A_2 = \{ P(a,b,0), B(a, 1)\}$ and $\q_2 = (\R, \varkappa)$, we also have $\ans(\q_2,\A_2) = \{(a,0)\}$ because $(a^\I, b^\I) \in Q^{\I(1)}$, for every model~$\I$ of $\A_2$ and~$\R$; see Fig.~\ref{fig:ex1}b.
For $\A_3 = \{ A(a,0), B(a, 1)\}$ and $\q_3 = (\T \cup \R, \varkappa)$, we again have $\ans(\q_3,\A_3) = \{(a,0)\}$ because, in every model $\I$ of $(\T \cup \R, \A_3)$,  there exists $u \in \Delta^\I$ with $(a^\I, u) \in P^{\I(0)}$ and $(a^\I, u) \in Q^{\I(1)}$; see Fig.~\ref{fig:ex1}c.

Consider  $\varrho = P \sqcap \Rnext Q$ with $\T$ and $\R$ as above.
For $\A_4 = \{ A(a,0)\}$ and $\q_4 = (\T \cup \R, \varrho)$, we obviously have $\ans(\q_4,\A_4) = \{(a,b,0)\}$, while $\ans(\q_5,\A_4) = \emptyset$ for $\q_5 = (\T \cup \T', \varrho)$, where  $\T' = \{\, \exists P^- \sqsubseteq \Rnext \exists Q^-\,\}$.
\end{example}

By the \emph{\OMPIQ{} answering problem for} $\DL^{\op}_{\smash{\frag/\fragr}}$ we understand the decision problem for the set $\ans(\q,\A)$, where $\q$ is a $\DL^{\op}_{\smash{\frag/\fragr}}$ \OMPIQ{} and $\A$ an ABox.
In the context of OBDA, we are usually interested in the \emph{data complexity} of this problem when $\q$ is considered to be fixed or negligibly small compared to data $\A$, which is regarded as the only input to the problem~\cite{Vardi82}. In the atemporal case, the data complexity of answering  conjunctive queries mediated by ontologies from the \DL{} family is well understood~\cite{CDLLR07,ACKZ09}: it ranges from \ACz{}---which guarantees FO-rewritability (see below)---to \PTime{} and further to \coNP{} (see also Table~\ref{DL-sat})\footnote{The fine-grained combined complexity of  answering conjunctive queries mediated by \OWLQL{} ontologies was investigated by~\citeauthor{DBLP:journals/jacm/BienvenuKKPZ18}~\citeyear{DBLP:journals/jacm/BienvenuKKPZ18}.}\!.
The data complexity of answering \LTL{} OMQs is either \ACz{} or \NCo~\cite{AIJ21}; we remind the reader that
\begin{equation*}
\ACz ~\subsetneq~ \NCo ~\subseteq~ \LogSpace ~\subseteq~ \PTime ~\subseteq~ \coNP.
\end{equation*}
As our aim here is to identify families of OMPIQs answering which can be done in \ACz{} or \NCo, we now look at these two complexity classes in more detail.

Let $\A$ be an ABox with $\ind(\A) = \{a_0,\dots,a_m\}$. Without loss of generality, we always assume that $\max \A \ge m$ (if this is not so, we simply add a `dummy' $D(a_0,m)$ to $\A$).
We represent $\A$ as a first-order structure $\SA$ with domain $\tem(\A)$ ordered by $<$ such that
\begin{equation*}
\SA \models A(k,\ell) \quad \text{iff} \quad  A(a_k,\ell)\in \A \qquad \text{and} \qquad
\SA \models P(k,k',\ell) \quad \text{iff} \quad P(a_k,a_{k'},\ell)\in \A,
\end{equation*}
for any concept and role names $A$, $P$ and any $k,k',\ell\in\tem(\A)$. To simplify notation, we often identify an individual name $a_k \in \ind(\A)$ with its numerical representation $k \in \tem(\A)$.
As a technical tool in our constructions, we also use infinite first-order structures $\mathfrak{S}_\A^{\Z}$ with domain $\Z$ that are defined in the same way as $\SA$ but over the whole $\Z$.

The structure $\SA$ represents a temporal database over which we can evaluate first-order formulas (queries) with data atoms of the form $A(x,t)$ and $P(x,y,t)$ as well as atoms with the order $t_1 < t_2$  and congruence $t \equiv 0\pmod n$ predicates, for $n > 1$.
It is well-known~\cite{Immerman99} that the evaluation problem for $\FO$-formulas with arbitrary numerical predicates is in non-uniform $\smash{\ACz{}}$ for data complexity, the class of languages computable by bounded-depth polynomial-size circuits with unary \textsc{not}-gates and unbounded fan-in \textsc{and}- and \textsc{or}-gates. In this article, we require \emph{$\FO(<)$-formulas}, which only use data and order atoms, and \mbox{$\FOE$}-\emph{formulas}, which in addition may contain congruence atoms of the form $t \equiv 0 \pmod n$.
Evaluation of $\FOE$-formulas is known to be in \LogTime-uniform $\smash{\ACz{}}$ for data complexity~\cite{Immerman99}.
We also use $\FO(\RPR)$-\emph{formulas}, that is, $\FO$-formulas extended with \emph{relational primitive recursion} (\RPR), whose evaluation is in~\NCo{} for data complexity~\cite{DBLP:journals/iandc/ComptonL90}, the class computed by a family of polynomial-size logarithmic-depth circuits with gates of at most two inputs.
We remind the reader that, using \RPR, we can construct formulas $\Phi(\avec{z},\avec{z}_1,\dots,\avec{z}_n)$ such as
\begin{equation*}
\left[ \begin{array}{l}
Q_{1}(\avec{z}_1,t) \equiv \varTheta_1\big(\avec{z}_1,t,Q_1(\avec{z}_1,t-1),\dots,Q_n(\avec{z}_n,t-1)\big)\\
\dots\\
Q_{n}(\avec{z}_n,t) \equiv \varTheta_n\big(\avec{z}_n,t,Q_1(\avec{z}_1,t-1),\dots,Q_n(\avec{z}_n,t-1)\big)
\end{array}\right] \ \Psi(\avec{z},\avec{z}_1,\dots,\avec{z}_n),
\end{equation*}
where the part of $\Phi$ within $[\dots]$ defines recursively, via the $\FO(\RPR)$-formulas $\varTheta_i$, the interpretations of the predicates~$Q_i$ in the $\FO(\RPR)$-formula $\Psi$.
The recursion starts at $t=0$ assuming that $Q_{i}(\avec{z}_i,-1)$ is false for all $Q_{i}$ and $\avec{z}_i$, $1 \leq i\leq n$. Thus, the truth value of $Q_{i}(\avec{z}_i,0)$ is computed by substituting falsehood $\bot$ for all $Q_{i}(\avec{z}_i,-1)$.
We allow the relation variables $Q_i$ to occur in only one recursive definition $[\dots]$, so it makes sense to write $\SA\models Q_i(\avec{n}_i,k)$, for any tuple $\avec{n}_i$ in $\tem(\A)$ and $k \in \tem(\A)$, if the computed value is `true'.
Using thus defined truth-values, we compute inductively the truth-relation $\SA\models \Psi(\avec{n},\avec{n}_1,\dots,\avec{n}_n)$, and so $\SA\models \Phi(\avec{n},\avec{n}_1,\dots,\avec{n}_n)$, as usual in first-order logic. It is to be noted that $\FO(\RPR)$ offers a rather limited form of recursion that can only infer  a predicate at $t$ if certain predicates hold at $t-1$, which puts $\RPR$ near the bottom in the hierarchy of fixed-point operations, below the (deterministic) transitive closure, least fixed-point, etc.~\cite{Immerman99}.

\begin{definition}\label{rewriting}\em
Let $\lang$ be one of the three classes of FO-formulas introduced above. A constant-free $\lang$-formula $\rew(x,t)$ is said to be an $\lang$-\emph{rewriting of} an \OMPIQ{} $\q=(\TO, \varkappa)$ if $\ans(\q,\A) = \bigl\{\, (a,\ell) \in \ind(\A)\times \tem(\A) \mid \SA \models \rew(a,\ell)\,\bigr\}$, for any ABox~$\A$.
Similarly, a constant-free $\lang$-formula $\rew(x,y,t)$ is an $\lang$-\emph{rewriting of} an \OMPIQ{} $\q = (\TO,\varrho)$ if we have $\ans(\q,\A) = \bigl\{\,(a,b,\ell) \in \ind(\A)\times \ind(\A)\times\tem(\A)\mid \SA \models \rew(a,b,\ell)\,\bigr\}$, for any ABox~$\A$.
\end{definition}

It follows from the definition that answering $\FO(<)$- and $\FOE$-rewritable \OMPIQ{}s is in $\ACz$ for data complexity, and answering $\FO(\RPR)$-rewritable \OMPIQ{}s is in \NCo{}. We illustrate these three types of FO-rewriting by instructive examples.

\begin{example}\label{exampleBvNB}\em
Consider the $\DL\Xnext_{\core}$  OMPIQs $\q = (\T \cup \R, Q)$ and $\q' = (\T \cup \R, \varkappa)$, where $\T$, $\R$ and $\varkappa$ are as in Example~\ref{ex1}. It is readily seen that $\rew(x,y,t) = Q(x,y,t) \lor P(x,y, t-1)$ is an $\FO(<)$-rewriting of $\q$, where $P(x,y, t-1)$ abbreviates
\begin{equation*}
\exists t'\,\bigl((t' < t) \land \neg\exists s\, \bigl((t' < s) \land (s < t)\bigr)\land P(x,y, t') \bigr);
\end{equation*}
consult~\citeA[Remark~3]{AIJ21}.\ The ABoxes $\A_1$--$\A_3$ from Example~\ref{ex1} show three sets of atoms at least one of which  must be present in an ABox $\A$ for $(a,0)$ to be a certain answer to~$\q'$ over $\A$; see Fig.~\ref{fig:ex1}.  This observation implies that an $\FO(<)$-rewriting of $\q'$ can be defined as $\rew'(x,t) = \rew_1(x,t) \lor \rew_2(x,t) \lor \rew_3(x,t)$, where
\begin{align*}
\rew_1(x,t) ~=~& \exists y, z \,\bigl( P(x,y,t) \land Q(z,y,t+1) \land B(z, t+1) \bigr),\\
\rew_2(x,t) ~=~& \exists y \, \bigl( P(x,y,t) \land B(x, t+1) \bigr),\\
\rew_3(x,t) ~=~& A(x,t) \land B(x, t+1),
\end{align*}
and $B(x, t+1)$ and $B(z, t+1) $ are shortcuts similar to  $P(x,y, t-1)$ above.
\end{example}

\begin{example}\label{ex:4}\em
Consider the $\DL\Xbox_{\horn/\core}$ OMPIQ $\q = (\T \cup \R, A \sqcap \exists S)$ with
\begin{equation*}
\T = \{\, \exists R \sqsubseteq \exists S \,\}, \quad \R = \{\, S \sqsubseteq \Lbox T, \ S \sqsubseteq \Rbox \Rbox \Rbox P, \ \Lbox \Lbox \Lbox T \sqcap \Rbox P \sqsubseteq R \,\}.
\end{equation*}
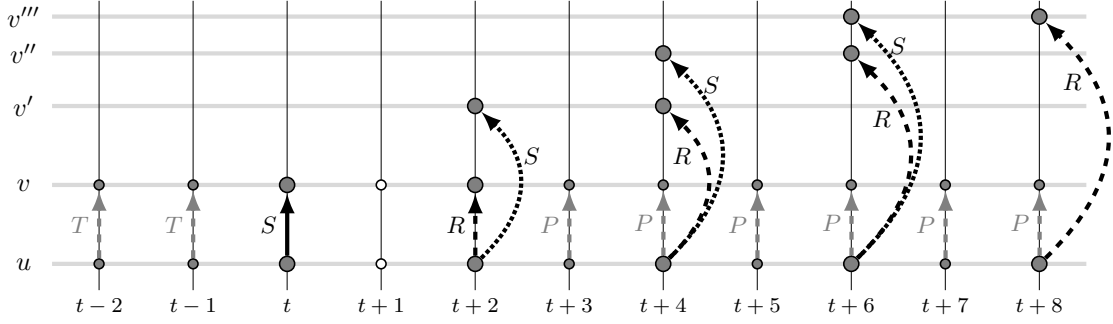
\begin{figure}[t]
\centerline{%
\begin{tikzpicture}[xscale=1.25, yscale=0.7]\footnotesize
\begin{scope}\small
\draw[object-timeline] (0.5, 0) -- ++(11,0); \node at (0.2, 0) {$u$};
\draw[object-timeline] (0.5, 1.5) -- ++(11,0); \node at (0.2, 1.5) {$v$};
\draw[object-timeline] (0.5, 3) -- ++(11,0); \node at (0.2, 3) {$v'$};
\draw[object-timeline] (0.5, 4) -- ++(11,0); \node at (0.2, 4) {$v''$};
\draw[object-timeline] (0.5, 4.7) -- ++(11,0); \node at (0.2, 4.7) {$v'''$};
\end{scope}
\begin{scope}\scriptsize
\draw[time-guideline] (1,-0.5) -- ++(0,5.5); \node at (1,-0.8) {$t-2$};
\draw[time-guideline] (2,-0.5) -- ++(0,5.5); \node at (2,-0.8) {$t-1$};
\draw[time-guideline] (3,-0.5) -- ++(0,5.5); \node at (3,-0.8) {$t$};
\draw[time-guideline] (4,-0.5) -- ++(0,5.5); \node at (4,-0.8) {$t+1$};
\draw[time-guideline] (5,-0.5) -- ++(0,5.5); \node at (5,-0.8) {$t+2$};
\draw[time-guideline] (6,-0.5) -- ++(0,5.5); \node at (6,-0.8) {$t+3$};
\draw[time-guideline] (7,-0.5) -- ++(0,5.5); \node at (7,-0.8) {$t+4$};
\draw[time-guideline] (8,-0.5) -- ++(0,5.5); \node at (8,-0.8) {$t+5$};
\draw[time-guideline] (9,-0.5) -- ++(0,5.5); \node at (9,-0.8) {$t+6$};
\draw[time-guideline] (10,-0.5) -- ++(0,5.5); \node at (10,-0.8) {$t+7$};
\draw[time-guideline] (11,-0.5) -- ++(0,5.5); \node at (11,-0.8) {$t+8$};
\end{scope}
\node[spoint] (a1) at (1,0) {};
\node[spoint] (b1) at (1,1.5) {};
\node[spoint] (a2) at (2,0) {};
\node[spoint] (b2) at (2,1.5) {};
\node[ppoint] (a3) at (3,0) {};
\node[ppoint] (b3) at (3,1.5) {};
\node[point] (a4) at (4,0) {};
\node[point] (b4) at (4,1.5) {};
\node[ppoint] (a5) at (5,0) {};
\node[ppoint] (b5) at (5,1.5) {};
\node[ppoint] (c5) at (5,3) {};
\node[spoint] (a6) at (6,0) {};
\node[spoint] (b6) at (6,1.5) {};
\node[ppoint] (a7) at (7,0) {};
\node[spoint] (b7) at (7,1.5) {};
\node[ppoint] (c7) at (7,3) {};
\node[ppoint] (d7) at (7,4) {};
\node[spoint] (a8) at (8,0) {};
\node[spoint] (b8) at (8,1.5) {};
\node[ppoint] (a9) at (9,0) {};
\node[spoint] (b9) at (9,1.5) {};
\node[ppoint] (d9) at (9,4) {};
\node[ppoint] (e9) at (9,4.7) {};
\node[spoint] (a10) at (10,0) {};
\node[spoint] (b10) at (10,1.5) {};
\node[ppoint] (a11) at (11,0) {};
\node[spoint] (b11) at (11,1.5) {};
\node[ppoint] (e11) at (11,4.7) {};
\begin{scope}[ultra thick]
\draw[->] (a3) to  node[midway,left] {$S$} (b3);
\draw[->,dashed,gray] (a2) to  node[midway,left] {$T$} (b2);
\draw[->,dashed,gray] (a1) to  node[midway,left] {$T$} (b1);
%\draw[->,dashed,gray] (a0) to  node[midway,left] {$T$} (b0);
\draw[->,dashed,gray] (a6) to  node[midway,left] {$P$} (b6);
\draw[->,dashed,gray] (a7) to  node[midway,left] {$P$} (b7);
\draw[->,dashed,gray] (a8) to  node[midway,left] {$P$} (b8);
\draw[->,dashed,gray] (a9) to  node[midway,left] {$P$} (b9);
\draw[->,dashed,gray] (a10) to  node[midway,left] {$P$} (b10);
\draw[->,dashed,gray] (a11) to  node[midway,left] {$P$} (b11);
\draw[->,dashed] (a5) to  node[midway,left] {$R$} (b5);
\draw[->,densely dotted,bend right] (a5) to  node[pos=0.7,right] {$S$} (c5);
\draw[->,dashed,bend right] (a7) to  node[pos=0.7,left] {$R$} (c7);
\draw[->,densely dotted,bend right] (a7) to  node[pos=0.87,right] {$S$} (d7);
\draw[->,dashed,bend right] (a9) to  node[pos=0.7,left] {$R$} (d9);
\draw[->,densely dotted,bend right] (a9) to  node[pos=0.9,right] {$S$} (e9);
\draw[->,dashed,bend right] (a11) to  node[pos=0.75,left] {$R$} (e11);
\end{scope}
\end{tikzpicture}
}%
\caption{Structure of models in Example~\ref{ex:4} (roles $T$ and $P$ are shown only for $(u,v)$).}\label{fig:ex4}
\end{figure}%
We first observe that $\R \models S \sqsubseteq \Rnext \Rnext R$, and so $\T \cup \R \models \exists S \sqsubseteq \Rnext \Rnext \exists S$, although we do not have $\T \cup \R \models S \sqsubseteq \Rnext \Rnext S$; see Fig.~\ref{fig:ex4}. Now, for every ABox $\A_n = \{\,S(a,b,0), A(a,n) \,\}$, $n \geq 0$, we have $\ans(\q,\A_n) = \{(a,n)\}$ iff $n$ is even; otherwise, $\ans(\q,\A_n) = \emptyset$.
This suggests the following $\FOE$-rewriting of $\q$:
\begin{equation*}
\rew(x,t) ~=~ A(x,t) \land \exists y,s\,\bigl(\bigl[R(x,y,s) \lor S(x,y,s)\bigr] \land (t - s \in 0 + 2\N)\bigr),
\end{equation*}
where $t - s \in 0 + 2\N$ is an $\FOE$-formula $\varphi(t,s)$ such that $\SA \models \varphi(n,m)$ iff $n - m \in 0 + 2\N$ and $p + q\N$ is the set $\{\, p + qk \mid k \geq 0\,\}$, for $p, q \geq 0$~\cite[Remark~3]{AIJ21}.
Note, however, that $\q$ is \emph{not} $\FO(<)$-rewritable since properties such as `the size of the domain is even' are not definable by $\FO(<)$-formulas, which can be established using a standard Ehrenfeucht-Fra\"iss\'e argument~\cite{Straubing94,Libkin}. On the other hand, all $\LTL\Xbox_{\horn}$ OMPIQs are $\FO(<)$-rewritable~\cite{AIJ21}.
\end{example}

\begin{example}\label{exampleBvNB:2}\em
As shown by~\citeA[Example 5]{AIJ21}, the \OMAQ{} $\q = (\TO,B_0)$ with
\begin{equation*}
\TO ~=~ \{\, \Lnext B_k \sqcap A_0 \sqsubseteq B_k, \ \Lnext B_{1-k} \sqcap A_1 \sqsubseteq B_k \mid k = 0, 1 \,\}
\end{equation*}
encodes \parity{} in the following sense. For any binary word  $\avec{e} = (e_1,\dots,e_{n}) \in \{0,1\}^n$, let
$\A_{\avec{e}}  =   \{\,B_0(a, 0)\,\} \cup  \{\, A_{e_i}(a, i) \mid 0 < i \leq n\,\}$.
Then $(a,n)$ is a certain answer to $\q$  over $\A_{\avec{e}}$ iff the number of 1s in~$\avec{e}$ is even. As \parity{} is not in $\ACz$~\cite{DBLP:journals/mst/FurstSS84}, $\q$ is not rewritable to any FO-formula with numerical predicates. An $\FO(\RPR)$-rewriting of $\q$ inspired by~\citeA[Example 5]{AIJ21} is shown below:
\begin{equation*}
\rew(x,t) ~=~  \left[ \begin{array}{l}
Q_{0}(x,t) \equiv \varTheta_0(x,t, Q_{0}(x,t-1), Q_{1}(x,t-1))\\
Q_{1}(x,t) \equiv \varTheta_1(x,t, Q_{0}(x,t-1), Q_{1}(x,t-1))
\end{array}\right] \ Q_0(x,t),
\end{equation*}
where, for $k = 0,1$, the formula $\varTheta_k(x,t, Q_{0}(x,t-1), Q_{1}(x,t-1))$ is
\begin{equation*}
 B_k(x,t) \lor \bigl(Q_k(x,t-1)\land A_0(x,t)\bigr) \lor \bigl(Q_{1-k}(x,t-1)\land A_1(x,t)\bigr).
\end{equation*}
\end{example}

We obtain our FO-rewritability results for temporal \DL{} \OMPIQ{}s in three steps. First, we reduce FO-rewritability of \OMAQ{}s to FO-rewritability of $\bot$-free  \OMPIQ{}s of some restricted form (whose ontology is consistent with all ABoxes). Then, we further reduce, where possible, FO-rewritability of those $\bot$-free restricted \OMPIQ{}s to FO-rewritability of \LTL{} \OMAQ{}s, which has been thoroughly investigated by~\citeauthor{AIJ21}~\citeyear{AIJ21}. This `projection' of a 2D formalism onto one of its `axes' relies on the fact that atemporal \DL{} ontologies with Horn or Krom RIs can be encoded in the one-variable fragment of first-order logic~\cite{ACKZ09}. The third step reduces FO-rewritability of some $\DL_{\horn}\Xallop{}$ \OMPIQ{}s to FO-rewritability of \OMAQ{}s using model-theoretic considerations.

In the remainder of Section~\ref{sec:tomq}, we first briefly remind the reader of \LTL{} \OMAQ{}s and \OMPIQ{}s~\cite{AIJ21} and make two basic observations on the reducibility of answering $\DL^{\op}_{\smash{\frag/\fragr}}$ \OMAQ{}s that do not contain interacting concepts and roles to answering $\LTL^{\op}_{\smash{\frag}}$ and $\LTL^{\op}_{\smash{\fragr}}$ \OMAQ{}s.

%*********

\subsection{\LTL{} OMAQs and OMPIQs}\label{sec:ltl-can-period}

$\LTL^{\op}_{\frag}$ \emph{\OMAQ{}s} and \emph{\OMPIQ{}s} can be defined as $\DL^{\op}_{\frag}$ \OMAQ{}s and \OMPIQ{}s of the form $\q = (\TO, A)$ and, respectively, $\q = (\TO, \varkappa)$ that contain no occurrences of roles. In this case, $\TO$ is referred to as an $\LTL^{\op}_{\frag}$ \emph{ontology}. We assume (without loss of generality) that an $\LTL$ \emph{ABox} $\A$ has \emph{one fixed} individual name, $a$, and consists of assertions of the form $C(a,\ell)$ with a role-free temporalised concept $C$. In an $\LTL$ interpretation $\I$, we have $\Delta^\I = \{ a \}$ and $P_i^{\I(n)} = \emptyset$ for all role names $P_i$ and $n \in \Z$.
To simplify notation, we write $\I\models C(n)$ instead of $a^\I \in C^{\I(n)}$ and  $C(\ell) \in \A$ instead of $C(a,\ell) \in \A$.

The key property of $\LTL_\horn\Xallop$ we need is that every \emph{consistent} $\LTL_\horn\Xallop$ KB $(\TO, \A)$ has a \emph{canonical model} $\C_{\TO, \A}$, which can be defined as the intersection of all of its models:
\begin{equation}\label{intersection}
\C_{\TO, \A} \models A_i(n) \quad \text{ iff }\quad \I\models A_i(n), \text{ for all models } \I \text{ with } \I\models (\TO,\A),
\end{equation}
for every concept name $A_i$~\cite{AKRZ:LPAR13}. (An alternative, syntactic, definition of the canonical models is given in Section~\ref{sec:canonical-model}.)

Denote by $\subo$ the set of all subconcepts occurring in $\TO$ and their negations. By~$\type_{\TO,\A}(n)$ we denote the set of all $C\in\subo$ with $\can \models C(n)$, where we write $\can\models \neg C'(n)$ iff $\can\not\models C'(n)$.  It is known that every satisfiable $\LTL$-formula is satisfied in an \emph{ultimately periodic model}~\cite{DBLP:journals/jacm/SistlaC85}. Since $\C_{\TO,\A}$ is the minimal model in the sense of~\eqref{intersection},  it is also ultimately periodic. This is formalised in the next lemma, which follows immediately from Lemmas~19 and~21 by~\citeA{AIJ21}. The periodic structure will be used for defining our FO-rewritings in the sequel.
\begin{lemma}\label{period:A}
$(i)$ For any consistent $\LTL_{\smash{\horn}}\Xallop$ KB $(\TO, \A)$ with $\{ C \mid C(\ell) \in \A \} \subseteq \subo$, there are positive integers $s_{\TO,\A} \le 2^{|\TO|}$ and $p_{\TO,\A} \le 2^{2|\TO|}$ such that
\begin{align}\label{eq:period:A}
\begin{array}{rl}
&\type_{\TO,\A}(n) = \type_{\TO,\A}(n - p_{\TO,\A}), \text{ for } n \le \min \A -s_{\TO,\A}, \\[3pt]
&\type_{\TO,\A}(n) = \type_{\TO,\A}(n + p_{\TO,\A}), \text{ for }n \ge \max \A + s_{\TO,\A}.
\end{array}
\end{align}
If $\TO$ is an $\LTL_\horn\Xbox$ ontology, one can take $s_{\TO,\A} \le |\TO|$ and  $p_{\TO,\A} =1$.

$(ii)$ For any $\LTL_{\smash{\horn}}\Xallop$ ontology $\TO$, there are positive integers $s_{\TO} \le 2^{|\TO|}$ and $p_{\TO} \le 2^{2|\TO|\cdot 2^{|\TO|}}$ such that, for any ABox  $\A$ consistent with $\TO$ and with $\{ C \mid C(\ell) \in \A \} \subseteq \subo$, we have
\begin{align}\label{eq:period}
\begin{array}{rl}
&\type_{\TO,\A}(n) = \type_{\TO,\A}(n - p_{\TO}),  \text{ for } n \le \min \A -s_{\TO}, \\[3pt]
&\type_{\TO,\A}(n) = \type_{\TO,\A}(n + p_{\TO}),  \text{ for }n \ge \max \A + s_{\TO}.
\end{array}
\end{align}
If $\TO$ is an $\LTL_\horn\Xbox$ ontology, one can take $s_{\TO} \le |\TO|$ and  $p_{\TO} =1$.
\end{lemma}

%*********

\subsection{Projecting Temporal $\DL$ OMAQs to \LTL: Initial Observations}\label{sec:proj}

Consider a $\DL^{\op}_{\smash{\frag/\fragr}}$ \OMAQ{} $\q = (\T, B_0)$ with a basic concept $B_0$ such that the ontology does not have role inclusions. 
Define an $\LTL^{\op}_{\frag}$ OMAQ $\q^{\smash{\dagger}} = (\T,B_0)^{\smash{\dagger}}$ as follows. For basic concepts $A$ and $\exists S$, set
\begin{equation*}
A^\dagger = A  \qquad \text{ and }\qquad (\exists S)^\dagger = E_S,
\end{equation*}
where $E_S$ is a fresh concept name, the \emph{surrogate} of $\exists S$. The TBox $\T^\dagger$ of $(\T,B_0)^{\smash{\dagger}}$ is obtained from~$\T$ by replacing every basic concept~$B$ with~$B^\dagger$;
the \LTL-translation $C^\dag$ of a temporalised concept $C$ is defined analogously. Given an ABox $\A$ and $a\in\ind(\A)$, denote by~$\A_a^\dagger$ an \LTL{} ABox that consists of all ground atoms $A(\ell)$ for $A(a,\ell)\in \A$ and $E_S(\ell)$
for $S(a,b,\ell)\in\A$ (remembering that $S^-(b,a,\ell)\in\A$ iff $S(a,b,\ell)\in\A$). 

Suppose now that the TBox~$\T$ does not contain $\bot$ in the sense that all inclusions have at least one temporalised concept on the right-hand side; such $\T$ is called \emph{$\bot$-free}.
We claim that
\begin{equation}\label{observ1}
\ans^{\Z}(\q,\A) ~=~ \bigl\{\, (a,n) \mid a\in\ind(\A) \text{ and } n \in \ans^{\Z}(\q^\dagger,\A^\dagger_{\smash{a}}) \,\bigr\}.
\end{equation}
Indeed, it suffices to show that, for any $a \in \ind(\A)$ and $n \in \Z$, there is a model $\I$ of $(\T,\A)$ with $a \notin B_0^{\I(n)}$ iff there is a model $\I_a$ of $(\T^\dagger,\A_a^\dagger)$ with $\I_a\not\models B_0^\dag(n)$. Direction $(\Rightarrow)$ is easy: every model $\I$ of $(\T,\A)$ gives rise to the \LTL-interpretations~$\I_b$, for $b\in\ind(\A)$, defined by taking $\I_b\models A(n)$ iff $b \in A^{\I(n)}$, and $\I_b \models  E_S(n)$ iff $b \in (\exists S)^{\I(n)}$. It is readily seen by induction that, for any temporalised concept $C$ and any $n \in \Z$, we have $b \in C^{\I(n)}$ iff $\I_b \models C^\dag(n)$, and so the $\I_b$ are models of $(\T^\dagger,\A_b^\dagger)$.

For $(\Leftarrow)$, a bit more work is needed. For each $b\in\ind(\A)$ different from the given $a$, we take any model $\I_b$ of $(\T^\dagger,\A_b^\dagger)$. Further, for any role $S$ and any $n \in \Z$, we take a model $\I_{S,n}$ of $(\T, \{\exists S^-(w_S,n)\})$, for some fresh individual name $w_S$. These models exist because $\T$ is $\bot$-free. Finally, define an interpretation $\I$ by taking, for all $b,c \in \ind(\A)$ and $n \in \Z$: $b \in A^{\I(n)}$ iff $\I_b\models A(n)$, $(b,c) \in S^{\I(n)}$ iff $S(b,c,n) \in \A$, and $(b,w_S) \in S^{\I(n)}$ iff $\I_b\models E_S(n)$. Note that $b \in (\exists S)^{\I(n)}$ iff $\I_n\models  E_S(n)$. Using this fact, it is not hard to show by induction that, for any temporalised concept $C$ and any $n \in \Z$, we have $b \in C^{\I(n)}$ iff $\I_b\models C^\dag(n)$. It follows that $\I$ is the model of $(\T,\A)$ as required.

Equality~\eqref{observ1} gives us the following:
\begin{proposition}\label{prop:conceptOMAQs}
Let $\lang$ be $\FO(<)$, $\FOE$, or $\FO(\RPR)$. A $\bot$-free $\DL^{\op}_{\smash{\frag/\fragr}}$ \OMAQ{} $(\T, B)$ with a basic concept $B$ is $\lang$-rewritable if the $\LTL^{\op}_{\frag}$ \OMAQ{} $(\T,B)^\dagger$ is $\lang$-rewritable.
\end{proposition}
\begin{proof}
We obtain an $\lang$-rewriting of $(\T,B)$ from an $\lang$-rewriting $\rew^\dagger(t)$ of $(\T,B)^\dagger$ by replacing every atom of the form~$A(s)$ with $A(x,s)$, every $E_P(s)$ with $\exists y\, P(x,y,s)$ and every $E_{P^-}(s)$ with $\exists y\, P(y,x,s)$. In the case of $\FO(\RPR)$, we additionally replace every $Q(t_1,\dots, t_k)$, for a relation variable $Q$, with $Q(x,t_1,\dots, t_k)$.
\end{proof} 

Consider now a $\smash{\DL^{\op}_{\frag/\fragr}}$ \OMPIQ{} $\q = (\TO, \varrho)$ with $\TO = \T\cup \R$ and a positive temporal role~$\varrho$. 
For every role name~$P$, we introduce two concepts names~$A_P$ and $A_{P^-}$. Let $\R^\ddag$ and~$\varrho^\ddag$ be the results of replacing each role $S$ in them  with $A_S$, and let $\q^\ddag =  (\TO,\varrho)^\ddag = (\R^\ddagger,\varrho^\ddag)$. 
Given an ABox~$\A$ and $a,b\in \ind(\A)$, denote by $\A^\ddagger_{\smash{a,b}}$ the \LTL{} ABox with  the atoms $A_S(\ell)$ such that $S(a,b,\ell) \in \A$; by the assumption made in Section~\ref{sec:tdl}, $\A^\ddagger_{\smash{b,a}}$ will contain $A_{S^-}(\ell)$. 

First, observe that ontologies without TBoxes have the following important property: 
\begin{proposition}\label{prop:RBox:consistency}
A $\DL_{\smash{\frag/\fragr}}^{\op}$ KB $(\R, \A)$ is consistent iff the $\LTL_{\fragr}^{\op}$ 
KB $(\R^\ddagger,\A_{a,b}^\ddagger)$ is consistent, for every $a,b\in\ind(\A)$. 
\end{proposition}
\begin{proof}
Every interpretation~$\I$ with domain $\ind(\A)$ and empty concept names is the `sum' of \LTL{} interpretations~$\I_{a,b}$ for the concepts $S^\ddagger$, with $a,b\in\ind(\A)$: if the $\I_{a,b}$ and $\I_{b,a}$ agree on role inverses, then we take $P^{\I(n)} = \{ (a,b) \mid \I_{a,b} \models P^\ddagger(n) \}$, for each role name $P$. It remains to observe that every consistent $(\R, \A)$ has a model with domain $\ind(\A)$ and the empty interpretation of concept names.
\end{proof}

Suppose now that $\q = (\TO, \varrho)$  is $\bot$-free.  We claim that
\begin{equation}\label{observ2}
\ans^{\Z}(\q,\A) = \bigl\{\, (a,b, n) \mid a,b\in\ind(\A) \text{ and } n \in \ans^{\Z}(\q^\ddagger,\A^\ddagger_{\smash{a,b}}) \,\bigr\}.
\end{equation}
Indeed, since  $\R$ is $\bot$-free and so trivially consistent, models of $(\R,\A)$ are `sums' of models of $(\R^\ddagger,\A^\ddagger_{\smash{a,b}})$; see Proposition~\ref{prop:RBox:consistency}. As $\T$ is $\bot$-free, the TBox does not affect $\ans^{\Z}(\q,\A)$. Indeed, every model of $(\TO, \A)$  is a model of $(\R, \A)$, and conversely, every model $\I$ of $(\R,\A)$ gives rise to a model $\J$ of $(\TO, \A)$ that has the same interpretation of roles as $\I$ on the ABox individuals. More precisely, the model $\J$ is obtained by extending the domain of $\I$ with a special object $w$ and interpreting all concept names by the whole domain; the interpretation of roles on the ABox is inherited from $\I$; in addition, the $(a^\J, w)$, for individual names $a$, and $(w,w)$ belong to all roles at all time instants. That $\J$ is a `saturated' model of  $(\TO, \A)$ follows from the assumption that $\TO$ is $\bot$-free.
Thus, we obtain the following transfer result:
\begin{proposition}\label{prop:roleOMIQs}
Let $\lang$ be $\FO(<)$, $\FOE$ or $\FO(\RPR)$.
A $\bot$-free $\DL^{\op}_{\smash{\frag/\fragr}}$ \OMPIQ{} $(\TO, \varrho)$ with a positive temporal role $\varrho$ is $\lang$-rewritable if the $\LTL^{\op}_{\fragr}$ \OMPIQ{} $(\TO, \varrho)^\ddagger$ is $\lang$-rewritable.
\end{proposition}
\begin{proof}
Suppose $\rew^\ddagger(t)$ is an $\lang$-rewriting of $(\TO,\varrho)^\ddagger$. We can assume that $\rew^\ddagger(t)$ is constructed from atoms of the form~$A_{P}(s)$ and $A_{P-}(s)$, with $P$ occurring in $(\TO, \varrho)$,  and the built-in predicates and constructs of $\lang$. We obtain the required $\lang$-rewriting $\rew(x,y,t)$ from $\rew^\ddagger(t)$  by replacing every~$A_P(s)$ with $P(x,y,s)$, every $A_{P-}(s)$ with $P(y,x,s)$ and, in the case of $\FO(\RPR)$, by additionally replacing every  occurrence of $Q(t_1,\dots, t_k)$, for a relation variable~$Q$, with  $Q(x,y,t_1,\dots, t_k)$.
\end{proof}

In Section~\ref{sec:DL-Lite}, we show that answering \OMAQ{}s in temporal \DL{} can be reduced to answering restricted OMPIQs with $\bot$-free ontologies. This reduction relies on our ability to decide consistency of temporal \DL{} KBs, which will be studied in the next section.

% !TEX root = TDL-Lite.tex

\section{Consistency of Temporal $\DL$ Knowledge Bases}\label{sec:consistency}

The consistency problem for temporal \DL{} KBs that do \emph{not} contain CIs can be reduced by Proposition~\ref{prop:RBox:consistency} to \LTL{} satisfiability, which is known to be decidable and in fact \PSpace-complete~\cite{DBLP:journals/jacm/SistlaC85}.
On the other hand, due to the interaction of Boolean RIs with CIs, both consistency and \OMAQ{} answering with $\DL_{\core/\gbool}\Xnext{}$ ontologies turn out to be undecidable (even for data complexity):
\begin{theorem}\label{thm:undec}
$(i)$ Checking consistency of $\DL_{\core/\gbool}\Xnext{}$ KBs is undecidable.

$(ii)$ There 
is a $\DL_{\core/\gbool}\Xnext{}$ ontology $\TO$ such that the following problems are undecidable\textup{:} whether a given ABox~$\A$ is consistent with $\TO$, whether 
the pair $(a,0)$ is a certain answer  over a given ABox $\A$ to \OMAQ{} $(\TO, A)$ with a concept name $A$, and whether $(a,b,0)$ is a certain answer  over a given ABox~$\A$ to \OMAQ{}  $(\TO, P)$  with a role name $P$.
\end{theorem}
\begin{proof}
$(i)$ The proof is by reduction of the  $\N \times \N$-tiling problem, which is known to be undecidable~\cite{Berger1966-BERTUO-6}: given a finite set $\mathfrak T$ of tile types $\{1,\dots,m\}$,  decide whether $\mathfrak T$  can tile the $\N\times\N$ grid. For each $\mathfrak{T}$, we denote by $\textit{up}(i)$, $\textit{down}(i)$, $\textit{left}(i)$ and $\textit{right}(i)$ the colours on the four edges of any tile type $i\in\mathfrak{T}$. Define a $\DL_{\core/\gbool}\Xnext{}$ ontology~$\TO_{\mathfrak T}$, where $T$ is a role name and the $T_i$ and the~$R_i$ are role names associated with tile types $i \in \mathfrak T$, by taking
\begin{align*}
& T \ \sqsubseteq \  \bigsqcup_{i\in \mathfrak T} R_i, \qquad\qquad
 \exists R_i \sqcap \exists R_j   \ \sqsubseteq \ \bot, \quad \text{ for } i,j\in \mathfrak T \text{ with } i \ne j, \\
& \exists R^-_i \ \sqsubseteq \ \exists T_i, \qquad T_i \ \ \sqsubseteq \hspace*{-1em} \bigsqcup_{\begin{subarray}{c}j \in \mathfrak T\\\textit{up}(i) = \textit{down}(j)\end{subarray}}\hspace*{-1.5em} R_j \quad\text{ and }\quad
R_i \ \ \sqsubseteq \hspace*{-1em} \bigsqcup_{\begin{subarray}{c}j \in \mathfrak T\\\textit{right}(i)  = \textit{left}(j)\end{subarray}} \hspace*{-1.5em}\Rnext R_j, \qquad \text{ for } i \in \mathfrak T.
\end{align*}
It is readily checked that $\{T(a,b,0)\}$ is consistent with $\TO_{\mathfrak T}$ iff $\mathfrak T$ can tile the $\N \times \N$ grid.

\smallskip

$(ii)$ Using the representation of the universal Turing machine by means of tiles~\cite{Borgeretal97}, we obtain a set~$\mathfrak U$ of tile types for which the following problem is undecidable: given a finite sequence of tile types $i_0,\dots,i_n$, decide whether $\mathfrak U$ can tile the $\N \times \N$ grid so that tiles of types $i_0,\dots,i_n$ are placed on $(0,0),\dots,(n,0)$, respectively. Given such $i_0,\dots,i_n$, we take the ABox $\A = \{\,T(a,b,0), R_{i_0}(a,b,0),\dots,R_{i_n}(a,b,n)\,\}$. Then $\mathfrak U$ can tile $\mathbb N \times \mathbb N$ with $i_0,\dots,i_n$ on the first row iff  $\A$ is consistent with $\TO_{\mathfrak U}$ iff $A(a,0)$ is \emph{not} a certain answer to \OMAQ{} $(\TO_{\mathfrak U},A)$ over $\A$, where $A$ is a fresh concept name. Similar considerations apply to the case of a fresh role $P$.
\end{proof}

%\begin{table}
%\newcommand{\cmplx}[1]{\scriptsize #1}
%\centering%
%\begin{tabular}{cc}
%\toprule
%fragment $\DL^\op_{\frag/\fragr}$ & complexity \\
%\midrule
%bool  and g-bool  ($\nxt$ only) & undecidable \cmplx{[Th.~\ref{thm:undec}]}\\
%bool/krom ($\nxt$ only) & \ExpSpace{} \cmplx{[Th.~\ref{thm:bool-krom:hardness}]}\\
%bool/horn and horn & \ExpSpace{} \cmplx{[Th.~\ref{th:bool-sat-expspace}]}\\
%bool/core and horn/core & \PSpace{} \cmplx{[Th.~\ref{th:core-ris:pspace}]}\\
%\bottomrule
%\end{tabular}
%\caption{Combined complexity of the consistency problem for $\DL^\op_{\frag/\fragr}$.}\label{table:consistency:tdl-lite}
%\end{table}

\begin{table}[t]
\centering%
\newcommand{\cmplx}[2]{#1~\scriptsize{}[#2]}
\renewcommand{\tabcolsep}{2pt}%
\begin{tabular}{ccccc}\toprule
role & \multicolumn{4}{c}{concept inclusions}\\
inclusions & \bool & \horn  &   \krom & \core\\\midrule
(\textit{g}-)\bool &  ? & ? & ? & ? \\ 
  \krom &   ? & ? & ? & ? \\
  \horn &   \multicolumn{4}{c}{\cmplx{in \ExpSpace{}}{Th.~\ref{th:bool-sat-expspace}}}  \\
  \core &   \multicolumn{4}{c}{\cmplx{in \PSpace{}}{Th.~\ref{th:core-ris:pspace}}}  \\\bottomrule\\[-6pt]
 \multicolumn{5}{c}{\small a) $\op = \Box$.}
\end{tabular}
\hfil
\renewcommand{\tabcolsep}{7pt}%
\begin{tabular}{ccc}\toprule
role & \multicolumn{2}{c}{concept inclusions}\\
inclusions & \bool \hfil \horn  &  \krom \hfil \core\\\midrule
(\textit{g}-)\bool & \multicolumn{2}{c}{\cmplx{undecidable}{Th.~\ref{thm:undec}}}    \\
 \krom & \cmplx{\ExpSpace$^*$}{Th.~\ref{thm:bool-krom:hardness}} & \cmplx{in \PSpace$^*$}{Th.~\ref{thm:krom:qtl-krom}} \\
 \horn &  \cmplx{\ExpSpace}{Th.~\ref{th:bool-sat-expspace}} & \cmplx{in \ExpSpace}{Th.~\ref{th:bool-sat-expspace}} \\
 \core &  \cmplx{\PSpace}{Th.~\ref{th:core-ris:pspace}} &  \cmplx{in \PSpace}{Th.~\ref{th:core-ris:pspace}} \\\bottomrule\\[-6pt]
 \multicolumn{3}{c}{\small b) $\op = \nxt$ or $\op = \Box\nxt$.}
\end{tabular}
\caption{Combined complexity of the consistency problem for
  $\DL^\op_{\frag/\fragr}$ (the upper bounds in cases with $^*$ are only known for $\op = \nxt$).}
\label{table:consistency:tdl-lite}
\end{table}

In Sections~\ref{sec:HornRI:sat} and~\ref{sec:KromRI:sat}, we show that, by admitting Horn or Krom RIs only, we make KB consistency decidable in \ExpSpace; see Table~\ref{table:consistency:tdl-lite}. 
These results are established by reduction to the one-variable fragment \FOLTLI{} of first-order \LTL{}, which is known to be \ExpSpace-complete~\cite{DBLP:journals/jcss/HalpernV89,gkwz}.
We remind the reader that formulas in \FOLTLI{}  are constructed from predicates with individual constants and a single individual variable $x$ using the Boolean  connectives, first-order quantifiers $\forall x$, $\exists x$ and \LTL{} operators. Interpretations,~$\mathfrak M$, for first-order  \LTL{} are similar to the DL interpretations~\eqref{interpretation}: for all $n \in \Z$, we have standard FO-structures~$\mathfrak M(n)$ with the same domain $\Delta^{\mathfrak M}$ and the same interpretations $a^{\mathfrak M}$ of individual constants~$a$  but possibly varying interpretations $P^{\smash{\mathfrak M(n)}}$ of predicate symbols $P$. At each time instant $n$, the quantifiers $\forall x$ and $\exists x$ are interpreted over the domain of~$\mathfrak M$ as usual, while the \LTL{} operators are interpreted over $(\Z,<)$, also as usual, given an assignment of $x$ to some domain element of $\Delta^{\mathfrak M}$. For a formula $\varphi$ with a free variable and an individual constant $a$, we write $\mathfrak M,n \models \varphi(a)$ to say that the formula~$\varphi$ is true at moment~$n$ in $\mathfrak M$ under the assignment $x \mapsto a^{\mathfrak M}$. We abbreviate $ \Rbox\Lbox$ by $\Box$ and a sequence of $\ell$-many $\Rnext$ and $\Lnext$  by $\Rnext^\ell$ and $\Lnext^\ell$, respectively. Also, $\nxt^k\vartheta$ stands for $\Rnext^k\vartheta$ if $k > 0$, $\vartheta$ if $k = 0$, and $\Lnext^{-k} \vartheta$ if $k < 0$.

\subsection{Consistency of $\DL_{\bool/\horn}\Xallop{}$ KBs}\label{sec:HornRI:sat}

Let $\K = (\TO,\A)$ be a $\DL_{\smash{\bool/\horn}}\Xallop{}$ KB with $\TO = \T \cup \R$. As in Section~\ref{sec:ltl-can-period}, denote by $\subt$ the set of subconcepts in~$\T$ 
and their negations.  For $\tp\subseteq\subt$, let $\tp^\dagger$ be the result of replacing each~$B$  in~$\tp$ with~$B^\dagger$ (see Section~\ref{sec:proj}). A \emph{concept type} $\tp$ for~$\T^\dagger$ is a maximal subset $\tp$ of $\subt$ such that $\tp^\dagger$ is consistent with $\T^\dagger$  (note, however, that $\tp$ is not necessarily consistent with~$\T$: for example, $\tau^\dagger$ can be consistent with $\T^\dagger$ even if  $\exists P\in\tau$ and  $\T$ contains $\exists P^-\sqsubseteq \bot$).
A \emph{beam} $\beam$ for~$\T$ is a function from $\Z$ to the set of all concept types for $\T^\dagger$ such that, for all $n\in\Z$, we have the following:
\begin{align}
\label{eq:beam:0}
& \Rnext C \in \beam(n) ~\text{ iff }~  C\in\beam(n+1),
&& \Lnext C \in \beam(n) ~\text{ iff }~  C\in\beam(n-1),\\
\label{eq:beam:1}
& \Rbox C\in\beam(n) ~\text{ iff }~  C \in\beam(k), \text{ for all } k > n,
&& \Lbox C\in\beam(n) ~\text{ iff }~  C \in\beam(k), \text{ for all } k < n.
\end{align}
Given an interpretation $\I$ and $u \in \Delta^\I$, the function $\beam^\I_u\colon n\mapsto \bigl\{\, C\in\subt \mid u \in C^{\I(n)}\,\bigr\}$ is a beam; we refer to it as \emph{the beam of $u$ in~$\I$}. Here and below, we normally omit the negated elements of types, which is unambiguous due to their maximality.

Assume that $\R$ contains all roles that occur in $\T$. Denote by $\subr$ the set of
subroles in~$\R$ and their negations.  A \emph{role type} $\rtp$ for $\R$ is a maximal  subset of $\subr$ consistent with~$\R$ (equivalently, by Proposition~\ref{prop:RBox:consistency}, $\rtp^\ddagger$ is consistent with the $\LTL_\horn\Xallop$ ontology $\R^\ddagger$)
and  a \emph{rod} $\rod$ for~$\R$ is a function from $\Z$ to the set of all role types for $\R$ such that~\eqref{eq:beam:0} and~\eqref{eq:beam:1} hold for all $n\in\Z$ with $\beam$ replaced by~$\rod$ and $C$ by temporalised roles~$R$. Given an interpretation~$\I$ and $u,v\in\Delta^\I$, the function $\rod_{u,v}^\I\colon n \mapsto \bigl\{\, R \in\subr \mid (u,v)\in R^{\I(n)}\,\bigr\}$ is a rod for $\R$; we call it the \emph{rod of $(u,v)$ in $\I$}.

In this section we consider Horn RIs, and the key property of $\LTL_\horn\Xallop$ is that every consistent $\LTL_\horn\Xallop$ KB $(\R^\ddagger,\A^\ddagger)$ has a canonical model  $\C_{\R^\ddagger,\A^\ddagger}$; see Section~\ref{sec:ltl-can-period}. 
Since the RIs in $\R$ are Horn, given any \LTL{} ABox $\A^\ddagger$, which is a set of atoms of the form $S^\ddagger(\ell)$, we define the \emph{$\R$-canonical rod} $\rod_{\A^\ddagger}$ for $\A^\ddagger$ (provided that~$\A^\ddagger$ is consistent with~$\R^\ddagger$) by taking
$\rod_{\A^\ddagger}\colon n\mapsto \bigl\{\, R \in\subr \mid R^\ddagger(n)\in \C_{\R^\ddagger,\A^\ddagger}\,\bigr\}$.
In other words, $\R$-canonical rods are the minimal rods for $\R$ `containing' all atoms of $\A^\ddagger$:  for any $R$ and $n\in\Z$,
\begin{equation}\label{eq:rods:minimal}
R\in\rod_{\A^\ddagger}(n) \ \text{ iff }\  R\in\rod(n), \text{ for all rods } \rod \text{ for } \R \text{ such that } S\in\rod(\ell) \text{ for all } S^\ddagger(\ell)\in\A^\ddagger.
\end{equation}
Finally, given a beam $\beam$, we say a rod $\rod$ is \emph{$\beam$-compatible} if $\exists S\in\beam(n)$ whenever $S\in\rod(n)$, for all $n\in \Z$ and basic concepts $\exists S$ in $\T$.
We are now fully equipped to prove the following characterisation of  $\DL_{\bool/\horn}\Xallop{}$ KBs consistency.
\begin{lemma}\label{lem:unravelling}
Let
$\K = (\TO,\A)$ be a $\DL_{\bool/\horn}\Xallop{}$ KB with
$\TO = \T \cup \R$. Let
\begin{equation*}
\Xi = \ind(\A)\cup\bigl\{\, w_P, w_{P^-}\mid P \text{ a role name in } \TO\,\bigr\}.
\end{equation*}
Then $\K$ is consistent iff  there are beams $\beam_w$,  $w\in \Xi$, for $\T$ satisfying the following conditions\textup{:}
\begin{align}
\label{eq:quasimodel:beam:abox}
& A\in\beam_a(\ell), \text{ for all } A(a,\ell)\in\A,\\
\label{eq:quasimodel:inv}
& \text{if } \exists S\in\beam_w(n), \text{ then } \exists S^-\in \beam_{w_{S^-}}(k), \text{ for some } k\in\Z,\\
\label{eq:quasimodel:rod:abox}
& \text{for every } a,b\in\ind(\A), \text{ there is a $\beam_a$-compatible rod } \rod \text{ for } \R\\\nonumber & \hspace*{16em}\text{ such that } S\in\rod(\ell) \text{ for all } S(a,b,\ell)\in\A,\\
\label{eq:quasimodel:rod:witness}
& \text{if }\exists S\in\beam_w(n), \text{ then there is a $\beam_w$-compatible rod } \rod \text{ for } \R \text{ such that } S\in \rod(n).
\end{align}
Moreover, for any beams~$\beam_w$,  $w\in \Xi$, for $\T$ satisfying~\eqref{eq:quasimodel:beam:abox}--\eqref{eq:quasimodel:rod:witness}, there is a model $\I$ of $\K$ such that
\begin{itemize}
\item[--] for any $a\in\ind(\A)$, the beam of $a^\I$ in $\I$ coincides with $\beam_a$,
\item[--] for any $u\in\Delta^\I\setminus\{a^\I\mid a\in\ind(\A)\}$, there are $S$ and $n\in\Z$ with $\beam_u^\I(k) = \beam_{w_S}(k +n)$, for all  $k\in\Z$, 
\item[--]  for any $a,b\in\ind(\A)$, the rod of $(a^\I,b^\I)$ in $\I$ is the $\R$-canonical rod for $\A^\ddagger_{\smash{a,b}}$, and
\item[--]  for any $u\in\Delta^\I$, $v\in\Delta^\I\setminus\{a^\I\mid a\in\ind(\A)\}$, the rod of $(u,v)$ in $\I$ is the $\R$-canonical rod for some $\{S^\ddagger(n)\}$.
\end{itemize}
\end{lemma}

Before proving the lemma, we illustrate the construction by an example.
\begin{example}\em\label{ex:kdagger-sat}
Consider the KB $\K = (\TO,  \{ Q(a,b,0)\})$, where $\TO$ consists of
\begin{equation*}
\exists Q \sqcap \Rbox A \sqsubseteq \bot,\qquad \top \sqsubseteq A \sqcup \exists P\qquad\text{ and }\qquad P^- \sqsubseteq \Rnext Q,
\end{equation*}
which is the result of converting
$\exists Q \sqsubseteq \Rdiamond \exists P$ and $P^- \sqsubseteq \Rnext Q$
into normal form~\eqref{axiom1}.
\begin{figure}[t]
\centerline{%
\begin{tikzpicture}[xscale=0.6, yscale=0.7, semithick]\footnotesize
\foreach \x/\i in {-2/0,-1/1,0/2,1/3,2/4,3/5,4/6,5/7} {
\node at (\x,-1.3) {\scriptsize $\i$};
\draw[time-guideline] (\x,-1.1) -- ++(0,2.1);
\draw[time-guideline] (\x,1.6) -- ++(0,2);
\draw[time-guideline] (\x,4.2) -- ++(0,1.2);
}
\foreach \y/\j in {-0.8/0,0/1,2.6/4,5.2/7} {
\draw[object-timeline] (-2.5,\y) -- ++(8,0);
\foreach \x/\i in {-2/0,-1/1,0/2,1/3,2/4,3/5,4/6,5/7} {
\node (a\j\i) at (\x,\y) [point]{};
}
}
\foreach \y/\j in {0.8/2,1.8/3,3.4/5,4.4/6} {
\draw[very thin,draw=gray] (-2,\y) -- ++(7,0);
\draw[very thin,draw=gray,dashed] (-2,\y) -- ++(0.7,0);
\draw[very thin,draw=gray,dashed] (5,\y) -- ++(0.7,0);
\foreach \x/\i in {-2/0,-1/1,0/2,1/3,2/4,3/5,4/6,5/7} {
\node (a\j\i) at (\x,\y) [graypoint]{};
}
}
\fill[rounded corners=1mm,gray,fill opacity=0.2] (-2.5,-0.7) rectangle +(8,0.6);
\fill[rounded corners=1mm,pattern=north west lines, pattern color=gray!50] (-2.5,-0.7) rectangle +(8,0.6);
\node at (3,-0.45) {$\rod_{a,b}$};
\draw[ultra thick,{Bracket[]}-,black!60] (2.5,-0.1) -- +(0,-0.6);
\node at (5,-0.4) {$\rod_{b,a}$};
\draw[ultra thick,{Bracket[]}-,black!60] (4.5,-0.7) -- +(0,0.6);
\foreach \x/\y/\l in {1.5/2.7/4,-0.5/0.1/2} {
\fill[rounded corners=1mm,gray,fill opacity=0.5] (\x,\y) rectangle +(4,0.6);
\node at ($(\x+3.55,\y+0.3)$) {$\rod_{P,\l}$};
\draw[ultra thick,{Bracket[]}-,black!80] ($(\x+3,\y)$) -- +(0,0.6);
}
\foreach \x/\y/\l in {-2.5/1.9/0} {
\fill[rounded corners=1mm,gray,fill opacity=0.5] (\x,\y) rectangle +(4.2,0.6);
\node at ($(\x+3.7,\y+0.2)$) {$\rod_{\smash{P^-\!\!,0}}$};
\draw[ultra thick,-{Bracket[]},black!80] ($(\x+3,\y)$) -- +(0,0.6);
}
\foreach \x/\y/\l in {-2.5/4.5/0} {
\fill[rounded corners=1mm,gray,fill opacity=0.5] (\x,\y) rectangle +(4,0.6);
\node at ($(\x+3.5,\y+0.2)$) {$\rod_{P,\l}$};
\draw[ultra thick,-{Bracket[]},black!80] ($(\x+3,\y)$) -- +(0,0.6);
}
\node at (-3.2,-0.8) {$\beam_b$};
\node at (-3.2,0) {$\beam_a$};
\node at (-3.2,0.8) {\textcolor{gray}{$u_{P,2}$}};
\node at (-3.2,1.8) {\textcolor{gray}{$u_{P^-,0}$}};
\node at (-3.2,2.6) {$\beam_{w_{P^-}}$};
\node at (-3.2,3.4) {\textcolor{gray}{$u_{P,4}$}};
\node at (-3.2,4.4) {\textcolor{gray}{$u_{P,0}$}};
\node at (-3.2,5.2) {$\beam_{w_P}$};
%
% a0 = b
% a1 = a
% a2 = uPa3
% ---
% a3 = uP-wP-0
% a4 = wP-
% a5 =  uPwP-5
%
% a6 = uPwP0
% a7 = wP
%
\node (a10p) at (a10) [ppoint]{};
\node (a00p) at (a00) [qpoint]{};
\node (a12p) at (a12) [ppoint]{};
\node (a13p) at (a13) [qpoint]{};
\node (a40p) at (a40) [qpoint]{};
\node (a41p) at (a41) [ppoint]{};
\node (a44p) at (a44) [ppoint]{};
\node (a45p) at (a45) [qpoint]{};
\node (a70p) at (a70) [ppoint]{};
\node (a71p) at (a71) [qpoint]{};
\begin{scope}[thick]
\draw[<-] (a00p)  to node [right]{$Q$}  (a10p);
\draw[->] (a12p)  to node [right]{$P$} (a22);
\draw[<-] (a13p) to node [right]{$Q$}  (a23);
\draw[->] (a30)  to node [right]{$P$} (a40p);
\draw[->] (a41p)  to node [right]{$Q$}  (a31);
\draw[->] (a44p)  to node [right]{$P$} (a54);
\draw[->] (a55)  to node [right]{$Q$} (a45p);
\draw[->] (a70p)  to node [right]{$P$} (a60);
\draw[->] (a61)  to node [right]{$Q$} (a71p);
\end{scope}
%
%%%%%%%%%%%%%%%%%%%%%%%%%%%%%
\begin{scope}[xshift=107mm,yshift=2mm]
\foreach \y in {-1,0.2,1.4,2.6,3.8} {
\draw[object-timeline] (-2.5,\y) -- ++(14.2,0);
}
\foreach \x/\i in {-2/0,-1/1,0/2,1/3,2/4,3/5,4/6,5/7,6/8,7/9,8/10,9/11,10/12,11/13} {
\node at (\x,-1.5) {\scriptsize $\i$};
\draw[time-guideline] (\x, -1.2) -- ++(0, 5.5);
\draw[dotted] (\x,4.35) -- +(0,0.5);
}
\fill[rounded corners=1mm,gray,fill opacity=0.2] (-2.5,-0.9) rectangle +(14.2,1);
\fill[rounded corners=1mm,pattern=north west lines, pattern color=gray!50] (-2.5,-0.9) rectangle +(14.2,1);
\node at (6,-0.4) {$\rod_{\smash{a,b}}$};
\draw[ultra thick,-{Bracket[]},black!60] (5.5,-0.9) -- +(0,1);
\node at (8,-0.4) {$\rod_{\smash{b,a}}$};
\draw[ultra thick,{Bracket[]}-,black!60] (7.5,-0.9) -- +(0,1);
\fill[rounded corners=1mm,gray,fill opacity=0.50] (-0.5,0.3) rectangle +(4.1,1);
\node at (3.1,0.8) {$\rod_{P,2}$};
\draw[ultra thick,{Bracket[]}-,black!80] (2.5,0.3) -- +(0,1);
\fill[rounded corners=1mm,gray,fill opacity=0.50] (3.5,1.5) rectangle +(4.1,1);
\node at (7.1,2) {$\rod_{P,6}$};
\draw[ultra thick,{Bracket[]}-,black!80] (6.5,1.5) -- +(0,1);
\fill[rounded corners=1mm,gray,fill opacity=0.50] (7.5,2.7) rectangle +(4.2,1);
\node at (11.2,3.2) {$\rod_{P,10}$};
\draw[ultra thick,{Bracket[]}-,black!80] (10.5,2.7) -- +(0,1);
{\footnotesize
\node at (-3.2,-1) {$b$};
\node at (-3.2,0.2) {$a$};
\node at (-3.2,1.4) {$aP^2$};
\node at (-3.4,2.6) {$aP^2P^6$};
\node at (-3.7,3.8) {$aP^2P^6P^{10}$};
}
\foreach \y/\j in {-1/0,0.2/1,1.4/2,2.6/3,3.8/4} {
\foreach \x/\i in {-2/0,-1/1,0/2,1/3,2/4,3/5,4/6,5/7,6/8,7/9,8/10,9/11,10/12,11/13} {
\node (a\j\i) at (\x,\y) [point]{};
}
}
\node (a00p) at (a00) [qpoint]{};
\node (a10p) at (a10) [ppoint]{};
\node (a12p) at (a12) [ppoint]{};
\node (a13p) at (a13) [qpoint]{};
\node (a22p) at (a22) [qpoint]{};
\node (a23p) at (a23) [ppoint]{};
\node (a26p) at (a26) [ppoint]{};
\node (a36p) at (a36) [qpoint]{};
\node (a27p) at (a27) [qpoint]{};
\node (a37p) at (a37) [ppoint]{};
\node (a310p) at (a310) [ppoint]{};
\node (a410p) at (a410) [qpoint]{};
\node (a311p) at (a311) [qpoint]{};
\node (a411p) at (a411) [ppoint]{};
\begin{scope}[->,thick]
\draw (a10p) to node [right]{$Q$} (a00p);
\draw (a12p) to node [right]{$P$} (a22p);
\draw (a23p) to node [right]{$Q$} (a13p);
\draw (a26p) to node [right]{$P$} (a36p);
\draw (a37p) to node [right]{$Q$} (a27p);
\draw (a310p) to node [right]{$P$} (a410p);
\draw (a411p) to node [right]{$Q$} (a311p);
\end{scope}
\node at (-3.2,4.5) {$\cdots$};
\end{scope}
\end{tikzpicture}%
}%
\caption{The beams and rods in Example~\ref{ex:kdagger-sat} (left) and the temporal interpretation $\I$ obtained by unravelling them (right).}\label{fig:FOLTL-pieces}
\end{figure}
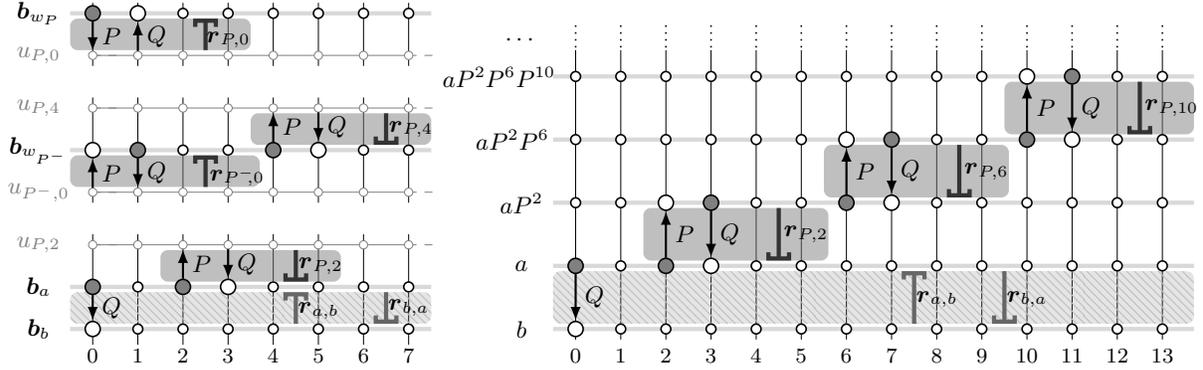
Beams $\beam_a$, $\beam_b$, $\beam_{\smash{w_{P^-}}}$ and $\beam_{w_P}$ are depicted in Fig.~\ref{fig:FOLTL-pieces} (left) by horizontal lines: the type contains $\exists P$ or $\exists Q$ whenever the large node is grey; similarly, the type contains $\exists P^-$ or $\exists Q^-$ whenever the large node is white (the label of the arrow specifies the role); we omit $A$ to avoid clutter. The rods are the arrows between the pairs of horizontal lines. For example, the rod required by~\eqref{eq:quasimodel:rod:abox} for $a$ and $b$ is labelled by $\rod_{a,b}$: it contains only $Q$ at $0$ (we specify only the positive components of the types); the rod required by~\eqref{eq:quasimodel:rod:abox} for $b$ and $a$ is labelled by $\rod_{b,a}$, and in this case, it is the mirror image of $\rod_{a,b}$. In fact, if we choose $\R$-canonical rods in~\eqref{eq:quasimodel:rod:abox}, then the rod for any $b,a$ will be the mirror image of the rod for $a,b$. The rod $\rod_{P,2}$ required by~\eqref{eq:quasimodel:rod:witness} for~$\exists P$ on~$\beam_a$ at moment $2$ is depicted between $\beam_a$ and $u_{P,2}$: it contains $P$ at $2$ and $Q$ at~$3$. In fact, it should be clear that, if  we choose $\R$-canonical rods in~\eqref{eq:quasimodel:rod:witness}, then they will all be isomorphic copies of at most $|\TO|$ rods: more precisely, the rods will be of the form $\rod_{\{S^\ddagger(n)\}}$, for a role $S$ from~$\TO$.
In the proof of Lemma~\ref{lem:unravelling}, we show how this collection of beams and $\R$-canonical rods can be used to obtain a  model $\I$ of $\K$ depicted on the right-hand side of  Fig.~\ref{fig:FOLTL-pieces} (with $A$ omitted again). We now proceed with the proof of the lemma.
\end{example}

\begin{proof}
$(\Leftarrow)$ Suppose that we have the required collection of beams $\beam_w$ for $\T$. We construct by induction on $m < \omega$ a sequence of temporal interpretations $\I_m = (\Delta^{\I_m}, \cdot^{\I_m(n)})$ and maps $f_m \colon \Delta^{\I_m} \times \Z \to \Xi \times \Z$. The meaning of $f_m$ is as follows: if $f_m(\lambda, n) = (w, k)$, then the type of the element $\lambda$ in $\I_m$ at moment $n$ is `copied' from the beam $\beam_w$ at  moment~$k$. For the basis of induction, set $\Delta^{\I_0} = \ind(\A)$, $A^{\I_0(n)} = \{a \mid A\in\beam_a(n)\}$ and $P^{\I_0(n)} = \{(a,b) \mid P\in\rod_{a,b}(n)\}$, for all concept and roles names $A$ and $P$ and all $n \in\Z$, and $f_0(a,n) = (a,n)$, for all $a\in \ind(\A)$ and $n \in \Z$,  where $\rod_{a,b}$ is the $\R$-canonical rod for $\A^\ddagger_{\smash{a,b}}$, which exists by~\eqref{eq:quasimodel:rod:abox} and which, by~\eqref{eq:rods:minimal}, is compatible with $\beam_a$, while its inverse is compatible with $\beam_b$.

Suppose next that $\I_m$ and $f_m$, for $m \ge 0$, have already been defined, that the elements of~$\Delta^{\I_m}$ are words of the form $\lambda = aS_1^{n_1} \dots S_l^{n_l}$, for $a \in \ind(\A)$, $n_i \in \Z$ and $l \ge 0$. We call a pair $(\lambda,n)$ an $S$-\emph{defect} in $\I_m$ if $f_m(\lambda,n) = (w,k)$, $\exists S\in \beam_w(k)$, and $(\lambda, \lambda') \notin S^{\I_m(n)}$ for any $\lambda' \in \Delta^{\I_m}$.
For any role~$S$ and any $S$-defect $(\lambda,n)$ in $\I_m$, we add the word $\lambda S^n$ to~$\Delta^{\I_m}$ and denote the result by $\Delta^{\I_{m+1}}$. By~\eqref{eq:quasimodel:inv}, we have $\exists S^-\in\beam_{w_{S^-}}(k')$, for some $k' \in \Z$. We fix one such $k'$ and extend $f_m$ to $f_{m+1}$ by setting $f_{m+1}(\lambda S^n,n') = (w_{S^-},n'-n+k')$, for any~$n' \in \Z$. We also define $A^{\I_{m+1}(n')}$ by extending $A^{\I_m(n')}$ with those $\lambda S^n$ for which $A\in\beam_{w_{S^-}}(n'-n+k')$, and  define~$P^{\I_{m+1}(n')}$ by extending $P^{\I_m(n')}$ with $(\lambda, \lambda S^n)$ for which $P\in\rod_{S,n}(n')$ and with $(\lambda S^n, \lambda)$ for which $P^-\in\rod_{S,n}(n')$, where~$\rod_{S,n}$ is the $\R$-canonical rod for $\{S^\ddagger(n)\}$, which exists by~\eqref{eq:quasimodel:rod:witness} and, by~\eqref{eq:rods:minimal}, is compatible with the respective beams.

Finally, let $\I$ and $f$ be the unions of all $\I_m$ and $f_m$, for $m < \omega$, respectively.
We show that $\I$ is a model of $\K$. It follows immediately from the construction that $\I$ is a model of~$\R$ and~$\A$. To show that $\I$ is also a model of $\T$, it suffices to prove that, for any $\lambda \in\Delta^\I$ and any role $Q$, we have $\lambda \in (\exists Q)^{\I(n')}$ iff $\exists Q\in\beam_w(k')$, where $f(\lambda, n') = (w,k')$. The implication $(\Leftarrow)$ follows directly from the procedure of `curing defects'\!. Let $\lambda \in (\exists Q)^{\I(n')}$, and so $(\lambda,\lambda') \in Q^{\I(n')}$, for some $\lambda' \in\Delta^\I$. Two cases are possible now.
\begin{itemize}
\item[--] If $\lambda,\lambda' \in \ind(\A)$, then $Q\in\rod_{\lambda,\lambda'}(n')$. Then, by~\eqref{eq:quasimodel:rod:abox}, $\exists Q\in\beam_{\lambda}(n')$. It remains to recall that $f(\lambda, n') = f_0(\lambda, n') =  (\lambda,n')$.
\item[--]  If $\lambda' \notin \ind(\A)$, then $\lambda' = \lambda S^n$, for some $S$ and $n$, and $Q\in\rod_{S,n}(n')$. We also have $\exists S\in\beam_w(k)$, where $f(\lambda,n) = (w,k)$. By~\eqref{eq:quasimodel:rod:witness}, there is a rod $\rod$ for $\R$ such that $S\in \rod(n)$, and so, we must have $Q\in\rod(n')$. Since $\rod$ is compatible with $\beam_w$,  we obtain $\exists Q\in\beam_w(n' - n + k)$. It remains to recall that $f(\lambda, n') = (w, n' - n + k)$.
\end{itemize}

$(\Rightarrow)$ Given a model $\I$ of $\K$, we construct beams $\beam_w$ for $\T$ as follows. Set  $\beam_a=\beam^\I_{\smash{a^\I}}$, for all $a\in\ind(\A)$.  For each~$S$, if  $S^{\I(n)}\ne\emptyset$, for some $n\in\Z$, then set $\beam_{w_S} = \beam^\I_{u}$, for $u \in (\exists S)^{\I(n)}$; otherwise, set $\beam_{w_S}=\beam^\I_{\smash{a^\I}}$, for an arbitrary~$a\in\ind(\A)$.
It is  straightforward to check that these beams are as required.
\end{proof}

We now reduce the existence of the required collection of beams to the \ExpSpace-complete satisfiability problem for \FOLTLI~\cite{DBLP:journals/jcss/HalpernV89,gkwz}, thereby establishing the upper complexity bound for $\DL_{\bool/\horn}\Xallop{}$. Let
$\K = (\TO,\A)$ be a $\DL_{\bool/\horn}\Xallop{}$ KB with
$\TO = \T \cup \R$ 
and let $\Xi$ be as in Lemma~\ref{lem:unravelling}. We treat its elements as constants in the first-order language and
define a translation~$\Psi_\K$ of $\K$ into \FOLTLI{} with a single individual variable $x$ as a conjunction of the following sentences, for all constants~$w \in \Xi$:
\begin{align}
\label{eq:TBox}
& \Box \bigl( C_1(w) \land \dots \land C_k(w) \to C_{k+1}(w)\lor \dots \lor C_{k+m}(w) \bigr),\hspace*{-4em}\\
\nonumber&&&\hspace*{-6em} \text{for all } C_1\sqcap \dots \sqcap C_k  \sqsubseteq C_{k+1}\sqcup \dots \sqcup C_{k+m} \text{ in } \T,\\
\label{eq:RBox}
& \Box \forall x \, \bigl( R_1(w,x)  \land \dots \land R_k(w,x) \to R(w,x)\bigr), && \text{for all } R_1\sqcap \dots \sqcap R_k \sqsubseteq R \text{ in } \R,\\
\label{eq:RBox:b}
& \Box\forall x \, \bigl( R_1(w,x)  \land \dots \land R_k(w,x) \to \bot\bigr), && \text{for all } R_1\sqcap \dots \sqcap R_k \sqsubseteq \bot \text{ in } \R,\\
& \Rnext^\ell A(a), &&  \text{for all } A(a,\ell) \text{ in } \A,\\
\label{abox-role1}
& \Rnext^\ell P(a,b), &&  \text{for all }  P(a,b,\ell) \text{ in  }\A,\\
\label{eq:unrav1}
& \Box \bigl(\exists S(w) \to \Rdiamond \Ldiamond \exists S^-(w_{S^-})\bigr), && \text{for all roles $S$ in $\TO$},\\
\label{eq:unrav2}
& \Box \bigl(\exists S(w) \leftrightarrow \exists x\, S(w,x)\bigr), && \text{for all roles $S$ in $\TO$}.
\end{align}
Thus, in $\Psi_\K$, we regard concept names $A$ and basic concepts $\exists S$ as \emph{unary predicates} and roles~$S$ as \emph{binary predicates}, assuming that $(P^-)^-$ is $P$ and $S^-(w,x)$ is $S(x,w)$. The interesting conjuncts in $\Psi_\K$ are~\eqref{eq:unrav1} and~\eqref{eq:unrav2}, which reflect the interaction between~$\T$~and~$\R$.

\begin{lemma}\label{lemma:horn-red}
A $\DL\Xallop_{\bool/\horn}$ KB $\K$ is consistent iff $\Psi_\K$ is satisfiable.
\end{lemma}
\begin{proof}
Each collection of beams $\beam_w$, $w\in \Xi$, for $\T$ gives rise to a model $\mathfrak{M}$ of $\Psi_\K$: the domain of $\mathfrak{M}$ comprises $\Xi$ and elements $u_{S,m}$, for a role $S$ and $m\in\Z$. We fix $\R$-canonical rods $\rod_{a,b}$ for $\A_{\smash{a,b}}^\ddagger$, which are guaranteed to exist by~\eqref{eq:quasimodel:rod:abox}, and $\R$-canonical rods $\rod_{S,m}$ for $\{S^\ddagger(m)\}$ for each $S$ and $m\in\Z$  with $\exists S\in\beam_w(m)$, for some $w\in\Xi$, which are guaranteed to exist by~\eqref{eq:quasimodel:rod:witness}, and set
\begin{align*}
\mathfrak{M},n\models B(w) &~\text{ iff }~ B\in\beam_w(n), \text{ for all } w\in\Xi, n\in\Z \text{ and basic concepts } B,\\
\mathfrak{M},n\models P(a,b) &~\text{ iff }~ P\in\rod_{a,b}(n), \text{ for all } a,b\in\ind(\A), n\in\Z, \text{ and role names } P,\\
\mathfrak{M},n\models S'(w,u_{S,m}) &~\text{ iff }~ S'\in\rod_{S,m}(n),\\ & \hspace*{4em} \text{ for all } w\in\Xi, n,m\in\Z \text{ and roles } S, S'\text{ with } \exists S\in\beam_w(m).
\end{align*}
It is readily checked that $\mathfrak{M}$ is as required; see also Fig.~\ref{fig:FOLTL-pieces}, where, in the context of Example~\ref{ex:kdagger-sat}, the $u_{S,m}$ are represented explicitly by grey horizontal lines. Conversely, it can be verified that every model $\mathfrak{M}$ of $\Psi_\K$ gives rise to the required collection of beams for $\T$.
\end{proof}

\begin{theorem}\label{th:bool-sat-expspace}
Checking consistency of $\DL_{\bool/\horn}\Xallop{}$ and $\DL_{\horn}\Xnext{}$ KBs is \ExpSpace-complete.
\end{theorem}
\begin{proof}
The upper bound follows from Lemma~\ref{lemma:horn-red} and \ExpSpace-completeness of \FOLTLI{}.
The hardness is proved by reduction of the non-halting problem for deterministic Turing machines with exponential tape. We assume that the head of a given machine~$M$ never runs beyond the first $2^n$ cells of its tape on an input word $\avec{a}$ of length $m$, where $n = p(m)$ for some polynomial~$p$; we also assume that it never attempts to access cells before the start of the tape. We construct a $\DL_\horn\Xnext{}$ ontology that encodes the computation of~$M$ on $\avec{a}$ using a single individual. 
The initial configuration is spread over the time instants~$1, \dots, 2^n$, from which the first $m$ instants represent $\avec{a}$ and the remaining ones encode the blank symbol~$\#$ (time instant 0 also contains~$\#$). The second configuration uses the next $2^n$ instants, $2^n+1, \dots, 2^n + 2^n$, etc. The configurations are encoded with the following concept names:
\begin{itemize}
\item[--] $H_{q,a}$  contains the ABox individual at the moment $i2^n+j$ whenever the machine head scans the $j$th cell of the $i$th configuration and sees symbol $a$, with $q$ being the current state of the machine;
\item[--] $C_a$ contains the ABox individual at $i2^n+j$ whenever the $j$th cell of the $i$th configuration contains $a$ but is \emph{not scanned} by the head.
\end{itemize}
For example, if $\avec{a} =a_1,\dots,a_m$ is the input and $q_0$ is the initial state of $M$,  then the ABox, which encodes the initial configuration, consists of
\begin{equation*}
I(e, 0), \ C_{\#}(e,0), \ H_{q_0,a_1}(e, 1), \ C_{a_2}(e, 2), \ \dots, \ C_{a_m}(e, m), \ E(e, m + 1),
\end{equation*}
where concept names $I$ and $E$ mark the start and the end of the input $\avec{a}$, respectively. Then, we use $\DL_\horn\Xnext{}$ CIs
\begin{equation*}
I \sqsubseteq \Rnext^{2^n} F,\qquad
F  \sqsubseteq \Lnext F,\qquad
E  \sqsubseteq \Rnext E,\qquad
E \sqcap F  \sqsubseteq C_{\#}
\end{equation*}
to fill the rest of the tape in the initial configuration with $\#$: concept name $F$ marks all the cells of the first configuration (along with all negative time points), and $E$ is propagated to all the cells starting from the end of the input~$\avec{a}$ in the first configuration along with all cells in all other configurations. 
Next, we encode computations of the deterministic Turing machine $M$ with tape alphabet $\Gamma$ and transition function $\delta\colon Q\times\Gamma \to Q \times \Gamma \times \{R, L\}$ by means of the following $\DL_\horn\Xnext{}$ CIs:
\begin{align*}
 \Lnext C_{a'} \sqcap H_{q,a} \sqcap \Rnext C_{a''} & \sqsubseteq  \Rnext^{2^n} (\Lnext C_{a'} \sqcap C_{b} \sqcap\Rnext H_{q', a''}), && \text{for } \delta(q,a) = (q',b,R) \text{ and } a',a''\in\Gamma,\\
 \Lnext C_{a'}\sqcap  H_{q,a} \sqcap \Rnext C_{a''} & \sqsubseteq \Rnext^{2^n} (\Lnext H_{q', a'}\sqcap  C_{b} \sqcap \Rnext C_{a''}), && \text{for } \delta(q,a) = (q',b,L)\text{ and } a',a''\in\Gamma,\\
  \Lnext C_{a'} \sqcap C_a \sqcap \Rnext C_{a''} & \sqsubseteq \Rnext^{2^n} C_a, && \text{for } a,a',a''\in\Gamma,
\end{align*}
where $\vartheta \sqsubseteq \Rnext^n(\vartheta_1 \sqcap \dots \sqcap \vartheta_k)$ abbreviates the $\vartheta\sqsubseteq \Rnext^n\vartheta_i$, $1\leq i \leq k$. Finally, CIs 
\begin{equation*}
H_{q,a} \sqsubseteq \bot, \hspace*{2em} \text{for each accepting and rejecting state } q \text{ and } a\in\Gamma,
\end{equation*}
ensures that accepting and rejecting states never occur in the encodings of computations.
These CIs are of exponential size, and our next task is to show how to convert them into a $\DL_\horn\Xnext{}$ ontology of polynomial size.

Consider a CI of the form $A \sqsubseteq \Rnext^{2^n} B$. First, we replace the temporalised concept $\Rnext^{2^n} B$ by $\exists P$ and add the CI $\exists Q \sqsubseteq B$ to the TBox, where $P$ and $Q$ are fresh role names. Then, we add the following RIs to the RBox, $\R$, for fresh role names $P_0,\dots,P_{n-1},\bar P_0,\dots, \bar P_{n-1}$:
\begin{align*}
  & P \sqsubseteq \Rnext (\bar P_{n-1} \sqcap \dots \sqcap \bar P_0),\\*
  & P_{n-1} \sqcap \dots \sqcap P_{0} \sqsubseteq Q,\\
& \bar P_{k}  \sqcap P_{k-1} \sqcap \dots \sqcap P_{0}\sqsubseteq \Rnext (P_{k} \sqcap  \bar P_{k-1}\sqcap \dots \sqcap \bar{P}_0),  &&  \text{for } 0 \leq k < n,\\*
& \bar P_{j} \sqcap \bar P_{k} \sqsubseteq \Rnext \bar P_{j} \quad  \text{ and }\quad %
 P_{j} \sqcap \bar P_{k} \sqsubseteq \Rnext P_{j}, &&  \text{for }  0 \leq k < j < n.
\end{align*}
Intuitively, these RIs encode a binary counter from $0$ to $2^n - 1$, where roles $\bar P_i$ and $P_i$ stand for `the $i$th bit of the counter is 0' and, respectively, `is 1', and ensure that $\R \models P \sqsubseteq \Rnext^{2^n} Q$ but $\R \not \models P \sqsubseteq \Rnext^i Q$ for any $i \neq 2^n$ (remember that $\exists P$ generates different $P$-successors at different time points). Further details are left to the reader. Note that the encoding of $\Rnext^{2^n}$ on concepts using an RBox of size polynomial in $n$ is based on the same idea as the encoding of~$\Rnext$ on concepts using an RBox with $\Box$ operators only in Example~\ref{exampleBvNB:2} and Theorem~\ref{thm:unexpected}.
\end{proof}

\subsection{Consistency of $\smash{\DL_{\bool/\core}\Xallop{}}$ KBs}\label{sec:corelRI:sat}

We now adapt the technique of Section~\ref{sec:HornRI:sat} to reduce consistency of $\smash{\DL_{\bool/\core}\Xallop{}}$ KBs to that of \LTL{} KBs. The modification is based on the observation that the logical consequences of core RIs, and so the auxiliary $\R$-canonical rods as well, have a simpler structure than those of Horn RIs: intuitively, while \core{} RIs  can entail $\R \models P \sqsubseteq \nxt^n Q$ for $n$ exponentially large in the size of $\R$, it must also be the case that $\R \models P \sqsubseteq \nxt^i Q$ for either all $i < n$ or all $i > n$. Using this property and a trick with binary counters (see below), we reduce satisfiability of~$\K$ to satisfiability of a polynomial-size \LTL{} KB.

Let $\R$ be a $\smash{\DL_{\bool/\core}\Xallop{}}$ RBox and $\rod$ the $\R$-canonical rod for some $\A^\ddagger_R = \{ R^\ddagger(0) \}$. Then $S\in \rod(n)$ iff one of the following conditions holds:
\begin{itemize}\itemsep=0pt
\item[--] $(\R')^\ddagger, \A^\ddagger_R \models S^\ddagger (n)$, where $\R'$ is obtained from $\R$ by removing the RIs with $\Box$,
\item[--] there is $m > n$ with $|m| \leq 2^{|\R|}$ and $\Lbox S\in\rod(m)$,
\item[--] there is $m < n$ with $|m| \leq 2^{|\R|}$ and $\Rbox S\in\rod(m)$.
\end{itemize}
Let $\min_{R,S}$ be the minimal  integer $m$ with $\Rbox S\in\rod(m)$; if it exists, then $|\min_{R,S}| \leq \smash{2^{|\R|}}$. The maximal integer $m$
 with $\Lbox S\in\rod(m)$ is also bounded by  $\smash{2^{|\R|}}$ (if defined at all). 
The following example shows that these integers can indeed be exponential in $|\R|$.

\begin{example}\em
Let $\R$ be the $\DL_{\bool/\core}\Xallop{}$ RBox with the following RIs:
\begin{align*}
&P \sqsubseteq  R_0, && R_i \sqsubseteq \Rnext R_{i + 1\hspace*{-0.5em}\pmod 2},\quad \text{ for } 0 \leq i < 2, && R_1 \sqsubseteq S,\\
& P \sqsubseteq Q_0, && Q_i \sqsubseteq \Rnext Q_{i + 1\hspace*{-0.5em}\pmod 3}, \quad\text{ for } 0 \leq i < 3, && Q_1 \sqsubseteq S, && Q_2 \sqsubseteq S,\\
& && && P \sqsubseteq S,  && P \sqsubseteq \Lbox S.
\end{align*}
Clearly, $\R \models P \sqsubseteq \Rnext^6 \Lbox S$. If instead of the 2- and 3-cycles we use $p_1$, $p_2$, \dots, $p_n$-cycles, where $p_i$ is the $i$th prime number, for $1 \leq i \leq n$, then $\R \models P \sqsubseteq \Rnext^{p_1 \times \dots \times p_n} \Lbox S$.
\end{example}

In any case, the restriction to core RIs brings down the complexity:
\begin{theorem}\label{th:core-ris:pspace}
Checking consistency of $\DL_{\bool/\core}\Xallop{}$ and $\DL_{\horn/\core}\Xnext{}$ KBs is $\PSpace$-complete.
\end{theorem}
\begin{proof}
We encode the given KB $\K$ in \LTL{} following the proof of Lemma~\ref{lemma:horn-red} and representing~$\Psi_\K$ as an \LTL-formula with propositional variables of the form $C_u$ and $R_{u,v}$, for $u,v \in \Xi$, assuming that $R^-_{u,v} = R_{v,u}$; in particular, $(\exists S)_w$ denotes the propositional variable for the unary atom $\exists S(w)$ of  $\Psi_\K$. Sentences~\eqref{eq:TBox}--\eqref{eq:unrav1} can be translated into \LTL{} by simply instantiating all universal quantifiers by constants in~$\Xi$.  Sentences~\eqref{eq:unrav2}, however, require a special treatment.
First, we take
\begin{align}\label{eq:core-RI-conn}
& \Box\bigl(\nxt_1 (\exists S_1)_{w} \to \nxt_2 (\exists S_2)_{w}\bigr),  & \text{ for every } \nxt_1 S_1 \sqsubseteq \nxt_2 S_2 \text{ in } \R,
\end{align}
where each $\nxt_i$ is $\Rnext$, $\Lnext$ or blank. We also require the consequences of $\R$ of the form \mbox{$\exists R \sqsubseteq \nxt^{\max_{R,S}} \Lbox \exists S$} and $\exists R \sqsubseteq \nxt^{\min_{R,S}} \Rbox \exists S$, for all $R$ and $S$ with defined $\max_{R,S}$ and $\min_{R,S}$, that are not entailed by~\eqref{eq:core-RI-conn}. First, we use a \PSpace{} subroutine to check the existence and compute the  $\min_{R,S}$ and $\max_{R,S}$: we guess the binary representation of, say, $\min_{R,S}$ and check whether the conjunction of $\R^\ddagger$, $R^\ddagger(0)$ and $\smash{\Rnext^{\max_{R,S}}} \neg\Lbox S^\ddagger(0)$ is unsatisfiable; as the latter conjunct can be exponential in $|\R|$, we need to represent it by $O(|\R|^2)$ conjuncts using the binary counter (similarly to the proof of Theorem~\ref{thm:bool-krom:hardness}). By Savitch's theorem and the \PSpace{} membership for \LTL{} satisfiability~\cite{DBLP:journals/jacm/SistlaC85}, such a non-deterministic subroutine can be implemented in \PSpace{}.
Next, we encode the required consequences of $\R$: for example, if $\max_{R,S} \geq 0$, then $\exists R \sqsubseteq \nxt^{\max_{R,S}} \Lbox \exists S$ gives rise to
\begin{align}
\label{eq:core-RI-1}
& \Box\bigl(\Rbox \Rdiamond (\exists R)_w \to \Rbox (\exists S)_w\bigr),\\
\label{eq:core-RI-2}
&\Box\bigl((\exists R)_w \land \neg \Rdiamond  (\exists R)_w \to \Rnext^{\max_{R,S}}\Lbox (\exists S)_w\bigr),
\end{align}
where $\Rnext^{\max_{R,S}}$ can again be expressed by $O(|\R|^2)$ binary counter formulas. 
To explain the meaning of~\eqref{eq:core-RI-1}--\eqref{eq:core-RI-2}, consider any $w \in \Delta^\I$ in a model $\I$ of $\K$. If  $w \in (\exists R)^{\I(n)}$ for infinitely many $n > 0$, then $w \in (\exists S)^{\I(n)}$ for all $n$, which is captured by~\eqref{eq:core-RI-1}. Otherwise, there is $n$ such that $w \in (\exists R)^{\I(n)}$ and $w \notin (\exists R)^{\I(m)}$, for $m>n$, whence $w \in (\exists S)^{\I(k)}$, for any $k < n+ \max_{R,S}$, which is captured  by~\eqref{eq:core-RI-2}.
The \LTL{} translation~$\Psi'_\K$ of $\K$ is a conjunction of~\eqref{eq:TBox}--\eqref{eq:unrav1} and \eqref{eq:core-RI-conn}, as well as~\eqref{eq:core-RI-1}--\eqref{eq:core-RI-2} for all $R$ and $S$ with defined $\max_{R,S}$ and their counterparts for $\exists R \sqsubseteq \nxt^{\min_{R,S}} \Rbox \exists S$.
One can show that $\K$ is satisfiable iff $\Psi'_\K$ is satisfiable.

The \PSpace{} lower bound follows from the fact that $\LTL\Xnext_\horn$ is \PSpace-complete and every $\LTL\Xnext_\horn$-formula is equisatisfiable with some  $\DL_{\horn/\core}\Xnext$ KB.
\end{proof}

\subsection{Consistency of $\smash{\DL\Xnext_{\bool/\krom}}$ KBs}\label{sec:KromRI:sat}

Let  $\K = (\T \cup \R,\A)$ be a $\smash{\DL\Xnext_{\bool/\krom}}$ KB. We assume that $\K$ has no nested temporal operators and that, in RIs of the form $\top \sqsubseteq R_1 \sqcup R_2$ and $R_1 \sqcap R_2\sqsubseteq \bot$ from $\R$, both $R_i$ are plain (atemporal) roles.
We construct a \FOLTLI{} sentence $\Phi_\K$ with one variable~$x$. First, we set $\Phi_\K = \bot$ if $(\R, \A)$ is inconsistent, which can be checked in polynomial time by Proposition~\ref{prop:RBox:consistency} and Lemma~5.3 by~\citeA{DBLP:journals/tocl/ArtaleKRZ14}.
If, however, $(\R, \A)$ is consistent, then we treat  basic concepts in $\K$ as unary predicates (see Lemma~\ref{lemma:horn-red})
and define $\Phi_\K$ as a conjunction of the following \FOLTLI{} sentences:
\begin{align}
& \Box\forall x\, \bigl(C_1(x) \land \dots \land C_k(x) \to  C_{k+1}(x) \lor \dots \lor C_{k+m}(x)\bigr), \\\nonumber{} &&&\hspace*{-9em}  \text{ for  all } C_1 \sqcap \dots \sqcap C_k \sqsubseteq C_{k+1} \sqcup \dots \sqcup C_{k+m} \text{ in }\T,\\
\label{R1}
& \Box \forall x\,  \bigl(\exists S_1(x) \lor \exists S_2(x)\bigr)\ \land \  \Box \bigl(\forall x\,  \exists S_1(x) \lor \forall x\, \exists S_2^-(x)\bigr), && \text{ for all  } \top \sqsubseteq S_1 \sqcup S_2 \text{ in }\R,\\
\label{R8}
& \Rnext^\ell A(a), && \text{ for all } A(a,\ell) \text{ in }\A,\\
\label{R8a}
&   \Rnext^\ell \exists P(a) \text{ and } \Rnext^\ell \exists P^-(b), &&\text{ for all } P(a,b,\ell) \text{ in } \A,\\
\label{R7}
& \Box \bigl(\exists x\, \exists P(x) \leftrightarrow \exists x\, \exists P^-(x)\bigr), && \text{ for all role names } P \text{ in }\T,\\
\label{R5}
& \Box\forall x\,  \bigl(\nxt_1 \exists S_1(x) \to \nxt_2 \exists S_2(x)\bigr), && \text{ for all } \nxt_1 S_1 \sqsubseteq \nxt_2 S_2 \\*\nonumber &&& \hspace*{1em}\text{ with  }\R \models \nxt_1 S_1 \sqsubseteq \nxt_2 S_2,
\end{align}
where 
each $\nxt_i$ is $\Rnext$, $\Lnext$ or blank, and $S_1$ can be $\top$ and $S_2$ can be $\bot$ (we assume that if $S_1$ is $\top$, then $\exists S_1(x)$ is simply $\top$ rather than an atom with a unary predicate; similarly, $\exists S_2(x)$ is $\bot$ if $S_2$ is $\bot$).
Note that the problem of checking whether $\R \models \nxt_1 S_1 \sqsubseteq \nxt_2 S_2$ is in~\PTime~\cite[Lemma 5.3]{DBLP:journals/tocl/ArtaleKRZ14}, and so $\Phi_\K$ can be constructed in polynomial time.

\begin{lemma}\label{red}
A $\DL\Xnext_{\bool/\krom}$ KB $\K$ is consistent iff $\Phi_\K$ is satisfiable.
\end{lemma}
\begin{proof}
$(\Rightarrow)$ Suppose $\I \models \K$. Treating $\I$ as a temporal FO-interpretation, we show that $\I \models \Phi_\K$. The only non-standard sentences are~\eqref{R1}. Suppose $\I \models\top \sqsubseteq S_1 \sqcup S_2$ and $\I,n \not\models \exists S_1(d)$, for some $d\in\Delta^\I$ and $n\in\Z$. Then, for every $e\in \Delta^\I$, we have $\I,n \models S_2(d,e)$, and so $\I,n \models \exists S_2^-(e)$.

\smallskip

$(\Leftarrow)$ Suppose $\mathfrak{M} \models \Phi_\K$. We require the following property of $\mathfrak{M}$, which follows, for any RI $\top \sqsubseteq S_1 \sqcup S_2$ in $\R$, from~\eqref{R1}: for any $n \in \Z$ and $d,e \in \Delta^{\mathfrak{M}}$, either
\begin{equation}\label{eq:balanced}
  \mathfrak{M}, n \models \exists S_1(d) \text{ and } \mathfrak{M}, n \models \exists S_1^-(e) \quad \text{ or }\quad
   \mathfrak{M}, n \models \exists S_2(d) \text{ and }  \mathfrak{M}, n \models \exists S_2^-(e).
\end{equation}
We construct a model $\I$ of $\K$ in a step-by-step manner, regarding $\I$ as a set of ground atoms. To begin with, we put in $\I$ all $P(a,b,n) \in \A$ and then proceed in three steps.

\emph{Step} 1: If $\top \sqsubseteq S_1 \sqcup S_2$ is in $\R$ with $\mathfrak{M}, n \models \exists S_1(a)$ and either $\mathfrak{M}, n \not\models \exists S_2(a)$ or $\mathfrak{M}, n \not\models \exists S_2^-(b)$, for $n \in \Z$ and $a,b \in \ind(\A)$, then, by~\eqref{eq:balanced}, $\mathfrak{M}, n \models \exists S_1^-(b)$, and we add $S_1(a,b,n)$ to $\I$. We do the same for $\exists S_1^-$, $\exists S_2$ and $\exists S_2^-$ (recall that $\R$ is closed under role inverses).
We now show that the constructed interpretation $\I$ is consistent with $\R$. Suppose otherwise, that is, there are some $P(a,b,n)$ and $R(a,b,n)$ in $\I$ with $P \sqcap R \sqsubseteq \bot$ in $\R$. Two cases need consideration. $(i)$~If both atoms were added at Step~1 because of some RIs $\top \sqsubseteq P \sqcup Q$ and $\top \sqsubseteq R \sqcup S$, then $\R \models P \sqsubseteq S$ and so, by~\eqref{R5},
$\mathfrak{M},n \models \exists S(a)$ and $\mathfrak{M},n \models \exists S^-(b)$, contrary to the definition of Step 1. $(ii)$~Otherwise, as $\A$ is consistent with~$\R$, the only other possibility is that $R(a,b,n) \in \A$ and $P(a,b,n)$ was added at Step~1 because of some RI $\top\sqsubseteq P \sqcup Q$. In this case, $\R \models R \sqsubseteq Q$, whence, by~\eqref{R8a}, $\mathfrak{M},n \models \exists Q(a)$ and $\mathfrak{M},n \models \exists Q^-(b)$, contrary to $P(a,b,n)$ being added at Step~1.

\emph{Step} 2: For all roles $P$ and $R$ and all $n,k\in\Z$, if $P(a,b,n) \in \I$ and $\R \models P \sqsubseteq \nxt^k R$, then we add $R(a,b,n+k)$ to~$\I$. We show that the resulting $\I$ remains consistent with $\R$. Suppose otherwise, that is, there are  $R_1(a,b,n), R_2(a,b,n) \in \I$ with $R_1 \sqcap R_2 \sqsubseteq \bot$ in $\R$. Suppose that $R_i(a,b,n)$, for $i = 1,2$, was added to $\I$ for $\R \models P_i \sqsubseteq \nxt^{k_i} R_i$ and a $P_i$-atom constructed at Step 1. As $\R \models \neg R_i \sqsubseteq \neg \nxt^{-k_i} P_i$, we arrive to a contradiction with the consistency of $\I$ at Step~1.

\emph{Step} 3: For each RI $\top \sqsubseteq P \sqcup Q$ in $\R$, each $a,b\in\ind(\A)$ and each $n\in\Z$ such that $\mathfrak{M}, n \models \exists P(a)$, $\mathfrak{M}, n \models \exists P^-(b)$, $\mathfrak{M}, n \models \exists Q(a)$, $\mathfrak{M}, n \models \exists Q^-(b)$, but neither $P(a,b,n)$ nor $Q(a,b,n)$ are in~$\I$, we add one of them, say $P(a,b,n)$, to~$\I$. The result remains consistent with $\R$: indeed, if we had $S(a,b,n)\in\I$ with $P \sqcap S \sqsubseteq \bot$ in $\R$, then  $Q(a,b,n)$ would have been added to $\I$ at Step 2 because $\R \models S \sqsubseteq Q$. We take the closure of $P(a,b,n)$ as at Step~2 and repeat the process.

We conclude the first stage of constructing $\I$ by extending it with all $B(a,n)$ such that $\mathfrak{M},n \models B(a)$, for $n \in \Z$ and $a \in \ind(\A)$. By construction, $P(a,b,n) \in \I$ implies $\exists P(a,n)$ and  $\exists P^-(b,n) \in \I$, but not necessarily the other way round.
So, suppose $\exists P(a,n) \in \I$ but there is no $P(a,b,n)$ in $\I$. Take a fresh individual $w_P$, add it to the domain of $\I$ and add $P(a,w_P,n)$ to $\I$. By~\eqref{R5} and~\eqref{R7}, the result is consistent with $\R$. By~\eqref{R7}, there is $d \in \Delta^{\mathfrak{M}}$ with $\mathfrak{M},n \models \exists P^-(d)$. So, we add $B(w_P,m)$ to $\I$ for each basic concept $B$ and $m \in \Z$ with $\mathfrak{M},m \models B(d)$. We then apply to $\I$ the three-step procedure described above and repeat this \emph{ad infinitum}.

It is readily seen that the obtained  interpretation $\I$ is a model of $\K$ (the complete argument is left to the reader, noting that a similar
%\nb{RG14: I do not find that having a forward reference to the inside of another proof to conclude this proof is very nice to the reader.\\moved here} 
unravelling construction is given in detail in the proof of Lemma~\ref{lem:unravelling}).
\end{proof}

It follows from the proof of Lemma~\ref{red} that it is always possible to construct a model~$\I$ of $\R$ from a model $\mathfrak{M}$ of $\Phi_\K$ if $\mathfrak{M}$ satisfies the domain/range restrictions for the roles in~$\R$ encoded by~\eqref{R1} and~\eqref{R5}. One reason why this encoding is enough is that, e.g., the CI $\Rnext \exists P \sqsubseteq \exists Q$ is sufficient to capture the effect of the RI $\Rnext P \sqsubseteq Q$ on the domains/ranges. However, $\Rbox P \sqsubseteq Q$ does not entail $\Rbox \exists P \sqsubseteq \exists Q$, and it is not clear what domain/range axioms can be used to capture the impact of the RI on domains/ranges in the presence of $\Box$-operators. The complexity of the consistency problem for $\DL\Xbox_{\bool/\krom}$ remains open.

\begin{theorem}\label{thm:bool-krom:hardness}
Checking consistency of $\DL\Xnext_{\bool/\krom}$ and $\DL\Xnext_{\horn/\krom}$ KBs is\linebreak \ExpSpace-complete. 
\end{theorem}
\begin{proof}
The upper bound follows from Lemma~\ref{red}. We prove the matching lower bound by reduction of the non-halting problem for deterministic Turing machines with exponential tape, as in the proof of Theorem~\ref{th:bool-sat-expspace}. We also make the same assumptions as in that proof and use the concept names $C_a$ and $H_{q,a}$ with the same meaning. This time, though, the computation is encoded on $2^n$ successors of a single ABox element, $e$; see Fig.~\ref{fig:horn-krom:hardness}.

\begin{figure}[t]
\centerline{%
\begin{tikzpicture}[xscale=1.25, yscale=0.7, 
rpoint/.style={rectangle,semithick,fill=gray!20,draw=black,minimum size=3mm,inner sep=0pt, 
                 append after command={  (\tikzlastnode.north east) edge[draw](\tikzlastnode.south west) }, 
                 append after command={  (\tikzlastnode.south east) edge[draw](\tikzlastnode.north west) }}]\footnotesize
\begin{scope}\small
\draw[object-timeline] (0.5, 0) -- ++(11.2,0); \node at (0.1, 0) {$e^\I$};
\draw[object-timeline] (0.5, 1.5) -- ++(11.2,0); \node at (0.1, 1.5) {$v_1$};
\draw[object-timeline] (0.5, 3) -- ++(11.2,0); \node at (0.2, 3) {$\dots$};
\draw[object-timeline] (0.5, 4.5) -- ++(11.2,0); \node at (0.2, 4.5) {$v_{2^n - 1}$};
\draw[object-timeline] (0.5, 6) -- ++(11.2,0); \node at (0.2, 6) {$v_{2^n}$};
\end{scope}
\begin{scope}\scriptsize
\draw[time-guideline] (1,-0.6) -- ++(0,7.1); \node at (1,-0.9) { $0$};
\draw[time-guideline] (2,-0.6) -- ++(0,7.1); \node at (2,-0.9) {$1$};
\draw[time-guideline] (3,-0.6) -- ++(0,7.1); \node at (3,-0.9) {$\dots$};
\draw[time-guideline] (4,-0.6) -- ++(0,7.1); \node at (4,-0.9) {$2^n - 1$};
\draw[time-guideline] (5,-0.6) -- ++(0,7.1); \node at (5,-0.9) {$2^n$};
\draw[time-guideline] (6,-0.6) -- ++(0,7.1); \node at (6,-0.9) {$2^n + 1$};
\draw[time-guideline] (7,-0.6) -- ++(0,7.1); \node at (7,-0.9) {$\dots$};
\draw[time-guideline] (8,-0.6) -- ++(0,7.1); \node at (8,-0.9) {$2\cdot 2^n - 1$};
\draw[time-guideline] (9,-0.6) -- ++(0,7.1); \node at (9,-0.9) {$2 \cdot 2^n$};
\draw[time-guideline] (10,-0.6) -- ++(0,7.1); \node at (10,-0.9) {$2 \cdot 2^n +1$};
\draw[time-guideline] (11,-0.6) -- ++(0,7.1); \node at (11,-0.9) {$\dots$};
\end{scope}
\begin{scope}[semithick]\scriptsize
\draw (1.8,-1.1) -- ++(0,-0.2) -- ++(3.4,0) node[midway,below] {initial configuration} -- ++(0,0.2);
\draw (5.8,-1.1) -- ++(0,-0.2) -- ++(3.4,0) node[midway,below] {configuration 2} -- ++(0,0.2);
\draw (9.8,-1.1) -- ++(0,-0.2) -- ++(1.9,0) node[pos=0.55,below] {configuration 3};
\end{scope}
\node[rpoint, label=below left:{$I,F$}, label={[above,fill=white]{$C_{\#}$}}] (a1) at (1,0) {};
\node[qpoint, label=below left:{$F$}] (a2) at (2,0) {};
\node[rpoint, label={[above,fill=white]{$H_{q_0,a_1}$}}] (b2) at (2,1.5) {};
\node[qpoint, label=below left:{$F$}] (a3) at (3,0) {};
\node[rpoint, label={[above,fill=white]{$C_{a_i}$}}] (b3) at (3,3) {};
\node[ppoint, label=below left:{$F$}, label=below right:{$E$}] (a4) at (4,0) {};
\node[rpoint, label={[above,fill=white]{$C_{\#}$}}] (b4) at (4,4.5) {};
\node[ppoint, label=below left:{$F$}, label=below right:{$E$}] (a5) at (5,0) {};
\node[rpoint, label={[above,fill=white]{$C_{\#}$}}] (b5) at (5,6) {};
\node[qpoint, label=below right:{$E$}] (a6) at (6,0) {};
\node[rpoint, label={[above,fill=white]{$C_{a}$}}] (b6) at (6,1.5) {};
\node[qpoint, label=below right:{$E$}] (a7) at (7,0) {};
\node[rpoint, label={[above,fill=white]{$C_{a'}$}}] (b7) at (7,3) {};
\node[qpoint, label=below right:{$E$}] (a8) at (8,0) {};
\node[rpoint, label={[above,fill=white]{$C_{\#}$}}] (b8) at (8,4.5) {};
\node[qpoint, label=below right:{$E$}] (a9) at (9,0) {};
\node[rpoint, label={[above,fill=white]{$C_{\#}$}}] (b9) at (9,6) {};
\node[qpoint, label=below right:{$E$}] (a10) at (10,0) {};
\node[rpoint, label={[above,fill=white]{$C_{a''}$}}] (b10) at (10,1.5) {};
\node[qpoint, label=below right:{$E$}] (a11) at (11,0) {};
\node[rpoint, label={[above,fill=white]{$C_{a'''}$}}] (b11) at (11,3) {};
\begin{scope}[ultra thick]\scriptsize
\draw[->,out=45,in=135,looseness=20] (a1.east) to node[above,fill=white] {$QC_{\#}$} (a1.west);
\draw[->] (a2) to node[left] {$\boldsymbol P$} node[right,pos=0.42] {$QH_{q_0,a_1}$} (b2);
\draw[->] (a3) to node[left,pos=0.4] {$\boldsymbol P$} node[right,pos=0.37] {$QC_{a_i}$} (b3);
\draw[->] (a4) to node[left] {$\boldsymbol P_{\#}$} node[right,pos=0.5] {$QC_{\#}$} (b4);
\draw[->] (a5) to node[left,pos=0.63] {$\boldsymbol P_{\#}$} node[right,pos=0.63] {$QC_{\#}$} (b5);
\draw[->] (a6) to node[right,pos=0.42] {$QC_{a}$} (b6);
\draw[->] (a7) to node[right,pos=0.36] {$QC_{a'}$} (b7);
\draw[->] (a8) to node[right,pos=0.48] {$QC_{\#}$} (b8);
\draw[->] (a9) to node[right,pos=0.63] {$QC_{\#}$} (b9);
\draw[->] (a10) to node[right,pos=0.42] {$QC_{a''}$} (b10);
\draw[->] (a11) to node[right,pos=0.36] {$QC_{a'''}$} (b11);
\end{scope}
\begin{scope}[<->,dashed,thick]
\draw(b2) -- (b6);
\draw(b6) -- (b10);
\draw(b3) -- (b7);
\draw(b7) -- (b11);
\draw(b4) -- (b8);
\draw(b5) -- (b9);
\draw (b2) -- (b3);
\draw (b3) -- (b4);
\draw (b4) -- (b5);
\draw (b6) -- (b7);
\draw (b7) -- (b8);
\draw (b8) -- (b9);
\draw (b10) -- (b11);
\draw[<-] (b8) -- (11.7, 4.5);
\draw[<-] (b9) -- (11.7, 6);
\draw[<-] (b11) -- (11.7, 3.95);
\end{scope}
\end{tikzpicture}
}%
\caption{Structure of models in the proof of Theorem~\ref{thm:bool-krom:hardness}. The horizontal dashed arrows connect representations of the same cell in successive configurations, and diagonal dashed arrows connect successive cells in the same configuration. Roles $SC_a$ and $SH_{q,a}$ are not shown, and roles $QC_a$ and $QH_{q,a}$ are only partially depicted: if, for example, $QC_a$ contains some $(e^\I,v_i)$ at $t$, then it contains all $(u,v)$ at~$t$.}\label{fig:horn-krom:hardness}
\end{figure}
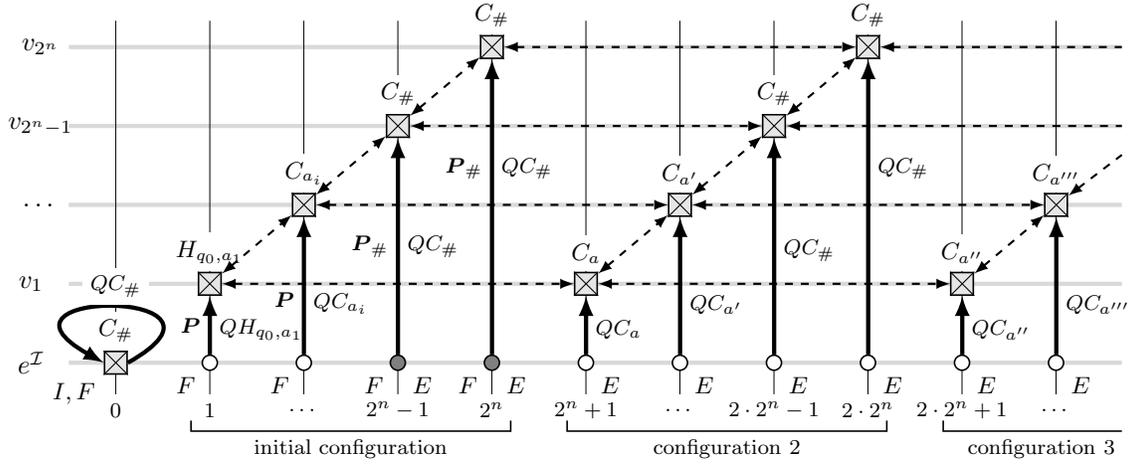%

If the input $\avec{a}$ is $a_1,\dots,a_m$ and $q_0$ is the initial state of the Turing machine $M$, then the ABox consists of the following:
\begin{multline*}
I(e, 0), \ C_{\#}(e,0), \ P(e, v_1, 1), \ H_{q_0,a_1}(v_1, 1), \ P(e,v_2,2), \ C_{a_2}(v_2, 2), \ \dots,\\ P(e,v_m,m), \ C_{a_m}(v_m, m), \ E(e, m + 1),
\end{multline*}
where concept names $I$ and $E$ mark the start and the end of the input $\avec{a}$, respectively. Then we use the CIs
\begin{equation*}
I \sqsubseteq \Rnext^{2^n} F,\qquad
F  \sqsubseteq \Lnext F,\qquad
E  \sqsubseteq \Rnext E,\qquad
E \sqcap F  \sqsubseteq \exists P_{\#}, \qquad 
\exists P_{\#} \sqsubseteq C_{\#}
\end{equation*}
to fill the rest of the tape in the initial configuration with $\#$, similarly to the proof of Theorem~\ref{th:bool-sat-expspace}. So, we obtain $m$ `named' $P$-successors $v_1,\dots, v_m$ of $e$ and $2^n - m$ `anonymous' $P_{\#}$-successors of $e$, which will be denoted $v_{m+1},\dots,v_n$. Each of the successors is labelled with the $H_{q,a}$ or $C_a$ from the initial configuration of $M$ on $\avec{a}$. Next, we use a Krom RI and CI to `propagate' the information from the labels $C_a$ and $H_{q,a}$ to other domain elements: for example,
\begin{equation*}
\top \sqsubseteq SC_a \sqcup QC_a \qquad\text{ and }\qquad C_a \sqcap \exists SC_a^- \sqsubseteq \bot
\end{equation*}
mean that, if any element $v$ is labelled by $C_a$ at moment $n$, then any pair $(v',v)$ must belong  to role $QC_a$ at moment $n$. We use similar RIs and CIs with roles $SH_{q,a}$ and  $QH_{q,a}$ for labels~$H_{q,a}$. Moreover, we state that the domains of distinct $QC_a$ and $QH_{q,a}$ are disjoint:
\begin{multline*}
\exists QC_a \sqcap \exists QC_{a'} \sqsubseteq \bot, \text{ for } a\ne a', \qquad \exists QC_a\sqcap \exists QH_{q,a} \sqsubseteq \bot,\\ \exists QH_{q,a} \sqcap \exists QH_{q'a'} \sqsubseteq \bot, \text{ for } (q,a)\ne (q',a'). 
\end{multline*}
This ensures that if a successor $v_i$ of $e$ is labelled with $C_a$, then $(e,v_i)$ and indeed every $(e,v_j)$ belongs to $QC_a$; the same applies to labels $H_{q,a}$ and roles $QH_{q,a}$. Next, we encode computations of $M$ with tape alphabet $\Gamma$ and transition function $\delta\colon Q\times\Gamma \to Q \times \Gamma \times \{R, L\}$ by means of the following CIs, for  $\delta(q,a) = (q',b,R)$ and $a',a''\in\Gamma$:
\begin{align*}
 C_{a'} \sqcap \Rnext \exists QH^-_{q,a}  & \sqsubseteq  \Rnext^{2^n} C_{a'},\\
 H^-_{q,a} & \sqsubseteq  \Rnext^{2^n} C_{b},\\
 \Lnext \exists QH^-_{q,a} \sqcap  C_{a''} & \sqsubseteq  \Rnext^{2^n} H_{q', a''}.
%
% \Lnext C_{a'}\sqcap  H_{q,a} \sqcap \Rnext C_{a''} & \sqsubseteq \Rnext^{2^n} (\Lnext H_{q', a'}\sqcap  C_{b} \sqcap \Rnext C_{a''}), && \text{for } \delta(q,a) = (q',b,L)\text{ and } a',a''\in\Gamma,\\
 %
%  \Lnext C_{a'} \sqcap C_a \sqcap \Rnext C_{a''} & \sqsubseteq \Rnext^{2^n} C_a, && \text{for } a,a',a''\in\Gamma,
\end{align*}
Note that only one of the  concepts on the left-hand side of the CIs is a label ($C_a$ or $H_{q,a}$), while the other (if present) is the range of the respective role: this ensures that labels occur only every $2^n$ moments on each of the successors $v_i$. We use similar CIs for $\delta(q,a) = (q',b,L)$ and $a',a''\in\Gamma$. We also need the following CIs, for  $a,a',a''\in\Gamma$:
\begin{align*}
\Lnext \exists QC_{a'} \sqcap C_a \sqcap \Rnext \exists QC_{a''} & \sqsubseteq \Rnext^{2^n} C_a,
\end{align*}
to ensure that the contents of the tape does not change unless it's overwritten by the head of the Turing machine.
Finally, CIs
\begin{equation*}
H_{q,a} \sqsubseteq \bot, \hspace*{2em} \text{for each accepting and rejecting state } q \text{ and } a\in\Gamma,
\end{equation*}
ensures that accepting and rejecting states never occur in the encodings of computations.
These CIs are of exponential size, and our next task is to show how to convert them into a $\DL_{\horn/\krom}\Xnext{}$ ontology of polynomial size.

Consider, for example, a CI $I \sqsubseteq \Rnext^{2^n} F$. We express it using the following CIs:
\begin{align*}
  & I \sqsubseteq \Rnext (\bar B_{n-1} \sqcap \dots \sqcap \bar B_0), \\
  & B_{n-1} \sqcap \dots \sqcap B_{0} \sqsubseteq F,\\
& \bar B_{k}  \sqcap B_{k-1} \sqcap \dots \sqcap B_{0}\sqsubseteq \Rnext (B_{k} \sqcap  \bar B_{k-1}\sqcap \dots \sqcap \bar B_0), && \text{ for } 0 \leq k < n, \\
%
%&& \bar P_{k} \sqcap P_{k-1} \sqcap \dots \sqcap P_{0} \sqsubseteq , && \text{for } 0 \leq k < n,  \\
%
& \bar B_{j} \sqcap \bar B_{k} \sqsubseteq \Rnext \bar B_{j}  \quad \text{ and } \quad  B_{j} \sqcap \bar B_{k} \sqsubseteq \Rnext B_{j}, && \text{ for }  0 \leq k < j < n,
\end{align*}
which have to be converted into normal form~\eqref{axiom1}. Intuitively, they encode a binary counter from $0$ to $2^n - 1$, where~$\bar B_i$ and~$B_i$ stand for `the $i$th bit of the counter is $0$' and, respectively, `the $i$th bit of the counter is $1$'\!. Other CIs of the form $C_1\sqsubseteq \Rnext^{2^n} C_2$ are handled similarly, each with a fresh set of concept names $\bar B_i$ and~$B_i$.
\end{proof}

The next result follows from Lemma~\ref{red} and an observation that $\Phi_\K$ is a formula in the Krom fragment of \FOLTLI{}, the satisfiability in which is known to be \PSpace-complete~\cite[Theorem~9]{DBLP:conf/time/ArtaleKLWZ07}: 
\begin{theorem}\label{thm:krom:qtl-krom}
Checking consistency of $\DL\Xnext_{\krom}$ KBs is in \PSpace. 
\end{theorem}

% !TEX root = TDL-Lite.tex

\section{Rewriting $\DL_{\bool/\horn}\Xallop$ OMAQs by Projection to \LTL}\label{sec:DL-Lite}

Our aim in this section is to identify classes of FO-rewritable $\DL_{\bool/\horn}\Xallop$ OMAQs, which will be done by projecting them to \LTL{} OMAQs and using the classification of~\citeA{AIJ21}. Initial steps in this direction have been made in Section~\ref{sec:proj} for OMAQs without $\bot$ and interacting concepts and roles. The results of this section are summarised in Table~\ref{TDL-table-omaq}.

\begin{table}[t]
\centering%
\tabcolsep=4pt
\begin{tabular}{ccc}\toprule
 \rule[-3pt]{0pt}{13pt} & $\DL_{\frag/\fragr}\Xbox$ &  $\DL_{\frag/\fragr}\Xnext$  {\footnotesize and} $\DL_{\frag/\fragr}\Xallop$  \\\midrule
$*/(\mathit{g}\text{-})\bool$ & \coNP-hard {\scriptsize [Th.~\ref{app:th:coNP-gbool}]} & undecidable {\scriptsize [Th.~\ref{thm:undec}]} \\\midrule
$*/\krom$ & ? & ? \\\midrule
$\bool/\horn$, $\horn$ & \FO(\RPR) {\scriptsize [Th.~\ref{cor:dl-omaq-rewritability}/\ref{cor:role-ompiqs}]}, {\small \NCo-hard} {\scriptsize [Th.~\ref{thm:unexpected}\,(\emph{i})]}
& \multirow{4}{*}[-1\dimexpr \aboverulesep + \belowrulesep + \cmidrulewidth]{\renewcommand{\tabcolsep}{0pt}\begin{tabular}{c}\FO(\RPR)  {\scriptsize [Th.~\ref{cor:dl-omaq-rewritability}\,(\emph{ii})/\ref{cor:role-ompiqs}\,(\emph{iii})]}\\[4pt] {\small \NCo-hard}\\[-4pt] {\scriptsize \cite[Th.~10]{AIJ21}}\end{tabular}} \\ \cmidrule(lr){2-2}
$\krom/\horn$, $\core/\horn$ & $\FOE$ {\scriptsize [Th.~\ref{cor:dl-omaq-rewritability}~(\emph{i})/\ref{cor:role-ompiqs}\,(\emph{ii}) \& Th.~\ref{thm:unexpected}\,(\emph{ii})]}  &  \\ \cmidrule(lr){2-2}
$*/\rhorn$  & \multirow{3}{*}[-0.5\dimexpr \aboverulesep + \belowrulesep + \cmidrulewidth]{$\FO(<)$ {\scriptsize [Th.~\ref{ex:dl-core-mon-fo}/\ref{cor:role-ompiqs}\,(\emph{i})]}} & \\[3pt]
$\bool/\core$, $\horn/\core$  & & \\\cmidrule(lr){3-3}
$\krom/\core$, $\core$  & &  $\FOE$ {\scriptsize [Th.~\ref{cor:dl-omaq-rewritability}~(\emph{i})/\ref{cor:role-ompiqs}\,(\emph{ii})]}
 \\\bottomrule
\end{tabular}
\caption{Rewritability and data complexity of $\DL_{\frag/\fragr}^{\op}$ OMAQs and \OMPIQ{}s of the form~$(\TO,\varrho)$ with a positive temporal role $\varrho$, where * denotes any of $\bool$, $\horn$, $\krom$ or $\core$.
}
\label{TDL-table-omaq}
\end{table}

We begin by showing how to get rid of $\bot$ from $\DL_{\bool/\horn}\Xallop$ ontologies.

\begin{lemma}\label{lemma:consistency}
Let $\lang$ be one of $\FO(<)$, $\FOE$, or $\FO(\RPR)$, $\fragr \in \{\core,\rhorn,\horn\}$.
Suppose $\TO = \T \cup \R$ is a $\DL_{\frag/\fragr}^\op$ ontology. Denote by $\TO'$ the $\bot$-free $\DL_{\frag/\fragr}^\op$ ontology obtained by removing all  disjointness axioms $\vartheta_1 \sqcap \dots \sqcap \vartheta_k \sqsubseteq \bot$ from $\TO$, and let $\varkappa_\bot' = \Ldiamond\Rdiamond\varkappa_\bot$  and $\varrho_\bot' = \Ldiamond\Rdiamond\varrho_\bot$ with
\begin{equation*}
\varkappa_\bot \ = \ \hspace*{-2.5em}\bigsqcup_{C_1 \sqcap \dots \sqcap C_k \sqsubseteq \bot \text{ in } \T}\hspace*{-2.5em} (C_1 \sqcap \dots \sqcap C_k) \ \sqcup \hspace*{-1.5em}\bigsqcup_{\begin{subarray}{c}S \text{ is a role with}\\\TO\models S\sqsubseteq \bot\end{subarray}}\hspace*{-1.75em} \exists S 
\qquad\text{ and }\qquad
\varrho_\bot  \ =  \ \hspace*{-2.5em}\bigsqcup_{R_1 \sqcap \dots \sqcap R_k \sqsubseteq \bot \text{ in } \R}\hspace*{-2.5em} (R_1 \sqcap \dots \sqcap R_k).
\end{equation*}
Assume that $\rew^\T_\bot(x,t)$ and $\rew^\R_\bot(x,y,t)$ are $\lang$-rewritings of the \OMPIQ{}s $(\TO', \varkappa'_\bot)$ and $(\TO',\varrho'_\bot)$, respectively, and 
$\chi_\bot = \exists x,t\,\rew_\bot^\T(x, t)\  \lor \  \exists x,y,t\,\rew_\bot^\R(x,y, t)$.
Then the following hold\emph{:} 
\begin{itemize}
\item[$(i)$] for any basic concept~$B$, if $\rew'(x,t)$ is an $\lang$-rewriting of the \OMPIQ{} $(\TO',  \varkappa'_\bot \sqcup B)$, then
$\rew'(x,t)  \lor  \chi_\bot$
is an $\lang$-rewriting of the \OMAQ{}~$(\TO,B)$\textup{;}
\item[$(ii)$]  for any positive temporal role $\varrho$, if $\rew'(x,y,t)$ is an $\lang$-rewriting of the \OMPIQ{} $(\TO',\varrho)$, then
$\rew'(x,y,t)   \lor \chi_\bot$
is an $\lang$-rewriting of the \OMPIQ{}~$(\TO,\varrho)$.
\end{itemize}
\end{lemma}
\begin{proof}
As we establish below, $\chi_\bot$ is true when evaluated over an ABox $\A$ (that is, $\SA \models \chi_\bot$) iff $\TO$ and $\A$ are inconsistent. 
Thus, it suffices to show that, for any ABox $\A$,
\begin{itemize}
\item[--] $\TO$ and $\A$ are consistent iff $\ans(\TO',\varkappa'_\bot,\A) = \emptyset$ and $\ans(\TO',\varrho'_\bot,\A) = \emptyset$;

\item[--] if $\TO$ and $\A$ are consistent, then $\ans^{\Z}(\TO,B,\A) = \ans^{\Z}(\TO', \varkappa'_\bot \sqcup B,\A)$, for any basic concept $B$, and $\ans^{\Z}(\TO,\varrho,\A) = \ans^{\Z}(\TO', \varrho,\A)$, for any positive temporal role~$\varrho$.
\end{itemize}
If $\TO$ and $\A$ are consistent, any model $\I$ of $\TO$ and~$\A$ is trivially a model of $\TO'$ with $\varkappa_{\smash{\bot}}^{\I(n)} = \emptyset$ and $\varrho_{\smash{\bot}}^{\I(n)} = \emptyset$, for all $n\in \Z$. So  $\ans(\TO',\varkappa'_{\smash{\bot}},\A) = \emptyset$ and $\ans(\TO',\varrho'_{\smash{\bot}},\A) = \emptyset$. Also, we clearly have $\ans^{\Z}(\TO,B,\A) \supseteq \ans^{\Z}(\TO',\varkappa'_\bot \sqcup B,\A)$ and $\ans^{\Z}(\TO,\varrho,\A) \supseteq \ans^{\Z}(\TO',\varrho,\A)$. 

Next, suppose $\ans(\TO', \varkappa'_\bot,\A) = \emptyset$ and $\ans(\TO', \varrho'_\bot,\A)=\emptyset$. We show how to construct a model $\I$ of $(\TO,\A)$. By definition, $\TO'$ and $\A$ are consistent.  Since $\ans(\TO', \varkappa'_\bot,\A) = \emptyset$,  for each $a \in \ind(\A)$, there is a model $\I_a$ of~$(\TO',\A)$ such that $a^{\I_{a}} \notin \varkappa_{\smash{\bot}}^{\I_{a}(n)}$ for all~$n \in \Z$. Also, for each role $S$ consistent with $\TO$, there is a model~$\I_S$ of $(\TO',\{S(w,u,0)\})$ with $w^{\I_S} \notin \varkappa_{\smash{\bot}}^{\I_S(n)}$ for all $n \in \Z$.
We take, for each $a \in \ind(\A)$, the beam $\beam_a$ for $a^{\I_a}$ in $\I_a$, and, for each role $S$ consistent with $\TO$, the beam $\beam_{w_S}$ of $w^{\I_S}$ in~$\I_S$. Observe that  $\varrho_{\smash{\bot}}$ and the second group of disjuncts in $\varkappa_{\smash{\bot}}$ ensure that, for the chosen beams,  there are compatible $\R'$-canonical rods, which also satisfy disjointness axioms in $\R$.
We then apply Lemma~\ref{lem:unravelling} to obtain a model~$\I$ of $(\TO', \A)$. By construction, $\varkappa_{\smash{\bot}}^{\I(n)}= \emptyset$, for all $n\in\Z$, and so $\I\models\T$. Since the $\R'$- and $\R$-canonical rods coincide, we also have $\I\models\R$. Thus, $\I$ is a model of $(\TO,\A)$.

It remains to show that $\ans^{\Z}(\TO,B,\A) \subseteq \ans^{\Z}(\TO',\varkappa'_{\smash{\bot}} \sqcup B,\A)$ for consistent $\TO$ and~$\A$. Suppose $(a,\ell) \notin \ans^{\Z}(\TO', \varkappa'_{\smash{\bot}}\sqcup B,\A)$. Then there is a model $\I_a$ of $(\TO',\A)$ such that $a^{\I_a} \notin (\varkappa'_{\smash{\bot}}\sqcup B)^{\I_{a}(\ell)}$, whence  $a^{\I_a} \notin B^{\I_{a}(\ell)}$ and $a^{\I_a} \notin \varkappa_{\smash{\bot}}^{\I_{a}(n)}$, for all $n\in\Z$. Take the beam $\beam_a$ of $a^{\I_a}$ in $\I_a$. As $\TO$ and $\A$ are consistent, $\ans(\TO',\varkappa'_{\smash{\bot}},\A) = \emptyset$, and so, for every $b \in \ind(\A)\setminus \{ a\}$, there is a model $\I_b$ of $(\TO',\A)$ such that $b^{\I_b} \notin \varkappa_{\smash{\bot}}^{\I_{b}(n)}$, for all $n\in \Z$. Take the beam $\beam_b$  of~$b^{\I_b}$ in $\I_b$. We now construct a model $\I$ of $(\TO,\A)$ with $a^\I\notin B^{\I(\ell)}$ from the chosen beams using Lemma~\ref{lem:unravelling} in the same way as in the previous paragraph. It follows that  $(a,\ell) \notin \ans^{\Z}(\TO, B,\A)$. 
The case of role \OMPIQ{}s is similar except that no individual requires any special treatment like $a$ above.
\end{proof}

It is to be noted that the asymmetricity of the \OMPIQ{}s $(\TO',  \varkappa'_\bot \sqcup B)$ and $(\TO',\varrho)$ in $(i)$ and $(ii)$ above is explained by the fact that $\frag$ can be $\bool$ or $\krom$, while $\fragr$ is always $\horn$. To illustrate, consider $\TO$ with CIs $C \sqsubseteq B \sqcup E$ and  $D \sqcap E \sqsubseteq \bot$. Clearly, $\TO\models C \sqcap D \sqsubseteq B$, and so, over $\A = \{ C(a,0), D(a,0) \}$, which is consistent with $\TO$, we have $(\TO,\A)\models B(a,0)$. On the other hand, $(\TO',\A) \not \models B(a,0)$ because there are models $\I$ of $(\TO',\A)$ with $a^\I \notin B^{\I(0)}$ and $a^\I \in E^{\I(0)}$. The disjunct $\varkappa_\bot' = \Ldiamond\Rdiamond (D \sqcap E)$ is needed to capture such models.

Proposition~\ref{prop:roleOMIQs} allows us to project the $\bot$-free $\DL_{\smash{\frag/\fragr}}^\op$ \OMPIQ{}s of the form $(\TO', \varrho)$ and $(\TO', \varrho_{\smash{\bot}}')$ to  \LTL{} \OMPIQ{}s. Thus, it remains to deal with $\bot$-free $\DL_{\smash{\frag/\fragr}}^\op$ \OMPIQ{}s of the form $(\TO', \varkappa_{\smash{\bot}}'\sqcup B)$ and  $(\TO', \varkappa_{\smash{\bot}}')$. Observe  that $\varkappa_\bot'$ contains no \emph{qualified existential restrictions}---that is, subqueries of the form $\exists S.\lambda$ with $\lambda \ne \top$---and conjunctions in $\varkappa_\bot'$ are prefixed by $\Ldiamond\Rdiamond$ and match the form of CIs allowed by~$\frag$. So we next focus on projecting such \OMPIQ{}s with interacting concepts and roles to the \LTL{} axis.

%******

\subsection{$\lang$-Rewritability of $\DL_{\bool/\horn}\Xallop{}$ OMAQs}\label{sec:rewrOMAQ}

We begin with two examples illustrating the interaction between the DL and temporal dimensions in $\DL_{\bool/\horn}\Xallop{}$ we need to take into account when constructing rewritings.

\begin{example}\label{ex:connexion}\em
Let $\T = \{ \, B \sqsubseteq \exists P, \ \exists Q \sqsubseteq  A\,\}$ and $\R = \{\, P \sqsubseteq \Rnext Q \,\}$. An obvious idea of constructing a rewriting for the OMAQ $\q = (\T \cup \R,\, A)$ would be to find first a rewriting of the $\LTL$ \OMPIQ{} $(\T,A)^\dagger$ obtained from $(\T, A)$ by replacing the basic concepts $\exists P$ and~$\exists Q$ with their surrogates $(\exists P)^\dagger = E_P$ and $(\exists Q)^\dagger= E_Q$, respectively. This would give us the first-order query $A(t) \lor E_Q(t)$. By restoring the intended meaning of $A$ and $E_Q$ (see the proof of Proposition~\ref{prop:conceptOMAQs}), we would then obtain $A(x,t) \lor \exists y \, Q(x,y,t)$. The second step would be to rewrite, using the RBox~$\R$, the atom $Q(x,y,t)$ into $Q(x,y,t) \lor P(x,y,t-1)$. Alas, the resulting formula
\begin{equation*}
A(x,t) \lor \exists y \, \bigl(Q(x,y,t) \lor P(x,y, t-1) \bigr)
\end{equation*}
falls short of being an $\FO(<)$-rewriting of $\q$ as it does not return the certain answer $(a,1)$ over $\A = \{\, B(a,0) \,\}$.
The reason is that, in our construction, we did not take into account the CI $\exists P \sqsubseteq \Rnext \exists Q$, which is a consequence of $\R$. If we now add the `connecting axiom'  $(\exists P)^\dagger \sqsubseteq \Rnext (\exists Q)^\dagger$ to $\T^\dagger$, then in the first step we obtain $A(t) \lor E_Q(t) \lor E_P(t-1) \lor B(t-1)$, which gives us the correct $\FO(<)$-rewriting
\begin{equation*}
A(x,t) \ \ \lor \ \ \exists y \, \big( Q(x,y,t)  \lor P(x,y, t-1)\big) \ \ \lor \ \ \exists y \, P(x,y, t-1) \ \ \lor \ \ B(x, t-1)
\end{equation*}
of $\q$, where the third disjunct is obviously redundant and can be omitted.
\end{example}

\begin{example}\label{ex:shrub}\em
Consider now the OMAQ $\q = (\T \cup \R,\, A)$ with
\begin{equation*}
\T = \{\,\exists Q \sqsubseteq \Lbox A\,\}, \qquad
\R = \{\, P \sqsubseteq \Rbox P_1,\ T \sqsubseteq \Rbox T_1, \ T_1\sqsubseteq \Rbox T_2,\ P_1 \sqcap T_2 \sqsubseteq Q \,\}.
\end{equation*}
The two-step construction outlined in Example~\ref{ex:connexion} would give us first the formula
\begin{equation*}
\Phi(x,t) \ \ =\ \  A(x,t) \ \lor \ \exists t' \, \big( (t < t') \land \exists y \, Q(x,y,t') \big)
\end{equation*}
as a rewriting of $(\T,A)$. It it readily seen that the following formula is a rewriting of $(\R,Q)$:
\begin{multline*}
\Psi(x,y,t') \ \ = \ \  Q(x,y,t') \ \lor{} \
\big( \bigl[ P_1(x,y,t')  \lor  \exists t'' \, \bigl( (t'' < t') \land P(x,y,t'') \bigr) \bigr]  \land{}\\ \bigl[ T_2(x,y,t') \lor \exists t'' \, \bigl( (t'' < t') \land \bigl(T_1(x,y,t'')\lor \exists t'''\,\bigl((t'''<t'') \land T(x,y,t''')\bigr)\bigr) \bigr) \bigr] \bigr).
\end{multline*}
However, the result of replacing $Q(x,y,t')$ in $\Phi(x,t)$ with $\Psi(x,y,t')$ is not an $\FO$-rewriting of $(\TO, A)$: when evaluated over $\A = \{\, T(a,b,0),\ P(a,b,1)\,\}$,  it does not return the certain answers $(a,0)$ and $(a,1)$; see Fig.~\ref{fig:ex:connexion}. (Note that these answers would  be found had we evaluated the obtained `rewriting' over $\Z$ rather than $\{0,1\}$.)

\begin{figure}[t]
\centerline{
\begin{tikzpicture}[xscale=1.2, yscale=0.9]\footnotesize
\foreach \y in {0,1} {
\draw[object-timeline] (-5.4,\y) -- ++(11,0);
}
\foreach \x/\i in {-4/0,-2/1,0/2,2/3,4/4} {
\node at (\x,-0.6) {\scriptsize$\i$};
\draw[time-guideline] (\x, -0.4) -- ++(0,1.7);
}
\fill[gray,fill opacity=0.2] (-4,-0.4) rectangle +(2,1.7);
%\draw[ultra thin] (-4.8,-0.8) -- ++(0,2.1);
%\draw[ultra thin] (-1.2,-0.8) -- ++(0,2.1);
%
%
\node (a) at (-5.8,0) {$a$};
\begin{scope}\scriptsize
\node (a0) at (-4,0) [wpoint]{$A$};
\node (a1) at (-2,0) [wpoint]{$A$};
\node (a2) at (0,0) [wpoint]{$A$};
\node (a3) at (2,0) [wpoint]{$A$};
\node (a4) at (4,0) [wpoint]{$A$};
\end{scope}
\node (c) at (-5.8,1) {$b$};
\node (c0) at (-4,1) [qpoint]{};
\node (c1) at (-2,1) [qpoint]{};
\node (c2) at (0,1) [qpoint]{};
\node (c3) at (2,1) [qpoint]{};
\node (c4) at (4,1) [qpoint]{};
\begin{scope}[ultra thick,draw=black!90]\scriptsize
\draw[->] (a0)  to node [right]{$\boldsymbol{T}$} (c0);
\draw[->] (a1)  to node [right]{$\boldsymbol{P},T_1$}  (c1);
\draw[->] (a2)  to node [right]{$P_1,T_1,T_2,Q$}  (c2);
\draw[->] (a3)  to node [right]{$P_1,T_1,T_2,Q$}  (c3);
\draw[->] (a4)  to node [right]{$P_1,T_1,T_2,Q$}  (c4);
\end{scope}
\end{tikzpicture}%
}%
\caption{A typical model
for $(\T\cup\R, \{T(a,b,0), P(a,b,1)\})$ in Example~\ref{ex:shrub}.}\label{fig:ex:connexion}
\end{figure}
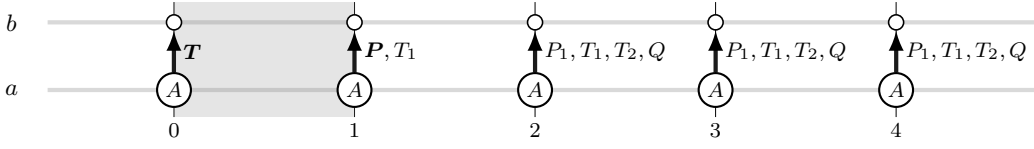

This time, in the two-step construction of the rewriting, we are missing the `consequence' $\exists (\Rbox P_1 \sqcap \Rbox T_2) \sqsubseteq \Rbox\exists Q$ of $\R$ and $\T$.
To fix the problem, we can take a fresh role name~$G_\rtp$, for $\rtp = \{\,\Rbox P_1, \Rbox T_2 \,\}$, and add the `connecting axiom' $\exists G_\rtp \sqsubseteq \Rbox \exists Q$ to $\T$. Then, in the first step, we rewrite the extended TBox and $A$ into the formula
\begin{equation*}
\Phi'(x,t) \ \ =\ \  A(x,t) \ \lor \ \exists t' \, \bigl( (t < t') \land \exists y \, Q(x,y,t')\bigr) \ \lor \  \exists t'\exists y\, G_\rtp(x,y,t'),
\end{equation*}
where we replace $Q(x,y,t')$ with $\Psi(x,y,t')$ as before, and restore the intended meaning of $G_\rtp(x,y,t')$ by rewriting $(\R,\Rbox P_1 \sqcap \Rbox T_2)$ into
\begin{equation*}
P(x,y,t') \land \bigl(T_1(x,y,t') \lor \exists t''\,\bigl((t'' < t') \land T(x,y,t'')\bigr)\bigr)
\end{equation*}
and substituting it for $G_\rtp(x,y,t')$ in $\Phi'(x,t)$.
\end{example}

We now formally define the connecting axioms for a given $\DL_{\bool/\horn}\Xallop{}$ ontology $\TO = \T \cup \R$. Again, we assume that $\R$ contains all role names from $\T$. Recall that a role type $\rtp$ for $\R$ is a maximal  subset of $\subr$ consistent with $\R$. As before, we only specify the positive part of role types and say that a role type is \emph{non-empty} if it contains some role $R$.
Given a role type $\rtp$, we consider the $\R$-canonical rod~$\rod_\rtp$  (see Section~\ref{sec:HornRI:sat}) for $\{R^\ddagger(0) \mid R\in\rtp \}$. Note that, by definition, we have $\rod_\rtp(0) = \rtp$. By Lemma~\ref{period:A} $(i)$, we can find positive integers $s^{\,\rtp} \leq |\R|$ and $p^{\,\rtp} \le 2^{2|\R|}$ such that
\begin{equation*}
\rod_\rtp(n) = \rod_\rtp(n - p^{\,\rtp}), \quad \text{ for }n \leq  - s^{\,\rtp}, \qquad\text{ and }\qquad \rod_\rtp(n) = \rod_\rtp(n + p^{\,\rtp}), \quad \text{ for } n \ge s^{\,\rtp}.
\end{equation*}
For a role type $\rtp$ for $\R$, we take a fresh role name $G_{\rtp}$ and fresh concept names $D^n_{\rtp}$, for $- s^{\,\rtp} - p^{\,\rtp} < n < s^{\,\rtp}+ p^{\,\rtp}$, and define the following CIs:
\begin{multline*}
\exists G_\rtp \sqsubseteq D^0_\rtp, \qquad\qquad
D_{\rtp}^n \sqsubseteq \Rnext D_{\rtp}^{n + 1}, \  \text{ for } 0 \le n <  s^{\,\rtp} + p^{\,\rtp} - 1, \qquad\quad
D_{\rtp}^{s^{\,\rtp} + p^{\,\rtp} - 1} \sqsubseteq \Rnext D_\rtp^{s^{\,\rtp}},\\
\text{ and } \qquad D_{\rtp}^n  \sqsubseteq \exists S, \text { for roles } S\in \rod_{\rtp}(n) \text{ and } 0 \le n < s^{\,\rtp} + p^{\,\rtp},
\end{multline*}
together with symmetrical CIs for $-s^{\,\rtp} - p^{\,\rtp} \le n \le 0$ for the past-time `loop'. Let {\bf (con)} be the set of all such CIs for all possible role types~$\rtp$ for $\R$, and let $\T_\R = \T \cup \textbf{(con)}$.

\begin{example}\em
In Example~\ref{ex:connexion}, for the role type $\rtp = \{P,\Rnext Q\}$, we have $s^{\,\rtp} = 2$, $p^{\,\rtp} = 1$, and so $\T_\R$ contains the following:
\begin{equation*}
\exists P \sqsubseteq D^0_\rtp, \quad D^0_\rtp \sqsubseteq \Rnext D^1_\rtp, \quad D^1_\rtp \sqsubseteq \Rnext D^2_\rtp,  \quad D^2_\rtp \sqsubseteq \Rnext D^2_\rtp, \quad \text{ and } \quad D^0_\rtp \sqsubseteq \exists P, \quad D^1_\rtp \sqsubseteq \exists Q,
\end{equation*}
which imply $\exists P \sqsubseteq \Rnext \exists Q$.
In the context of Example~\ref{ex:shrub}, for the role type $\rtp = \{\Rbox P_1, \Rbox T_2\}$, we have $s^{\,\rtp} = 1$, $p^{\,\rtp} = 1$, and so $\T_\R$ contains the following CIs:
\begin{equation*}
\exists G_\rtp \sqsubseteq D^0_\rtp, \quad D^0_\rtp \sqsubseteq \Rnext D^1_\rtp, \quad D^1_\rtp \sqsubseteq \Rnext D^1_\rtp,\quad \text{ and }
\quad D^1_\rtp \sqsubseteq \exists P_1, \quad  D^1_\rtp \sqsubseteq  \exists T_2, \quad D^1_\rtp \sqsubseteq \exists Q.
\end{equation*}
Note that, in this case, instead of two CIs $D^0_\rtp \sqsubseteq \Rnext D^1_\rtp$ and $D^1_\rtp \sqsubseteq \Rnext D^1_\rtp$, we could use a single $D^0_\rtp \sqsubseteq \Rbox D^1_\rtp$.
\end{example}

Denote by $\T^\dagger_{\smash{\R}}$ the $\LTL_\bool\Xallop$ ontology obtained from $\T_\R$ by replacing every basic concept~$B$ in it with~$B^\dagger$. Consider an ABox $\A$. For any $a,b\in\ind(\A)$, let $\rod_{a,b}$  be the $\R$-canonical rod for $\A^\ddagger_{\smash{a,b}}$. We split $\A$ into the concept and role components, $\mathcal{U}$  and $\mathcal{B}$, as follows:
\begin{align*}
\mathcal{U} & = \bigl\{\,A(a,\ell) \mid A(a,\ell) \in \A\,\bigr\}, \\
\mathcal{B} & = \bigl\{\, \exists G_\rtp(a,\ell) \mid a\in\ind(\A),\ \ell\in\tem(\A), \rtp = \rod_{a,b}(\ell) \text{ is non-empty}, \text{ for some } b\in\ind(\A) \,\bigr\}.
\end{align*}
We denote by $\mathcal{U}^\dagger_a$ and $\mathcal{B}^\dagger_{a}$  the sets of all atoms $A^\dagger(\ell)$, for $A(a,\ell)\in\mathcal{U}$, and $(\exists G_\rtp)^\dagger(\ell)$, for $\exists G_\rtp(a,\ell)\in\mathcal{B}$, respectively. Observe that the connecting axioms are such that ${\bf (con)}^\dagger$ is an $\LTL_\core\Xnext$ ontology, and the ABox~$\mathcal{B}$ is defined so that, for any $a \in \ind(\A)$ and $n \in \Z$,
\begin{equation}\label{prop-con}
S \in \rod_{a,b}(n) \text{ for some } b\in\ind(\A) \ \ \text{ iff } \ \ (\exists S)^\dagger(n) \in \C_{\textbf{(con)}^\dagger,\mathcal{B}_a^\dagger} \ \ \text{for any role } S \text{ in } \R.
\end{equation}
Indeed, let $S \in \rod_{a,b}(n)$. If $n \in \tem(\A)$, then $S\in\rtp$ and $\exists G_\rtp(a,n) \in \mathcal{B}$, for $\rtp = \rod_{a,b}(n)$, whence~$(\exists G_\rtp)^\dagger(n) \in \mathcal{B}_a^\dagger$, and so $(\exists S)^\dagger(n)\in\C_{\textbf{(con)}^\dagger,\mathcal{B}_a^\dagger}$. If $n > \max\A$, then we consider $\rtp = \rod_{a,b}(\max \A)$. It should be clear that the canonical model of $(\R^\ddagger, \{R^\ddagger(\max\A) \mid R\in \rtp\})$ contains $S(n)$. So $\rtp$ is non-empty and $\exists G_\rtp(a,\max \A)\in\mathcal{B}$, whence $(\exists G_\rtp)^\dagger(\max \A) \in \mathcal{B}_a^\dagger$, and so $(\exists S)^\dagger(n)\in\C_{\textbf{(con)}^\dagger,\mathcal{B}_a^\dagger}$
The case $n < \min\A$ is symmetric. The converse implication follows from the definition of {\bf (con)}.
We use~\eqref{prop-con} to establish the following key result:

\begin{lemma}\label{thm:technical}
Let $(\TO,\varkappa)$ be a  $\bot$-free $\DL_{\bool/\horn}\Xallop{}$ \OMPIQ{} that contains no qualified existential restrictions. Then, for any ABox $\A$, we have
\begin{equation*}
\ans^\Z(\TO,\varkappa,\A) = \bigl\{\, (a,n) \mid a\in\ind(\A) \text{ and } n \in \ans^\Z(\T_\R^{\smash{\dagger}},\varkappa^\dag,\mathcal{U}^\dagger_a \cup \mathcal{B}_a^\dagger) \,\bigr\}.
\end{equation*}
\end{lemma}
\begin{proof}
$(\subseteq)$ Suppose that $n \notin \ans^\Z(\T_\R^{\smash{\dagger}},\varkappa^\dagger,\mathcal{U}_a^\dagger \cup \mathcal{B}_a^\dagger)$. Then there is an \LTL{} model $\I_a$ of $(\T_\R^{\smash{\dagger}},\mathcal{U}_a^\dagger \cup \mathcal{B}_a^\dagger)$ with $\I_a\not\models \varkappa^\dag(n)$. We define a model $\I$ of $(\TO,\A)$ with $a^\I \notin \varkappa^{\I(n)}$ using unravelling (Lemma~\ref{lem:unravelling}) similarly to the proof of Lemma~\ref{lemma:consistency}.
To begin with, we take the beam $\beam_a\colon n\mapsto \bigl\{\, C\in\subt \mid \I_a\models C^\dagger(n) \,\bigr\}$; note that $\beam_a$ is a beam for $\T$ because $\T_{\smash{\R}}^\dagger$ extends~$\T^\dagger$. By~\eqref{prop-con}, the $\R$-canonical rod~$\rod_{a,b}$ for $\A_{\smash{a,b}}^\ddagger$ is $\beam_a$-compatible, for each $b\in\ind(\A)$.
Next, we fix a model $\J$ of $(\TO,\A)$ and, for every $b \in \ind(\A)\setminus \{ a\}$, take the beam $\beam_b$ of~$b^\J$ in $\J$. By Lemma~\ref{lem:unravelling}, we obtain a model $\I$ of $(\TO,\A)$ with $a^\I \notin \varkappa^{\I(n)}$.

The inclusion $(\supseteq)$ is straightforward.
\end{proof}

We now use this technical result to generalise Proposition~\ref{prop:conceptOMAQs} and construct rewritings for \OMPIQ{}s $(\TO, \varkappa)$ that contain no qualified existential quantifiers  from rewritings of suitable \LTL{} \OMPIQ{}s, where we identify a role type $\rtp$ with the intersection of all $R\in\rtp$, i.e., $\rtp = \bigsqcap_{R\in\rtp} R$, and $\rtp^\ddagger$ with the intersection of all $R^\ddagger$, for $R\in\rtp$:
\begin{lemma}\label{lemma:basis}
Let $\lang$ be one of $\FO(<)$, $\FOE$, or $\FO(\RPR)$, and $\fragr \in \{\core,\rhorn,\horn\}$.
A  $\bot$-free $\DL_{\frag/\fragr}^{\op}$ \OMPIQ{} $\q = (\TO, \varkappa)$ that contains no qualified existential restrictions is $\lang$-rewritable whenever
\begin{itemize}
\item[--] the $\LTL_\frag^{\op\cup\{{\scriptscriptstyle\bigcirc}\}}$ \OMPIQ{} $\q^\dagger = (\T_\R, \varkappa)^\dagger$ is $\lang$-rewritable and

\item[--] the $\LTL_\fragr^{\op}$ \OMPIQ{}s $\q^\ddagger_{\rtp} = (\R, \rtp)^\ddagger$ are $\lang$-rewritable, for role types $\rtp$ for $\R$.
\end{itemize}
\end{lemma}
\begin{proof}
By Proposition~\ref{prop:roleOMIQs}, we have an $\lang$-rewriting $\rew_\rtp(x,y,t)$ of every $\q_\rtp = (\R,\rtp)$. Similarly to the proof of Proposition~\ref{prop:conceptOMAQs}, we claim that the $\lang$-formula
$\rew(x,t)$ obtained from an $\lang$-rewriting~$\rew^\dagger(t)$ of $\q^\dagger$ by replacing every $A(s)$  with $A(x,s)$, every $(\exists P)^\dagger(s)$ with $\exists y\,P(x,y,s)$, every $(\exists P^-)^\dagger(s)$ with $\exists y\,P(y,x,s)$, every $(\exists G_\rtp)^\dagger(s)$ with $\exists y \, \rew_\rtp(x,y,s)$, and, in the case of $\FO(\RPR)$, by replacing every  $Q(t_1,\dots, t_k)$, for a relation variable $Q$, with $R(x,t_1,\dots, t_k)$ is an $\lang$-rewriting of~$\q$.

Indeed, we show that $\TO,\A \models \varkappa(a, \ell)$ iff  $\SA \models \rew(a,\ell)$, for any ABox $\A$, any $\ell \in \tem (\A)$ and any $a \in \ind (\A)$. 
By Lemma~\ref{thm:technical}, $\TO,\A \models \varkappa(a,\ell)$ iff $\T^{\smash{\dagger}}_\R, \mathcal{U}_a^\dagger \cup \mathcal{B}_a^\dagger \models \varkappa^\dagger(\ell)$.
As $\rew^\dagger(t)$ is an $\lang$-rewriting of $\smash{(\T^{\smash{\dagger}}_{\smash{\R}},\varkappa^\dagger)}$, the latter is equivalent to $\mathfrak{S}_{\mathcal{U}^\dagger_a \cup \mathcal{B}_a^\dagger} \models \rew^\dagger(\ell)$. Now, since $\rew_\rtp(x,y,t)$ is an $\lang$-rewriting of $\q_\rtp$, for all $b\in\ind(\A)$ and $n\in \tem(\A)$, we have  $\smash{(\exists G_\rtp)^\dagger(n) \in \mathcal{B}_b^\dagger}$ iff $\rtp = \rod_{b,c}(n)$ for the $\R$-canonical rod for $\A^\ddagger_{\smash{b,c}}$, for some $c \in \ind(\A)$, iff $\SA \models \exists y \, \rew_\rtp(b,y,n)$. Then, $\mathfrak{S}_{\mathcal{U}_a^\dagger \cup \mathcal{B}_a^\dagger} \models \rew^\dagger(\ell)$ iff $\SA \models \rew(a,\ell)$, as required.
\end{proof}

We are now in a position to obtain our first set of concrete rewritability results for \OMAQ{}s. Let $\q = (\TO,B)$ be a $\DL_{\smash{\krom/\core}}\Xallop$ \OMAQ{} with $\TO = \T\cup \R$ and a basic concept $B$. By Lemma~\ref{lemma:consistency}~(\emph{i}), $\q$ is $\lang$-rewritable whenever \OMPIQ{}s $(\TO',\varrho'_\bot)$, $(\TO', \varkappa'_\bot)$ and $(\TO',  \varkappa'_\bot \sqcup B)$ are $\lang$-rewritable. Suppose $\TO' = \T'\cup \R'$. By Proposition~\ref{prop:roleOMIQs}, the first \OMPIQ{} is $\lang$-rewritable if the $\LTL_{\core}\Xallop$ \OMPIQ{} $(\R',\varrho'_\bot)^\ddagger$ is $\lang$-rewritable, while, by Lemma~\ref{lemma:basis}, the other two are $\lang$-rewritable whenever the $\LTL_{\core}\Xallop$ \OMPIQ{}s $(\R',\varrho)^\ddagger$, for role types $\varrho$ for $\R'$, and the
$\LTL_\krom\Xallop$ \OMPIQ{}s $(\T'_{\R'}, \varkappa'_\bot)^\dagger$ and $(\T'_{\R'}, \varkappa'_\bot\sqcup B)^\dagger$ are $\lang$-rewritable. Theorems~24 and 27 of~\citeA{AIJ21} show that
 all $\LTL_{\smash{\core}}\Xallop$ and $\LTL_{\smash{\krom}}\Xallop$ \OMPIQ{}s  are $\FOE$ and $\FO(\RPR)$-rewritable, respectively. Note, however, that the consistency checking concept $\varkappa'_\bot$ in the two $\LTL_{\krom}\Xallop$ \OMPIQ{}s above has a special shape  (see Lemma~\ref{lemma:consistency}), and so the two \OMPIQ{}s are in fact equivalent to the $\LTL_{\krom}\Xallop$ \OMAQ{}s $(\T_\R,\bot)^\dagger$ and $(\T_\R, B)^\dagger$, respectively, where $\T_\R$ is constructed for $\TO$. The latter class of \OMAQ{}s is $\FOE$-rewritable~\cite[Theorem~16]{AIJ21}, which gives us our first optimal result. The case of $\DL_{\krom/\core}\Xallop$ \OMAQ{}s of the form $(\TO,S)$ with a basic role $S$ is similar; note that the \OMPIQ{}  with the consistency checking concept~$\varkappa'_\bot$ is required even for role \OMAQ{}s.
In the same way,  $\DL_{\krom/\horn}\Xbox$ \OMAQ{}s are reducible to $\LTL\Xallop_{\krom}$ \OMPIQ{}s of special form (equivalent to \OMAQ{}s) and $\LTL\Xbox_{\horn}$ \OMPIQ{}s, which, by Theorems~16 and  24 of~\citeA{AIJ21}, are $\FOE$- and $\FO(<)$-rewritable, respectively. It is worth pointing out that the TBox $\T_\R$ in Lemma~\ref{lemma:basis} requires the use of the $\nxt$ operators, which leads to the increased  language expressivity for $\LTL\Xallop_{\krom}$ \OMAQ{}s. Finally, $\DL_{\smash{\bool/\horn}}\Xallop$ \OMAQ{}s are reducible to  $\LTL_{\smash{\bool}}\Xallop$ and  $\LTL_{\smash{\horn}}\Xallop$ \OMPIQ{}s, which are known to be $\FO(\RPR)$-rewritable~\cite[Theorems~8]{AIJ21}. Hence, we obtain:
\begin{theorem}\label{cor:dl-omaq-rewritability}
$(i)$ All $\DL_{\smash{\krom/\core}}\Xallop$ and $\DL_{\smash{\krom/\horn}}\Xbox$  \OMAQ{}s are $\FOE$-rewritable.

$(ii)$ All $\DL_{\smash{\bool/\horn}}\Xallop$ \OMAQ{}s are $\FO(\RPR)$-rewritable.
\end{theorem}

Despite $\FO(<)$-rewritability of $\LTL_\horn\Xbox$ \OMPIQ{}s, we cannot obtain $\FO(<)$-rewritability of $\DL_{\smash{\horn}}\Xbox$ \OMAQ{}s because the connecting axioms {\bf (con)} use the next-time operator~$\Rnext$. We now show that actually there are  $\DL_{\smash{\horn}}\Xbox$ \OMAQ{}s that are not $\FO(<)$- and not even $\FOE$-rewritable:
\begin{theorem}\label{thm:unexpected}
$(i)$ There is a $\DL_{\smash{\horn}}\Xbox$ \OMAQ{}, answering which is $\NCo$-hard for data complexity. 

$(ii)$ There is a non-$\FO(<)$-rewritable $\DL_{\smash{\core/\horn}}\Xbox$ \OMAQ{}.
\end{theorem}
\begin{proof}
$(i)$ Consider the $\DL_\horn\Xbox$ \OMAQ{} $\q = (\TO, \exists S_0)$ with $\TO = \T \cup \R$ and
\begin{align*}
\T & = \bigl\{\,\exists R_k \sqcap A_0 \sqsubseteq \exists S_k, \  \ \exists R_k \sqcap A_1 \sqsubseteq \exists S_{1-k} \mid k=0,1\,\bigr\}\ \cup \ \bigl\{\,  B \sqsubseteq \exists S_0\,\bigr\},\\
\R & = \bigl\{\,S_k \sqsubseteq F_k, \ \ S_k \sqsubseteq \Rbox F_k, \ \ S_k \sqsubseteq \Lbox P_k, \ \ \Rbox F_k \sqcap P_k \sqsubseteq R_k \mid k=0,1\,\bigr\};
\end{align*}
cf.~Example~\ref{exampleBvNB:2}.
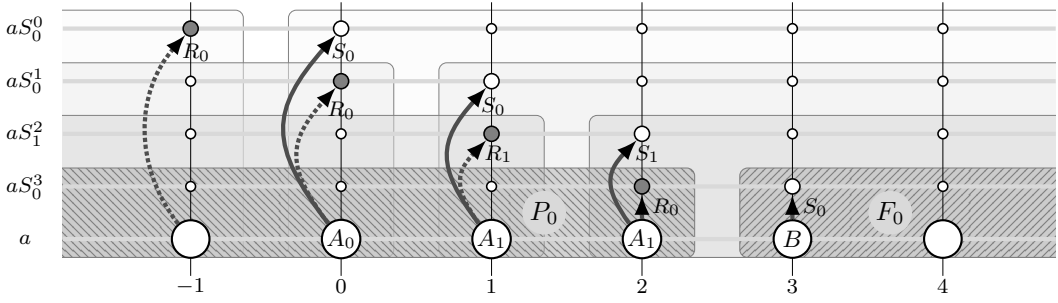
\begin{figure}[t]
\centering%
\begin{tikzpicture}[xscale=1, yscale=0.7, semithick, 
fpoint/.style={wpoint,inner sep=0.5pt,minimum size=5mm}]\footnotesize
\begin{scope}
\clip (-7.7,-0.5) rectangle +(13.2,5.35);
\filldraw[draw=gray,fill=gray!5,ultra thin,rounded corners=3pt,fill opacity=0.5] (-4.7,-0.35)  rectangle +(10.5,4.7);
\filldraw[draw=gray,fill=gray!5,ultra thin,rounded corners=3pt,fill opacity=0.5] (-5.3,-0.35)  rectangle +(-2.6,4.7);
\filldraw[draw=gray,fill=gray!15,ultra thin,rounded corners=3pt,fill opacity=0.5] (-2.7,-0.35)  rectangle +(8.5,3.7);
\filldraw[draw=gray,fill=gray!15,ultra thin,rounded corners=3pt,fill opacity=0.5] (-3.3,-0.35)  rectangle +(-4.6,3.7);
\filldraw[draw=gray,fill=gray!30,ultra thin,rounded corners=3pt,fill opacity=0.5] (-0.7,-0.35)  rectangle +(6.5,2.7);
\filldraw[draw=gray,fill=gray!30,ultra thin,rounded corners=3pt,fill opacity=0.5] (-1.3,-0.35)  rectangle +(-6.6,2.7);
\fill[fill=gray!60,thin,rounded corners=3pt,fill opacity=0.5] (1.3,-0.35)  rectangle +(4.5,1.7);
\fill[fill=gray!60,thin,rounded corners=3pt,fill opacity=0.5] (0.7,-0.35)  rectangle +(-8.6,1.7);
\filldraw[draw=gray,pattern=north east lines, pattern color=black!50,thin,rounded corners=3pt] (1.3,-0.35)  rectangle +(4.5,1.7);
\filldraw[draw=gray,pattern=north west lines, pattern color=black!50,thin,rounded corners=3pt] (0.7,-0.35)  rectangle +(-8.6,1.7);
\end{scope}
\foreach \y in {0,1,2,3,4} {
\draw[object-timeline] (-7.7,\y) -- ++(13.2,0);
}
\foreach \x/\l in {-6/-1,-4/0,-2/1,0/2,2/3,4/4} {
\draw[time-guideline] (\x,-0.7) -- ++(0,5.2);
\node at (\x,-0.9) {\scriptsize $\l$};
}
\node (a) at (-8.2,0) {\scriptsize $a$};
\begin{scope}\footnotesize
\node (am1) at (-6,0) [fpoint]{};
\node (a0) at (-4,0) [fpoint]{$A_0$}; 
\node (a1) at (-2,0) [fpoint]{$A_1$};
\node (a2) at (0,0) [fpoint]{$A_1$};
\node (a3) at (2,0) [fpoint]{$B$};
\node (a4) at (4,0) [fpoint]{};
\end{scope}
\node (r3) at (-8.2,1) {\scriptsize $aS_0^3$};
\node (r3-m1) at (-6,1) [point]{};
\node (r3-0) at (-4,1) [point]{};
\node (r3-1) at (-2,1) [point]{};
\node (r3-2) at (0,1) [ppoint]{};
\node (r3-3) at (2,1) [qpoint]{};
\node (r3-4) at (4,1) [point]{};
\node (r2) at (-8.2,2) {\scriptsize $aS_1^2$};
\node (r2-m1) at (-6,2) [point]{};
\node (r2-0) at (-4,2) [point]{};
\node (r2-1) at (-2,2) [ppoint]{};
\node (r2-2) at (0,2) [qpoint]{};
\node (r2-3) at (2,2) [point]{};
\node (r2-4) at (4,2) [point]{};
\node (r1) at (-8.2,3) {\scriptsize $aS_0^1$};
\node (r1-m1) at (-6,3) [point]{};
\node (r1-0) at (-4,3) [ppoint]{};
\node (r1-1) at (-2,3) [qpoint]{};
\node (r1-2) at (0,3) [point]{};
\node (r1-3) at (2,3) [point]{};
\node (r1-4) at (4,3) [point]{};
\node (r0) at (-8.2,4) {\scriptsize $aS_0^0$};
\node (r0-m1) at (-6,4) [ppoint]{};
\node (r0-0) at (-4,4) [qpoint]{};
\node (r0-1) at (-2,4) [point]{};
\node (r0-2) at (0,4) [point]{};
\node (r0-3) at (2,4) [point]{};
\node (r0-4) at (4,4) [point]{};
\begin{scope}[ultra thick,draw=black!70]
\draw[->] (a3)  to node [right]{\scriptsize $S_0$} (r3-3);
\draw[->,densely dotted] (a2)  to node [right]{\scriptsize $R_0$}  (r3-2);
\node[fill=gray!30,circle,inner sep=1pt] at (3.3,0.5) {\small $F_0$};
\node[fill=gray!30,circle,inner sep=1pt] at (-1.3,0.5) {\small $P_0$};
\draw[->,bend left,looseness=1.2] (a2)  to node [right,pos=0.85]{\scriptsize $S_1$} (r2-2);
\draw[->,densely dotted,bend left,looseness=1.2] (a1)  to node [right,pos=0.85]{\scriptsize $R_1$} (r2-1);
\draw[->,bend left,looseness=1.2] (a1)  to node [right,pos=0.88]{\scriptsize $S_0$} (r1-1);
\draw[->,densely dotted,bend left,looseness=1.2] (a0)  to node [right,pos=0.85]{\scriptsize $R_0$} (r1-0);
\draw[->,bend left,looseness=1.2] (a0)  to node [right,pos=0.92]{\scriptsize $S_0$} (r0-0);
\draw[->,densely dotted,bend left,looseness=0.9] (am1)  to node [right,pos=0.9]{\scriptsize $R_0$} (r0-m1);
\end{scope}
\end{tikzpicture}%
\caption{Proof of Theorem~\ref{thm:unexpected}.}\label{fig:unexpected}
\end{figure}
For $\avec{e} = (e_0,\dots,e_{n-1}) \in \{0,1\}^n$, let $\A_{\avec{e}} =  \{B(a,n)\} \cup  \mbox{$\{ A_{e_i}(a,i) \mid i < n\,\}$}$. Then $(a,0)$ is a certain answer to~$\q$ over $\A_{\avec{e}}$ iff the number of 1s in $\avec{e}$ is even---see Fig.~\ref{fig:unexpected} and note that $\exists S_k(a,n)$ always generates a \emph{fresh} witness $aS_k^n$.
The same idea can be used to simulate arbitrary NFAs as in the proof of Theorem~10 by~\citeA{AIJ21}, which shows \smash{\NCo}-hardness. We leave details to the reader. The case of \OMAQ{}s of the form $(\TO, S)$, for a role $S$, can be dealt with similarly to the proof of Theorem~\ref{thm:undec} above.

To prove $(ii)$, we can use the RBox $\R$ defined above and the fact that the \OMAQ{} $(\{A_0 \sqsubseteq \Lnext A_1,A_1 \sqsubseteq \Lnext A_0\},A_0)$ is not $\FO(<)$-rewritable~\cite{AIJ21}.
\end{proof}

\subsection{$\FO(<)$-Rewritability of Role-Monotone $\DL_{\bool/\horn}\Xallop{}$ OMAQs}\label{sec:rewrOMAQ:mon}

As follows from Theorem~\ref{thm:unexpected},  $\DL_\horn\Xbox$ turns out to be fundamentally different from $\LTL_\horn\Xbox$:
in the proof, we used  basic concepts of the form $\exists Q$  and Horn RIs with $\Rbox$ and $\Lbox$ to encode the $\Lnext$ operator on concepts in the sense that $\TO \models  \exists S_k \sqsubseteq  \Lnext \exists R_k$.
Our aim now is to give a sufficient condition under which the connecting axioms {\bf (con)} can be expressed in $\DL_\core\Xbox$, which would guarantee $\FO(<)$-rewritability of $\q^\dagger$ in Lemma~\ref{lemma:basis}.

We say that a $\DL_{\bool/\horn}\Xallop{}$ RBox $\R$ (and an ontology with such $\R$) is \emph{role-monotone} if, for any role type $\rtp$ for~$\R$ and any role $S$,
\begin{equation}\label{monotone}
S \in \rod_\rtp(n) \ \text{ and } \  n \neq 0 \quad \text{ implies  } \quad S \in \rod_\rtp(k) \text{ for all } k \ge n \text{ or for all } k \le n.
\end{equation}
In other words, $\bigl\{\, n\in\Z\mid S\in\rod_\rtp(n) \,\bigr\}  = I^S_\PP\cup I^S_0\cup I^S_\FF$, where each of $I^S_\PP$, $I^S_0$ and $I^S_\FF$ is either empty or an interval of the form $(-\infty,m']$,  $\{0\}$ and $[m,\infty)$, respectively; see Fig.~\ref{fig:monotone} for an illustration of the possible cases. The proof of Theorem~\ref{thm:unexpected} gives an example of an RBox that is not role-monotone: for the role type $\rtp$ containing $S_k$, $F_k$, $\Rbox F_k$ and~$\Lbox P_k$, we have $R_k\in\rod_\rtp(-1)$, but neither $R_k\in\rod_\rtp(-2)$ nor $R_k\in\rod_\rtp(0)$.

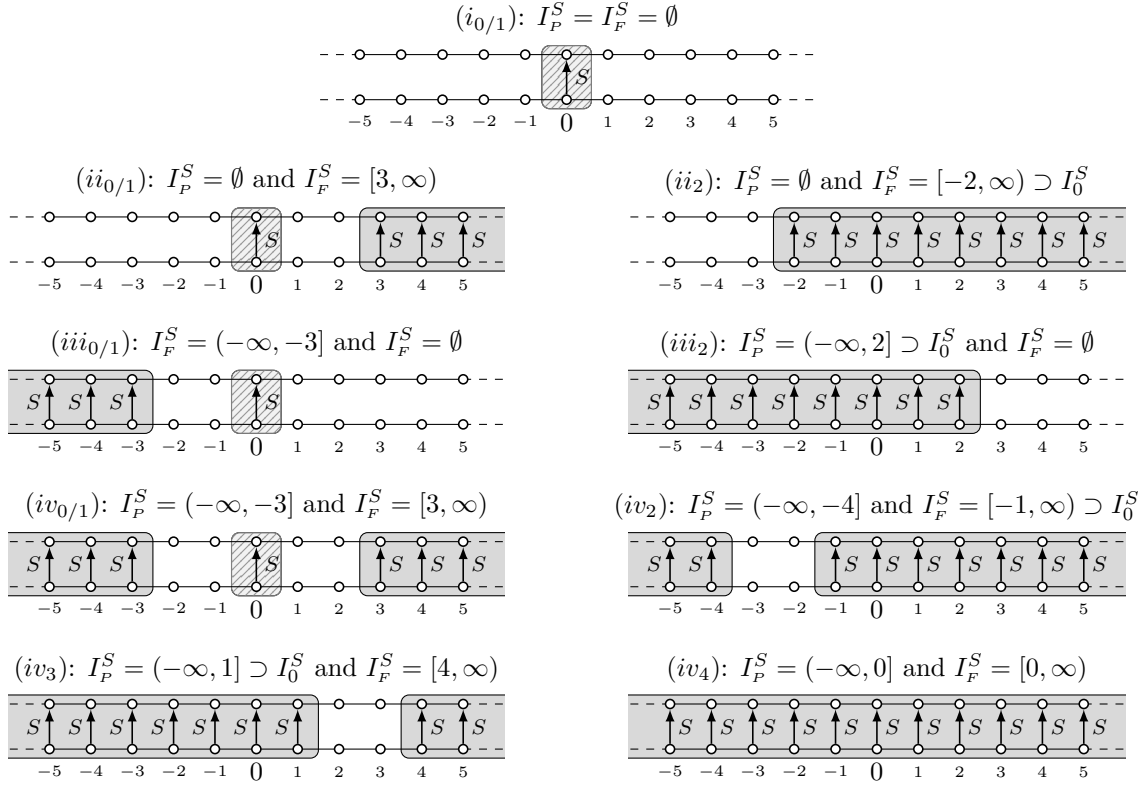
\begin{figure}[t]
\centerline{%
\begin{tikzpicture}[xscale=0.55,yscale=0.6, semithick]
\begin{scope}[yshift=36mm,xshift=75mm]
\node at (0,1.8) {\small $(i_{0/1})$: $I^S_\PP = I^S_\FF = \emptyset$};
\begin{scope}
\clip (-3,-0.5) rectangle +(9,2);
\fill[fill=gray!10,draw=black,thin,rounded corners=3pt] (-0.6,-0.2)  rectangle +(1.2,1.4);
\filldraw[draw=gray,pattern=north east lines, pattern color=black!40,thin,rounded corners=3pt] (-0.6,-0.2)  rectangle +(1.2,1.4);
\end{scope}
\foreach \y in {0,1} {
\draw[thin] (-5,\y) -- +(10,0);
\draw[dashed,thin] (-5,\y) -- +(-1,0);
\draw[dashed,thin] (5,\y) -- +(1,0);
\foreach \x in {-5,-4,-3,-2,-1,0,1,2,3,4,5} {
\node[tpoint] (a\x\y) at (\x,\y) {};
}
}
\foreach \x in {-5,-4,-3,-2,-1,1,2,3,4,5} {
\node at (\x,-0.5) {\tiny $\x$};
}
\node at (0,-0.5) {\small $0$};
\begin{scope}[semithick]\scriptsize
\draw[->] (a00) to node[right] {$S$} (a01);
\end{scope}
\end{scope}
%
%%%%%%%%%%%%%%%%%%%
\begin{scope}[xshift=0mm]
\node at (0,1.8) {\small $(ii_{0/1})$: $I^S_\PP = \emptyset$ and $I^S_\FF = [3,\infty)$};
\begin{scope}
\clip (-3,-0.5) rectangle +(9,2);
\filldraw[fill=gray!30,draw=black,thin,rounded corners=3pt] (2.5,-0.2)  rectangle +(5,1.4);
\fill[fill=gray!10,draw=black,thin,rounded corners=3pt] (-0.6,-0.2)  rectangle +(1.2,1.4);
\filldraw[draw=gray,pattern=north east lines, pattern color=black!40,thin,rounded corners=3pt] (-0.6,-0.2)  rectangle +(1.2,1.4);
\end{scope}
\foreach \y in {0,1} {
\draw[thin] (-5,\y) -- +(10,0);
\draw[dashed,thin] (-5,\y) -- +(-1,0);
\draw[dashed,thin] (5,\y) -- +(1,0);
\foreach \x in {-5,-4,-3,-2,-1,0,1,2,3,4,5} {
\node[tpoint] (a\x\y) at (\x,\y) {};
}
}
\foreach \x in {-5,-4,-3,-2,-1,1,2,3,4,5} {
\node at (\x,-0.5) {\tiny $\x$};
}
\node at (0,-0.5) {\small $0$};
\begin{scope}[semithick]\scriptsize
\draw[->] (a00) to node[right] {$S$} (a01);
\foreach \x in {3,4,5} {
\draw[->] (a\x0) to node[right] {$S$} (a\x1);
}
\end{scope}
\end{scope}
%
%%%%%%%%%%%%%%%%%%%%
\begin{scope}[xshift=150mm]
\node at (0,1.8) {\small $(ii_2)$: $I^S_\PP = \emptyset$ and $I^S_\FF = [-2,\infty) \supset I^S_0$};
\begin{scope}
\clip (-3,-0.5) rectangle +(9,2);
\filldraw[fill=gray!30,draw=black,thin,rounded corners=3pt] (-2.5,-0.2)  rectangle +(10,1.4);
\end{scope}
\foreach \y in {0,1} {
\draw[thin] (-5,\y) -- +(10,0);
\draw[dashed,thin] (-5,\y) -- +(-1,0);
\draw[dashed,thin] (5,\y) -- +(1,0);
\foreach \x in {-5,-4,-3,-2,-1,0,1,2,3,4,5} {
\node[tpoint] (a\x\y) at (\x,\y) {};
}
}
\foreach \x in {-5,-4,-3,-2,-1,1,2,3,4,5} {
\node at (\x,-0.5) {\tiny $\x$};
}
\node at (0,-0.5) {\small $0$};
\begin{scope}[semithick]\scriptsize
\foreach \x in {-2,-1,0,1,2,3,4,5} {
\draw[->] (a\x0) to node[right] {$S$} (a\x1);
}
\end{scope}
\end{scope}
%
%%%%%%%%%%%%%%%%%%%
\begin{scope}[xshift=0mm,yshift=-36mm]
\node at (0,1.8) {\small $(iii_{0/1})$: $I^S_\FF = (-\infty,-3]$ and $I^S_\FF = \emptyset$};
\begin{scope}
\clip (3,-0.5) rectangle +(-9,2);
\filldraw[fill=gray!30,draw=black,thin,rounded corners=3pt] (-2.5,-0.2)  rectangle +(-5,1.4);
\fill[fill=gray!10,draw=black,thin,rounded corners=3pt] (-0.6,-0.2)  rectangle +(1.2,1.4);
\filldraw[draw=gray,pattern=north east lines, pattern color=black!40,thin,rounded corners=3pt] (-0.6,-0.2)  rectangle +(1.2,1.4);
\end{scope}
\foreach \y in {0,1} {
\draw[thin] (-5,\y) -- +(10,0);
\draw[dashed,thin] (-5,\y) -- +(-1,0);
\draw[dashed,thin] (5,\y) -- +(1,0);
\foreach \x in {-5,-4,-3,-2,-1,0,1,2,3,4,5} {
\node[tpoint] (a\x\y) at (\x,\y) {};
}
}
\foreach \x in {-5,-4,-3,-2,-1,1,2,3,4,5} {
\node at (\x,-0.5) {\tiny $\x$};
}
\node at (0,-0.5) {\small $0$};
\begin{scope}[semithick]\scriptsize
\draw[->] (a00) to node[right] {$S$} (a01);
\foreach \x in {-5,-4,-3} {
\draw[->] (a\x0) to node[left] {$S$} (a\x1);
}
\end{scope}
\end{scope}
%
%%%%%%%%%%%%%%%%%%%%
\begin{scope}[xshift=150mm,yshift=-36mm]
\node at (0,1.8) {\small $(iii_2)$: $I^S_\PP = (-\infty,2]  \supset I^S_0$ and $I^S_\FF = \emptyset$};
\begin{scope}
\clip (3,-0.5) rectangle +(-9,2);
\filldraw[fill=gray!30,draw=black,thin,rounded corners=3pt] (2.5,-0.2)  rectangle +(-10,1.4);
\end{scope}
\foreach \y in {0,1} {
\draw[thin] (-5,\y) -- +(10,0);
\draw[dashed,thin] (-5,\y) -- +(-1,0);
\draw[dashed,thin] (5,\y) -- +(1,0);
\foreach \x in {-5,-4,-3,-2,-1,0,1,2,3,4,5} {
\node[tpoint] (a\x\y) at (\x,\y) {};
}
}
\foreach \x in {-5,-4,-3,-2,-1,1,2,3,4,5} {
\node at (\x,-0.5) {\tiny $\x$};
}
\node at (0,-0.5) {\small $0$};
\begin{scope}[semithick]\scriptsize
\foreach \x in {-5,-4,-3,-2,-1,0,1,2} {
\draw[->] (a\x0) to node[left] {$S$} (a\x1);
}
\end{scope}
\end{scope}
%
%%%%%%%%%%%%%%%%%%%
\begin{scope}[xshift=0mm,yshift=-72mm]
\node at (0,1.8) {\small $(iv_{0/1})$: $I^S_\PP = (-\infty,-3]$ and $I^S_\FF = [3,\infty)$};
\begin{scope}
\clip (-6,-0.5) rectangle +(12,2);
\filldraw[fill=gray!30,draw=black,thin,rounded corners=3pt] (-2.5,-0.2)  rectangle +(-5,1.4);
\filldraw[fill=gray!30,draw=black,thin,rounded corners=3pt] (2.5,-0.2)  rectangle +(5,1.4);
\fill[fill=gray!10,draw=black,thin,rounded corners=3pt] (-0.6,-0.2)  rectangle +(1.2,1.4);
\filldraw[draw=gray,pattern=north east lines, pattern color=black!40,thin,rounded corners=3pt] (-0.6,-0.2)  rectangle +(1.2,1.4);
\end{scope}
\foreach \y in {0,1} {
\draw[thin] (-5,\y) -- +(10,0);
\draw[dashed,thin] (-5,\y) -- +(-1,0);
\draw[dashed,thin] (5,\y) -- +(1,0);
\foreach \x in {-5,-4,-3,-2,-1,0,1,2,3,4,5} {
\node[tpoint] (a\x\y) at (\x,\y) {};
}
}
\foreach \x in {-5,-4,-3,-2,-1,1,2,3,4,5} {
\node at (\x,-0.5) {\tiny $\x$};
}
\node at (0,-0.5) {\small $0$};
\begin{scope}[semithick]\scriptsize
\draw[->] (a00) to node[right] {$S$} (a01);
\foreach \x in {-5,-4,-3} {
\draw[->] (a\x0) to node[left] {$S$} (a\x1);
}
\foreach \x in {3,4,5} {
\draw[->] (a\x0) to node[right] {$S$} (a\x1);
}
\end{scope}
\end{scope}
%
%%%%%%%%%%%%%%%%%%%
\begin{scope}[xshift=150mm,yshift=-72mm]
\node at (0,1.8) {\small $(iv_2)$: $I^S_\PP = (-\infty,-4]$ and $I^S_\FF = [-1,\infty)\supset I_0^S$};
\begin{scope}
\clip (-6,-0.5) rectangle +(12,2);
\filldraw[fill=gray!30,draw=black,thin,rounded corners=3pt] (-3.5,-0.2)  rectangle +(-4,1.4);
\filldraw[fill=gray!30,draw=black,thin,rounded corners=3pt] (-1.5,-0.2)  rectangle +(8,1.4);
\end{scope}
\foreach \y in {0,1} {
\draw[thin] (-5,\y) -- +(10,0);
\draw[dashed,thin] (-5,\y) -- +(-1,0);
\draw[dashed,thin] (5,\y) -- +(1,0);
\foreach \x in {-5,-4,-3,-2,-1,0,1,2,3,4,5} {
\node[tpoint] (a\x\y) at (\x,\y) {};
}
}
\foreach \x in {-5,-4,-3,-2,-1,1,2,3,4,5} {
\node at (\x,-0.5) {\tiny $\x$};
}
\node at (0,-0.5) {\small $0$};
\begin{scope}[semithick]\scriptsize
\foreach \x in {-5,-4} {
\draw[->] (a\x0) to node[left] {$S$} (a\x1);
}
\foreach \x in {-1,0,1,2,3,4,5} {
\draw[->] (a\x0) to node[right] {$S$} (a\x1);
}
\end{scope}
\end{scope}
%
%%%%%%%%%%%%%%%%%%%
\begin{scope}[xshift=0mm,yshift=-108mm]
\node at (0,1.8) {\small $(iv_3)$: $I^S_\PP = (-\infty,1]\supset I_0^S$ and $I^S_\FF = [4,\infty)$};
\begin{scope}
\clip (-6,-0.5) rectangle +(12,2);
\filldraw[fill=gray!30,draw=black,thin,rounded corners=3pt] (3.5,-0.2)  rectangle +(4,1.4);
\filldraw[fill=gray!30,draw=black,thin,rounded corners=3pt] (1.5,-0.2)  rectangle +(-8,1.4);
\end{scope}
\foreach \y in {0,1} {
\draw[thin] (-5,\y) -- +(10,0);
\draw[dashed,thin] (-5,\y) -- +(-1,0);
\draw[dashed,thin] (5,\y) -- +(1,0);
\foreach \x in {-5,-4,-3,-2,-1,0,1,2,3,4,5} {
\node[tpoint] (a\x\y) at (\x,\y) {};
}
}
\foreach \x in {-5,-4,-3,-2,-1,1,2,3,4,5} {
\node at (\x,-0.5) {\tiny $\x$};
}
\node at (0,-0.5) {\small $0$};
\begin{scope}[semithick]\scriptsize
\foreach \x in {-5,-4,-3,-2,-1,0,1} {
\draw[->] (a\x0) to node[left] {$S$} (a\x1);
}
\foreach \x in {4,5} {
\draw[->] (a\x0) to node[right] {$S$} (a\x1);
}
\end{scope}
\end{scope}
%
%%%%%%%%%%%%%%%%%%%
\begin{scope}[xshift=150mm,yshift=-108mm]
\node at (0,1.8) {\small $(iv_4)$: $I^S_\PP = (-\infty,0]$ and  $I^S_\FF = [0,\infty)$};
\begin{scope}
\clip (-6,-0.5) rectangle +(12,2);
\filldraw[fill=gray!30,draw=black,thin,rounded corners=3pt] (-7,-0.2)  rectangle +(14,1.4);
\end{scope}
\foreach \y in {0,1} {
\draw[thin] (-5,\y) -- +(10,0);
\draw[dashed,thin] (-5,\y) -- +(-1,0);
\draw[dashed,thin] (5,\y) -- +(1,0);
\foreach \x in {-5,-4,-3,-2,-1,0,1,2,3,4,5} {
\node[tpoint] (a\x\y) at (\x,\y) {};
}
}
\foreach \x in {-5,-4,-3,-2,-1,1,2,3,4,5} {
\node at (\x,-0.5) {\tiny $\x$};
}
\node at (0,-0.5) {\small $0$};
\begin{scope}[semithick]\scriptsize
\foreach \x in {-5,-4,-3,-2,-1,0,1,2,3,4,5} {
\draw[->] (a\x0) to node[right] {$S$} (a\x1);
}
\end{scope}
\end{scope}
\end{tikzpicture}}
\caption{Examples of possible configurations for a role in a role-monotone RBox: $I^S_0 = \emptyset$ in $(i_0)$ and $I^S_0 = \{0\}$ in $(i_1)$; similarly for other cases.}\label{fig:monotone}
\end{figure}

We now show how to replace {\bf (con)} for a role-monotone $\R$ with a set of $\DL_\core\Xbox$ CIs. Let $\rtp$ be a role type for~$\R$. We consider the following four groups of cases for each role $S$.
\begin{description}
\item[$(i)$] If $I^S_\PP = I^S_\FF = \emptyset$, then we take the CI
\begin{align}
\tag{$i_1$}
& \exists G_\rtp \sqsubseteq \exists S, && \text{ if } I_0^S \ne\emptyset;
\end{align}
otherwise---that is, in case $S\notin\rod_\rtp(n)$, for all $n \in \Z$---we do not need any CIs.
\item[$(ii)$] If $I^S_\PP = \emptyset$ and $I^S_\FF = [m,\infty)$,
then we take the following CIs, depending on $m$ and $I^S_0$ (note that $I^S_\emptyset \subset I^S_\FF$ if $m \leq 0$):
\begin{align}
\tag{$ii_0$}
& \exists G_\rtp \sqsubseteq \Rbox^m \exists S, && \text{if } m > 0 \text{ and } I_0^S =\emptyset;\\
\tag{$ii_1$}
& \exists G_\rtp \sqsubseteq \exists S\ \text{ and }\ \exists G_\rtp \sqsubseteq \Rbox^m \exists S, && \text{if } m > 0 \text{ and } I_0^S \ne\emptyset;\\
\tag{$ii_2$}
& \exists G_\rtp \sqsubseteq \Rbox D_{\rtp,S}\ \text{ and }\ \Rbox^{-m+1} D_{\rtp,S} \sqsubseteq \exists S, && \text{if } m \leq 0, \text{ for a fresh  } D_{\rtp,S}.
\end{align}

\item[$(iii)$] If $I^S_\PP = (-\infty, m']$ and $I^S_\FF = \emptyset$,
then we take the mirror image of $(ii)$, with  $\Lbox$, $<$, $\geq$ and $m'$ in place of $\Rbox$,  $>$,  $\leq$ and~$m$, respectively.
\item[$(iv)$] If $I^S_\PP = (-\infty,m']$ and $I^S_\FF = [m,\infty)$,
then we take the respective CIs from both cases $(ii)$ and $(iii)$.
\end{description}
The resulting set of CIs for all role types $\rtp$ for $\R$ and all roles $S$ is denoted by {\bf (mon-con)}. It is readily seen that the connecting axioms {\bf (con)} can be replaced by {\bf (mon-con)} in the definition of $\T_\R$; the resulting TBox is denoted by~$\T_{\textbf{mon-}\R}$. Thus, we obtain: 
\begin{lemma}\label{lemma:forbox}
For any  role-monotone $\TO$, Lemmas~\ref{thm:technical} and~\ref{lemma:basis} hold for $\T_{\textbf{mon-}\R}$ in place of~$\T_\R$.
\end{lemma}

We show now that all $\DL_{\smash{\core}}\Xbox$ RBoxes as well as
$\DL_{\smash{\horn}}\Xbox$ RBoxes without $\Box$-operators on the
left-hand side of RIs are role-monotone. Recall that the fragment of $\DL_{\smash{\bool/\horn}}\Xbox$ with
RIs of the latter type is denoted by
$\DL_{\smash{\bool/\rhorn}}\Xbox$.
Note that instead of RIs with $\Box$-operators on the
right-hand side only we can take RIs with
$\Diamond$-operators on the left-hand side only (see the discussion on \emph{TQL} in Section~\ref{sec:tdl}).

\begin{lemma}\label{lemma:role-monotone}
All $\DL_{\smash{\bool/\core}}\Xbox$ and $\DL_{\smash{\bool/\rhorn}}\Xbox$
ontologies are role-monotone.
\end{lemma}
\begin{proof}
Suppose first that $\TO = \T \cup \R$ is a $\DL_{\bool/\core}\Xbox$ ontology and $\rtp$ a role type for $\R$. Denote by~$\rod_\rtp^\alpha(n)$ the set of temporalised roles $R$ such that $R^{\smash{\ddagger}}(n)\in\cl_{\smash{\R^\ddagger}}^\alpha(\{R^{\smash{\ddagger}}(0) \mid R \in \rtp\})$; cf.~the canonical model construction in Section~\ref{sec:canonical-model} below or~\cite[Section~7]{AIJ21}. It is readily seen by induction that, for any $\alpha < \omega_1$, if $R \in \rod_\rtp^\alpha(n)$ and~$n \neq 0$, then~$R \in \rod_\rtp^\alpha(k)$ either for all~$k \ge n$ or for all $k \le n$.
In $\DL_{\smash{\bool/\rhorn}}\Xbox$, there are no
$\Box$-operators on the left-hand side of RIs, and so the
rules $(\Rbox^\leftarrow)$ and $(\Lbox^\leftarrow)$ in the canonical model construction of Section~\ref{sec:canonical-model} become
redundant. In view of this, we can easily show that if $R \in \rod_\rtp^\alpha(n)$ and either $0 <
n < k$ or $k < n < 0$, then
$R \in \rod_\rtp^\alpha(k)$.
\end{proof}

Following the lines of the argument for Theorem~\ref{cor:dl-omaq-rewritability}, Lemmas~\ref{lemma:forbox} and~\ref{lemma:role-monotone} and using $\FO(<)$-rewritability of $\LTL\Xbox_{\bool}$ \OMAQ{}s and $\LTL\Xbox_{\horn}$ \OMPIQ{}s~\cite[Theorems~11 and 24]{AIJ21}, we obtain:
\begin{theorem}\label{ex:dl-core-mon-fo}
All $\DL_{\smash{\bool/\core}}\Xbox$ and $\DL_{\smash{\bool/\rhorn}}\Xbox$
\OMAQ{}s are $\FO(<)$-rewritable.
\end{theorem}

% !TEX root =  TDL-Lite.tex

%**********************

\section{Rewriting $\DL_{\horn}\Xallop{}$ OMPIQs}\label{OBDAfor Temporal}

Our next aim is to lift, where possible, the rewritability results obtained in Section~\ref{sec:DL-Lite} for temporal \DL{} \OMAQ{}s to \OMPIQ{}s. 
First, we observe that, for \OMPIQ{}s of the form $(\TO,\varrho)$ with a positive temporal role $\varrho$, Lemma~\ref{lemma:consistency}~($ii$), Proposition~\ref{prop:roleOMIQs} and Lemmas~\ref{lemma:basis}, \ref{lemma:forbox} and~\ref{lemma:role-monotone} (for the consistency concept \OMPIQ{} in Lemma~\ref{lemma:consistency}~($ii$)) together with the \LTL-rewritability results of~\citeA{AIJ21} give us the following (see Table~\ref{TDL-table-omaq}):
\begin{theorem}\label{cor:role-ompiqs}
For \OMPIQ{}s of the form $(\TO,\varrho)$ with a positive temporal role $\varrho$,  

$(i)$ all $\DL_{\smash{\bool/\core}}\Xbox$ and $\DL_{\smash{\bool/\rhorn}}$ \OMPIQ{}s are $\FO(<)$-rewritable\textup{;}

$(ii)$ all $\DL_{\smash{\krom/\core}}\Xallop$ \OMPIQ{}s and $\DL_{\smash{\krom/\horn}}\Xbox$ are $\FOE$-rewritable\textup{;}

$(iii)$ all $\DL_{\smash{\bool/\horn}}\Xallop$ \OMPIQ{}s  are $\FO(\RPR)$-rewritable.
\end{theorem}

\begin{table}[t]
\centering%
\tabcolsep=0pt
\begin{tabular}{ccc}\toprule
 \rule[-3pt]{0pt}{13pt} & $\DL_{\frag/\fragr}\Xbox$ &  $\DL_{\frag/\fragr}\Xnext$ {\footnotesize and}  $\DL_{\frag/\fragr}\Xallop$  \\\midrule
$*/(\mathit{g}\text{-})\bool$ & \coNP-hard  {\scriptsize \cite{DBLP:journals/jiis/Schaerf93}} & undecidable {\scriptsize [Th.~\ref{thm:undec}]} \\\midrule
$*/\krom$ &  \multicolumn{2}{c}{\multirow{4}{*}{\coNP-hard  {\scriptsize\cite{DBLP:journals/jiis/Schaerf93}} }}  \\
$\bool/\horn$, $\krom/\horn$ & \\
$\bool/\rhorn$, $\krom/\rhorn$  & &\\
$\bool/\core$, $\krom/\core$  & &  \\\midrule
$\horn$ &$\FO(\RPR)$ {\scriptsize [Th.~\ref{cor:dl-omiq-rewritability}\,(\emph{iii})]}, {\small \NCo-hard} {\scriptsize [Th.~\ref{thm:unexpected}\,(\emph{i})]} &
\multirow{4}{*}[-1\dimexpr \aboverulesep + \belowrulesep + \cmidrulewidth]{\renewcommand{\tabcolsep}{0pt}\begin{tabular}{c}\FO(\RPR) {\scriptsize [Th.~\ref{cor:dl-omiq-rewritability}~(\emph{iii})]}\\[4pt]
{\small \NCo-hard}\\[-4pt]{\scriptsize \cite[Th.~10]{AIJ21}}\end{tabular}} \\ \cmidrule(lr){2-2}
$\core/\horn$ & $\FOE$  {\scriptsize [Th.~\ref{cor:dl-omiq-rewritability}\,(\emph{ii})  \& Th.~\ref{thm:unexpected}\,(\emph{ii})]} \\ \cmidrule(lr){2-2}
$\horn/\rhorn$, $\core/\rhorn$ &  \multirow{3}{*}[-0.5\dimexpr \aboverulesep + \belowrulesep + \cmidrulewidth]{$\FO(<)$ {\scriptsize [Th.~\ref{cor:dl-omiq-rewritability}\,(\emph{i})]}} & \\[3pt]
$\horn/\core$&  &  \\ \cmidrule(lr){3-3}
$\core$&  & $\FOE$  {\scriptsize [Th.~\ref{cor:dl-omiq-rewritability}\,(\emph{ii})]} \\\bottomrule
\end{tabular}
\caption{Rewritability and data complexity of $\DL_{\frag/\fragr}^{\op}$ \OMPIQ{}s of the form $(\TO,\varkappa)$, for a positive temporal concept $\varkappa$, where * denotes any of $\bool$, $\horn$, $\krom$ or $\core$.}
\label{TDL-table-ompiq}
\end{table}

So it remains to consider \OMPIQ{}s of the form $(\TO,\varkappa)$ with a positive temporal concept~$\varkappa$. 
Recall that all \LTL{} \OMPIQ{}s are $\FO(\RPR)$-rewritable, and so answering them can always be done in $\NCo$ for data complexity. On the other hand, as follows from the results of~\citeA{DBLP:journals/jiis/Schaerf93}, answering atemporal $\DL_\krom$ \OMPIQ{}s with positive temporal concepts can encode 2+2SAT:
\begin{equation*}
\big( \{\, \top \sqsubseteq A \sqcup B \,\}, \ \exists R. (\exists P_1.A \sqcap \exists P_2.A \sqcap \exists N_1.B \sqcap \exists N_2.B)\big)
\end{equation*}
is $\coNP$-complete for data complexity. In a similar manner, one can show that answering \OMPIQ{}s with the $\DL\Xbox_\krom$ ontology $\{A \sqsubseteq \Rdiamond B\}$ is $\coNP$-complete because its normal form is $\{A \sqcap \Rbox A' \sqsubseteq \bot,\ \top \sqsubseteq A' \sqcup B\}$. There are also $\DL_\krom$ \OMPIQ{}s that are complete for \LogSpace, \NL{} and \PTime~\cite{DBLP:conf/kr/GerasimovaKKPZ20}, with the inclusion $\NCo \subseteq \LogSpace$ believed to be strict.  (Note that these results require ABoxes with an unbounded number of individual names.) Thus, in the remainder of Section~\ref{OBDAfor Temporal}, we focus on $\DL_{\horn}\Xallop{}$ \OMPIQ{}s; see Table~\ref{TDL-table-ompiq}.

\subsection{Canonical Models}\label{sec:canonical-model}

The proofs of the rewritability results for \OMPIQ{}s require canonical models for fragments of $\DL_\horn\Xallop$, which generalise the canonical models of~\citeA{AIJ21}.

Suppose we are given a $\DL_\horn\Xallop$ ontology $\TO$ and an ABox $\A$. Let $\Lambda$ be a countable set of atoms of the form $\bot$, $C(w,n)$ and $R(w_1,w_2,n)$, where $C$ is a temporalised concept, $R$ a temporalised role and $n\in\Z$. To simplify notation, we refer to the concept and role atoms $C(w,n)$ and $R(w_1,w_2,n)$ as $\vartheta(\avec{w}, n)$, calling $\avec{w}$---that is, $w$ or $(w_1,w_2)$---a \emph{tuple}.
Denote by $\cl_\TO(\Lambda)$ the result of applying non-recursively to $\Lambda$ the following rules, where $S$ is a role: %\nb{RG20: what does mp and cls stand for?}
\begin{description}
\item[\rm (mp)] if $\TO$ contains $\vartheta_1 \sqcap \dots \sqcap \vartheta_m \sqsubseteq \vartheta$ and $\vartheta_i(\avec{w}, n) \in \Lambda$ for all $i$, $1 \le i \le m$, then we add $\vartheta(\avec{w}, n)$ to $\Lambda$;

\item[\rm (cls)] if $\TO$ contains $\vartheta_1 \sqcap \dots \sqcap \vartheta_m \sqsubseteq \bot$  and $\vartheta_i(\avec{w}, n) \in \Lambda$ for  all $i$, $1 \le i \le m$, then we add $\bot$;

\item[\rm ({$\Rbox^\to$})] if $\Rbox \vartheta(\avec{w}, n) \in \Lambda$, then we  add all $\vartheta(\avec{w}, k)$ with $k > n$;

\item[\rm ({$\Rbox^\leftarrow$})] if $\vartheta(\avec{w}, k) \in \Lambda$ for all $k > n$, then we add $\Rbox \vartheta(\avec{w}, n)$;

\item[\rm ({$\Rnext^\to$})] if $\Rnext \vartheta(\avec{w}, n) \in \Lambda$,  then we  add $\vartheta(\avec{w}, n+1)$;

\item[\rm ({$\Rnext^\leftarrow$})] if $\vartheta(\avec{w}, n + 1) \in \Lambda$,  then we add $\Rnext \vartheta(\avec{w}, n)$;

\item[\rm ({$\Lbox^\to$}), ({$\Lbox^\leftarrow$}), ({$\Lnext^\to$}), ({$\Lnext^\leftarrow$})] are the past-time counterparts of the four rules above;

\item[\rm (inv)] if $S(w_1,w_2,n) \in \Lambda$, then we add $S^-(w_2,w_1,n)$ to $\Lambda$ (assuming that $(P^-)^- = P$);

\item[\rm $(\exists^\leftarrow)$] if $S(w_1,w_2,n) \in \Lambda$, then we add $\exists S(w_1,n)$;

\item[\rm $(\exists^\to)$] if $\exists S(w,n) \in \Lambda$,
then we add $S(w,wS^n,n)$, where $wS^n$ is an individual  name called the  \emph{witness} for $\exists S(w,n)$.
\end{description}
We set $\cl_\TO^0(\Lambda) = \Lambda$ and, for any successor ordinal $\xi +1$ and limit ordinal $\zeta$,
\begin{equation}\label{closures}
\cl_\TO^{\smash{\xi +1}}(\Lambda) = \cl_\TO(\cl_\TO^{\smash{\xi}}(\Lambda))\qquad \text{ and }\qquad  \cl^{\smash{\zeta}}_\TO (\Lambda) = \bigcup\nolimits_{\xi<\zeta} \cl_\TO^{\smash{\xi}}(\Lambda).
\end{equation}
Let  $\can = \cl_\TO^{\omega_1}(\A)$, where $\omega_1$ is the first uncountable ordinal (as
$\cl_\TO^{\smash{\omega_1}}(\Lambda)$ is countable, there is an ordinal $\alpha < \omega_1$ such that $\cl^{\smash{\alpha}}_\TO(\Lambda) = \cl^{\smash{\beta}}_\TO (\Lambda)$, for all $\beta \ge \alpha$).
We regard $\can$ as both a set of atoms of the form $\bot$ and $\vartheta(\avec{w}, n)$ and as an interpretation whose domain, $\Delta^\can$, comprises $\ind(\A)$ and the witnesses~$wS^n$ that are used in the construction of~$\can$, and the interpretation function is defined by taking $\avec{w} \in \eta^{\can(n)}$ iff $\eta(\avec{w},n) \in \can$, for any concept or role name~$\eta$.

\begin{example}\label{ex:canon}\em
The interpretation $\can$ for $\A = \{\,A(a,0)\,\}$ and
\begin{equation*}
\TO ~=~ \bigl\{\, A \sqsubseteq \Rbox \exists P, \  P \sqsubseteq \Rnext R, \ R \sqsubseteq \Rnext R, \ \Rbox R \sqsubseteq S, \ \exists S^- \sqsubseteq \exists T, \ \exists T^- \sqsubseteq A \,\bigr\}
\end{equation*}
is shown in Fig.~\ref{fig:dl-lite:canonical}.
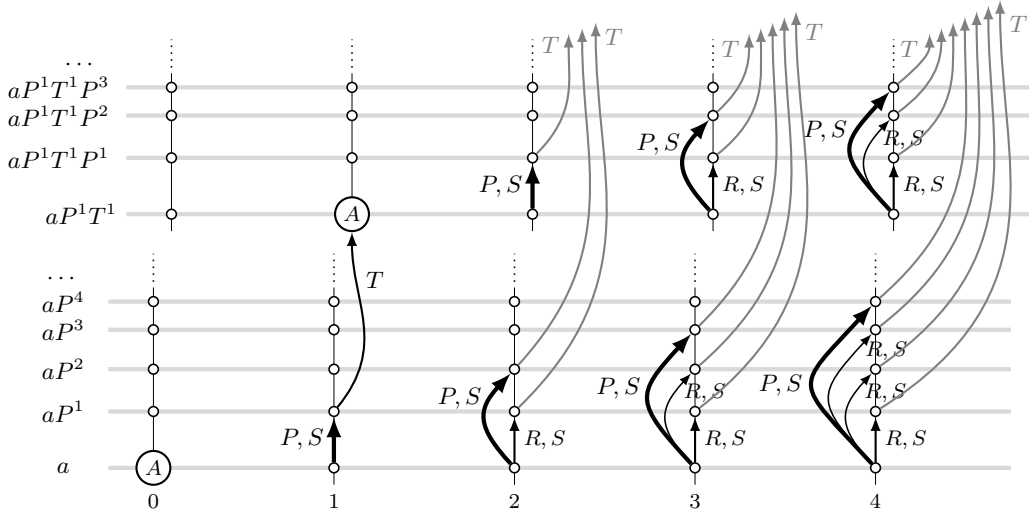
\begin{figure}[t]
\centerline{%
\begin{tikzpicture}[xscale=1.2, yscale=0.75, semithick]\footnotesize
\foreach \y in {0,1,1.75,2.45,2.95} {
\draw[object-timeline] (-4.5,\y) -- ++(10,0);
}
\foreach \x/\l in {-4/0,-2/1,-0/2,2/3,4/4} {
\node at (\x,-0.6) {\scriptsize $\l$};
\draw[time-guideline] (\x,-0.3) -- ++(0,3.5);
\draw[dotted] (\x,3.3) -- ++(0,0.5);
}
\node (a) at (-5,0) {$a$};
\node (a0) at (-4,0) [wpoint]{\scriptsize$A$};
\node (a1) at (-2,0) [point]{};
\node (a2) at (0,0) [point]{};
\node (a3) at (2,0) [point]{};
\node (a4) at (4,0) [point]{};
\foreach \y/\j in {1/1,1.75/2,2.45/3,2.95/4} {
\node (aP1) at (-5,\y) {$aP^\j$};
\foreach \x/\i in {-4/0,-2/1,0/2,2/3,4/4} {
\node (aP\j\i) at (\x,\y) [point]{};
}
}
\node at (-5,3.35) {$\cdots$};
\draw[->,ultra thick] (a1)  to node [left]{$P,S$} (aP11);
\draw[->,thick] (a2)  to node [right]{\scriptsize $R,S$}  (aP12);
\draw[->,ultra thick,bend left,looseness=1.2] (a2)  to node [left,pos=0.75]{$P,S$} (aP22);
\draw[->,thick] (a3)  to node [right]{\scriptsize $R,S$}  (aP13);
\draw[->,bend left,looseness=1.2] (a3)  to node [right,pos=0.8]{\scriptsize$R,S$} (aP23);
\draw[->,ultra thick, bend left,looseness=1.4] (a3)  to node [left,pos=0.6]{$P,S$} (aP33);
\draw[->,thick] (a4)  to node [right]{\scriptsize $R,S$}  (aP14);
\draw[->,bend left,looseness=1.2] (a4)  to node [right,pos=0.8]{\scriptsize$R,S$} (aP24);
\draw[->,bend left,looseness=1.4] (a4)  to node [right,pos=0.9]{\scriptsize$R,S$} (aP34);
\draw[->,ultra thick,bend left,looseness=1.6] (a4)  to node [left,pos=0.5]{$P,S$} (aP44);
\begin{scope}[xshift=2mm,yshift=45mm]
\foreach \y in {0,1,1.75,2.25} {
\draw[object-timeline] (-4.5,\y) -- ++(10,0);
}
\node at (-5,2.65) {$\cdots$};
\foreach \x/\i in {-4/0,-2/1,0/2,2/3,4/4} {
\draw[time-guideline] (\x,-0.3) -- ++(0,2.8);
\draw[dotted] (\x,2.6) -- ++(0,0.5);
}
\node (da) at (-5,0) {$aP^1T^1$};
\node (da0) at (-4,0) [point]{};
\node (da1) at (-2,0) [wpoint]{\scriptsize$A$};
\node (da2) at (0,0) [point]{};
\node (da3) at (2,0) [point]{};
\node (da4) at (4,0) [point]{};
\foreach \y/\j in {1/1,1.75/2,2.25/3} {
\node (daP1) at (-5.25,\y) {$aP^1T^1P^\j$};
\foreach \x/\i in {-4/0,-2/1,0/2,2/3,4/4} {
\node (daP\j\i) at (\x,\y) [point]{};
}
}
\draw[->,ultra thick] (da2)  to node [left]{$P,S$} (daP12);
\draw[->,thick] (da3)  to node [right]{\scriptsize $R,S$}  (daP13);
\draw[->,ultra thick,bend left,looseness=1.2] (da3)  to node [left,pos=0.75]{$P,S$} (daP23);
\draw[->,thick] (da4)  to node [right]{\scriptsize $R,S$}  (daP14);
\draw[->,bend left,looseness=1.2] (da4)  to node [right,pos=0.8]{\scriptsize$R,S$} (daP24);
\draw[->,ultra thick, bend left,looseness=1.4] (da4)  to node [left,pos=0.7]{$P,S$} (daP34);
\end{scope}
\draw[->,thick,out=62,in=-90] (aP11) to node[right,pos=0.75] {$T$} (da1);
\begin{scope}[gray]
\draw[->,thick,out=57,in=-90] (daP12) to node[left,pos=0.92] {$T$} ($(da1)+(2.4,3.2)$);
\draw[->,thick,out=60,in=-90] (aP22) to ($(da1)+(2.55,3.3)$);
\draw[->,thick,out=57,in=-90] (aP12) to node[right,pos=0.968] {$T$} ($(da1)+(2.7,3.4)$);
\draw[->,thick,out=54,in=-90] (daP23) to node[left,pos=0.83] {$T$} ($(da1)+(4.4,3.2)$);
\draw[->,thick,out=51,in=-90] (daP13) to ($(da1)+(4.53,3.3)$);
\draw[->,thick,out=57,in=-90] (aP33) to ($(da1)+(4.66,3.4)$);
\draw[->,thick,out=54,in=-90] (aP23) to ($(da1)+(4.79,3.5)$);
\draw[->,thick,out=51,in=-90] (aP13) to node[right,pos=0.955] {$T$} ($(da1)+(4.92,3.6)$);
\draw[->,thick,out=51,in=-90] (daP34) to node[left,pos=0.74] {$T$} ($(da1)+(6.4,3.2)$);
\draw[->,thick,out=48,in=-90] (daP24) to ($(da1)+(6.53,3.3)$);
\draw[->,thick,out=45,in=-90] (daP14) to ($(da1)+(6.66,3.4)$);
\draw[->,thick,out=54,in=-90] (aP44) to ($(da1)+(6.79,3.5)$);
\draw[->,thick,out=51,in=-90] (aP34) to ($(da1)+(6.92,3.6)$);
\draw[->,thick,out=48,in=-90] (aP24) to ($(da1)+(7.05,3.7)$);
\draw[->,thick,out=45,in=-90] (aP14) to node[right,pos=0.95] {$T$} ($(da1)+(7.18,3.8)$);
\end{scope}
\end{tikzpicture}%
}%
\caption{The interpretation $\can$ in Example~\ref{ex:canon}.}\label{fig:dl-lite:canonical}
\end{figure}
Note that its construction requires $\omega^2$ applications of $\cl_\TO$: $\omega + 2$ applications are needed to derive $S(a,aP^1,1)$ (in fact, all $S(a,aP^k, n)$, for $n \geq k > 0$), then a $T$-successor with $A$ is created, and the process is repeated again and again. The resulting interpretation has infinite branching and infinite depth.
\end{example}

Let $\I =(\Delta^\I,\cdot^{\I(n)})$ be a temporal interpretation. A \emph{homomorphism} from $\can$ to $\I$ is a map $h \colon \Delta^\can \to \Delta^\I$ such that $h(a) = a^\I$, for $a \in \ind(\A)$, and $\avec{w} \in \eta^{\can(n)}$ implies $h(\avec{w}) \in \eta^{\I(n)}$, for any concept or role name $\eta$, any tuple~$\avec{w}$ from~$\Delta^\can$ and any $n \in \Z$, where $h(w_1,w_2) = (h(w_1),h(w_2))$. We now show that the canonical models of $\DL_\horn\Xallop$ KBs satisfy the following important properties; cf.~Theorem~18 by~\citeA{AIJ21}.
\begin{theorem}\label{canonical-dl}
Let $\TO$ be a $\DL_\horn\Xallop$ ontology and $\A$ an ABox. Then the following hold\textup{:}
\begin{description}
\item[\rm (\emph{i})] for any temporalised concept or role $\vartheta$, any tuple $\avec{w}$ in $\Delta^\can$ and any $n \in \Z$, we have $\avec{w} \in \vartheta^{\can(n)}$ iff $\vartheta(\avec{w}, n) \in \can$\textup{;}

\item[\rm (\emph{ii})] for any model $\I$ of $\TO$ and $\A$, there exists a homomorphism $h$ from $\can$ to $\I$\textup{;} 

\item[\rm (\emph{iii})] if $\bot \in \can$, then $\TO$ and $\A$ are inconsistent\textup{;} otherwise, $\can$ is a model of $\TO$ and $\A$\textup{;}

\item[\rm (\emph{iv})] if $\TO$ and $\A$ are consistent, then, for any \OMPIQ{} $\q = (\TO,\varkappa)$ and any $n\in\Z$, we have $(a,n) \in \ans^\Z(\q,\A)$ iff $a \in \varkappa^{\can(n)}$\textup{;} similarly, for any \OMPIQ{} $\q = (\TO,\varrho)$, we have $(a,b,n) \in \ans^\Z(\q,\A)$ iff $(a,b) \in \varrho^{\can(n)}$.
\end{description}
\end{theorem}
\begin{proof}
Claim $(i)$ is proved by induction on the construction of $\vartheta$. The basis of induction (for a concept or role name~$\vartheta$) follows from the definition of $\can$. Suppose $\vartheta = \exists S$ and $w \in (\exists S)^{\can(n)}$. Then $(w,w') \in S^{\can(n)}$, for some $w' \in \Delta^\can$, and so, by the induction hypothesis, $S(w,w', n) \in \can$, which gives $\exists S(w,n) \in \can$ by $(\exists^\leftarrow)$. Conversely, if $\exists S(w,n) \in \can$ then, by $(\exists^\to)$, we have $S(w,wS^n,n) \in \can$, whence, by the induction hypothesis, $(w,wS^n) \in S^{\can(n)}$, and so $w \in (\exists S)^{\can(n)}$. The case of $\vartheta = P^-$ is straightforward by~(inv). For $\vartheta = \Rbox \vartheta_1$, suppose first that $\avec{w}\in (\Rbox \vartheta_1)^{\can(n)}$. Then $\avec{w}\in \vartheta_1^{\can(k)}$ for all $k > n$, whence, by the induction hypothesis, $\vartheta_1(\avec{w},k) \in \can$, and so, by ({$\Rbox^\leftarrow$}), we obtain $\Rbox \vartheta_1(\avec{w},n) \in\can$. Conversely, if $\Rbox \vartheta_1(\avec{w},n) \in\can$ then, by ({$\Rbox^\to$}), we have $\vartheta_1(\avec{w},k) \in \can$ for all $k > n$. By the induction hypothesis, $\avec{w} \in \smash{\vartheta_1^{\can(k)}}$ for all $k > n$, and so $\avec{w}\in (\Rbox \vartheta_1)^{\can(n)}$. The other temporal operators, $\Lbox$, $\Rnext$ and $\Lnext$, are treated similarly.

\smallskip

$(ii)$ Suppose $\I$ is a model of $\TO$ and $\A$. By induction on $\alpha < \omega_1$, we construct a map $h^\alpha$ from $\Delta^\can$ to $\Delta^\I$ such that, for any $\vartheta$, any tuple $\avec{w}$ in $\Delta^\can$ and any $n \in \Z$, if $\vartheta(\avec{w},n) \in \cl_\TO^\alpha(\A)$ then $h^\alpha(\avec{w}) \in \vartheta^{\I(n)}$. For the basis of induction, we set  $h^0(a) = a^\I$. By~(\emph{i}), the basis is chosen correctly. Next, for a successor ordinal $\xi + 1$, we define $h^{\xi+1}$ by extending $h^\xi$ to the new witnesses $wS^n$ introduced by the rule $(\exists^\to)$ at step $\xi + 1$. If $\exists S(w,n) \in \cl_\TO^{\smash{\xi}}(\A)$, then $h^\xi(w) \in (\exists S)^{\I(n)}$ by the induction hypothesis, and so there is $w' \in \Delta^\I$ such that $(h^\xi(w),w') \in S^{\I(n)}$. We set $h^{\xi +1}(wS^n) = w'$ for all such $wS^n$ and keep $h^{\xi +1}$ the same as $h^\xi$ on the domain elements of $\cl_\TO^{\smash{\xi}}(\A)$. That $\vartheta(\avec{w},n) \in \cl^{\smash{\xi + 1}}_\TO(\A)$ implies $h^{\xi + 1}(\avec{w}) \in \vartheta^{\I(n)}$ follows from the induction hypothesis and the fact that $\I$ is `closed' under all of our rules save $(\exists^\to)$. For a limit ordinal $\zeta$, we take $h^\zeta$ to be the union of all $h^\xi$, for~$\xi < \zeta$.

\smallskip

(\emph{iii}) Suppose first that $\bot\in \can$. Then there is an axiom  $\vartheta_1 \sqcap \dots \sqcap \vartheta_m \sqsubseteq \bot$ in $\TO$ and $n\in \Z$ such that $\vartheta_{\smash{i}}(\avec{w},n)\in \can$ for all $i$ $(1 \le i \le m)$. By (\emph{i}) and (\emph{ii}), $\avec{w} \in \vartheta_{\smash{i}}^{\I(n)}$, for all models $\I$ of $\TO$ and $\A$ and all $i$ ($1\leq i\leq m$). But then $\A$ is inconsistent with $\TO$. On the other hand, if~$\A$ is consistent with $\TO$, then $\can$
is a model of $\TO$ by (\emph{i}) and closure under the rules~(mp) and~(cls) since $\bot\notin\can$; $\can$ is a model of $\A$ by definition.

\smallskip

To show (\emph{iv}), observe that, for any interpretations $\I$ and $\J$, if $\eta^{\I(n)}\subseteq \eta^{\J(n)}$
for all $n\in\Z$ and all concept and role names $\eta$, then $\vartheta^{\I(n)}\subseteq\vartheta^{\J(n)}$ for all $n\in\Z$ and all positive
temporal concepts and roles $\vartheta$.
Now, that $(\avec{w},n) \in \ans^\Z(\q,\A)$ implies $\avec{w}\in\vartheta^{\can(n)}$
follows by~(\emph{iii}) from the fact that $\can$ is a model of $\TO$ and $\A$, while the converse direction follows
from (\emph{i}), (\emph{ii}) and the observation above.\mbox{}\hfill
\end{proof}

Items~(\emph{iii}) and~(\emph{iv}) allow us to establish the following simpler analogue of Lemma~\ref{lemma:consistency}~$(i)$, which holds for Horn OMPIQs (rather than OMAQs with Boolean CIs):
\begin{lemma}\label{lemma:consistency:horn}
Let $\lang$ be one of $\FO(<)$, $\FOE$, or $\FO(\RPR)$, $\fragr \in \{\core,\rhorn,\horn\}$ and $\frag\in\{\core,\horn\}$.
Let $\TO$ be a $\DL_{\smash{\frag/\fragr}}^\op$ ontology and $\TO'$, $\varkappa_{\smash{\bot}}'$  and $\varrho_{\smash{\bot}}'$ as in  Lemma~\ref{lemma:consistency}.
For any positive temporal concept~$\varkappa$, if $\rew'(x,t)$, $\rew^\T_{\smash{\bot}}(x,t)$ and $\rew^\R_{\smash{\bot}}(x,y,t)$ are $\lang$-rewritings of the \OMPIQ{}s $(\TO',  \varkappa)$, $(\TO', \varkappa'_\bot)$ and $(\TO',\varrho'_\bot)$, respectively, then
$\rew'(x,t)   \lor  \chi_{\smash{\bot}}$
is an $\lang$-rewriting of the \OMPIQ{}~$(\TO,\varkappa)$, where $\chi_{\smash{\bot}} = \exists x,t\,\rew_\bot^\T(x, t)  \lor  \exists x,y,t\,\rew_\bot^\R(x,y, t)$.
\end{lemma}

We also require the following simple properties of the canonical models:
\begin{lemma}\label{lem:witness-interaction}
  Let $\TO$ be a $\DL_\horn\Xallop$ ontology with $\TO = \T \cup \R$ and $\A$  an ABox.
  If~$\TO$ and~$\A$ are consistent, then, for any
  role~$S$ and $n\in \Z$, we have
\begin{description}
\item[\rm (\emph{i})] \label{lem:witness-inter:item0} $wS^n\in\Delta^\can$ iff $\exists S(w,n)\in \can$, for any $w\in \Delta^\can$\textup{;}
\item[\rm (\emph{ii})] \label{lem:witness-inter:item1}
$R(w,wS^n,m)\in \can$ iff
  $\R\models S \sqsubseteq \nxt^{m-n}R$, for any temporalised role $R$,  $wS^n\in\Delta^\can$ and $m\in\Z$\textup{;}
\item[\rm (\emph{iii})] \label{lem:witness-inter:item2}
$C(wS^n,m)\in \can$ iff
  $\TO\models \exists S^- \sqsubseteq \nxt^{m-n}C$, for any temporalised concept
  $C$, any $wS^n\in\Delta^\can$ and any $m\in\Z$.
\end{description}
\end{lemma}
\begin{proof}
Claim~(\emph{i}) follows immediately from rule~$(\exists^\rightarrow)$ in the definition of $\can$.

(\emph{ii}) Suppose $\R\models S \sqsubseteq \nxt^{m-n}R$ and
  $wS^n\in \Delta^\can$. Then $S(w,wS^n,n)\in \can$ by the definition of $\can$. By Theorem~\ref{canonical-dl}~(\emph{iii}),
  $\can\models \R$, and so $R(w,wS^n,m)\in \can$.
Conversely, suppose $R(w,wS^n,m)\in \can$. By the construction of $\can$, we have $S(w,wS^n,n)\in \can$ and $R(w,wS^n,m) \in \C_{\R,\{S(w,wS^n,n)\}}$. It follows, by~\eqref{intersection}, that $\R\models S \sqsubseteq \nxt^{m-n}R$.

The proof of (\emph{iii}) is similar.
\end{proof}

\subsection{Rewriting $\DL_{\horn}\Xallop{}$ OMPIQs}\label{sec:OMPIQs}

In this section, we prove our main technical result that reduces rewritability of $\bot$-free $\DL_{\horn}\Xallop{}$ \OMPIQ{}s $\q= (\TO,\varkappa)$ with a positive temporal concept~$\varkappa$ to rewritability of $\DL_{\horn}\Xallop{}$ OMAQs.
Without loss of generality, we assume that all concept and roles names in $\varkappa$ occur in $\TO = \T \cup \R$ and that a role name occurs in $\T$ iff it occurs in~$\R$.

In order to encode the structure of the infinite canonical model in a finite way, we use phantoms~\cite{AIJ21}, which are formulas that encode certain answers to \OMPIQ{}s  beyond the active temporal domain $\tem(\A)$. In the context of \DL{}, we require the following modification of the original definition.
\begin{definition}\label{def:phantomDL}\em
Let $\lang$ be one of $\FO(<)$, $\FOE$, or $\FO(\RPR)$. An $\lang$-\emph{phantom} of the given \OMPIQ{} $\q = (\TO,\varkappa)$ for $k \ne 0$ is an $\lang$-formula $\Phi^\kpar_{\q}(x)$ such that, for any ABox $\A$,
\begin{equation*}
\SA \models \Phi^\kpar_{\q}(a) \quad \text{ iff } \quad (a,\ka{\kpar})\in \ans^\Z(\q,\A), \qquad \text{  for any } a \in \ind (\A),
\end{equation*}
where
\begin{equation*}
\ka{\kpar} = \begin{cases}\max \A + \kpar, & \text{ if } \kpar > 0, \\  \min \A + \kpar, & \text{ if } \kpar < 0.\end{cases}
\end{equation*}
Similarly, an $\lang$-\emph{phantom} of an \OMPIQ{} $\q = (\TO,\varrho)$ for $\kpar\ne 0$ is an $\lang$-formula $\Phi^\kpar_{\q}(x,y)$ such that, for any ABox $\A$,
\begin{equation*}
\SA \models \Phi^\kpar_{\q}(a,b) \quad \text{ iff } \quad  (a,b, \ka{\kpar})\in \ans^\Z(\q,\A), \qquad \text{  for any } a,b \in \ind (\A).
\end{equation*}
\end{definition}
For $\bot$-free \OMPIQ{}s of the form $\q = (\TO,\varrho)$,  we can, by~\eqref{observ2}, construct $\lang$-phantoms $\Phi^\kpar_{\q}(x,y)$ from $\lang$-phantoms $\Phi^\kpar_{\smash{\q^\ddagger}}$ for the $\LTL$ \OMPIQ{}s $\q^\ddagger = (\R,\varrho)^\ddagger$, provided that they exist.
For $\bot$-free \OMAQ{}s of the form $\q = (\TO, B)$, we have the following analogue of Lemma~\ref{lemma:basis}:
\begin{lemma}\label{lemma:basis-k}
Let $\lang$ be one of $\FO(<)$, $\FOE$, or $\FO(\RPR)$. A $\bot$-free $\DL_{\horn}\Xallop$ \OMAQ{}  $(\TO,B)$ with $\TO = \T \cup \R$ and a basic concept $B$ has an $\lang$-phantom for $k \neq 0$ whenever
\begin{itemize}
\item[--] the $\LTL_\horn\Xallop$ \OMAQ{} $\q^\dagger = (\T_\R, B)^\dagger$ has an $\lang$-phantom for $\kpar$ and

\item[--] the $\LTL_\horn\Xallop$ \OMPIQ{}s $\q^\ddagger_\rtp = (\R, \rtp)^\ddagger$ are $\lang$-rewritable, for role types $\rtp$ for~$\R$.
\end{itemize}
\end{lemma}
The proof relies on the fact that Lemma~\ref{thm:technical} applies to all $n \in \Z$ and proceeds by taking an $\lang$-phantom for $\q^\dagger$ and replacing every $A(s)$ with $A(x,s)$, every $(\exists P)^\dagger(s)$ with $\exists y\,P(x,y,s)$, every $(\exists P^-)^\dagger(s)$ with $\exists y\,P(y,x,s)$ and every $(\exists G_\varrho)^\dagger(s)$ with $\exists y\,\rew_\rtp(x,y,s)$, where $\rew_\rtp(x,y,s)$ is an $\lang$-rewriting for $(\R, \rtp)$ provided by  Proposition~\ref{prop:roleOMIQs} for each $\smash{\q^\ddagger_\rtp}$.

Since the \LTL{} canonical models have ultimately periodic structure, there are only finitely many non-equivalent phantoms for any \OMPIQ{} $\q= (\TO,\varkappa)$. We also have the following:
\begin{restatable}{lemma}{dlperiod}\label{th:dl-period}
For any $\DL_{\horn}\Xallop$ ontology $\TO$,
there are positive integers $s_{\TO}$ and $p_{\TO}$  such that, for any  positive temporal concept $\varkappa$ \textup{(}all of whose concept and role names occur in
$\TO$\textup{)}, any ABox $\A$ consistent with $\TO$, and any $w\in  \Delta^{\can}$, 
\begin{align*}
& \can, \kpar\models \varkappa(w) \quad
    \text{iff} \quad \can, \kpar + p_{\TO} \models
    \varkappa(w), \quad \text{ for any } \kpar \geq M^w_\A + s_{\TO} +
    |\varkappa| p_{\TO},\\*
&    \can, \kpar \models \varkappa(w) \quad \text{iff} \quad
    \can, \kpar - p_{\TO} \models \varkappa(w), \quad \text{ for any } \kpar
    \leq \bar{M}^w_\A - (s_{\TO} + |\varkappa| p_{\TO}),
\end{align*}
where, for $w = a S_1^{m_1} \dots S_l^{m_l}\in \Delta^{\can}$, with $l \ge 0$, we denote
\begin{equation*}
M^w_\A = \max \{\, \max \A, m_1, \dots, m_l \,\} \qquad \text{ and }\qquad \bar{M}^w_\A = \min \{\, \min \A, m_1, \dots, m_l \,\}.
\end{equation*}
\end{restatable}
The bounds on $s_{\TO}$ and $p_{\TO}$ are double- and triple-exponential in the size of $\TO$, and the proof of this lemma, which generalises Lemma~\ref{period:A} in Section~\ref{sec:ltl-can-period}, can be found in Appendix~\ref{app:psi}; cf.~Lemma 22 by~\citeA{AIJ21} for the \LTL{} case.

We now have all of the ingredients to prove Lemma~\ref{lemma:mainDL}, our main technical result in Section~\ref{OBDAfor Temporal}. First, we illustrate the rather involved construction in its proof by an example.
\begin{example}\em\label{ex:phantoms}
Consider the \OMPIQ{} $\q = (\TO,\varkappa)$ with $\varkappa = \exists Q. \Rdiamond B$ and $\TO = \T \cup \R$, where
\begin{equation*}
\T = \bigl\{\, A \sqsubseteq \Rnext \exists P,\ \ \exists P^- \sqsubseteq \Rnext^2 B, \ \ \exists Q \sqsubseteq \exists Q \,\bigr\} \qquad\text{ and }\qquad \R = \bigl\{\, P \sqsubseteq \Lnext Q \,\bigr\}
\end{equation*}
(recall that we require that the TBox contains the same role names as the RBox).
We construct an $\FOE$-rewriting $\rew_{\TO, \varkappa}(x,t)$ of $\q$ by induction. To illustrate, consider an ABox $\A$ consistent with $\TO$, $a \in \ind(\A)$ and $\ell \in \tem(\A)$ such that $a \in (\exists Q. \Rdiamond B)^{\can(\ell)}$. Our rewriting has to cover two main cases; see Fig.~\ref{fig:phantoms}.

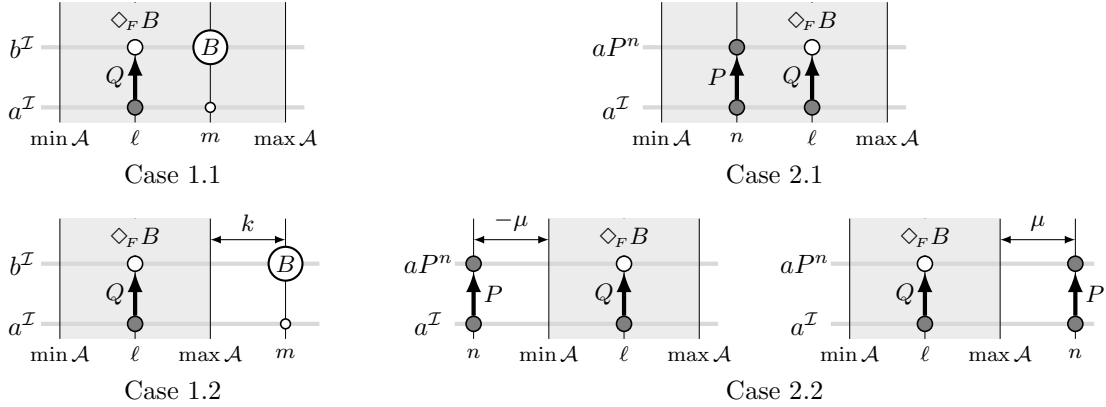
\begin{figure}[t]
\centering%
\begin{tikzpicture}[xscale=1, yscale=0.8]\footnotesize
\begin{scope}
\fill[gray!15] (0,-0.25) rectangle +(2,2);
\begin{scope}\small
\draw[object-timeline] (-0.25, 0) -- ++(3.7,0); \node at (-0.5, 0) {$a^\I$};
\draw[object-timeline] (-0.25, 1) -- ++(3.7,0); \node at (-0.5, 1) {$b^\I$};
\end{scope}
\begin{scope}\scriptsize
\draw[time-guideline] (1,-0.25) -- ++(0,2); \node at (1,-0.5) {$\ell$};
\draw[time-guideline] (3,-0.25) -- ++(0,2); \node at (3,-0.5) {$m$};
\draw[time-guideline] (0,-0.25) -- ++(0,2); \node at (0,-0.5) {$\min \A$};
\draw[time-guideline] (2,-0.25) -- ++(0,2); \node at (2,-0.5) {$\max \A$};
\end{scope}
\draw[<->] (2,1.4) -- ++(1,0) node[midway, above] {$k$};
\node[ppoint] (al) at  (1,0) {};
\node[qpoint,label={[fill=gray!15,above]{$\Rdiamond B$}}] (bl) at  (1,1) {};
\node[point] (am) at  (3,0) {};
\node[wpoint] (bm) at  (3,1) {$B$};
\begin{scope}[ultra thick]
\draw[->] (al) to node[left,midway] {$Q$} (bl);
\end{scope}
\node at (1.5,-1.1) {\small Case 1.2};
\end{scope}
\begin{scope}[yshift=36mm]
\fill[gray!15] (0,-0.25) rectangle +(3,2);
\begin{scope}\small
\draw[object-timeline] (-0.25, 0) -- ++(3.5,0); \node at (-0.5, 0) {$a^\I$};
\draw[object-timeline] (-0.25, 1) -- ++(3.5,0); \node at (-0.5, 1) {$b^\I$};
\end{scope}
\begin{scope}\scriptsize
\draw[time-guideline] (1,-0.25) -- ++(0,2); \node at (1,-0.5) { $\ell$};
\draw[time-guideline] (2,-0.25) -- ++(0,2); \node at (2,-0.5) {$m$};
\draw[time-guideline] (0,-0.25) -- ++(0,2); \node at (0,-0.5) {$\min \A$};
\draw[time-guideline] (3,-0.25) -- ++(0,2); \node at (3,-0.5) {$\max \A$};
\end{scope}
\node[ppoint] (al) at  (1,0) {};
\node[qpoint,label={[fill=gray!15,above]{$\Rdiamond B$}}] (bl) at  (1,1) {};
\node[point] (am) at  (2,0) {};
\node[wpoint] (bm) at  (2,1) {$B$};
\begin{scope}[ultra thick]
\draw[->] (al) to node[left,midway] {$Q$} (bl);
\end{scope}
\node at (1.5,-1.1) {\small Case 1.1};
\end{scope}
\begin{scope}[xshift=80mm,yshift=36mm]
\fill[gray!15] (0,-0.25) rectangle +(3,2);
\begin{scope}\small
\draw[object-timeline] (-0.25, 0) -- ++(3.5,0); \node at (-0.6, 0) {$a^\I$};
\draw[object-timeline] (-0.25, 1) -- ++(3.5,0); \node at (-0.6, 1) {$aP^n$};
\end{scope}
\begin{scope}\scriptsize
\draw[time-guideline] (1,-0.25) -- ++(0,2); \node at (1,-0.5) {$n$};
\draw[time-guideline] (2,-0.25) -- ++(0,2); \node at (2,-0.5) {$\ell$};
\draw[time-guideline] (0,-0.25) -- ++(0,2); \node at (0,-0.5) {$\min \A$};
\draw[time-guideline] (3,-0.25) -- ++(0,2); \node at (3,-0.5) {$\max \A$};
\end{scope}
\node[ppoint] (an) at  (1,0) {};
\node[ppoint] (bn) at  (1,1) {};
\node[ppoint] (al) at  (2,0) {};
\node[qpoint,label={[fill=gray!15,above]{$\Rdiamond B$}}] (bl) at  (2,1) {};
\begin{scope}[ultra thick]
\draw[->] (an) to node[left,midway] {$P$} (bn);
\draw[->] (al) to node[left,midway] {$Q$} (bl);
\end{scope}
\node at (1.5,-1.1) {\small Case 2.1};
\end{scope}
\begin{scope}[xshift=80mm,yshift=0mm]
\begin{scope}[xshift=-25mm]
\fill[gray!15] (1,-0.25) rectangle +(2,2);
\begin{scope}\small
\draw[object-timeline] (-0.25, 0) -- ++(3.5,0); \node at (-0.6, 0) {$a^\I$};
\draw[object-timeline] (-0.25, 1) -- ++(3.5,0); \node at (-0.6, 1) {$aP^n$};
\end{scope}
\begin{scope}\scriptsize
\draw[time-guideline] (0,-0.25) -- ++(0,2); \node at (0,-0.5) {$n$};
\draw[time-guideline] (2,-0.25) -- ++(0,2); \node at (2,-0.5) {$\ell$};
\draw[time-guideline] (1,-0.25) -- ++(0,2); \node at (1,-0.5) {$\min \A$};
\draw[time-guideline] (3,-0.25) -- ++(0,2); \node at (3,-0.5) {$\max \A$};
\end{scope}
\node[ppoint] (an) at  (0,0) {};
\node[ppoint] (bn) at  (0,1) {};
\node[ppoint] (al) at  (2,0) {};
\node[qpoint,label={[fill=gray!15,above]{$\Rdiamond B$}}] (bl) at  (2,1) {};
\draw[<->] (0,1.4) -- ++(1,0) node[midway, above] {$-\mu$};
\begin{scope}[ultra thick]
\draw[->] (an) to node[right,midway] {$P$} (bn);
\draw[->] (al) to node[left,midway] {$Q$} (bl);
\end{scope}
\end{scope}
\begin{scope}[xshift=15mm]
\fill[gray!15] (1,-0.25) rectangle +(2,2);
\begin{scope}\small
\draw[object-timeline] (0.75, 0) -- ++(3.5,0); \node at (0.4, 0) {$a^\I$};
\draw[object-timeline] (0.75, 1) -- ++(3.5,0); \node at (0.4, 1) {$aP^n$};
\end{scope}
\begin{scope}\scriptsize
\draw[time-guideline] (4,-0.25) -- ++(0,2); \node at (4,-0.5) {$n$};
\draw[time-guideline] (2,-0.25) -- ++(0,2); \node at (2,-0.5) {$\ell$};
\draw[time-guideline] (1,-0.25) -- ++(0,2); \node at (1,-0.5) {$\min \A$};
\draw[time-guideline] (3,-0.25) -- ++(0,2); \node at (3,-0.5) {$\max \A$};
\end{scope}
\node[ppoint] (an) at  (4,0) {};
\node[ppoint] (bn) at  (4,1) {};
\node[ppoint] (al) at  (2,0) {};
\node[qpoint,label={[fill=gray!15,above]{$\Rdiamond B$}}] (bl) at  (2,1) {};
\draw[<->] (3,1.4) -- ++(1,0) node[midway, above] {$\mu$};
\begin{scope}[ultra thick]
\draw[->] (an) to node[right,midway] {$P$} (bn);
\draw[->] (al) to node[left,midway] {$Q$} (bl);
\end{scope}
\end{scope}
\node at (1.5,-1.1) {\small Case 2.2};
\end{scope}
\end{tikzpicture}%
\caption{Cases in Example~\ref{ex:phantoms}: $\varkappa = \exists Q. \Rdiamond B$.}\label{fig:phantoms}
\end{figure}

\smallskip

\noindent\emph{Case} 1: If there is $b \in \ind(\A)$ with $(a,b) \in Q^{\can(\ell)}$ and $b \in (\Rdiamond B)^{\can(\ell)}$, then we can describe this configuration by the formula
\begin{equation}\label{abox}
\exists y \, \bigl[\rew_{\R, Q}(x,y,t) \land \rew_{\TO, \Rdiamond B}(y,t)\bigr],
\end{equation}
where $\rew_{\R, Q}(x,y,t)$ is a rewriting of the \OMAQ{} $(\R, Q)$ provided by Proposition~\ref{prop:roleOMIQs}, and $\rew_{\TO, \Rdiamond B}(x,t)$ is a rewriting of a simpler \OMPIQ{} $(\TO, \Rdiamond B)$, which can be defined as follows:
\begin{equation*}
\rew_{\TO, \Rdiamond B}(x,t)  ~=~ \exists s \, \bigl( (s > t) \land \rew_{\TO, B}(x,s) \bigr) \ \lor \ \bigvee_{k \in (0,N]} \Phi_{\TO, B}^k(x),
\end{equation*}
where $\rew_{\TO, B}(x,t)$ is the rewriting  and the $\Phi_{\TO, B}^k(x)$ are the phantoms of the OMAQ $(\TO, B)$ provided by Lemmas~\ref{lemma:basis} and~\ref{lemma:basis-k}, respectively, and $N = s_\TO + |\varkappa|p_\TO$; see Lemma~\ref{th:dl-period}.
That is, $N$ is a suitable integer that depends on~$\q$ and reflects the periodicity of the canonical model of~$\TO$: a `witness' $B(b,m)$ for $(\Rdiamond B)(b,\ell)$ is either in the active temporal domain $\tem(\A)$ (see Case 1.1 in Fig.~\ref{fig:phantoms}) or at a distance $k \leq N$  from $\max\A$ (see Case 1.2 in Fig.~\ref{fig:phantoms}). The latter requires phantoms~$\Phi_{\TO, B}^k(x)$; note that due to $\Rdiamond$ we look only at integers larger than~$\ell$ (in particular, only positive~$k$).

\smallskip

\noindent
\emph{Case} 2: If there exists $n \in \Z$ such that
\begin{itemize}
\item[--] the canonical model for $(\TO, \A)$ contains $P(a,aP^n,n)$,

\item[--] the canonical model for $(\R,\{P(a,aP^n,n)\})$ contains $Q(a,aP^n,\ell)$,

\item[--] the canonical model for $(\TO,\{P(a,aP^n,n)\})$ contains $(\Rdiamond B)(aP^n, \ell)$,
\end{itemize}
then we consider two further options.

\noindent\emph{Case} 2.1: If $n \in \tem(\A)$, then, for the first item, we  use a rewriting $\rew_{\TO, \exists P}(x,s)$ of the \OMAQ{} $(\TO, \exists P)$. For the second, we need a formula $\Theta_{P \leadsto Q}(s, t)$ such that $\SA\models\Theta_{P \leadsto Q}(n,\ell)$ iff the canonical model for $(\R,\{P(a,b,n)\})$ contains $Q(a,b,\ell)$, for $n,\ell \in \tem(\A)$. To capture the third condition, we need a formula $\tilde\Psi_{P, \Rdiamond B}(x, s, t)$ such that $\SA\models \tilde\Psi_{P, \Rdiamond B}(a, n, \ell)$ iff the canonical model for $(\TO, \{P(a,b,n)\})$ contains $(\Rdiamond B)(b,\ell)$, for $n,\ell \in \tem(\A)$. Assuming that these formulas are available, we can express Case~2.1 by the formula
\begin{equation}\label{case2.1}
\exists s \, \bigl[\rew_{\TO, \exists P}(x,s) \ \land  \ \Theta_{P \leadsto Q}(s, t) \ \land \ \tilde\Psi_{P, \Rdiamond B}(x, s, t)\bigr].
\end{equation}
\emph{Case} 2.2: If $n \notin \tem(\A)$, then we need phantom versions $\Phi^\mu_{\TO, \exists P}(x)$, $\Theta^\mu_{P \leadsto Q}(t)$ and  $\tilde\Psi^\mu_{P, \Rdiamond B}(x, t)$ of the subformulas in~\eqref{case2.1}. For example, $\tilde\Psi^\mu_{P, \Rdiamond B}(x, t)$ should be such that $\SA \models \tilde\Psi^\mu_{P, \Rdiamond B}(a, \ell)$ iff the canonical model for $(\TO, \{P(a,b,\ka{\mu})\})$ contains $(\Rdiamond B)(b,\ell)$, for $\ell\in\tem(\A)$. Provided that such phantoms are available, Case~2.2 can be represented by the formula
\begin{equation}\label{case2.2}
\bigvee_{\mu\in [-N,0) \cup (0, N]} \bigl[\Phi^{\mu}_{\TO, \exists P}(x) \ \land \ \Theta_{P \leadsto Q}^{\mu}(t)\  \land \ \tilde\Psi_{P, \Rdiamond B}^{\mu}(x,t)\bigr];
\end{equation}
notice that the quantifier $\exists s$ (corresponding to the choice of $n$) is replaced by the (finite) disjunction over $\mu$, while the argument $s$ of the subformulas is now shifted to the superscript~$\mu$ of the phantoms.

The required $\FOE$-rewriting $\rew_{\TO, \varkappa}(x,t)$ of $\q$ is a disjunction of~\eqref{abox}--\eqref{case2.2}.
All auxiliary formulas mentioned above are constructed using the same type of analysis that will be explained in full detail below. For example, in formula~\eqref{case2.1}, that is, Case~2.1 with $n\in\tem(\A)$, one can take
\begin{equation*}
\tilde\Psi_{P, \Rdiamond B}(x, s, t) \ \ = \ \ \exists s' \, \bigl( (s'>t) \land \tilde\Psi_{P,B}(x, s, s')\bigr) \ \ \lor \bigvee_{k \in (0,N]} \bar \Psi_{P, B}^{k}(x,s),
\end{equation*}
where $\bar\Psi_{P, B}^{k}(x,s)$ is such that $\SA \models \bar\Psi^\kpar_{P, B}(a, n)$ iff the canonical model for $(\TO,\{P(a,b, n)\})$ contains $B(b,\ka{\kpar})$. Note the similarity to formula $\rew_{\TO, \Rdiamond B}(x,t)$ in Case 1 above: the first disjunct covers the case when the canonical model contains $B(b,m)$ with $n$ in the active temporal domain $\tem(\A)$, while the second disjunct covers the case when $B(b,m)$ is found beyond $\tem(\A)$, and so we have to resort to the phantoms.
\end{example}

\begin{table}[t]
\centering\begin{tabular}{lcl}\toprule
$\SA \models \Theta_{S \leadsto S_1}(n, n_1)$ & $\Longleftrightarrow$ &
$\R\models S\sqsubseteq \nxt^{n_1 - n} S_1$
\\[2pt]
$\SA \models \Theta^k_{S \leadsto S_1}(n_1)$ & $\Longleftrightarrow$ &
$\R\models S\sqsubseteq \nxt^{n_1 -  \ka{\kpar}} S_1$
\\[2pt]
$\SA \models \bar \Theta^{\kpar_1}_{S \leadsto S_1}(n)$ & $\Longleftrightarrow$ &
$\R\models S\sqsubseteq \nxt^{\ka{\kpar_1} - n} S_1$
\\[2pt]
$\SA \models \Theta^{\kpar, \kpar_1}_{S \leadsto S_1}$ & $\Longleftrightarrow$ &
$\R\models S\sqsubseteq \nxt^{\ka{\kpar_1} - \ka{\kpar}} S_1$
\\[2pt]\midrule
$\SA \models \Xi_{B \leadsto B_1}(n, n_1)$ & $\Longleftrightarrow$ &
$\TO\models B\sqsubseteq \nxt^{n_1 - n} B_1$
\\[2pt]
$\SA \models \Xi^k_{B \leadsto B_1}(n_1)$ & $\Longleftrightarrow$ &
$\TO\models B\sqsubseteq \nxt^{n_1 - \ka{\kpar}} B_1$
\\[2pt]
$\SA \models \bar \Xi^{\kpar_1}_{B \leadsto B_1}(n)$ & $\Longleftrightarrow$ &
$\TO\models B\sqsubseteq \nxt^{\ka{\kpar_1} - n} B_1$
\\[2pt]
$\SA \models \Xi^{\kpar, \kpar_1}_{B \leadsto B_1}$ & $\Longleftrightarrow $ &
$\TO\models B\sqsubseteq \nxt^{\ka{\kpar_1} - \ka{\kpar}} B_1$
\\[2pt]\midrule
\multicolumn{3}{l}{$\SA \models  \Psi^{\mu_1,\dots,\mu_l}_{S_1 \dots S_l, \varkappa}(a,n_1, \dots, n_l, n) $ $\Longleftrightarrow$}\\
%&
\multicolumn{3}{r}{\hspace*{14em}$S_1(w_0, w_1,m_1), \dots, S_l(w_{l-1}, w_l,m_l) \in \C$ and $w_l \in \varkappa^{\C(n)}$,}\\
\multicolumn{3}{l}{$\SA \models  \Psi^{\mu_1,\dots,\mu_l,k}_{S_1 \dots S_l, \varkappa}(a,n_1, \dots, n_l)$ $\Longleftrightarrow$}\\
%&
\multicolumn{3}{r}{$S_1(w_0, w_1,m_1), \dots, S_l(w_{l-1}, w_l,m_l)\in \C$ and  $w_l \in \varkappa^{\C(\ka{\kpar})}$,}\\[4pt]
\multicolumn{3}{l}{where $\C = \C_{\TO,\{\exists S_1(a,m_1)\}}$, $w_0 = a$ and $w_i = w_{i-1}S_i^{m_i}$, for all $i$, $1 \leq i \leq l$,
}\\
\multicolumn{3}{r}{and $m_i = \begin{cases}
  n_i, & \text{if $\mu_i = 0$},\\
  \ka{\mu_i}, & \text{otherwise},
 \end{cases}$ for all $i$, $1 \leq i \leq l$.}
\\[3pt]\bottomrule
\end{tabular}
\caption{Semantic characterisations of the auxiliary formulas in rewritings: $a\in \ind(\A)$, $n, n_1, \dots, n_l\in\tem(\A)$, $k\in\Z\setminus\{0\}$ and $\mu_1,\dots,\mu_l\in \Z$.}\label{aux-formulas}
\end{table}

\begin{lemma}\label{lemma:mainDL}
Let $\lang$ be one of $\FO(<)$, $\FOE$, or $\FO(\RPR)$.  A $\bot$-free $\DL_{\horn}\Xallop{}$ \OMPIQ{} $\q = (\TO,\varkappa)$ with $\TO = \T \cup \R$  and a positive temporal concept $\varkappa$ is $\lang$-rewritable if
\begin{itemize}
\item[--] for any \OMAQ{} $\q = (\TO, B)$ with $B$ in the alphabet of $\TO$, there exist an $\lang$-rewriting $\rew(x,t)$ of $\q$  and $\lang$-phantoms~$\Phi_{\q}^\kpar(x)$ of $\q$ for all $\kpar\ne 0$\textup{;}

  \item[--] for any OMAQ $\q = (\TO, P)$ with $P$ from $\TO$, there exist an $\lang$-rewriting $\rew(x,y,t)$ of $\q$ and $\lang$-phantoms $\Phi_{\q}^\kpar(x,y)$ of $\q$ for all $\kpar\ne 0$.
\end{itemize}
\end{lemma}
\begin{proof}
We require a number of auxiliary FO-formulas (similar to $\Theta^k_{P \leadsto Q}(t)$ and  $\Psi^k_{P, \Rdiamond B}(x, t)$ in Example~\ref{ex:phantoms}), which are characterised semantically in Table~\ref{aux-formulas} and defined syntactically in Appendix~\ref{sec:appendix}.
For the last two items in Table~\ref{aux-formulas}, examples of the canonical model $\C$ are illustrated in Fig.~\ref{fig:successors} for $l = 3$; note that, in general, the order of the time instants $m_1,\dots,m_l,n$ can be arbitrary.
Note that if $\mu_i \ne 0$, then the formula $\Psi^{\mu_1,\dots,\mu_l}_{S_1\dots S_l,\varkappa}(x,t_1,\dots,t_l, t)$ does not depend on $t_i$.
In the context of Example~\ref{ex:phantoms}, we have $\tilde\Psi_{P, \varkappa}(x,s, s') = \Psi^{0}_{P, \varkappa}(x, s, s')$, $\tilde\Psi^\mu_{P, \varkappa}(x, t) = \Psi^{\mu}_{P, \varkappa}(x, 0, t)$ and $\bar \Psi^k_{P, \varkappa}(x, s) = \Psi^{0, k}_{P, \varkappa}(x, s)$, where $0$ is used in place of the dummy argument because $\Psi^{\mu}_{P, \varkappa}(x,t_1,t)$
does not depend on $t_1$.

As we shall see in Appendix~\ref{sec:appendix}, the formulas in Table~\ref{aux-formulas} are all in $\FOE$ if $\TO$ is a $\DL_{\horn}\Xallop$ ontology, and in $\FO(<)$ if  $\TO$ is in $\DL_{\horn/\core}\Xbox$ and $\DL_{\horn/\rhorn}\Xbox$.
Using these formulas, we now construct an $\lang$-rewriting $\rew_{\TO,\varkappa}(x,t)$ of $\q$ and $\lang$-phantoms $\Phi^k_{\TO,\varkappa}(x)$ for $k \ne 0$. For the given $\DL_\horn\Xallop$ ontology $\TO = \T \cup \R$, we take $s_{\TO}$ and $p_{\TO}$ as defined in Lemma~\ref{th:dl-period}.

\smallskip

\noindent\emph{Case} $\varkappa = A$. An $\lang$-rewriting $\rew_{\TO, A}(x,t)$ and $\lang$-phantoms $\Phi_{\TO, A}^k(x)$ are given by the formulation of the theorem.

\smallskip

\noindent \emph{Cases} $\varkappa = \varkappa_1 \sqcap \varkappa_2$ and $\varkappa = \varkappa_1 \sqcup \varkappa_2$ are trivial.

\smallskip

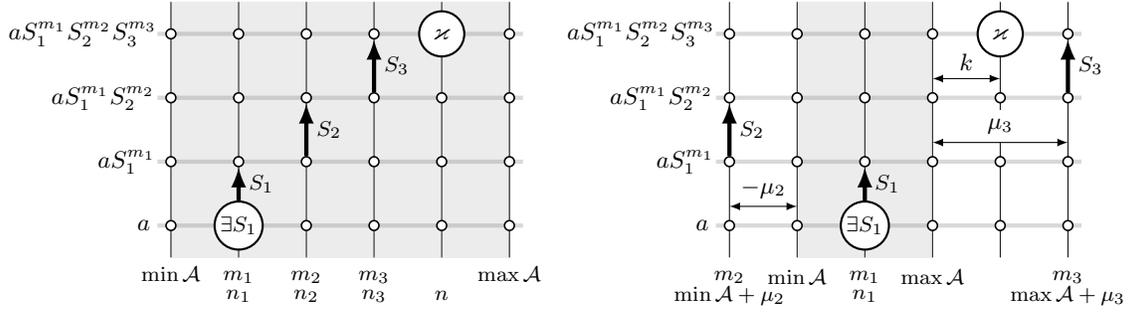
\begin{figure}[t]
\centering%
\begin{tikzpicture}[xscale=0.45, yscale=0.85, semithick]\footnotesize
\foreach \y in {0,1,2,3} {
\draw[object-timeline] (-6.4,\y) -- ++(10.8,0);
}
\foreach \x/\l/\c in {-4/$m_1$/$n_1$,-2/$m_2$/$n_2$,0/$m_3$/$n_3$,2//$n$} {
\draw[time-guideline] (\x,-0.5) -- ++(0,4);
\node at (\x,-0.8) {\scriptsize\l};
\node at (\x,-1.1) {\scriptsize\c};
}
\fill[gray,fill opacity=0.15] (-6,-0.5) rectangle +(10,4);
\draw[time-guideline] (-6,-0.5) -- ++(0,4); \node at (-6,-0.75) {\scriptsize$\min \A$};
\draw[time-guideline] (4,-0.5) -- ++(0,4); \node at (4,-0.75) {\scriptsize$\max \A$};
\node  at (-6.8,0) {$a$};
\node (a-2) at (-6,0) [point]{};
\node (a-1) at (-4,0) [wpoint]{\scriptsize $\exists S_1$};
\node (a0) at (-2,0) [point]{};
\node (a1) at (0,0) [point]{};
\node (a2) at (2,0) [point]{};
\node (a3) at (4,0) [point]{};
\node at (-7.3,1) {$aS_1^{m_1}$};
\node (c-2) at (-6,1) [point]{};
\node (c-1) at (-4,1) [point]{};
\node (c0) at (-2,1) [point]{};
\node (c1) at (0,1) [point]{};
\node (c2) at (2,1) [point]{};
\node (c3) at (4,1) [point]{};
\node  at (-8,2) {$aS_1^{m_1}S_2^{m_2}$};
\node (d-1-2) at (-6,2) [point]{};
\node (d-1-1) at (-4,2) [point]{};
\node (d-10) at (-2,2) [point]{};
\node (d-11) at (0,2) [point]{};
\node (d-12) at (2,2) [point]{};
\node (d-13) at (4,2) [point]{};
\node at (-8.6,3) {$aS_1^{m_1}S_2^{m_2}S_3^{m_3}$};
\node (d0-2) at (-6,3) [point]{};
\node (d0-1) at (-4,3) [point]{};
\node (d00) at (-2,3) [point]{};
\node (d01) at (0,3) [point]{};
\node (d02) at (2,3) [wpoint,minimum size=6mm]{$\varkappa$};
\node (d03) at (4,3) [point]{};
\draw[->,ultra thick] (a-1)  to node [right]{\scriptsize $S_1$} (c-1);
\draw[->,ultra thick] (c0)  to node [right]{\scriptsize $S_2$} (d-10);
\draw[->,ultra thick] (d-11)  to node [right]{\scriptsize $S_3$} (d01);
\begin{scope}[xshift=165mm]
\foreach \y in {0,1,2,3} {
\draw[object-timeline] (-6.4,\y) -- ++(10.8,0);
}
\foreach \x/\l/\c in {-6/$m_2$/$\min \A + \mu_2$,-4/$\min \A$/,-2/$m_1$/$n_1$,0/$\max \A$/,2//,4/$m_3$/$\max\A + \mu_3$} {
\draw[time-guideline] (\x,-0.5) -- ++(0,4);
\node at (\x,-0.8) {\scriptsize\l};
\node at (\x,-1.1) {\scriptsize\c};
}
\fill[gray,fill opacity=0.15] (-4,-0.5) rectangle +(4,4);
\node at (-6.8,0) {$a$};
\node (a0) at (-6,0) [point]{};
\node (a1) at (-4,0) [point]{};
\node (a2) at (-2,0) [wpoint]{\scriptsize $\exists S_1$};
\node (a3) at (0,0) [point]{};
\node (a4) at (2,0) [point]{};
\node (a5) at (4,0) [point]{};
\node at (-7.3,1) {$aS_1^{m_1}$};
\node (b0) at (-6,1) [point]{};
\node (b1) at (-4,1) [point]{};
\node (b2) at (-2,1) [point]{};
\node (b3) at (0,1) [point]{};
\node (b4) at (2,1) [point]{};
\node (b5) at (4,1) [point]{};
\node  at (-8,2) {$aS_1^{m_1}S_2^{m_2}$};
\node (c0) at (-6,2) [point]{};
\node (c1) at (-4,2) [point]{};
\node (c2) at (-2,2) [point]{};
\node (c3) at (0,2) [point]{};
\node (c4) at (2,2) [point]{};
\node (c5) at (4,2) [point]{};
\node at (-8.6,3) {$aS_1^{m_1}S_2^{m_2}S_3^{m_3}$};
\node (d0) at (-6,3) [point]{};
\node (d1) at (-4,3) [point]{};
\node (d2) at (-2,3) [point]{};
\node (d3) at (0,3) [point]{};
\node (d4) at (2,3)  [wpoint,minimum size=6mm]{$\varkappa$};
\node (d5) at (4,3) [point]{};
\draw[->,ultra thick] (a2)  to node [right]{\scriptsize $S_1$} (b2);
\draw[->,ultra thick] (b0)  to node [right]{\scriptsize $S_2$} (c0);
\draw[->,ultra thick] (c5)  to node [right]{\scriptsize $S_3$} (d5);
\begin{scope}\footnotesize
\draw[thin,<->] (0,2.3) -- node[above]{$k$} +(2,0);
\draw[thin,<->] (0,1.3) -- node[above,fill=white]{$\mu_3$} +(4,0);
\draw[thin,<->] (-4,0.3) -- node[above]{$-\mu_2$} +(-2,0);
\end{scope}
\end{scope}
\end{tikzpicture}%
\caption{Canonical model for $\Psi^{0,0,0}_{S_1,S_2,S_3,\varkappa}(a, n_1, n_2, n_3,n)$ and $\Psi^{0,\mu_2,\mu_3,k}_{S_1,S_2,S_3,\varkappa}(a,n_1,0,0)$.}\label{fig:successors}
\end{figure}%

\noindent \emph{Case} $\varkappa = \Rbox \varkappa'$. Using Lemma~\ref{th:dl-period}, we take  $N = s_{\TO} + |\varkappa| p_{\TO}$ and set
\begin{align*}
\rew_{\TO,\Rbox\varkappa'}(x,t) & = \forall s\, \bigl[(s > t) \to \rew_{\TO,\varkappa'}(x,s)\bigr] \ \ \land 
\hspace*{-0.5em}\bigwedge_{k \in (0, N]} \hspace*{-0.5em}\Phi_{\TO, \varkappa'}^k(x),\\
\Phi_{\TO,\Rbox\varkappa'}^k(x) & =   \begin{cases}\displaystyle
\bigwedge_{i \in (k, k +N]} \hspace*{-0.5em}\Phi_{\TO, \varkappa'}^i(x), & \text{if } k >0,\\
\displaystyle\rule{0pt}{18pt}
\bigwedge_{i \in (k,0)} \hspace*{-0.5em}\Phi_{\TO, \varkappa'}^i(x)  \land \rew_{\TO,\Rbox\varkappa'}(x,0), & \text{if } k <0,
\end{cases}
\end{align*}

\noindent
The cases of the other temporal operators are similar and left to the reader.

\smallskip

\noindent\emph{Case} $\varkappa = \exists S. \varkappa'$. The rewriting reflects the possible configurations depicted in Fig.~\ref{fig:phantoms}. Case~1 is straightforward and involves individual names only, but in Case~2 we need to pinpoint the time instant when a suitable witness is created. Case~2.1 deals with the (finite) active domain, and so an existential quantifier can be used to fix the required time instant. Case~2.2, however, deals with the infinitely many time instants outside the active domain, but Lemma~\ref{th:dl-period-cor} from Appendix~\ref{sec:appendix} shows that it actually is enough to consider only a bounded number of time instants before $\min \A$ and after $\max \A$, which can be referred to by suitable phantoms.
Using this observation, we take $N = s_\TO + |\varkappa|p_\TO$ and set
\begin{multline*}
\rew_{\TO,\exists S. \varkappa'} (x, t) \ \ = \ \  \exists y\, \bigl(\rew_{\R,S}(x,y,t) \land \rew_{\TO,\varkappa'}(y,t)\bigr) \  \ \vee\\[4pt]
 \bigvee_{\text{role $S_1$ in }\TO} \Bigl[ \exists t_1\, \bigl(\rew_{\TO,\exists S_1} (x, t_1) \land \Theta_{S_1 \leadsto S}(t_1,t) \land \Psi^{0}_{S_1,\varkappa'}(x,t_1,t)\bigr) \ \ \vee \\[-4pt]
 \hspace*{3em}\bigvee_{ \mu\in [-N,0)\cup (0, N]} \hspace*{-1em}\bigl(\Phi_{\TO,\exists S_1}^\mu (x) \land \Theta^\mu_{S_1 \leadsto S}(t) \land \Psi^{\mu}_{S_1,\varkappa'}(x,0,t)\bigr)\Bigr].
\end{multline*}
The three groups of disjuncts correspond to Cases 1,  2.1 and~2.2 in Example~\ref{ex:phantoms}, respectively.
We define the phantoms for $k > 0$ as follows:
\begin{multline*}
\Phi_{\TO,\exists S.\varkappa'}^k (x) \ \  = \ \  \exists y\, \bigl(\Phi^k_{\R,S}(x,y) \land \Phi^k_{\TO,\varkappa'}(y)\bigr) \ \ \vee \\[4pt]
 \bigvee_{\text{role $S_1$ in }\TO} \Bigl[\exists t_1\, \bigl(\rew_{\TO,\exists S_1} (x, t_1) \land \bar \Theta_{S_1 \leadsto S}^k(t_1) \land \Psi^{0,k}_{S_1,\varkappa'}(x,t_1)\bigr) \ \ \vee \\[-4pt]
  \hspace*{3em}\bigvee_{\mu\in [-N,0)\cup (0,k + N]} \hspace*{-2em}\bigl(\Phi_{\TO,\exists S_1}^\mu (x) \land \Theta^{\mu,k}_{S_1 \leadsto S} \land \Psi^{\mu, k}_{S_1,\varkappa'}(x, 0)\bigr)\Bigr].
\end{multline*}
If $k<0$, then we replace the bounds of the last disjunction by $[k-N,0) \cup (0, N]$.
\end{proof}

As a consequence of Lemmas~\ref{lemma:consistency:horn},~\ref{lemma:basis-k} and~\ref{lemma:mainDL} together with the \LTL-rewritability results of~\citeA{AIJ21} we obtain:
\begin{theorem}\label{cor:dl-omiq-rewritability}
For \OMPIQ{}s of the form $(\TO,\varkappa)$ with a positive temporal concept $\varkappa$,  

$(i)$ all $\DL_{\horn/\core}\Xbox$ and
$\DL_{\smash{\horn/\rhorn}}\Xbox$ \OMPIQ{}s are
$\FO(<)$-rewritable\textup{;}

$(ii)$ all $\DL_{\smash{\core}}\Xallop$ \OMPIQ{}s and $\DL_{\smash{\core/\horn}}\Xbox$ are $\FOE$-rewritable\textup{;}

$(iii)$ All $\DL_{\smash{\horn}}\Xallop$ \OMPIQ{}s are $\FO(\RPR)$-rewritable.
\end{theorem}

%

%%%%%%%%%%%%%%%%%%%%%%%%%%%%%%%%%%%%%%%%%%%%%%%%%%%%%%%%%%%%%%%%%%%%%%

%%% Local Variables:
%%% mode: latex
%%% TeX-master: "aij-main"
%%% End:

% !TEX root = tdl-lite.tex

\section{First-Order Temporal OMQs under the Epistemic Semantics}\label{sec:omq}

We use positive temporal concepts and roles as building blocks for our most expressive query language. It is inspired by the epistemic queries introduced by~\citeA{CalvaneseGLLR07} and the SPARQL~1.1 entailment regimes~\cite{GlimmOgbuji13}; cf.~\citeA{Motik12,GutierrezHV07}. The definition and main rewritability result are straightforward, so this section will be brief. A (\emph{temporal}) \emph{ontology-mediated query} (\OMQ{}) is a pair $\q=(\TO, \psi(\avec{x},\avec{t}))$, in which $\TO$ is an ontology and $\psi(\avec{x},\avec{t})$  a first-order formula built from atoms of the form $\varkappa(x,t)$, $\varrho(x,y,t)$, and $t<t'$, where $\varkappa$ and $\varrho$ are a positive temporal concept and role, respectively, $x$ and $y$ are \emph{individual variables}, and $t$ and~$t'$ \emph{temporal variables}; the free variables $\avec{x}$ and $\avec{t}$ of $\psi$ are called the \emph{answer variables} of~$\q$.
Given an ABox $\A$, the \OMQ~$\q$ is evaluated over a two-sorted  structure $\GOA$ with domain $\ind(\A)\cup \tem(\A)$, where the extension of $\varkappa$, $\rho$ and $<$ in $\GOA$ is defined by setting,
for any assignment $\mathfrak{a}$ mapping individual and temporal variables to elements of $\ind(\A)$ and $\tem(\A)$:
\begin{align*}
\GOA\models^\mathfrak{a}\varkappa(x,t)\quad & \text{ iff }\quad (\mathfrak{a}(x),\mathfrak{a}(t))\in \ans(\TO,\varkappa,\A),\\
\GOA\models^\mathfrak{a}\!\varrho(x,y,t)\quad & \text{ iff }\quad (\mathfrak{a}(x),\mathfrak{a}(y),\mathfrak{a}(t)) \in  \ans(\TO,\varrho,\A),\\
\GOA\models^\mathfrak{a} t < t' \quad & \text{ iff } \quad \mathfrak{a}(t) < \mathfrak{a}(t').
\end{align*}
We extend these to arbitrary formulas $\psi$ using the standard clauses for the Boolean connectives and first-order quantifiers over both $\ind(\A)$ and $\tem(\A)$.
Let $\avec{x} = x_1,\dots,x_k$ and $\avec{t} = t_1,\dots,t_m$ be the free variables of
$\psi$. We say that
$(a_1,\dots,a_k,\ell_1,\dots,\ell_m)$ is an \emph{answer} to the \OMQ{} $\q=(\TO,\psi(\avec{x},\avec{t}))$ over $\A$ if $\GOA\models^\mathfrak{a} \psi$, where $\mathfrak{a}(x_i) = a_i$, for all $i$, $1\leq i \leq k$, and $\mathfrak{a}(t_j) = \ell_j$, for all $j$, $1\leq j \leq m$.
Thus, we keep the open-world interpretation of positive temporal concepts and roles, the individual and temporal variables of \OMQ{}s range over the active
domains only, and the first-order constructs in $\psi$ are interpreted under the  epistemic semantics~\cite{CalvaneseGLLR07}.
Let $\lang$ be one of $\FO(<)$, $\FOE$, or $\FO(\RPR)$. We call an \OMQ{} $\q=(\TO, \psi(\avec{x},\avec{t}))$ $\lang$-\emph{rewritable} if there is an $\lang$-formula $\rew(\avec{x},\avec{t})$ such that, for any ABox $\A$ and any tuples $\avec{a}$ and $\avec{\ell}$ in $\ind(\A)$ and $\tem(\A)$, respectively, the pair $(\avec{a},\avec{\ell})$ is an answer to $\q$ over $\A$ iff~$\SA \models \rew(\avec{a},\avec{\ell})$.
It is straightforward to construct an $\lang$-rewriting of~$\q$ by replacing all occurrences of positive temporal concepts and roles in~$\psi$ with their $\lang$-rewritings.  Thus, we obtain:
\begin{theorem}\label{th:constitute}
Let $\lang$ be one of $\FO(<)$, $\FOE$, or $\FO(\RPR)$, and  $\q=(\TO, \psi)$ an \OMQ{}. If all \OMPIQ{}s $(\TO,\varkappa)$ and $(\TO,\varrho)$ with positive temporal concepts $\varkappa$  and roles $\varrho$  in $\psi$ are $\lang$-rewritable, then $\q$ is also $\lang$-rewritable.
\end{theorem}

% !TEX root = TDL-Lite.tex

\section{Conclusions}\label{conclusions}

In this article, aiming to extend the well-developed theory of ontology-based data access (OBDA) to temporal data, we designed a family of 2D ontology languages that combine logics from the \DL{} family for representing knowledge about object domains and clausal fragments of linear temporal logic \LTL{} over $(\Z,<)$ for capturing knowledge about the evolutions of objects in time. We also suggested a 2D query language that integrates first-order logic for querying the object domains with positive temporal concepts and roles as FO-atoms. The FO-constructs in these queries are interpreted under the epistemic semantics, while the temporal concepts and roles under the open-world semantics. The resulting ontology-mediated queries (OMQs) can be regarded as temporal extensions of SPARQL queries under the (generalised) \OWLQL{} direct semantics entailment regime~\cite{KontchakovRRXZ14}. 
Our main result is the identification of classes of OMQs that are $\FO(<)$-, $\FOE$-, or $\FO(\RPR)$-rewritable, with the first two types of rewriting guaranteeing OMQ answering in $\ACz$ and the last one in $\NCo$ for data complexity. In particular, we proved that all $\DL_{\core}^{\Box}$ OMQs are $\FO(<)$-rewritable and all $\DL_\core\Xallop$ OMQs are $\FOE$-rewritable, which means that classical atemporal OBDA with the W3C standard ontology language \OWLQL{} and SPARQL queries can be extended to temporal ontologies, queries and data without sacrificing the data complexity of OMQ answering.

Having designed suitable languages for temporal OBDA and established their efficiency in terms FO-rewritability and data complexity, we are facing a number of further open questions.
\begin{description}
\item[(succinctness)] What is the size of (various types of) minimal rewritings of temporal OMQs compared to that of \OWLQL{} OMQs investigated by~\citeA{DBLP:journals/jacm/BienvenuKKPZ18,DBLP:conf/pods/BienvenuKKPRZ17,DBLP:conf/dlog/BienvenuKKRZ17} The FO-rewritings constructed in this article seem to be far from optimal.

\item[(parameterised complexity)] What is the combined complexity of answering temporal  OMQs, in which, for example, the temporal depth of positive temporal concepts and roles is regarded as the parameter? Or what is the data  complexity of OMQ answering over data instances where the number of individual objects and/or the number of timestamps are treated as parameters.

\item[(non-uniform approach)] How complex is it to decide, given an arbitrary $\DL_{\bool/\horn}\Xallop$ OMAQ or  OMPIQ, whether it is $\FO(<)$-, $\FOE$-, or $\FO(\RPR)$-rewritable? For atemporal OMQs with DL-ontologies, this non-uniform approach to OBDA has been actively developed since the mid 2010s~\cite{DBLP:journals/tods/BienvenuCLW14,DBLP:conf/dlog/Bienvenu0LW16,DBLP:conf/ijcai/BarceloBLP18}. In the temporal case, first steps have recently been made for \LTL{} OMQs by~\citeA{DBLP:conf/time/21}.

\item[(two-sorted conjunctive queries)] Is it possible to generalise the results of this article from OMPIQs to two-sorted first-order conjunctive queries under the open-world semantics? Which of these two languages or their fragments would be more suitable for industrial users?

\item[(Krom RIs)] Are all $\DL_{\krom}^\Box$ OMAQs $\FO(<)$-rewritable? Are all $\DL_{\krom}\Xallop$ OMAQs $\FOE$-rewritable? What is the combined complexity of the consistency problem for $\DL\Xbox_{\bool/\krom}$ KBs? To answer these questions, one may need a type-based technique similar to the approach in the proof of Theorem~16 by~\citeA{AIJ21} as our Lemma~\ref{lemma:basis} is not applicable to Krom role inclusions.
\end{description}

We conclude by emphasising another aspect of the research project we are proposing in this article. It has recently been observed~\cite{DBLP:journals/dint/XiaoDCC19} that OBDA should be regarded as a principled way to integrate and access data via virtual knowledge graphs (VKGs). In VKGs, instead of structuring the integration layer as a collection of relational tables, the rigid structure of tables is replaced by the flexibility of graphs that are kept virtual and embed domain knowledge. In this setting, we propose to integrate temporal data via virtual temporal knowledge graphs.

%********

\section*{Acknowledgements}

This research was supported by the EPSRC UK grants EP/M012646 and EP/M012670 for the joint project  `\textsl{quant$^\textsl{MD}$}: Ontology-Based Management for Many-Dimensional Quantitative Data'\!, by the DFG grant 389792660 as part of TRR~248 -- CPEC~(\textsf{https://perspicuous-computing.science}), and by the RSF grant 22-11-00323 when M.~Zakharyaschev was visiting the HSE University. We are grateful to the anonymous referees of this article for their careful reading, valuable comments, and constructive suggestions.

%
%%%% Local Variables:
%%%% mode: latex
%%%% TeX-master: "aij-main"
%%%% End:

\appendix
% !TEX root = tdl-lite.tex
\section{Proof of Theorem~\ref{app:th:coNP-gbool}}\label{sec:app-atemporal}

\begin{theorem}\label{app:th:coNP-gbool}
There is a $\DL_{\core/\gbool}{}$ ontology $\TO$ such that answering \OMAQ{}s of the form $(\TO,B)$, for a basic concept $B$, and of the form $(\TO, S)$, for a role $S$, is \coNP-hard.
\end{theorem}
\begin{proof}
The proof is by reduction of the 2$+$2 SAT problem, which is known to be \NP-complete~\cite{DBLP:journals/jiis/Schaerf93}. A 2$+$2-CNF formula over a set $\Sigma$  of propositional variables is a CNF formula $\varphi=c_1 \land c_2 \land \dots \land c_n$ such that each clause $c_i$ is  of the form $l_i^1 \lor  l_i^2 \lor \neg l_i^3\lor \neg l_i^4$, for $1 \leq i \leq n$, where $l_i^j$  are elements of $\bar\Sigma = \Sigma \cup \{\mathsf{true}, \mathsf{false}\}$.

Let $\varphi$ be a 2$+$2-CNF. We construct  an \OMAQ{} with a $\DL_{\core/\gbool}{}$ ontology~$\TO $ and an ABox~$\A$ such that a fixed individual is a certain answer to the \OMAQ{} over $\A$ iff $\varphi$ is unsatisfiable.
The ABox $\A$ has an individual $l$ for each letter $l\in\bar\Sigma$ (including $\mathsf{true}$ and $\mathsf{false}$), an individual~$c_i$ for each clause $c_i$, and an individual $f$ for the formula.
It consists of the following atoms, with concepts $C_{\mathsf{true}}$, $C_{\mathsf{false}}$, and roles $L$, $P_1$, $P_2$, $N_3$, and $N_4$:
\begin{align*}
& C_{\mathsf{true}}(\mathsf{true}), \;
C_{\mathsf{false}}(\mathsf{false}),\\
& L(c_i,f), \; P_1(c_i,l_i^1), \;
P_2(c_i,l_i^2), \;
N_3(c_i,l_i^3), \;
N_4(c_i,l_i^4), \qquad \text{ for } 1 \leq i \leq n.
\end{align*}
The ontology $\TO = \T\cup \R$ has additional roles $T$, $F$, and $H_j,G_j,R_j$, for $1 \leq j \leq 4$:
\begin{align*}
\R = \{ \
& P_j \sqsubseteq T \sqcup F
\text{ and } N_j \sqsubseteq T \sqcup F, \text{ for } j=1,2, \\*
& P_1 \sqcap F \sqsubseteq H_1, \quad
P_2 \sqcap F \sqsubseteq H_2, \quad
N_3 \sqcap T \sqsubseteq H_3, \quad
N_4 \sqcap T \sqsubseteq H_4, \nonumber \\*
& L \sqsubseteq  G_j \sqcup R_j \text{ and } G_j \sqcap R_j \sqsubseteq \bot, \text{ for } 1 \leq j \leq 4, \quad
G_1 \sqcap G_2 \sqcap G_3 \sqcap G_4 \sqsubseteq Q
\ \} \quad  \text{ and }  \nonumber \\
\T = \{ \
& \exists T^- \sqcap \exists F^- \sqsubseteq \bot, \quad   % no literal is assigned differently
\exists T^- \sqcap C_{\mathsf{false}} \sqsubseteq \bot, \quad
\exists F^- \sqcap C_{\mathsf{true}} \sqsubseteq \bot, \nonumber \\*
& \exists H_j  \sqcap \exists R_j \sqsubseteq \bot \text{ for } 1 \leq j \leq 4 \ \}. \nonumber
\end{align*}
Informally, each model is viewed as a truth assignment to the propositional variables: a propositional letter belongs to the range of $T$ if it is assumed to be true (and false otherwise, that is, if it belongs to the range of $F$). Then, roles $H_j$ pick up those literals in every clause~$c_i$ that are false under this assignment, and role $Q$ aggregates this information: tuple $(c_i,f)$ belongs to $Q$ if all of its literals are false, and so is the clause $c_i$.  In other words, $f$ is a certain answer to \OMAQ{} $(\TO,\exists Q^-)$ just in case, for every truth assignment, there is a clause~$c_i$ whose positive variables are false, and whose negative variables are true.

We claim that
\begin{equation*}
\varphi \text{ is unsatisfiable \quad iff\quad } \ans(\TO,\exists Q^-,\A) = \{ f \}.
\end{equation*}
Suppose $\varphi$ is unsatisfiable. Let $\I$ be a model of $(\TO,\A)$. We define a truth assignment $\sigma$ by taking $\sigma (l)$ to be true iff $l \in (\exists T^-)^\I$. Since $\varphi$ is unsatisfiable, there is a clause~$c_i$ such that $\sigma(l_i^1 \lor  l_i^2 \lor \neg l_i^3\lor \neg l_i^4)$ is false under the resulting $\sigma$, which means $l_i^1\in (\exists F^-)^\I$, $l_i^2 \in (\exists F^-)^\I$, $l_i^3 \in (\exists T^-)^\I$ and  $l_i^4 \in (\exists T^-)^\I$.  Consequently, $c_i \in \exists H_j^\I$, and $(c_i,f)^\I\in G_j^\I$, for each $j$, $1 \leq j \leq 4$, and so $(c_i,f)^\I \in Q^\I$. It follows that $(\TO,\A)\models \exists Q^-(f)$.

Conversely, suppose $\varphi$ is satisfied under a truth assignment $\sigma$. Consider an interpretation~$\I$ such that $T^\I = \{(c_i, l_i^j)^\I \;| \; \sigma(l_i^j) \text{ is true} \}$, $F^\I = \{(c_i, l_i^j)^\I \;|\; i, j \} \setminus T^\I$,
$H_j^\I = \{(c_i, l_i^j)^\I \;| \;  (c_i, l_i^j)^\I \in F^\I \}$, for $j = 1,2$,
$H_j^\I = \{(c_i, l_i^j)^\I \;| \; (c_i, l_i^j)^\I \in T^\I \}$, for $j = 3,4$, and
$R_j^\I = \{(c_i,f)^\I \; |\; (c_i, l_i^j)^\I \not\in H_j^\I\}$ with $G_j^\I = \{ (c_i,f)^\I \;|\; i \} \setminus R_j^\I$. It can be easily verified that  $\I$ is a model of $(\TO,\A)$. Since $\varphi$ is satisfied under $\sigma$, for every $c_i$, there is either $j=1,2$ with true $\sigma(l^j_i)$ or $j=3,4$ with false $\sigma(l^j_i)$. Thus, for every $c_i$, there is $j$, $1\leq j\leq 4$, with $(c_i,l_i^j)^\I \notin H_j^\I$ and $(c_i,f)^\I\in R_j^\I$. Since $R_j$ and $G_j$ are disjoint, $\I \not \models \exists Q^-(f)$.
\end{proof}

\section{Definitions of Formulas in Table~\ref{aux-formulas} }\label{sec:appendix}

We provide syntactic definitions of the auxiliary formulas in Table~\ref{aux-formulas} for a given $\bot$-free $\DL_{\horn}\Xallop{}$ ontology $\TO$ with a TBox $\T$ and an RBox $\R$. Recall from Section~\ref{sec:proj} that $\R^\ddagger$ is the translation of $\R$ to $\LTL_{\horn}\Xallop{}$, which uses the concept names $A_S$ for roles $S$, while $\T^\dagger_\R$ is the translation of $\T_\R = \T \cup \textbf{(con)}$ to $\LTL_{\horn}\Xallop{}$ defined in Section~\ref{sec:rewrOMAQ}, which uses concept names from $\T$ along with concept names $(\exists S)^\dagger = E_S$ for roles $S$.

In the auxiliary formulas we use the $\FOE$-abbreviation
$t - t' \in r + p \N$, for $r,p \geq 0$, from Remark 3~$(ii)$ of~\citeA{AIJ21} with the following meaning: for $n, n'\in\tem(\A)$,
\begin{align*}
\SA \models n - n' \in r + p \N & \quad\text{ iff }\quad n' +r \in \tem(\A)  \text{ and } n = n' + r+pk, \text{ for some } k\in \N.
\end{align*}
Note that, if $n' + r > \max\A$, then the formula evaluates to~$\bot$.
We use two more shortcuts: $t - t' = r$ for $t - t'\in r + 0\N$ and $t \in r + p \N$ for $t - 0 \in r + p \N$.

Observe that, if $p = 1$, then $\FOE$-formula $t - t' \in r + p \N$ is equivalent to an $\FO(<)$-formula that we abbreviate by $t - t' > r$ and define by taking $t > t'$ for $r = 0$ and
\begin{equation*}
\exists t_1,\dots,t_r \bigl((t > t_r) \land (t_r > t_{r-1}) \land \dots \land (t_2 > t_1) \land (t_1 > t')\bigr) \qquad \text{ for $r > 0$.}
\end{equation*}

Also, we need symbols $\max$ and $\min$ that are used in place of temporal variables: $\max$ stands for a variable $t_{\max}$ additionally satisfying $\neg\exists t\,(t > t_{\max})$; similarly, $\min$ stands for a variable $t_{\min}$ with for $\neg\exists t\,(t < t_{\min})$. Finally, as usual, the empty disjunction is $\bot$.

\subsection{$\Theta$- and $\Xi$-formulas.}

To define the four types of $\Theta$-formulas, we consider the $\R$-canonical rod~$\rod$  for $\{S^\ddagger(0)\}$; see Section~\ref{sec:HornRI:sat}. Let $s = s_{\R^\ddagger, \{S^\ddagger(0)\}}$ and $p = p_{\R^\ddagger, \{S^\ddagger(0)\}}$ be the integers provided by Lemma~\ref{period:A}~$(i)$.
Let $0 \leq s_1 < \dots < s_l\leq s$ be all the numbers with $S_1 \in \rod(s_i)$ and let $1 \leq p_1 < \dots < p_m \leq p$ be all the numbers with $S_1 \in \rod(s + p_i)$. Symmetrically, let $0 \leq s_1' < \dots < s_{l'}' \leq s$ and $1 \leq p_1' < \dots < p'_{m'} \leq p$ be all the numbers with $S_1 \in \rod(-s_i')$ and $S_1 \in \rod(-s - p_i')$.

\smallskip

Formulas $\Theta_{S \leadsto S_1}(t, t_1)$ simply list the cases for the distance between $t_1$ and $t$ such that an $S$ at $t$ implies an $S_1$ at $t_1$: by Lemma~\ref{period:A}~$(i)$, the distance can be one of the $s_i$, or one of the $-s'_i$, or belong to one of the the arithmetic progressions $s + p_i + p\N$ or $-s - p_i' - p\N$. So, we define the formula by taking
\begin{align*}
\Theta_{S \leadsto S_1}(t, t_1) \ \ =   \ \ & \bigvee_{1 \leq i \leq l} (t_1  - t = s_i) \ \ \lor \bigvee_{1 \leq i \leq m} (t_1 - t \in s + p_i + p \N) \ \ \lor \\
 & \bigvee_{1 \leq i \leq l'} (t  - t_1 = s_i') \ \ \lor \bigvee_{1 \leq i \leq m'} (t - t_1 \in s + p_i' + p \N);
\end{align*}
note that the $s_i'$ and the $p_i'$ are non-negative, and we flip the sign of the arithmetic expressions when $t_1 < t$.

\smallskip

Formulas $\bar \Theta^{\kpar_1}_{S \leadsto S_1}(t)$ follow the same principle. If $\kpar_1  > 0$, then we list the cases when the distance between $t$ and $\max{} + \kpar_1$ is suitable. Note, however, that, as $t$ ranges over the active temporal domain only, we have $t <  \max{} + \kpar_1$, and so, we need only cases with the~$s_i$ and the~$p_i$, but not with the~$s_i'$ and the~$p_i'$. So, we set
\begin{align*}
\bar \Theta^{\kpar_1}_{S \leadsto S_1}(t)  \ \ = &  \bigvee_{1 \leq i \leq l} (\max{}  - t =  s_i - \kpar_1) \ \ \lor{}\\ & \bigvee_{1 \leq i \leq m} \Bigl[
(\max{} - t\in s + p_i + (p - \kpar_1 \bmod{p}) + p \N)  \quad
 \lor \hspace*{-1em}\bigvee_{\begin{subarray}{c}0 \leq q \leq s + p_i \text{ with}\\ \kpar_1 \in s + p_i-q + p\N\end{subarray}} \hspace*{-2em}(\max{} - t = q)\Bigr].
\end{align*}
Since the value of $\max + \kpar_1$ does not belong to the active temporal domain, we cannot use the $(\max + \kpar_1) - t\in s + p_i + p\N$ abbreviation directly and have to consider two cases, depending on whether $\max{} - t$ is larger or smaller than $s + p_i$ (see Fig.~\ref{fig:two-cases-theta}): these two cases are encoded in the two disjuncts in the second line of  the definition of $\bar \Theta^{\kpar_1}_{S \leadsto S_1}(t)$.
\begin{figure}[t]
\centering%
\begin{tikzpicture}[xscale=1.3, yscale=0.8, semithick]\footnotesize
\begin{scope}
\fill[gray!15] (0,-0.25) rectangle +(5,0.25);
\begin{scope}[fill=gray!70]
\fill (4,1) rectangle ++(1.5, 0.15);
\fill (7,1) rectangle ++(1.5, 0.15);
\fill (5.5,0.85) rectangle ++(1.5, 0.15);
\fill (8.5,0.85) rectangle ++(1.5, 0.15);
\end{scope}
\begin{scope}[thin,draw=gray]\scriptsize
\draw[ultra thick] (-0.25, 0) -- ++(10.75,0);
\draw (2,-0.4) -- ++(0,2); \node at (2,-0.7) { $t$};
\draw (0,-0.4) -- ++(0,1); \node at (0,-0.7) {$\min$};
\draw (5,-0.4) -- ++(0,2.1); \node at (5,-0.7) {$\max$};
\draw (10,-0.4) -- ++(0,2); \node at (10,-0.7) {$\max + \kpar_1$};
\draw (4, 0.8) -- (4, 1.7);
\draw (5.5, 0.6) -- (5.5, 1);
\end{scope}
\node[wpoint,minimum size=5mm] at (2,0) {$S$};
\node[wpoint,minimum size=5mm] at (10,0) {$S_1$};
\draw[semithick] (2,1) -- (10,1);
\begin{scope}[ultra thick]
\draw (2, 0.8) -- +(0,0.4);
\node at (3,1.3) {$s + p_i$};
\draw (4, 0.8) -- +(0,0.4);
\foreach \x in {5.5, 7, 8.5, 10} {
   \draw (\x, 0.8) -- ++(0,0.4);
   \node at ($(\x - 0.75, 1.3)$) {$p$};
}
\end{scope}
\draw[<->] (5.5, 0.7) -- (5,0.7) node[below,midway] {$\kpar_1 \bmod p$};
\draw[<->] (4, 1.6) -- (5,1.6) node[above,midway] {$p - \kpar_1 \bmod p$};
\node at (0.5,1) {\normalsize $\max{} - t > s + p_i$};
\end{scope}
\begin{scope}[yshift=-40mm]
\fill[gray!15] (0,-0.25) rectangle +(5,0.25);
\begin{scope}[fill=gray!70]
\fill (7,1) rectangle ++(1.5, 0.15);
\fill (5.5,0.85) rectangle ++(1.5, 0.15);
\fill (8.5,0.85) rectangle ++(1.5, 0.15);
\end{scope}
\begin{scope}[thin,draw=gray]\scriptsize
\draw[ultra thick] (-0.25, 0) -- ++(10.75,0);
\draw (2,-0.4) -- ++(0,2); \node at (2,-0.7) {$t$};
\draw (0,-0.4) -- ++(0,1); \node at (0,-0.7) {$\min$};
\draw (5,-0.4) -- ++(0,1.9); \node at (5,-0.7) {$\max$};
\draw (10,-0.4) -- ++(0,2); \node at (10,-0.7) {$\max + \kpar_1$};
\draw[gray] (5.5, 1) -- (5.5, 1.5);
\end{scope}
\node[wpoint,minimum size=5mm] at (2,0) {$S$};
\node[wpoint,minimum size=5mm] at (10,0) {$S_1$};
\draw[semithick] (2,1) -- (10,1);
\begin{scope}[ultra thick]
\draw (2, 0.8) -- +(0,0.4);
\node at (3.75,1.3) {$s + p_i$};
\draw (5.5, 0.8) -- +(0,0.4);
\foreach \x in {7, 8.5, 10} {
   \draw (\x, 0.8) -- ++(0,0.4);
   \node at ($(\x - 0.75, 1.3)$) {$p$};
}
\end{scope}
\draw[<->] (5, 0.7) -- (2,0.7) node[below,midway] {$q$};
\draw[<->] (5, 1.4) -- (5.5,1.4) node[above,midway] {$s + p_i - q$};
\node at (0.5,1) {\normalsize $\max{} - t \leq s + p_i$};
\end{scope}
\end{tikzpicture}
\caption{Two cases for $\bar \Theta^{\kpar_1}_{S \leadsto S_1}(t)$.}\label{fig:two-cases-theta}
\end{figure}%
The case $\kpar_1 < 0$ is similar and left to the reader.

\smallskip

Formulas $\Theta^k_{S \leadsto S_1}(t_1)$ are constructed similarly to $\bar \Theta^{\kpar_1}_{S \leadsto S_1}(t)$.

\smallskip

Formulas $\Theta^{k, k_1}_{S \leadsto S_1}$ have no free variables, but follow the same principle. We only consider the most interesting case of $\kpar < 0$ and $\kpar_1 > 0$ (leaving the three remaining cases to the reader): in this formula, we check that the difference between $\max{} + \kpar_1$ and $\min{} + \kpar$ either is  one of the $s_i$ (again, the $s'_i$ are irrelevant because the difference is positive) or belongs to one of the the arithmetic progressions $s + p_i + p\N$ (the $p'_i$ are irrelevant), where we again have two cases, with $\max{} - \min{}$ larger/smaller than $s + p_i$. We denote $\bar{\kpar} = \kpar_1 - \kpar$ and  set
\begin{align*}
\Theta^{k, k_1}_{S \leadsto S_1}  \ =    & \bigvee_{1 \leq i \leq l} (\max{} - \min{} =  s_i - \bar{\kpar}))\  \ \lor {}\\*
& \bigvee_{1 \leq i \leq m}
\Bigl[ (\max{} -\min{} \in s + p_i + (p - \bar{\kpar} \bmod p) + p \N) \quad \lor \hspace*{-1em}\bigvee_{\begin{subarray}{c}0 \leq q \leq s + p_i \text{ with}\\\bar{\kpar} \in s+p_i- q + p\N\end{subarray}} \hspace*{-2em}(\max{} - \min{} = q) \Bigr].
\end{align*}

\medskip

We note that if $\R$ is a $\DL_\horn\Xbox$ RBox, then, by Lemma~\ref{period:A} $(i)$, we can take $p = 1$, and so  $\Theta_{S \leadsto S_1}(t, t_1)$,  $\Theta^k_{S \leadsto S_1}(t_1)$, $\bar \Theta^{\kpar_1}_{S \leadsto S_1}(t)$ and $\Theta^{k, k_1}_{S \leadsto S_1}$ are equivalent to $\FO(<)$-formulas.

The $\Xi$-formulas are constructed similarly to the $\Theta$-formulas  using the ontology~$\smash{\T_\R^\dagger}$ (see Lemma~\ref{thm:technical}) and the beam in $\C_{\smash{\T_\R^\dagger}, \{ B^\dagger(0) \}}$ (see Section~\ref{sec:HornRI:sat}), instead of $\R^\ddagger$ and the $\R$-canonical rod for $\{S^\ddagger(0)\}$. Observe that, for role-monotone
$\DL_\horn\Xbox$-ontologies~$\TO$, those formulas are in $\FO(<)$ because
$\smash{\T_\R^\dagger}$ is in $\LTL_\horn\Xbox$.

To sum up, we have shown the following:
\begin{lemma}
Let $\TO = (\T, \R)$ be a $\bot$-free $\DL_\horn\Xallop$ ontology, $k,k_1 \in \Z\setminus\{0\}$, $B$, $B_1$
basic concepts, and $S$, $S_1$ roles. Then there are $\FOE$-formulas
$\Theta_{S \leadsto S_1}(t, t_1)$, $\Theta^\kpar_{S \leadsto S_1}(t_1)$, $\bar
\Theta^{\kpar_1}_{S \leadsto S_1}(t)$, $\Theta^{\kpar, \kpar_1}_{S \leadsto S_1}$, as well as
$\Xi_{B \leadsto B_1}(t, t_1)$, $\Xi^\kpar_{B \leadsto B_1}(t_1)$, $\bar \Xi^{\kpar_1}_{B
\leadsto B_1}(t)$, $\Xi_{B \leadsto B_1}^{\kpar, \kpar_1}$ satisfying the
characterisations in Table~\ref{aux-formulas}.

Moreover, if $\TO$ is a
$\smash{\DL_{\horn/\core}\Xbox}$ or $\smash{\DL_{\horn/\rhorn}\Xbox}$ ontology, then those formulas are
in $\FO(<)$.
\end{lemma}

\subsection{$\Psi$-formulas}\label{app:psi}
The formulas $\Psi^{\mu_1,\dots,\mu_l}_{S_1\dots S_l,\varkappa}(x, t_1,\dots,t_l, t)$ and
$\Psi^{\mu_1,\dots,\mu_l, k}_{S_1\dots S_l,\varkappa}(x, t_1,\dots,t_l)$ have, however,
more involved definitions, which take into account the periodicity of the canonical model $\can$.
The next
lemma characterises the temporal periodicity of $\can$ on both ABox and
non-ABox elements.
\dlperiod*

\begin{proof}
Since the number of role types is bounded by~$2^{|\R|}$, the  size of $\textbf{(con)}$ is bounded by~$2^{|\R|} (3+2(|\R|+2^{2|\R|})(1+|\R|)) \leq  2^{4|\R|}$ (assuming that $\R$ is non-empty). It follows that~$|\T_\R| = |\T| + |\textbf{(con)}| \leq |\T| + 2^{4|\R|} \leq 2^{4|\TO|}$. Therefore, we have $s_{\T_\R} < 2^{|\T_\R|} \leq 2^{2^{4|\TO|}}$ and $p_{\T_\R} < 2^{2|\T_\R|\cdot 2^{|\T_\R|}} \leq  2^{2^{(2^{4|\TO|} + 4|\TO| + 1)}} \leq 2^{2^{2^{5|\TO|}}}$. Set $s_\TO=s_{\T_\R} + p_{\T_\R}$ and $p_{\TO}= p_{\T_\R}$. Note that $s_\TO \geq p_\TO$. Also, as $\textbf{(con)}$ mimics the behaviour of roles using concept names, we can assume in applications of Lemma~\ref{period:A}~(\emph{ii}) to $\R$ below that $s_{\TO} \geq s_\R$ and $p_{\TO}$ is divisible by $p_\R$.

The proof of the lemma is by  induction on the construction of $\varkappa$. In fact,
we also prove (by simultaneous induction) the following auxiliary claim:\\[6pt]
\textsc{Claim 1.}
For any $wT_1^{n_1}\dots T_r^{n_r}\in  \Delta^{\can}$, with $r > 0$,
we have
\begin{align*}
& \can, \kpar \models \varkappa(w T_1^{n_1} \dots
    T_r^{n_r}) \quad \Longleftrightarrow \quad \can, \kpar + p_{\TO}
    \models \varkappa(w T_1^{n_1+p_\TO} \dots T_r^{n_r+p_\TO}),\\
    &\hspace*{10em}
    \text{ for any }\kpar \geq M^w_\A + s_{\TO} +
    |\varkappa| p_{\TO} \text{ provided that } n_1 \geq M^w_\A +s_\TO,\\
& \can, \kpar \models \varkappa(w T_1^{n_1} \dots T_r^{n_r}) \quad
    \Longleftrightarrow \quad \can, \kpar - p_{\TO} \models
    \varkappa(w T_1^{n_1-p_\TO} \dots T_r^{n_r-p_\TO}), \\
    &\hspace*{10em}
    \text{ for any } \kpar \leq \bar{M}^w_\A - (s_{\TO} + |\varkappa| p_{\TO}) \text{ provided that  } n_1 \leq \bar{M}^w_\A -s_\TO,
\end{align*}
We prove only the first equivalence in each pair; the other is shown by a symmetric argument.
First, observe that, by Lemma~\ref{period:A}~$(ii)$,  if $n_1 \geq M^w_\A +s_\TO$, then $w T_1^{n_1+p_\TO} \dots T_r^{n_r+p_\TO}\in\Delta^{\can}$, and similarly, if $n_1 \leq \bar{M}^w_\A - s_\TO$, then $w T_1^{n_1-p_\TO} \dots T_r^{n_r-p_\TO}\in\Delta^{\can}$. We now proceed by induction on the structure of $\varkappa$.

\smallskip

\noindent\textit{Cases $\varkappa = A$ and $\varkappa = \exists S.\top$}  for both the claim of the lemma and Claim~1 follow immediately from Lemmas~\ref{period:A}~$(ii)$ and~\ref{lem:witness-interaction}~(\emph{i}).

\smallskip

\noindent\textit{Case $\varkappa = \exists S. \varkappa'$} with $|\varkappa'|\geq 1$. We first prove the claim of the lemma. Suppose first that $\can, \kpar \models \varkappa'(w')$ and $\can, \kpar \models
    S(w, w')$, for some $w'\in\Delta^{\can}$.  There are two possible  locations for $w'$.
\begin{itemize}
\item[--] Let $w' = a
    S_1^{m_1} \dots S_{l-1}^{m_{l-1}}$ (if $l > 0$).   Since
    $M^{w'}_\A \leq M^{w}_\A$, and so $\kpar \ge M^{w'}_\A + s_\TO +
    |\varkappa'|p_\TO$, by the induction hypothesis, we obtain  $\can, \kpar + p_{\TO} \models
    \varkappa'(w')$.
    On the other hand, by Lemma~\ref{lem:witness-interaction}~(\emph{ii}), we have $S^-(w',w,\kpar) \in
    \C_{\R,\{S_l(w',w,m_l)\}}$, and so $S^- \in \rod_l(\kpar)$, where $\rod_l$ is the $\R$-canonical rod for $\{S_l(m_l)\}$. Since $\kpar \geq m_l  + s_\TO$,  by Lemma~\ref{period:A} $(ii)$, we obtain $S^- \in \rod_l(\kpar +
    p_{\TO})$ and $S^-(w',w,\kpar +
    p_{\TO}) \in \C_{\R,\{S_l(w',w,m_l)\}}$, whence, by Lemma~\ref{lem:witness-interaction}~(\emph{ii}), $\can,\kpar + p_{\TO}\models S(w,w')$.

\item[--] Let $w' = w S_{l+1}^{m_{l+1}}$. We have two further cases to consider.

If $m_{l+1} <  M^w_\A + s_{\TO}$,  then $M^{w'}_\A < M^w_\A+s_\TO$. As $s_\TO \leq p_\TO$,  we have $\kpar \geq M^{w'} + s_{\TO} + |\varkappa'| p_{\TO}$. By the induction hypothesis, we obtain $\can, \kpar + p_\TO \models \varkappa'(w')$; see Fig.~\ref{fig:appendix-witness-period}a.
On the other hand, since  $\kpar \geq m_{l+1} + s_\TO$, we can apply the argument with the $\R$-canonical rods as above and obtain $\can, \kpar + p_{\TO} \models S(w, w')$.

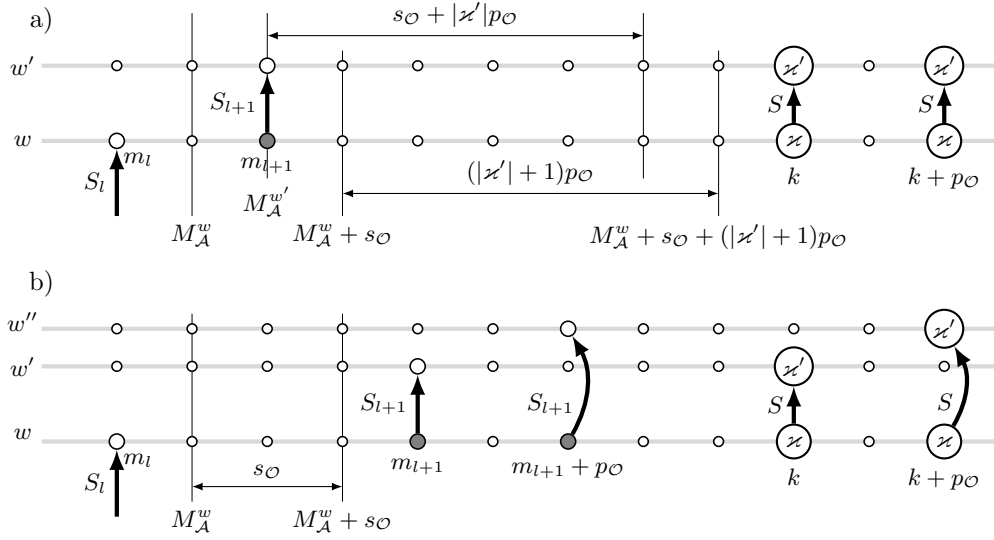
\begin{figure}[t]
\centerline{%
\begin{tikzpicture}[xscale=1, yscale=1, semithick]\footnotesize
\begin{scope}
\foreach \y in {0,1} {
    \draw[object-timeline] (-3,\y) -- ++(12.75,0);
}
\node (d) at (-3,1.6) {\small a)};
\node (a) at (-3.25,0) {$w$};
\node (d) at (-3.25,1) {$w'$};
\draw[time-guideline] (-1,-1) node[below]{$M^w_\A$}  -- ++(0,2.7);
\draw[time-guideline] (0,-0.5) node[below]{$M^{w'}_\A$}  -- ++(0,2.2);
\draw[time-guideline] (1,-1) node[below]{$M^w_\A + s_{\TO}$}  -- ++(0,2.2);
\draw[time-guideline] (5,-0.5)   -- ++(0,2.2);
\draw[time-guideline] (6,-1) node[below]{$M^w_\A+s_{\TO}+(|\varkappa'|+1)p_{\TO}$}  -- ++(0,2.2);
\node (a01) at (-2,0) [qpoint, label={[xshift=8pt, yshift=-16pt] $m_l$}]{};
\node at (-1,0) [point]{};
\node (a02) at (0,0) [ppoint, label=below:{$m_{l+1}$}]{};
\node (a12) at (0,1) [qpoint]{};
\node (a02') at (1,0) [point]{};
\node (a03) at (7,0) [wpoint, label=below:{$\kpar$}]{$\varkappa$};
\node (a04) at (9,0) [wpoint, label=below:{$\kpar+p_{\TO}$}]{$\varkappa$};
\foreach \x in {2,3,4,5,6,8} {
    \node at (\x,0) [point]{};
}
\foreach \x/\i in {-2/0,-1/1,1/3,2/4,3/5,4/6,5/9,6/10,8/11} {
     \node (a1\i) at (\x,1) [point]{};
 }
\node (a17) at (7,1) [wpoint]{$\varkappa'$};
\node (a18) at (9,1) [wpoint]{$\varkappa'$};
\begin{scope}[ultra thick]
\draw[->] (a03)  to node [left,midway]{$S$} (a17);
\draw[->] (a04)  to node [left,midway]{$S$} (a18);
\draw[->] ($(a01)+(0,-1)$)  to node [left,midway]{$S_l$} (a01);
\draw[->] (a02)  to node [left,midway]{$S_{l+1}$} (a12);
\end{scope}
\draw[<->,thin]	($(a13)+(0,-1.7)$) to node [above] {$(|\varkappa'|+1)p_{\TO}$} ($(a110)+(0,-1.7)$);
\draw[<->,thin]	($(a12)+(0,0.4)$) to node [above] {$s_{\TO}+|\varkappa'|p_{\TO}$} ($(a19)+(0,0.4)$);
\end{scope}%
\begin{scope}[yshift=-40mm]
\foreach \y in {0,1,1.5} {
    \draw[object-timeline] (-3,\y) -- ++(12.75,0);
}
\node (d) at (-3,2.1) {\small b)};
\node (a) at (-3.25,0.1) {$w$};
\node (d) at (-3.25,1.0) {$w'$};
\node (d) at (-3.25,1.6) {$w''$};
\draw[time-guideline] (-1,-0.8) node[below]{$M^w_\A$}  -- ++(0,2.5);
\draw[time-guideline] (1,-0.8) node[below]{$M^w_\A + s_{\TO}$}  -- ++(0,2.5);
\node (a01) at (-2,0) [qpoint, label={[xshift=8pt, yshift=-16pt] $m_l$}]{};
\node at (-1,0) [point]{};
\node (a02) at (2,0) [ppoint, label=below:{$m_{l+1}$}]{};
\node (a12) at (2,1) [qpoint]{};
\node (a02') at (4,0) [ppoint, label=below:{$m_{l+1}+p_{\TO}$}]{};
\node (a12') at (4,1.5) [qpoint]{};
\node (a03) at (7,0) [wpoint, label=below:{$\kpar$}]{$\varkappa$};
\node (a04) at (9,0) [wpoint, label=below:{$\kpar+p_{\TO}$}]{$\varkappa$};
\foreach \y in {0,1,3,5,6,8} {
    \node at (\y,0) [point]{};
}
    \foreach \x/\i in {-2/0,-1/1,0/4,1/3,3/5,4/6,5/9,6/10,7/12,8/11,9/13} {
        \node (a1\i) at (\x,1) [point]{};
    }
    \foreach \x/\i in {-2/0,-1/1,0/4,1/3,2/2,3/5,5/9,6/10,7/12,8/11,9/13} {
        \node (a2\i) at (\x,1.5) [point]{};
    }
\node (a17) at (7,1) [wpoint]{$\varkappa'$};
\node (a18) at (9,1.5) [wpoint]{$\varkappa'$};
\begin{scope}[ultra thick]
\draw[->] (a03)  to node [left,pos=0.5]{$S$} (a17);
\draw[->,bend right, pos=0.7] (a04)  to node [left,pos=0.3]{$S$} (a18);
\draw[->] ($(a01)+(0,-1)$)  to node [left,pos=0.5]{$S_l$} (a01);
\draw[->] (a02)  to node [left,pos=0.5]{$S_{l+1}$} (a12);
\draw[->,bend right, pos=0.7] (a02')  to node [left,pos=0.3]{$S_{l+1}$} (a12');
\end{scope}
\draw[<->,thin]	($(a11)+(0,-1.6)$) to node [above] {$s_{\TO}$} ($(a13)+(0,-1.6)$);
\end{scope}
\end{tikzpicture}%
}%
\caption{Proof of Lemma~\ref{th:dl-period}.}
\label{fig:appendix-witness-period}
\end{figure}

    If $m_{l+1} \geq M^w_\A + s_{\TO}$, then, by applying Claim~1 (as the induction hypothesis) to \mbox{$w'' = w S_{l+1}^{m_{l+1}+p_\TO}$}, we obtain $\can, \kpar+p_\TO \models
    \varkappa'(w'')$; see Fig.~\ref{fig:appendix-witness-period}b. On the other hand, $S(w,w',\kpar)
    \in \C_{\R,\{S_{l+1}(w,w',m_{l+1})\}}$, whence, by shifting the time instant for $S_{l+1}$, we obtain
    $S(w,w'',\kpar+p_\TO) \in
    \C_{\R,\{S_{l+1}(w,w'',m_{l+1}+p_\TO)\}}$, and so $\can,
    \kpar+p_\TO \models S(w, w'')$.
\end{itemize}
The converse, $\can, \kpar +
    p_{\TO} \models \varkappa'(w')$ and $\can, \kpar + p_{\TO} \models
    S(w, w')$ imply $\can, \kpar \models \varkappa'(w')$ and
    $\can, \kpar \models S(w, w')$, is shown similarly.
   
We now establish Claim~1 by distinguishing the following two cases.
\begin{itemize}
\item[--] Suppose $\can, \kpar \models \varkappa'(w T_1^{n_1} \dots T_{r-1}^{n_{r-1}})$ and $\can, \kpar \models S(w T_1^{n_1} \dots T_{r}^{n_r}, w T_1^{n_1} \dots T_{r-1}^{n_{r-1}})$. Observe that,  by shifting the time instant for $T_1$ (and therefore for all $T_2,\dots,T_r$), we obtain
    $\can, \kpar+p_\TO \models S(w T_1^{n_1+p_\TO} \dots
    T_{r}^{n_{r}+p_\TO}, w T_1^{n_1+p_\TO} \dots
    T_{r-1}^{n_{r-1}+p_\TO})$.
If $r=1$, then, by the induction hypothesis of the lemma,
    $\can,\kpar +p_\TO \models \varkappa'(w)$.
Otherwise, by Claim~1 as the induction hypothesis, we have $\can, \kpar+p_\TO \models \varkappa'(w
    T_1^{n_1+p_\TO}
    \dots T_{r-1}^{n_{r-1}+p_\TO})$.
    In either case, $\can, \kpar+p_\TO \models \exists
    S.\varkappa'(w T_1^{n_1+p_\TO} \dots T_r^{n_r+p_\TO})$.  The converse is similar.

\item[--] Suppose $\can, \kpar \models \varkappa'(w T_1^{n_1} \dots
    T_{r+1}^{n_{r+1}})$ and $\can, \kpar \models S(w T_1^{n_1} \dots
    T_{r}^{n_r}, w T_1^{n_1} \dots T_{r+1}^{n_{r+1}})$.
    Again, by shifting the time instant for $T_1$ (and therefore for all $T_2,\dots,T_{r+1}$), we obtain $\can, \kpar+p_\TO \models S(w T_1^{n_1+p_\TO} \dots
    T_{r}^{n_{r}+p_\TO}, w T_1^{n_1+p_\TO} \dots
    T_{r+1}^{n_{r+1}+p_\TO})$. On the other hand,
    by the induction hypothesis (Claim~1), $\can,\kpar+p_\TO \models \varkappa'(w
    T_1^{n_1+p_\TO} \dots T_{r+1}^{n_{r+1}+p_\TO})$. Thus, $\can, \kpar+p_\TO \models \exists S.\varkappa'(w
    T_1^{n_1+p_\TO} \dots T_r^{n_r+p_\TO})$.
    The converse implication is similar.
\end{itemize}

\noindent\textit{Cases $\varkappa = \varkappa_1 \sqcap \varkappa_2$ and $\varkappa = \varkappa_1\sqcup\varkappa_2$} are easy and left to the reader. 

\smallskip

\noindent\textit{Cases of temporal operators} follow from the proof of Lemma 22 by~\citeA{AIJ21}: we use~$|\varkappa| p_\TO$ instead of $\boldsymbol{l}_\varkappa p_\TO$, where $\boldsymbol{l}_\varkappa$ is the number of temporal operators in~$\varkappa$.
\end{proof}

\begin{lemma}\label{th:dl-period-cor}
In the context of Lemma~\ref{th:dl-period}, for any $w\in\Delta^{\can}$, any role $T$ from $\TO$, and any $n \in \tem(\A)$, we have
\begin{align*}
& \can, n \models \varkappa(w T^k) \quad
    \Longleftrightarrow \quad \can, n \models
    \varkappa(w T^{k+p_\TO}), \quad &&\text{for every $k \geq M^w_\A + s_\TO + |\varkappa|p_\TO$}, \\*
& \can, n \models \varkappa(w T^k) \quad
    \Longleftrightarrow \quad \can, n \models
    \varkappa(w T^{k-p_\TO}), \quad &&\text{for every $k \le \bar{M}^w_\A - (s_\TO + |\varkappa|p_\TO)$}.
\end{align*}
\end{lemma}
To prove this lemma, it is more convenient to show a more general result for an arbitrarily long role chain and disregard the active domain boundaries for $n$.
\\[12pt]
\textsc{Claim 2.}
For roles $T_1,\dots,T_r$ from $\TO$, $k_1,\dots,k_r \in \Z$, $r \geq 0$, the following holds:
\begin{align*}
\can, n
    \models \varkappa(w T_1^{k_1} \dots
    T_r^{k_r}) &\quad \Longleftrightarrow \quad \can, n
    \models \varkappa(w T_1^{k_1+p_\TO} T_{\smash{2}}^{k_2^+} \dots T_r^{k_r^+}), \\
& \text{for every $k_1 \geq M^w_\A + s_\TO$ and every $n \leq k_1 - (r \cdot s_\TO + |\varkappa|p_\TO)$,} \\
\can, n \models \varkappa(w T_1^{k_1} \dots T_r^{k_r}) &\quad
    \Longleftrightarrow \qquad \can, n \models
    \varkappa(w T_1^{k_1-p_\TO} T_{\smash{2}}^{k_2^-} \dots T_r^{k_r^-}), \\
& \text{for every $k_1 \le \bar{M}^w_\A - s_\TO$ and every $n \geq k_1 + r \cdot s_\TO + |\varkappa|p_\TO$,}
\end{align*}
where
\begin{align*}
& k_i^{+} = \begin{cases} k_i, & \text{if there is } 1 < j\leq i \text{ with } k_{j-1} - k_j \geq  s_\TO,\\ k_i + p_\TO, & \text{otherwise,} \end{cases} \\
& k_i^{-} = \begin{cases} k_i,  & \text{if there is } 1 < j\leq i \text{ with } k_{j} -k_{j-1} \geq  s_\TO,\\ k_i-  p_\TO, & \text{otherwise.} \end{cases}
\end{align*}
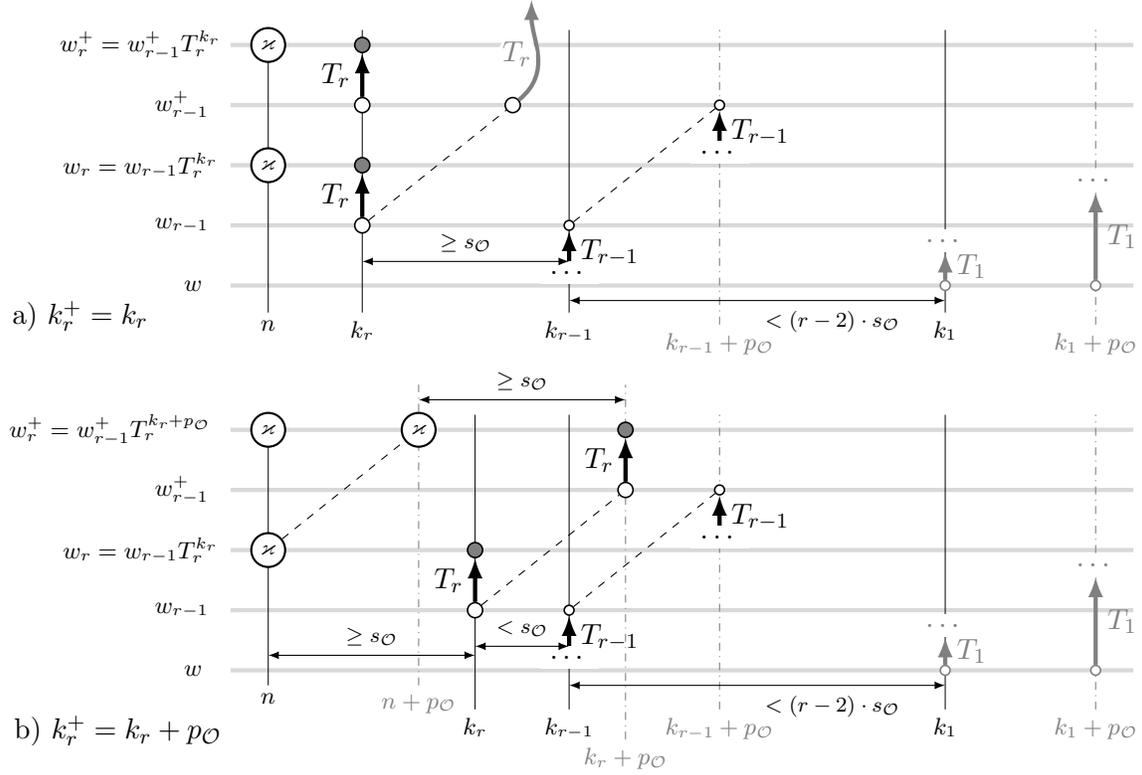
\begin{figure}[t]
\centerline{%
\begin{tikzpicture}[xscale=1, yscale=0.8]
\node at (-3.5,-0.5) {a) $k_r^+ = k_r$};
\begin{scope}\footnotesize
\draw[object-timeline] (-1.5,0) -- ++(12,0); \node (a) at (-2,0) {$w$};
\draw[object-timeline] (-1.5,1) -- ++(12,0); \node (a) at (-2.15,1) {$w_{r-1}$};
\draw[object-timeline] (-1.5,2) -- ++(12,0); \node (a) at (-2.7,2) {$w_r = w_{r-1}T_r^{k_r}$};
\draw[object-timeline] (-1.5,3) -- ++(12,0); \node (a) at (-2.15,3) {$w_{r-1}^+$};
\draw[object-timeline] (-1.5,4) -- ++(12,0); \node (a) at (-2.7,4) {$w_r^+ = w_{r-1}^+T_r^{k_r}$};
\end{scope}
\begin{scope}\footnotesize
\draw[time-guideline] (8,-0.45)  node[below]{$k_1$} -- ++(0,4.7);
\draw[time-guideline2] (10,-0.65)  node[below]{$k_1+p_\TO$} -- ++(0,4.9);
\draw[time-guideline] (3,-0.45)   node[below]{$k_{r-1}$} -- ++(0,4.7);
\draw[time-guideline2] (5,-0.65) node[below]{$k_{r-1}+p_\TO$} -- ++(0,4.9);
\draw[time-guideline] (0.25,-0.45)   node[below]{$k_r$} -- ++(0,4.7);
\draw[time-guideline] (-1,-0.45)  node[below]{$n$} -- ++(0,4.7);
\end{scope}
\node (w-kr) at (8,0) [point, draw=gray]{};
\node (w-kr') at (10,0) [point, draw=gray]{};
\draw[->,ultra thick,gray] (w-kr)  to node [right,midway]{$T_1$} +(0,0.55) node[label={[yshift=-5pt,fill=white] $\dots$}] {};
\draw[->,ultra thick,gray] (w-kr')  to node [right,midway]{$T_1$} +(0,1.55) node[label={[yshift=-5pt,fill=white] $\dots$}] {};
\node (w-kr-1) at (3,1) [point]{};
\node (w-kr-1') at (5,3) [point]{};
\draw[<-,ultra thick]  (w-kr-1) to node [right,midway]{$T_{r-1}$} ++(0,-0.8) node[label={[yshift=-9pt,fill=white] $\dots$}] {};
\draw[<-,ultra thick]   (w-kr-1') to node [right,midway]{$T_{r-1}$} ++(0,-0.8) node[label={[yshift=-9pt,fill=white] $\dots$}] {};
\node (w-kr) at (0.25,2) [ppoint]{};
\node (w-kr-s) at (0.25,1) [qpoint]{};
\draw[->,ultra thick]  (w-kr-s)  to node [left,midway]{$T_r$} (w-kr) ;
\node (w-kr') at (0.25,4) [ppoint]{};
\node (w-kr'-s) at (0.25,3) [qpoint]{};
\draw[->,ultra thick]  (w-kr'-s)  to node [left,midway]{$T_r$} (w-kr');
\node (w-kr'-s-1) at (2.25,3) [qpoint]{};
\draw[->,ultra thick,gray,out=45,in=-90]  (w-kr'-s-1)  to node [left,midway,fill=white,inner sep=0pt,circle]{$T_r$} ++(0.25,1.75);
\node (w-n) at (-1,2) [wpoint]{\scriptsize $\varkappa$};
\node (w-n'') at (-1,4) [wpoint]{\scriptsize $\varkappa$};
\draw[<->,thin]	(3,-0.25) to node [below,pos=0.7] {\scriptsize $< (r-2)\cdot s_{\TO}$} (8,-0.25);
\draw[<->,thin]	(0.25,0.4) to node [above] {\scriptsize $\geq s_{\TO}$} (3,0.4);
\draw[dashed] (w-kr-s) -- (w-kr'-s-1);
\draw[dashed] (w-kr-1) -- (w-kr-1');
\begin{scope}[yshift=-64mm]
\node at (-3,-1) {b) $k_r^+ = k_r + p_\TO$};
\begin{scope}\footnotesize
\draw[object-timeline] (-1.5,0) -- ++(12,0); \node (a) at (-2,0) {$w$};
\draw[object-timeline] (-1.5,1) -- ++(12,0); \node (a) at (-2.15,1) {$w_{r-1}$};
\draw[object-timeline] (-1.5,2) -- ++(12,0); \node (a) at (-2.7,2) {$w_r = w_{r-1}T_r^{k_r}$};
\draw[object-timeline] (-1.5,3) -- ++(12,0); \node (a) at (-2.15,3) {$w_{r-1}^+$};
\draw[object-timeline] (-1.5,4) -- ++(12,0); \node (a) at (-3.1,4) {$w_r^+ =w_{r-1}^+T_r^{k_r + p_\TO}$};
\end{scope}
\begin{scope}\footnotesize
\draw[time-guideline] (8,-0.65)  node[below]{$k_1$} -- ++(0,4.9);
\draw[time-guideline2] (10,-0.65)  node[below]{$k_1+p_\TO$} -- ++(0,4.9);
\draw[time-guideline] (3,-0.65)   node[below]{$k_{r-1}$} -- ++(0,4.9);
\draw[time-guideline2] (5,-0.65) node[below]{$k_{r-1}+p_\TO$} -- ++(0,4.9);
\draw[time-guideline] (1.75,-0.65)   node[below]{$k_r$} -- ++(0,4.9);
\draw[time-guideline2] (3.75,-1.15)   node[below]{$k_r + p_\TO$} -- ++(0,5.9);
\draw[time-guideline] (-1,-0.25)  node[below]{$n$} -- ++(0,4.5);
\draw[time-guideline2] (1,-0.25)  node[below]{$n + p_\TO$} -- ++(0,5);
\end{scope}
\node (w-kr) at (8,0) [point, draw=gray]{};
\node (w-kr') at (10,0) [point, draw=gray]{};
\draw[->,ultra thick,gray] (w-kr)  to node [right,midway]{$T_1$} +(0,0.55) node[label={[yshift=-5pt,fill=white] $\dots$}] {};
\draw[->,ultra thick,gray] (w-kr')  to node [right,midway]{$T_1$} +(0,1.55) node[label={[yshift=-5pt,fill=white] $\dots$}] {};
\node (w-kr-1) at (3,1) [point]{};
\node (w-kr-1') at (5,3) [point]{};
\draw[<-,ultra thick]  (w-kr-1) to node [right,midway]{$T_{r-1}$} ++(0,-0.8) node[label={[yshift=-9pt,fill=white] $\dots$}] {};
\draw[<-,ultra thick]   (w-kr-1') to node [right,midway]{$T_{r-1}$} ++(0,-0.8) node[label={[yshift=-9pt,fill=white] $\dots$}] {};
\node (w-kr) at (1.75,2) [ppoint]{};
\node (w-kr-s) at (1.75,1) [qpoint]{};
\draw[->,ultra thick]  (w-kr-s)  to node [left,midway]{$T_r$} (w-kr) ;
\node (w-kr') at (3.75,4) [ppoint]{};
\node (w-kr'-s) at (3.75,3) [qpoint]{};
\draw[->,ultra thick]  (w-kr'-s)  to node [left,midway]{$T_r$} (w-kr');
\node (w-n) at (-1,2) [wpoint]{\scriptsize $\varkappa$};
\node (w-n') at (1,4) [wpoint]{\scriptsize $\varkappa$};
\node (w-n'') at (-1,4) [wpoint]{\scriptsize $\varkappa$};
\draw[<->,thin]	(3,-0.25) to node [below,pos=0.7] {\scriptsize $< (r-2)\cdot s_{\TO}$} (8,-0.25);
\draw[<->,thin]	(1.75,0.4) to node [above] {\scriptsize $< s_{\TO}$} (3,0.4);
\draw[<->,thin]	(-1,0.25) to node [above] {\scriptsize $\geq s_{\TO}$} (1.75,0.25);
\draw[<->,thin]	(1,4.5) to node [above] {\scriptsize $\geq s_{\TO}$} (3.75,4.5);
\draw[dashed] (w-kr-s) -- (w-kr'-s);
\draw[dashed] (w-kr-1) -- (w-kr-1');
\draw[dashed] (w-n) -- (w-n');
\end{scope}
\end{tikzpicture}%
}%
\caption{Two cases for $k_i^+$ for the basis of induction in Lemma~\ref{th:dl-period-cor}: $k_{r-1}^+ = k_{r-1} + p_\TO$, $w_{r-1} = wT_1^{k_1}\dots T_{r-1}^{k_{r-1}}$,  and  $w_{r-1}^+ = wT_1^{k_1+p_\TO}\dots T_{r-1}^{k_{r-1}+p_\TO}$.}
\label{fig:appendix-dl-period-cor}
\end{figure}
\begin{proof}
We are going to show only the former statement; a proof for the latter will follow the same line of reasoning.

We abbreviate $wT_1^{k_1}T_2^{k_2} \dots T_r^{k_r}$ and  $wT_1^{k_1+p_\TO} T_{\smash{2}}^{k_2^+} \dots T_r^{k_r^+}$
by $w_r$ and $w^+_r$, respectively, and begin by observing that, if $k_1 \geq M^w_\A + s_\TO$, then
\begin{equation}\label{eq:L45:domain}
w_r \in \Delta^{\can} \quad\text{ iff }\quad w_r^+ \in \Delta^{\can}.
\end{equation}
The argument is by induction on $r$. For the basis case, $r = 1$, we have $w_1 = wT_1^{k_1}$ and $w_1^+ = wT_1^{k_1 + p_\TO}$, and~\eqref{eq:L45:domain} follows from Lemmas~\ref{period:A}~$(ii)$ and~\ref{lem:witness-interaction}~(\emph{i}). For the inductive step, suppose~\eqref{eq:L45:domain} holds for $r - 1$. If $k_{r-1}^+ = k_{r-1}$, then also $k_r^+ =  k_r$, and so~\eqref{eq:L45:domain} is immediate from the induction hypothesis. If $k_{r-1}^+ = k_{r-1} + p_\TO$, then we need to consider two further cases. If $k_r^+ = k_r + p_\TO$, then again~\eqref{eq:L45:domain} is immediate. If $k_r^+ = k_r$, then suppose first  $w_r\in\Delta^{\can}$. Then
$\can, k_r \models \exists T_r (w_{r-1})$. As $k_{r-1}^+ = k_{r-1} + p_\TO$, we observe that
$\can, k_r+p_\TO \models \exists T_r (w^+_{r-1})$, which, by Lemma~\ref{period:A}~(\emph{ii}), implies
$\can, k_r\models \exists T_r (w^+_{r-1})$ because  $k_r +p_\TO \leq (k_{r-1}+p_\TO) - s_\TO$. Thus, $w^+_r\in\Delta^{\can}$.  The converse direction is similar. This completes the argument for~\eqref{eq:L45:domain}.

Now, we proceed by induction on the construction of $\varkappa$ to show Claim~2.
\smallskip

\noindent\textit{Cases $\varkappa = A$ and $\varkappa = \exists S.\top$.}
If $k_r^+=k_r$, then the claim is immediate from Lemma~\ref{lem:witness-interaction}~(\emph{iii}) as the beams of both $w_r$ and $w^+_r$ are generated by $\exists T_r^-$ at $k_r$; see Fig.~\ref{fig:appendix-dl-period-cor}a.
Otherwise, $k_{j-1} - k_j< s_\TO$, for all $1< j\leq r$, and $w^+_r = \smash{wT_1^{k_1+p_\TO} T_2^{k_2+p_\TO} \dots T_r^{k_r+p_\TO}}$; see Fig.~\ref{fig:appendix-dl-period-cor}b. In this case, by Lemma~\ref{lem:witness-interaction}~(\emph{iii}), $\can, n \models \varkappa(w_r)$ iff \mbox{$\can, n +p_\TO \models \varkappa(w^+_r)$}, for all $n\in\Z$. By applying Lemma~\ref{period:A}~$(ii)$ to $w_r^+$,
we obtain
    $\can, n +p_\TO \models   \varkappa(w^+_r)$ iff  $\can, n \models \varkappa(w^+_r)$,  for every $n$ with  $n +p_\TO \leq  (k_r+p_\TO)- s_{\TO}$, that is, for every $n \leq k_r - s_\TO$, which, in view of our assumption on the $k_j$, means that $\can, n \models \varkappa(w_r)$ iff \mbox{$\can, n \models \varkappa(w^+_r)$} for every $n \leq k_1 - r \cdot s_\TO$ (for example, consider $n \leq k_2 - s_\TO$: as $k_1 - k_2 < s_\TO$, we have $k_2 > k_1 - s_\TO$, and so the equivalence holds for every $n \leq k_1 - 2s_\TO \leq k_2 - s_\TO$).

\medskip

\noindent\textit{Case $\varkappa = \exists S. \varkappa'$ with $|\varkappa| \geq 1$}. Suppose first $\can,n\models S(w_r, w')$  and $\can,n\models \varkappa'(w')$, for some $w'\in\Delta^{\can}$.  There are two possible  locations for $w'$.
\begin{itemize}
\item[--] If $w' = w_{r-1} = T_1^{k_1}T_2^{k_2} \dots T_{r-1}^{k_{r-1}}$, then we can repeat the argument from the basis case with Lemma~\ref{lem:witness-interaction}~(\emph{iii}) replaced by Lemma~\ref{lem:witness-interaction}~(\emph{ii}) and show that $\can, n \models S(w^+_r, w^+_{r-1})$, where $w^+_{r-1} = wT_1^{\smash{k_1+p_\TO}} T_2^{\smash{k_2^+}} \dots T_{r-1}^{\smash{k_{r-1}^+}}$, because $n \leq k_1 - r \cdot s_\TO$. On the other hand, by the induction hypothesis, $\can, n \models \varkappa'(w_{r-1}^+)$ whenever $n \leq k_1 - (r-1) \cdot s_\TO - |\varkappa'|p_\TO$, which, in particular, holds if $n\leq k_1 - r \cdot s_\TO - |\varkappa|p_\TO$ because $s_\TO, p_\TO\geq 0$.

\item[--]
Otherwise,  $w' = w_{r+1} = w_r T^{k_{r+1}}_{r+1}$. Let $w_{r+1}^+ = w_r^+ T^{k_{r+1}^+}_{r+1}$.
If $k_r^+ = k_r$, then also  $k^+_{r+1} = k_{r+1}$, and, by Lemma~\ref{lem:witness-interaction}~(\emph{ii}),
\mbox{$\can, n \models S(w^+_r, w^+_{r+1})$}.
Otherwise, $k_r^+ = k_r + p_\TO $, and we need to consider two further cases.
\begin{itemize}
\item[--] If $k_{r+1}^+ = k_{r+1}$, then, by Lemma~\ref{lem:witness-interaction}~(\emph{ii}), we have $\can, n \models S (w_{r}^+, w_{r+1}^+)$.

\item[--] If $k_{r+1}^+ = k_{r+1} + p_\TO$, then, by Lemma~\ref{lem:witness-interaction}~(\emph{ii}), we have \mbox{$\can, n +p_\TO \models S(w^+_r, w^+_{r+1})$}.
By  Lemma~\ref{period:A}~(\emph{ii}) applied to the rod, we obtain
\mbox{$\can, n \models S(w^+_r, w^+_{r+1})$} whenever  $n +p_\TO \leq  (k_{r+1}+p_\TO)- s_{\TO}$. Since
$k_{j-1} - k_j< s_\TO$, for all $1< j\leq r+1$, this holds, in particular, if $n \leq k_1 - (r+1) \cdot s_\TO$.
\end{itemize}
On the other hand, by the induction hypothesis, we have $\can, n \models \varkappa'(w_{r+1}^+)$ provided that $n \leq k_1 - (r+1) \cdot s_\TO - |\varkappa'|p_\TO$, which, in particular, holds if $n \leq k_1 - (r \cdot s_\TO + |\varkappa|p_\TO)$ because $s_\TO\leq p_\TO$.
\end{itemize}
Thus, in each case we  obtain $\can, n \models \varkappa(w_r^+)$ provided that $n \leq k_1 - (r \cdot s_\TO + |\varkappa|p_\TO)$.
The converse implication, if \mbox{$\can,n\models S(w_r^+, w')$} and \mbox{$\can,n\models \varkappa'(w')$}, for some $w'\in\Delta^{\can}$, is treated in a similar way.

\smallskip

\noindent\textit{Cases $\varkappa = \varkappa_1 \sqcap \varkappa_2$ and $\varkappa = \varkappa_1\sqcup\varkappa_2$} are easy and left to the reader. 

\smallskip

\noindent\textit{Cases of temporal operators} follow from the proof of Lemma 22 by~\citeA{AIJ21}.

\smallskip

This completes the proof of Claim~2.
\end{proof}

We can now define the remaining two formulas in Table~\ref{aux-formulas}, $\smash{\Psi^{\mu_1,\dots,\mu_l}_{S_1,\dots,S_l,\varkappa}(x, t_1,\dots,t_l, t)}$ and $\smash{\Psi^{\mu_1,\dots,\mu_l, k}_{S_1,\dots,S_l,\varkappa}(x, t_1,\dots,t_l)}$. The definition is by induction on the structure of $\varkappa$.
For convenience, we denote the sequence $S_1,\dots,S_l$ by $\avec{S}_l$, the sequence $t_1,\dots, t_l$  by $\avec{t}_l$, and the sequence $\mu_1,\dots,\mu_l$ by $\avec{\mu}_l$ and adopt similar notation for their prefixes: for example, $\avec{S}_0$ is the empty list of roles. We also assume that $\smash{\Psi^{\avec{\mu}_0}_{\avec{S}_0,\varkappa'}(x, t) =   \Phi_{\TO,\varkappa'}(x,t)}$ and $\smash{\Psi^{\avec{\mu_0},\kpar}_{\avec{S}_0,\varkappa'}(x, t) =   \Phi^\kpar_{\TO,\varkappa'}(x,t)}$, for any $\varkappa'$.

\medskip

\noindent\emph{Case} $\varkappa = A$.  If $\mu_l= 0$, then  $\Psi^{\avec{\mu}_l}_{\avec{S}_l,\varkappa}(x,
\avec{t}_l, t)=\Xi_{\exists S_l^- \leadsto A}(t_l, t)$ and $\Psi^{\avec{\mu}_l, k}_{\avec{S}_l,\varkappa}(x, \avec{t}_l)=\bar
\Xi^{k}_{\exists S_l^- \leadsto A}(t_l)$. Otherwise, we take
 $\Psi^{\avec{\mu}_l}_{\avec{S}_l,\varkappa}(x,
\avec{t}_l, t)=\Xi^{\mu_l}_{\exists S_l^- \leadsto A}(t)$ and $\Psi^{\avec{\mu}_l, k}_{\avec{S}_l,\varkappa}(x, \avec{t}_l)=\Xi_{\exists
S_l^- \leadsto A}^{\mu_l, k}$.

\medskip

\noindent\emph{Case} $\varkappa = \exists S. \varkappa'$. Let $N = (l+1) (s_\TO+|\varkappa|p_\TO)$; here we propagate the worst-case creation time boundaries from Lemma~\ref{th:dl-period-cor} for a witness of depth $l+1$.

\noindent If $\mu_l = 0$, then $\Psi^{\avec{\mu}_l}_{\avec{S}_l,\varkappa}(x, \avec{t}_l, t)$ is the following:
\begin{multline*}
\bigvee_{\text{role $S_{l + 1}$ in }\TO}%
\Bigl(\exists t_{l+1} \bigl[\Xi_{\exists S_l^- \leadsto \exists
    S_{l+1}}(t_l, t_{l +1})   \land \Theta_{S_{l+1} \leadsto
    S}(t_{l+1}, t) \land
\Psi^{\avec{\mu}_l,0}_{\avec{S}_l, S_{l+1},\varkappa'}(x,
\avec{t}_l, t_{l+1}, t)\bigr] \ \ \lor{}\\
\bigvee_{i \in[-N,0) \cup (0,N]}\hspace*{-1em} \bigl[\bar \Xi^{i}_{\exists S_l^-
    \leadsto \exists S_{l+1}}(t_l) \land \Theta^i_{S_{l+1} \leadsto
  S}(t) \land \Psi^{\avec{\mu}_l,i}_{\avec{S}_{l+1},\varkappa'}(x, \avec{t}_l, t_{l+1}, t)\bigr]\Bigr) \ \ \lor{} \\
\bigl[\Theta_{S_l \leadsto S^-}(t_l, t) \land \Psi^{\avec{\mu}_{l-1}}_{\avec{S}_{l-1},\varkappa'}(x, \avec{t}_{l-1}, t)\bigr],
\end{multline*}
and the $\Psi^{\avec{\mu}_l, \kpar}_{\avec{S}_l,\varkappa}(x, \avec{t}_l)$ are obtained from the above by removing occurrences of variable~$t$ and adding the $\kpar$ decoration instead:
\begin{multline*}
\bigvee_{\text{role $S_{l + 1}$ in }\TO}%
\Bigl(\exists t_{l+1} \bigl[\Xi_{\exists S_l^- \leadsto \exists
    S_{l+1}}(t_l, t_{l +1})   \land \bar \Theta_{S_{l+1} \leadsto
    S}^{\kpar}(t_{l+1}) \land
\Psi^{\avec{\mu}_l,0, \kpar}_{\avec{S}_l, S_{l+1},\varkappa'}(x,
\avec{t}_l, t_{l+1})\bigr] \ \ \lor{} \\*
\bigvee_{i\in [-N,0)\cup(0,N]} \bigl[\bar \Xi^{i}_{\exists S_l^- \leadsto
    \exists
  S_{l+1}}(t_l) \land \Theta^{i,k}_{S_{l+1} \leadsto S} \land \Psi^{\avec{\mu}_l,i, \kpar}_{\avec{S}_l, S_{l+1},\varkappa'}(x, \avec{t}_l, t_{l+1})\bigr]\Bigr) \ \ \lor{} \\*
\bigl[\bar\Theta_{S_l \leadsto S^-}^{\kpar}(t_l) \land \Psi^{\avec{\mu}_{l-1}, \kpar}_{\avec{S}_{l-1},\varkappa'}(x, \avec{t}_{l-1})\bigr].
\end{multline*}

\smallskip

\noindent If $\mu_l \neq 0$, then $\Psi^{\avec{\mu}_l}_{\avec{S}_l,\varkappa}(x,
\avec{t}_l, t)$ is the following:
\begin{multline*}
\bigvee_{\text{role $S_{l + 1}$ in }\TO}%
\Bigl(\exists t_{l+1} \bigl[\Xi^{\mu_l}_{\exists S_l^- \leadsto
    \exists S_{l+1}}(t_{l +1})   \land \Theta_{S_{l+1} \leadsto
    S}(t_{l+1}, t) \land
\Psi^{\avec{\mu}_l,0}_{\avec{S}_l, S_{l+1},\varkappa'}(x,
\avec{t}_l, t_{l+1}, t)\bigr] \ \ \lor{} \\
\bigvee_{i\in[-N,0)\cup(0,N]} \hspace*{-2em}\bigl[\Xi^{\mu_l, i}_{\exists S_l^- \leadsto \exists
    S_{l+1}} \land \Theta^i_{S_{l+1} \leadsto S}(t) \land \Psi^{\avec{\mu}_l, i}_{\avec{S}_l, S_{l+1},\varkappa'}(x, \avec{t}_l, t_{l+1},
    t)\bigr]\Bigr) \ \ \lor{} \\
\bigl[\Theta^{\mu_l}_{S_l \leadsto S^-}(t) \land \Psi^{\avec{\mu}_{l-1}}_{\avec{S}_{l-1},\varkappa'}(x,
    \avec{t}_{l-1}, t)\bigr],
\end{multline*}
and the $\Psi^{\avec{\mu}_l, k}_{\avec{S}_l,\varkappa}(x, \avec{t}_l)$ are again obtained by replacing $t$ with the $\kpar$ decorations:
\begin{multline*}
\bigvee_{\text{role $S_{l + 1}$ in }\TO}%
\Bigl(\exists t_{l+1} [\Xi^{\mu_l}_{\exists S_l^- \leadsto
    \exists S_{l+1}}(t_{l +1})   \land \bar \Theta_{S_{l+1} \leadsto
    S}^{\kpar}(t_{l+1}) \land
\Psi^{\avec{\mu}_l,0,\kpar}_{\avec{S}_l, S_{l+1},\varkappa'}(x,
\avec{t}_l, t_{l+1})] \ \ \lor{} \\
\bigvee_{i\in [-N,0)\cup(0,N]} [\Xi^{\mu_l, i}_{\exists S_l^- \leadsto \exists
  S_{l+1}} \land \Theta^{i,k}_{S_{l+1} \leadsto S} \land \Psi^{\avec{\mu}_l,i,\kpar}_{\avec{S}_l, S_{l+1},\varkappa'}(x, \avec{t}_l, t_{l+1})\bigr]\Bigr) \ \ \lor{} \\
\bigl[\Theta^{\mu_l,k}_{S_l \leadsto S^-} \land \Psi^{\avec{\mu}_{l-1},\kpar}_{\avec{S}_{l-1},\varkappa'}(x,
    \avec{t}_{l-1})\bigr].
\end{multline*}

\medskip

\noindent\emph{Case} $\varkappa = \Rbox \varkappa'$. By
Lemma~\ref{th:dl-period}, we take  again $N =  (l+1) (s_{\TO} + |\varkappa| p_{\TO})$
 and set
\begin{align*}
\Psi^{\avec{\mu}_l}_{\avec{S}_l,\varkappa}(x, \avec{t}_l, t) \ = \  & \forall s\, [(s > t) \to{} \Psi^{\avec{\mu}_l}_{\avec{S}_l,\varkappa'}(x,\avec{t}_l, s)] \ \ \land
  \bigwedge_{\kpar \in (0,N]} \Psi^{\avec{\mu}_l,\kpar}_{\avec{S}_l,\varkappa'}(x,\avec{t}_l),\\
\Psi^{\avec{\mu}_l, k}_{\avec{S}_l,\varkappa}(x, \avec{t}_l) \ = \ & \begin{cases}
\displaystyle\rule{0pt}{14pt} \bigwedge_{i \in (k, k+N]}\Psi^{\avec{\mu}_l,i}_{\avec{S}_l,\varkappa'}(x,\avec{t}_l), & \text{if } k > 0;\\
%  %
\displaystyle\rule{0pt}{20pt}  \forall s\, \Psi^{\avec{\mu}_l}_{\avec{S}_l,\varkappa'}(x,\avec{t}_l, s)  \land  \bigwedge_{i\in (k,0)\cup(0,N]} \Psi^{\avec{\mu}_l,i}_{\avec{S}_l,\varkappa'}(x,\avec{t}_l), &  \text{if } k < 0.
\end{cases}
\end{align*}

\medskip

\noindent The cases of other temporal operators are similar, and the cases $\varkappa =
\varkappa_1 \sqcap \varkappa_2$ and $\varkappa = \varkappa_1 \sqcup \varkappa_2$
are trivial. This completes the construction of the $\Psi$-formulas in Table~\ref{aux-formulas}. To sum up, we have shown the following:
\begin{lemma}
Let $(\TO,\varkappa)$ be a $\bot$-free $\DL_{\horn}\Xallop{}$ \OMPIQ{} with a positive temporal concept $\varkappa$, $k \in \Z\setminus\{0\}$, and $S_1, \dots, S_l$ roles. Then there are $\FOE$-formulas
$\Psi^{\mu_1,\dots,\mu_l}_{S_1\dots S_l,\varkappa}(x, t_1,\dots,t_l, t)$ and
$\Psi^{\mu_1,\dots,\mu_l, k}_{S_1\dots S_l,\varkappa}(x, t_1,\dots,t_l)$ satisfying the
characterisations in Table~\ref{aux-formulas}.

Moreover, if $\TO$ is a
$\smash{\DL_{\horn/\core}\Xbox}$ or $\smash{\DL_{\horn/\rhorn}\Xbox}$ ontology, then those formulas are
in $\FO(<)$.
\end{lemma}

%
%%%% Local Variables:
%%%% mode: latex
%%%% TeX-master: "aij-main"
%%%% End:

%%%%%%%%%%%%%%%%%%%%%%%%%%%%%%%%%%%%%%%%%%%%%%%%%%%%%%%%%%%%%%%%%%%%%%
%%%%%%%%%%%%%%%%%%%%%%%%%%%%%%%%%%%%%%%%%%%%%%%%%%%%%%%%%%%%%%%%%%%%%%

\bibliography{LTL,time17-bib}
\bibliographystyle{theapa}

%%%%%%%%%%%%%%%%%%%%%%%%%%%%%%%%%%%%%%%%%%%%%%%%%%%%%%%%%%%%%%%%%%%%%%
%%%%%%%%%%%%%%%%%%%%%%%%%%%%%%%%%%%%%%%%%%%%%%%%%%%%%%%%%%%%%%%%%%%%%%
\end{document}